\documentclass[1p]{elsarticle}

\usepackage{amsmath}
\usepackage{amssymb}
\usepackage{mathtools}
\usepackage{amsthm}
\usepackage{arydshln}
\usepackage{mathbbol}
\usepackage{mathrsfs}
\usepackage{stmaryrd}
\usepackage{graphicx}
\usepackage{tikz-cd}
\usepackage[colorlinks]{hyperref}
\usepackage{soul}

\def\C{\mathbb{C}}

\def\R{\mathbb{R}}
\def\T{\mathbb{T}}
\def\C{\mathbb{C}}
\def\N{\mathbb{N}}
\def\Z{\mathbb{Z}}

\def\bq{\begin{equation}}
\def\eq{\end{equation}}
\def\bqy{\begin{eqnarray}}
\def\eqy{\end{eqnarray}}

\def\bal#1\eal{\begin{align}#1\end{align}}

\def\ncr{\nonumber\\}

\newcommand{\vin}{\rotatebox[origin=c]{+90}{$\in$}}
\newcommand{\veq}{\rotatebox[origin=c]{+90}{$=$}}

\newcommand{\ol}[1]{\overline{#1}}

\let\div\undefined

\newcommand{\transpose}[1]{{}^{t}#1}

\newcommand{\suchthat}[0]{\colon\;}

\DeclareMathOperator{\rank}{rank}

\DeclareMathOperator{\div}{div}
\DeclareMathOperator{\curl}{curl}

\DeclareMathOperator{\dom}{dom}
\DeclareMathOperator{\rng}{ran}
\DeclareMathOperator{\spn}{span}

\newcommand{\fun}[1]{\mathcal{#1}}
\newcommand{\TT}[0]{\mathsf{T}}

\newcommand{\xx}[0]{\mathsf{x}}
\newcommand{\vv}[0]{\mathsf{v}}

\theoremstyle{plain}
\newtheorem{theorem}{Theorem}

\newtheorem{lemma}[theorem]{Lemma}
\newtheorem{proposition}[theorem]{Proposition}
\theoremstyle{definition}

\newtheorem{example}{Example}

\let\oldfinalMaketitle\finalMaketitle
\renewcommand{\finalMaketitle}{
  \oldfinalMaketitle
  \tableofcontents\vspace{\baselineskip}}

\graphicspath{{figures/}}

\begin{document}

\title{Metriplectic relaxation to equilibria}
  
\author[]{\texorpdfstring{C. Bressan\fnref{fn1}}{C. Bressan}} 
  
\author[label1]{M. Kraus}
\ead{michael.kraus@ipp.mpg.de}
  
\author[label1]{O. Maj}
\ead{omar.maj@ipp.mpg.de}
  
\author[label3]{P. J. Morrison}
\ead{morrison@physics.utexas.edu}
  
\fntext[fn1]{
  Formerly at the Max Planck Institute for Plasma Physics,
  Boltzmannstrasse 2, 85747, Garching, Germany. 
}

\address[label1]{
  Max Planck Institute for Plasma Physics,
  Boltzmannstrasse 2, 85747, Garching, Germany.
}

\address[label3]{
  Department of Physics and Institute for Fusion Studies,
  University of Texas at Austin, Austin TX 78712-1060, USA.
}

\journal{Commun. Nonlinear Sci. Numer. Simul.}


\begin{frontmatter}
  
  \begin{abstract}
    Metriplectic dynamical systems consist of a special combination of a
    Hamiltonian and a (generalized) entropy-gradient flow, such that the
    Hamiltonian is conserved and entropy is dissipated/produced (depending on a
    sign convention). It is natural to expect that, in the long-time limit, the
    orbit of a metriplectic system should converge to an extremum of entropy
    restricted to a constant-Hamiltonian surface.
    In this paper, we discuss sufficient conditions for this to occur.
    Then, we construct a class of metriplectic systems inspired by the
    Landau operator for Coulomb collisions in plasmas, which is included 
    as special case. For this class of brackets, checking the conditions for
    convergence reduces to checking two usually simpler conditions, and we
    discuss examples in detail.
    We apply these results to the construction of relaxation methods for the
    solution of equilibrium problems in fluid dynamics and plasma physics.  
  \end{abstract}


  
\end{frontmatter}

\section{Introduction}
\label{sec:introduction}

There are two main purposes of this paper:  investigate sufficient conditions
for metriplectic relaxation (reviewed in Section~\ref{sec:metriplectic}) to
occur and  use metriplectic relaxation  to find   equilibria of  fluid dynamics
and plasma physics systems.  Metriplectic dynamical systems, as introduced in
\cite{Morrison1984,pjm84b,Morrison1986},  are designed to formally converge to
an extremum of entropy while being restricted to a constant Hamiltonian surface.
Here we more rigorously examine the conditions for such relaxation and then
construct and investigate metriplectic systems that achieve the relaxation for
finding equilibria  of a collection of fluid and plasma physical systems. In the
remainder of this section, in Section~\ref{ssec:Oequil}  we first give an
overview of the  challenges and previous methods for calculating equilibria,
followed in Section~\ref{ssec:Mrelax} by our and others' previous metriplectic
relaxation methods. 
    
\subsection{Overview of equilibrium calculations}
\label{ssec:Oequil}

The calculation of equilibria of physical systems often leads to ill-posed
nonlinear problems, where the ill-posedness is due to the nonuniqueness of the
solution. Additional constraints are needed to define uniquely the equilibrium
of interest, depending on the application at hand. In some situations,
prescribing enough constraints to determine a unique equilibrium may not be
straightforward. This lack of uniqueness for equilibrium problems is precisely
discussed below in Section~\ref{sec:testcases} for examples taken from fluid
dynamics and magnetohydrodynamics (MHD): equilibria of the Euler equations in
vorticity form reduced to two dimensions \cite[p.488]{Morrison1998},
axisymmetric MHD equilibria \cite{Freidberg2014}, linear and nonlinear Beltrami
fields.

In some cases, after providing additional physical constraints, the equilibrium
conditions can be reformulated as a well-posed mathematical problem. This is
the case, for instance, for the Euler equations, axisymmetric MHD equilibria,
and linear Beltrami fields.

In more complicated situations, such as for nonlinear Beltrami fields and full
MHD equations in three dimensions, the problem of computing an equilibrium point
has no good solution.
The difficulties were shown in \cite{Berk1986} to be related to the
Kolmogorov-Arnold-Moser (KAM) theorem (see for example \cite{Llave2001}).
A mathematical perspective on these difficulties can be found in the introduction
of the paper by Bruno and Laurence \cite{Bruno1996}, cf. also the recent
developments by Enciso et al.\ \cite{Enciso2025spe}. A large fraction of MHD
equilibrium calculations in three dimensions are based on a reformulation of the
problem in which one assumes that the magnetic field is tangent to a family of
nested toroidal surfaces \cite{Bauer1978, Hirshman1983}. 
On the one hand, such a configuration is a natural generalization to
three-dimensions of the confined region in axisymmetric MHD equilibria. 
In addition, by basic considerations of topology, the confinement of a plasma
in a volume bounded by a closed orientable pressure isosurface, where the
pressure gradient is balanced by the electromagnetic force, requires the
surface to be a torus in the simplest case \cite{Kruskal1958a}. Therefore
searching for equilibria with nested toroidal flux surfaces appears to be the
simplest and most appealing approach to the three-dimensional equilibrium
problem. On the other hand, Grad conjectured that such equilibria may not
exist \cite{Grad1967} unless they are axisymmetric or we allow for weak
solutions characterized by singular current sheets localized on specific flux
surfaces, namely the resonant surfaces. The non-existence of smooth
non-axisymmetric equilibria with nested flux surfaces is referred to as the
\emph{Grad's conjecture}, and (to the best of our knowledge) it is still an
open question. Nevertheless, the variational formulation adopted for
equilibria with nested flux surfaces \cite{Bauer1978,Hirshman1983}, in
principle at least, allows for weak solutions
\cite{Garabedian1998,Garabedian2008}, although usually this possibility is not
exploited in state-of-the-art codes such as VMEC \cite{Hirshman1983},
DESC \cite{Dudt2020}, and GVEC \cite{Hindenlang2025}, which use a highly
regular representation of the magnetic and pressure fields. 
(For sake of completeness, we note that the GVEC code has the built-in
possibility of relaxing the regularity of the magnetic field allowing for
current sheets on prescribed surfaces, but this possibility has not been
exploited yet.) The singular current layers that are expected according to the
Grad's conjecture cannot be considered physical; these equilibria are regarded
as computationally efficient proxies for equilibria that may have a
complicated field-line topology in a neighborhood of some resonant surfaces,
but have nested toroidal surfaces elsewhere for good confinement.
This strategy has been extremely successful for the design of stellarators
\cite{Imbert-Gerard2024}, since it reduces significantly the complexity of the
problem. While this approach gives an acceptable representation of the magnetic
field in the core of a stellarator with modest computational cost, it cannot
account for more complicated magnetic field configurations, such as those with
magnetic islands and chaotic field lines, due to the built-in foliation of the
domain by toroidal surfaces. Yet magnetic-field islands and chaotic regions are
relevant in practice and calculations of MHD equilibria beyond the paradigm of
nested toroidal surfaces are needed. An iterative procedure for the calculation
of general MHD equilibria has been outlined by Grad \cite{Grad1958}, and a
similar iterative procedure is implemented in the PIES code \cite{Reiman1986}.
Such iteration schemes are purely heuristics: there is no theoretical control on
the convergence. 

Another approach to the computation of equilibria is based on artificial
relaxation. Relaxation methods solve the Cauchy problem for a fictitious
dissipative evolution law that contains a tailored dissipation mechanism. If the
dissipation mechanism is well designed, the solution of the Cauchy problem, with
a given (well-prepared) initial condition, exists globally in time, has a limit
for $t \to +\infty$, and the limit is an equilibrium of the considered physical
system. The dynamical evolution itself might not be physical, but the solution
should converge to a physical equilibrium as fast as possible. Some care may be
taken to preserve important properties of the solution. For the specific case of
MHD, for instance, one can evolve the magnetic field $B$ according to  
Faraday's equation, 
\begin{equation*}
  \partial_t B + c \curl E = 0, \quad B(0,x) = B_0(x)\,,
\end{equation*}
but with a properly chosen effective electric field $E$, which does not need to
have physical meaning. (Note, Gaussian units are used throughout this paper,
with $c$ being the speed of light in free space.) If the initial condition $B_0$
satisfies $\div B_0 = 0$, then $\div B = 0$ for all $t \geq 0$.
Probably the most intuitive relaxation method can be obtained by choosing
$E = - (U \times B)/c$, where the advecting velocity field $U$ solves the
viscous MHD momentum balance equation \cite{Moffatt1985}. The idea of this
method is physically intuitive: magnetic energy is converted into kinetic energy
and dissipated by viscosity. Since there is no resistivity, magnetic helicity is
preserved and this provides a lower bound for the dissipation of magnetic
energy \cite{Arnold1998,Moffatt2021}. The evolution of the system is not
physically consistent (the viscosity term is usually very simple and resistivity
is zero), but the relaxed state, if it is reached and it is smooth enough, is
guaranteed to be an ideal MHD equilibrium. However, while a relaxation method
in general seeks to find an ideal MHD equilibrium as the long-time limit of
the evolution of a given initial configuration, in most applications, we ask
for the answer to a different question: we seek an equilibrium that is
compatible with given data. For instance one might need to impose given
pressure and current profiles, i.e., the constant value of the pressure and
the concatenated plasma current on the surfaces tangent to the magnetic field
(flux surfaces). This raises the question of relaxing to an equilibrium
compatible with the given data starting from a suitable initial condition.
This problem is related to the concept of \emph{accessibility} of an
equilibrium since the relaxation mechanisms usually entail constraints: the
solution of a relaxation method, formally at least, evolves on the constrained
submanifold that contains the initial condition, but this submanifold may not
contain equilibria compatible with the given data. One therefore needs either
to prepare the initial condition appropriately (by making assumption on the
targeted equilibrium) or to adapt the solution during the evolution, using the
available data.  

For instance, with the choice of electric field $E = -(U\times B)/c$ mentioned
above, the Faraday's equation reduces to Lie-dragging of the magnetic field and
therefore, smooth solutions preserve the magnetic flux and the field-line
topology of the initial condition (frozen-in law \cite{Alfven1942}, cf. also
general MHD textbooks \cite{Freidberg2014,Davidson2001}). In this case, Moffatt
has introduced two different concepts \cite{Moffatt1985, Moffatt2021}: 
\begin{itemize}

\item \emph{Topological equivalence.} Two vector fields $B_0$ and $B_1$ are
  topologically equivalent if there is a diffeomorphism that maps one field
  into the other via push-forward. We may think of topologically equivalent
  fields as one being a smooth deformation of the other. 
    
\item \emph{Topological accessibility}. A vector field $B_1$ is topologically
  accessible from the vector field $B_0$ if $B_1(x) = B(t,x)$ for some
  $t \geq 0$, where $B$ is the solution of the Faraday's equation with
  $E = -(U \times B)/c$ and some velocity field $U$, i.e., the MHD induction
  equation. We note that in its original definition, Moffatt restricted $U$ to
  be solenoidal \cite{Moffatt1985}. 

\end{itemize}
If $B_1$ is topologically accessible from $B_0$ with $U$  sufficiently smooth
(e.g., of class $C^1$ with a $C^2$ flow, as functions of $(t,x)$), then $B_0$
and $B_1$ are also topologically equivalent. In general however, the solution
$B$ of the induction equation may develop a singularity in finite time.
More precisely, the magnetic field may develop tangential discontinuities at
certain surfaces that correspond to current sheets. In fact, according to an
argument put forward by Parker \cite{Parker1972}, the formation of current
sheets is a general occurrence for braided fields, i.e., ``most'' braided
initial conditions should develop current sheets. This is known as the
\emph{Parker's conjecture}. Recently, Enciso and Peralta-Salas
\cite{Enciso2025} have proven that, on axisymmetric toroidal domains, there
exists a set of smooth braided solenoidal vector fields that are \emph{not}
topologically equivalent to any MHD equilibrium. Furthermore, this set is
rather large, in the sense that it is dense in a nonempty open subset of the
space of smooth braided solenoidal fields (equipped  with the $C^\infty$
topology). This suggests that a relaxation method based on the MHD induction
equation should either allow for low-regularity solutions with the possible
formation of current sheets, as conjectured by Parker, or be complemented with
a way to prepare a suitable initial condition. We note that equilibria with
current sheets can be acceptable in some applications as discussed above in
relation to the Grad's conjecture. 

In summary, from the applications point of view, relaxation methods
constructed in this way are not fully satisfactory because of the following
drawbacks: 
(1) not all equilibrium points are accessible from a given initial condition.
(2) The method does not offer any mechanism to control important properties of
the equilibrium such as the pressure profile and current profiles.
(3) The relaxation mechanisms based solely on viscosity do not necessarily lead
to the shortest path from the initial condition down to an equilibrium point.  

An example of a relaxation method based on viscosity is implemented in the HINT
code \cite{Harafuji1989, Suzuki2006}. In HINT, pressure is relaxed with an
ad hoc algorithm, in a separate step, during the magnetic field relaxation. If
resistivity is accounted for in the relaxation of the magnetic field, then the
topology of the magnetic field lines can change, but with finite resistivity
helicity is no longer preserved and there is no lower bound for the magnetic
energy. 

Another relaxation method that seeks a faster way to relax the magnetic energy
is based on the variational principle for the equilibrium conditions (reviewed
in~\ref{sec:VP} for the case of Beltrami fields). This method has been proposed
by Chodura and Schl\"uter \cite{Chodura1981}, cf.\  also Moffatt
\cite[sec.\ 8.2]{Moffatt2021}, and it can be specialized to the case of
nonlinear Beltrami fields \cite{Wiegelmann2012}. The idea is to Lie drag both
the magnetic field and the pressure with an advecting velocity field $U$ chosen
to guarantee the maximum decay rate of the magnetic energy. These ideas are
strictly related to the modern theory of optimal transport of differential forms
\cite{Brenier2018}. 

The idea of Lie dragging both the magnetic field and the pressure has been
exploited in the SIESTA code as well \cite{Hirshman2011}, where each
displacement of both magnetic field and pressure is generated by an
infinitesimal ``Lie dragging step''.

\subsection{Overview of metriplectic relaxation}
\label{ssec:Mrelax}

Metriplectic dynamics is a class of dynamical systems. Its mathematical
structure has associated relaxation methods with desirable properties
for calculating equilibria. Metriplectic dynamics and its concomitant
variational principles for equilibria will be  thoroughly reviewed in Section
\ref{sec:testcases-and-VP}. 
In this subsection we give a  overview of the paper while describing  the usage
in this work, where we explore using {\it artificially} constructed metriplectic
dynamical systems in order to construct relaxation methods for the calculation
of equilibrium  points.    Metriplectic dynamics was  introduced by Morrison
\cite{Morrison1984,pjm84b,Morrison1986} as a generalization of noncanonical
Hamiltonian dynamics with the aim of including dissipative phenomena.
(See \cite{Coquinot2020,pjmU24,pjmZB24,pjmS24,Sato2025,pjmZ24} for recent
developments.) The equation of evolution is constructed in terms of two
algebraic structures: a Poisson bracket \cite{Morrison1998}, which is
antisymmetric and defines the Hamiltonian part of the equation, and a metric
bracket, which is symmetric and accounts for dissipation. In addition to the
brackets, a Hamiltonian function $\fun{H}$ and an entropy function $\fun{S}$ are
given, satisfying appropriate compatibility conditions. As a direct consequence
of the construction, the Hamiltonian $\fun{H}$ is conserved and the entropy
$\fun{S}$ is dissipated (more precisely, it is nonincreasing). The fact that
entropy is a monotonic function of time quantifies the dissipation in the
system. Defining entropy to be nonincreasing is inconsistent with its usual 
physical interpretation as a measure of uncertainty or ``disorder'', which
would require it to be nondecreasing. In this work however, entropy is treated
as a Lyapunov function \cite{Hirsch2013}, and thus we prefer to reverse the
sign and work with a nonincreasing entropy. Many physically relevant
mathematical models have been found to possess a metriplectic structure.
For instance the Vlasov-Maxwell-Landau system \cite{Morrison1986}, various fluid
mechanical systems \cite{pjm84b,pjmZB24,pjmZ24,Materassi2015}, visco-resistive MHD
\cite{Materassi2012}, and dissipative extended MHD \cite{Coquinot2020}.
This structure can also be exploited in order to design numerical schemes that
preserve the key features of the physical model \cite{Kraus2017a,pjmBZ24}. Here
instead we shall use metric brackets as equilibrium solvers. 
 
Not to be confused with metriplectic relaxation is a relaxation method that uses
the Hamiltonian structure. This method, which is based on squaring the
Poisson bracket, was introduced in \cite{Vallis-1989,Carnevale-1990} for
two-dimensional vortical motion of neutral fluids; it  was later generalized so
as  to make it more effective and work in a broader context in Flierl and
Morrison \cite{Flierl2011} and it was applied to reduced MHD problems by
Furukawa and coworkers
\cite{Chikasue2015a,Chikasue2015,Furukawa2017,Furukawa2018} (see \cite{pjmF24}
for recent review). This approach, which has been named ``double bracket''
dynamics or simulated annealing, has different properties as compared to
metriplectic dynamics: with double brackets, the Hamiltonian is dissipated while
the system evolves on the a specific hypersurface (a symplectic leaf) determined
by the constants of motion built into the Poisson brackets (Casimir invariants).
Another early approach is that of Brockett who used a version of the double
bracket for matrices constructed out of the commutator \cite{Brockett1991} (see
also  the work of Bloch and coauthors \cite{Bloch1992, Bloch2013}).  

Since metriplectic dynamics dissipates entropy but preserves the Hamiltonian,
one might expect that a global solution, if it exists, will approach
a minimum of the entropy function restricted to the level set of the
Hamiltonian that contains the initial condition. Specifically, if we denote by
$u$ a point in the phase space $V$ of the system, we might expect that the
solution of a generic metriplectic system with Hamiltonian $\fun{H}(u)$,
entropy $\fun{S}(u)$, and initial condition $u(0) = u_0 \in V$ would  converge 
to a solution of the variational principle 
\begin{equation}
  \label{eq:entropy-principle}
  \min \{\fun{S}(u) \colon u \in V,\; \fun{H}(u) = \fun{H}(u_0)\}.
\end{equation}

Variational principles of the form~(\ref{eq:entropy-principle}) are physically
relevant since, the resulting state of a physical relaxation mechanism can, in
certain cases, be well approximated by the solution of such a minimum entropy
principle. Typically one envisages a situation where the physical evolution of
the system is nearly ideal, that is, dissipation mechanisms are very small, but
such small nonideal effects are sufficient to dissipate one of the ideal
constant of motion, while causing negligible variations in the other constants
of motion. This is referred to as selective decay and plays an important role in
the relaxation of fluids and plasmas to a self-organized state
\cite{Hasegawa1985,Yoshida2002}. In MHD for instance, linear Beltrami fields are
found as a result of processes that dissipate magnetic energy, while
(approximately) preserving magnetic helicity as argued by Woltjer
\cite{Woltjer1958a}. The precise physical relaxation mechanism has been
discussed in detail by Taylor and it is referred to as Woltjer-Taylor
relaxation \cite{Taylor1974,Taylor1986,Qin2012}. In Section~\ref{sec:testcases}
we shall see that equilibrium problems can often be reformulated as a
variational principle of the form~(\ref{eq:entropy-principle}): among the many
possible solutions of the equilibrium conditions, the variational
principle~(\ref{eq:entropy-principle}) selects only those that minimize entropy
on a constant energy hypersurface, and thus reduces significantly the issue of
nonuniqueness of the equilibrium. If the equilibrium problem of interest is
formulated as a variational principle of the form (\ref{eq:entropy-principle}),
then a relaxation method for such a problem should converge to a minimum of
entropy on the constant-energy surface. A significant part of this work is
dedicated to understanding when metriplectic systems have this long-time
convergence property.

When the solution of a metriplectic system has a limit for $t \to +\infty$ and
the limit is a solution of~(\ref{eq:entropy-principle}), we say that  
the  system has \emph{completely relaxed}. Unfortunately, complete relaxation
does not always happen. It depends on the null space of the metric brackets,
which is defined precisely in Section~\ref{sec:metriplectic}. We shall
demonstrate complete relaxation (or lack thereof) by means of numerical
experiments. We shall propose and test a particular class of brackets modeled 
upon Morrison's brackets for the Landau collision operator
\cite{Morrison1984, Morrison1986} and show by means of numerical experiments
that these new brackets completely relax an initial condition. These new
brackets are referred to as collision-like metric bracket.  

Collision-like brackets have the disadvantage of generating integro-differential
evolution equations, that are usually computationally expensive (although
efficient methods exist \cite{Adams2017}). In an attempt to reduce the
computational cost of the relaxation methods, we have introduced a simplified
version of the collision brackets that are local and thus lead to pure partial
differential equations that have the structure of diffusion equations. We
refer to these simplified brackets as diffusion-like. We are able to
recover known relaxation methods, such as the metriplectic bracket based on
Nambu dynamics of \cite{Bloch2013}, which was also obtained and proposed in the
context of vortex dynamics in \cite{Gay-Balmaz2013,Gay-Balmaz2014}, and the
method of Chodura and Schl\"uter \cite{Chodura1981}, as special cases of
diffusion-like brackets. All of the brackets used in this paper follow naturally
from the inclusive 4-bracket construction given in \cite{pjmU24}.  

The remainder of the paper is structured as follows. As noted above, in
Section~\ref{sec:testcases-and-VP} we cover review material: we recall the
precise definition of metriplectic systems in Section \ref{sec:metriplectic} and
describe the equilibrium problems that we consider as test cases together with
corresponding variational principles in Section~\ref{sec:testcases}.
Section~\ref{sec:remarks-relax-equil} presents some mathematical results on
relaxation and the relaxation rate for metriplectic systems, including results
that we are unable to find in the literature.
In Sections~\ref{sec:finite-dim-Lyapunov} and \ref{sec:finite-dim-PL} we prove
extensions of the Lyapunov stability theorem and the Polyak--{\L}ojasiewicz
condition for the rate of relaxation, for finite-dimensional systems with the
inclusion of constraints, respectively, while in Section~\ref{sec:infinite-dim}
we make some comments on the extension of these results to infinite-dimensional
systems. In Section~\ref{sec:simple}, we address simple examples of metric
brackets and study the issue of complete relaxation, both analytically and
numerically. Section \ref{sec:metr-double-brackets} describes brackets,
which we refer to as metric double brackets, that fail to completely relax,
while Section~\ref{sec:projector-based-metric-bracket} describes projection
based brackets that do completely relax. Collision-like brackets are
introduced in Section~\ref{sec:coll-like-metr} with theory presented in
Sections \ref{sec:c-general}, \ref{sec:c-div-grad}, and \ref{sec:c-curl-curl}
and numerical  experiments in  Sections \ref{sec:app-euler} and \ref{sec:app-GS}
for the Euler equations in vorticity form and for the Grad-Shafranov MHD
equilibria, respectively. For these applications complete relaxation of the
solution is critical. Section~\ref{sec:diff-like-metr} is dedicated to
diffusion-like brackets, with their general construction and various forms
given in Sections \ref{sec:d-general}, \ref{sec:d-div-grad}, and
\ref{sec:d-curl-curl}. We shall see that their properties make them most
suitable for applications such as the calculation of \emph{nonlinear} Beltrami
fields, which are considered in Section \ref{sec:app-Beltrami}, for which the
property of complete relaxation is not needed. Finally, we conclude in
Section~\ref{sec:conclusions}.

\section{Metriplectic dynamics and variational principles for equilibria of
  fluids and plasmas} 
\label{sec:testcases-and-VP}

In this section, we review the necessary background material. 
First we briefly recall the definition and basic properties of metriplectic
dynamics. Then we formulate and discuss  examples of equilibrium problems,
and their variational formulations. These problems will be used as a test
bed for the metriplectic relaxation method proposed in this work.

\subsection{Metriplectic dynamics}
\label{sec:metriplectic}

Metriplectic dynamics, being a special kind of continuous-time dynamical system,
is determined by a phase space and a law governing the evolution in time of a
point in the phase space. In this work, we are mainly concerned with
infinite-dimensional metriplectic systems with a phase space given by a Banach
space $V$ of functions over a domain $\Omega \subset \R^d$ with values in
$\R^N$, where $d,N \in \N$. (We use ``domain''  as a shorthand for
``open, connected set''.) We shall however briefly discuss
finite-dimensional examples for the sake of clarity and simplicity. In the
finite-dimensional case the phase space is chosen to be an open connected subset
$\mathcal{Z} \subseteq \R^n$, $n \in \N$, with coordinates $z = (z^i)_{i=1}^n$.

For any $\fun{F} \in C^1(V)$, we denote by $D\fun{F}(u)$ its Fr\'echet
derivative \cite{Hunter2001} at the point $u \in V$. We shall also need the
functional derivative of $\fun{F}$. If $W$ is another Banach space with a
nondegenerate pairing
$\langle \cdot, \cdot \rangle_{V\times W} \colon V \times W \to \R$,
we can define the functional derivative $\delta \fun{F}/\delta u$ of $\fun{F}$
in $W$ as the unique element of $W$, if it exists, such that \cite{Marsden2001}
\begin{equation}
  \label{eq:functional-derivatives}
  D\fun{F}(u)v = \Big\langle v, \frac{\delta \fun{F}(u)}{\delta u}
  \Big\rangle_{V\times W}, \quad
  \forall v \in V.
\end{equation}
Unless otherwise stated, in this work we assume
$V \subseteq L^2(\Omega,\mu;\R^N)$, $W = L^2(\Omega,\mu; \R^N)$, and
the pairing is the $L^2$ product with respect to a given measure
$d\mu = m(x) dx$, where $m$ is a smooth and integrable function and $dx$ the
Lebesgue measure (volume element) on $\Omega$,
\begin{equation*}
  (u,v)_{L^2} = \int_\Omega  u(x) \cdot  v(x) d\mu(x), \quad
  u,v \in L^2(\Omega,\mu;\R^N).
\end{equation*}
(The nontrivial measure is needed in order to accommodate cases such as
Grad-Shafranov equilibria discussed below.)  

A metriplectic dynamical system is specified by giving two functions
$\fun{H}, \fun{S} \in C^\infty(V)$, namely the Hamiltonian and the entropy, 
respectively, together with compatible Poisson and metric brackets on $V$.
We recall the definitions
\cite{Morrison1984,Morrison1986,Morrison1998,Marsden2001,Morrison1982}.  

A \emph{Poisson bracket} on $V$ is a bilinear \emph{antisymmetric} map
\begin{equation*}
  \{\cdot,\cdot\} \colon C^\infty(V) \times C^\infty(V) \to C^\infty(V),
\end{equation*}
such that, for any $\fun{F}$, $\fun{G}$, and $\fun{H}$ in $C^\infty(V)$,
\begin{subequations}
  \label{eq:Poisson}
  \begin{align}
    \label{eq:Leibniz-for-Poisson}
    &\{\fun{F},\fun{G}\fun{H}\} = \{\fun{F},\fun{G}\} \fun{H} +
    \fun{G} \{\fun{F},\fun{H}\}, \\
    \label{eq:Jacobi}
    &\big\{\fun{F},\{\fun{G},\fun{H}\}\big\} +
    \big\{\fun{G},\{\fun{H},\fun{F}\}\big\} +
    \big\{\fun{H},\{\fun{F},\fun{G}\}\big\} = 0.
  \end{align}
\end{subequations}
Equations~(\ref{eq:Leibniz-for-Poisson}) and~(\ref{eq:Jacobi}) are referred to
as the Leibniz identity and the Jacobi identity, respectively. Then a Poisson
bracket defines a Lie algebra structure on $C^\infty(V)$ which in  addition is a
derivation in each argument. 

A \emph{metric bracket} on $V$ is a bilinear \emph{symmetric} map
\begin{equation*}
  (\cdot,\cdot) \colon C^\infty(V) \times C^\infty(V) \to C^\infty(V),
\end{equation*}
such that, for any $\fun{F}$ in $C^\infty(V)$, 
\begin{equation*}
  (\fun{F},\fun{F}) \geq 0.
\end{equation*}
By definition, the Poisson bracket must satisfy Leibniz and Jacobi
identities. Leibniz identity, in particular, implies at least formally that the
bracket can be written in term of the functional derivatives of its arguments,
cf.~\ref{sec:Leibniz} for a precise definition. Usually, the symmetric
bracket does not need to satisfy any condition other then bilinearity,
symmetry, and positive semidefiniteness. However, if one requires the symmetric
bracket to satisfy Leibniz identity, 
\begin{equation}
  \label{eq:Leibniz-for-metric}
  (\fun{F},\fun{G}\fun{H}) = (\fun{F},\fun{G}) \fun{H}
  + \fun{G} (\fun{F},\fun{H}),
\end{equation}
then both the Poisson and the symmetric brackets have a similar representation,
i.e.,  
\begin{subequations}
  \label{eq:J-K-kernels}
  \begin{equation}
    \label{eq:J-kernel}
    \{\fun{F},\fun{G}\} = \sum_{i,j=1}^N \int_\Omega \int_\Omega
    \frac{\delta \fun{F}(u)}{\delta u_{i}} (x)
    \mathscr{J}_{ij}(u;x,x')
    \frac{\delta \fun{G}(u)}{\delta u_j} (x') \, d\mu(x') \, d\mu(x), 
  \end{equation}
  and analogously
  \begin{equation}
    \label{eq:K-kernel}
    (\fun{F},\fun{G}) = \sum_{i,j =1}^{N} \int_\Omega \int_\Omega
    \frac{\delta \fun{F}(u)}{\delta u_i} (x)
    \mathscr{K}_{ij}(u;x,x')
    \frac{\delta \fun{G}(u)}{\delta u_j} (x') \, d\mu(x') \, d\mu(x), 
  \end{equation}
\end{subequations}
where the functional derivatives are computed with respect to the $L^2$ product
with a given measure $\mu$ on $\Omega$.
The kernels $\mathscr{J}(u)$ and $\mathscr{K}(u)$ define an anti-symmetric and a
symmetric, positive semidefinite operator, $J(u)$ and $K(u)$, respectively.
In finite dimensions, Eqs.~(\ref{eq:J-K-kernels}) take the form
\begin{align*}
  \{\fun{F}, \fun{G}\} &= J^{ij}(z)
  \frac{\partial \fun{F}(z)}{\partial z^{i}}
  \frac{\partial \fun{G}(z)}{\partial z^{j}},
  \\
  (\fun{F}, \fun{G}) &= K^{ij}(z) \frac{\partial \fun{F}(z)}{\partial z^i}
  \frac{\partial \fun{G}(z)}{\partial z^j},
\end{align*}
where here the sum over repeated indices ranges to $n$,
$J(z)$ is an  antisymmetric contravariant tensor, and $K(z)$ is a symmetric
positive semidefinite contravariant tensor over the domain
$\mathcal{Z} \subseteq \R^n$.
In particular $J$ is referred to as the Poisson tensor.

The evolution equation for a metriplectic system $u(t) \in V$ is formulated as
an evolution equation for arbitrary functions of $u(t)$, that is,
\begin{subequations}
  \label{eq:metriplectic-system}
  \begin{equation}
    \label{eq:metriplectic-system-equation}
    \frac{d\fun{F}}{dt} = \{\fun{F},\fun{H}\} - (\fun{F},\fun{S}), \quad
    \text{ for all } \fun{F} \in C^\infty(V),
  \end{equation}
  where $\fun{H},\fun{S} \in C^\infty(V)$  are the Hamiltonian and
  entropy functions, respectively, whereas $\{\cdot,\cdot\}$ and $(\cdot,\cdot)$
  are the Poisson and metric bracket on $V$, respectively, satisfying the
  compatibility conditions  
  \begin{equation}
    \label{eq:compatibility}
    \{\fun{F},\fun{S}\} = 0, \quad (\fun{F},\fun{H}) = 0, \quad
    \text{ for all } \fun{F} \in C^\infty(V).
  \end{equation}
\end{subequations}
If both brackets satisfy the Leibniz identity, the evolution equation then reads
\begin{equation*}
  \partial_t u = J(u) \frac{\delta \fun{H}(u)}{\delta u} -
  K(u) \frac{\delta \fun{S}(u)}{\delta u}.
\end{equation*}
In general, both $J(u)$ and $K(u)$ have nontrivial null spaces. Per definition,
the null space of a bracket is identified with that the corresponding operators
$J(u)$ and $K(u)$. We note that, in general, the null space of a bracket depends
on the phase-space point $u$, since $J$ and $K$ depend on $u$. The null space of
the Poisson bracket is due to the noncanonical form that often originates from a
reduction procedure based on the symmetries of the system
\cite{Meyer1973,Marsden1974a} (see for example the texts
\cite{Marsden1999,Holm2009}), while the null space of the metric bracket is due
to the requirement that at least energy is preserved,
cf. equation~(\ref{eq:compatibility}). In finite dimensions, the null space of
$J(z)$ is spanned by the gradients of Casimir invariants and equilibria from the
variational principle align with those of the equations of motion
\cite{Morrison1998}. We note, however, that there are subtleties at points where
the rank of $J(z)$ changes \cite{pjmE86} and equilibria at such points can
possess nearby behavior that is not Hamiltonian \cite{pjmYT17,pjmY20}.  
The null space of $K(z)$ contains at least the gradient of the Hamiltonian, due
to the compatibility condition~\eqref{eq:compatibility}; similar remarks on rank
changing could apply. For a metric bracket defined on a Banach space
$V \subseteq L^2(\Omega, \mu;\R^N)$ and corresponding to a bounded operator
$K(u)$ on $L^2(\Omega, \mu;\R^N)$, the null space at $u \in V$ can be
equivalently characterized as the space of functions $\fun{F}$ over $V$ such
that $\big(\fun{F},\fun{F}\big)(u) = 0$, since this is equivalent to
$\delta \fun{F}(u) / \delta u \in \ker K(u)$.

In general, the vector field
$J(u) \delta \fun{H}/\delta u$ can be viewed as a generalization of Hamiltonian
flow $\omega(u)^{-1} \delta \fun{H}/\delta u$ where the inverse of the
symplectic operator $\omega(u)$ is replaced by a possibly noninvertible Poisson
operator $J(u)$. On the other hand, the vector field
$-K(u) \delta \fun{H}/\delta u$ can be viewed as generalization of a gradient
flow $-G(u)^{-1} \delta \fun{H}/\delta u$ where the inverse of the metric
operator $G(u)$ is replaced by a symmetric, positive \emph{semi}definite
operator.  
Metriplectic dynamics combines the (generalized) symplectic and gradient flows.
Typically, the symplectic part describes the ideal dynamics, while the
gradient flow accounts for a nonideal relaxation mechanism. For this reason,
accepting a slight abuse of terminology, we refer to symmetric brackets
$(\cdot,\cdot)$ with the Leibniz property as \emph{metric brackets}. We shall
always tacitly assume that $\fun{H}$ and $\fun{S}$ are not functionally
dependent, i.e., $\delta \fun{S}(u) / \delta u$ is not everywhere parallel to 
$\delta \fun{H}(u) / \delta u$, otherwise the metric bracket part vanishes
identically.

We are mostly interested in the dynamical systems generated by metric brackets,
i.e., we shall drop the Poisson bracket part,
\begin{subequations}
  \label{eq:metric-system}  
  \begin{equation}
    \label{eq:metric-system-equation}
    \frac{d\fun{F}}{dt} = -(\fun{F},\fun{S}), \quad
    \text{ for all } \fun{F} \in C^\infty(V),
  \end{equation}
  where $\fun{S}$ is the entropy function and the metric bracket satisfies
  the compatibility condition 
  \begin{equation}
    \label{eq:metric-system-compatibility}
    (\fun{F},\fun{H}) = 0, \quad
    \text{ for all } \fun{F} \in C^\infty(V),
  \end{equation}
\end{subequations}
where $\fun{H}$ is the Hamiltonian.

A solution of either~(\ref{eq:metriplectic-system}) or~(\ref{eq:metric-system})
satisfies 
\begin{equation}
  \label{eq:metric-system-properties}
  \frac{d \fun{H}(u)}{dt} = 0, \qquad \frac{d\fun{S}(u)}{dt} \leq 0,
\end{equation}
that is, both system~(\ref{eq:metriplectic-system}) and~(\ref{eq:metric-system})
dissipate entropy on the surface of constant energy (Hamiltonian).

Because of~(\ref{eq:metric-system-properties}) one may expect that solutions of
the variational principle~(\ref{eq:entropy-principle}) are necessarily
equilibria of a metriplectic system. Indeed this is the case and it follows from
the method of Lagrange multipliers \cite{Marsden2001}, which gives a necessary
condition for~(\ref{eq:entropy-principle}): if $u$ is a solution
of~(\ref{eq:entropy-principle}), then there is a constant $\lambda \in \R$ 
(the Lagrange multiplier) such that  
\begin{equation}
  \label{eq:entropy-principle-2}
  D\fun{S}(u) - \lambda D\fun{H}(u) = 0, \quad \fun{H}(u) = \fun{H}_0\,.
\end{equation}
Alternatively, one can write the Lagrange condition using the functional
derivative, if they exist, 
\begin{equation*}
  \frac{\delta \fun{S}(u)}{\delta u} - \lambda \frac{\delta
    \fun{H}(u)}{\delta u} = 0, \quad \fun{H}(u) = \fun{H}_0\,.
\end{equation*}
Equations~(\ref{eq:entropy-principle-2}) constitute a system of two equations
for the pair $(u, \lambda) \in V \times \R$. Let us assume that the set of
solutions 
\begin{equation*}
  \mathfrak{C}_{\fun{H}_0} \coloneqq \{u \in V \colon  \exists \lambda \in \R
  \text{ such that $(u,\lambda)$ solves~(\ref{eq:entropy-principle-2})}\},
\end{equation*}
is nonempty ($\mathfrak{C}_{\fun{H}_0} \not=\emptyset$),
the restriction $\fun{S}$ to $\mathfrak{C}_{\fun{H}_0}$ is bounded from below,
and the minimum is attained, that is, there are points
$u_e \in \mathfrak{C}_{\fun{H}_0}$ where
$\fun{S}(u_e) = \min \{\fun{S}(u) \colon u \in \mathfrak{C}_{\fun{H}_0}\}$.
Then, the solutions of~(\ref{eq:entropy-principle}) correspond to those points
$u_e$ of $\mathfrak{C}_{\fun{H}_0}$ where $\fun{S}$ attains its minimum. 
The set $\mathfrak{C}_{\fun{H}_0}$ is the set of constrained critical points of
$\fun{S}$, that is, of critical points of $\fun{S}$ restricted to the energy
surface $\fun{H}(u) = \fun{H}_0$.

If the symmetric bracket satisfies the Leibniz identity, the Lagrange
condition together with either the compatibility
condition~(\ref{eq:metriplectic-system-equation})
or~(\ref{eq:metric-system-compatibility}) imply that any point in
$\mathfrak{C}_{\fun{H}_0}$, i.e., any constrained critical point of $\fun{S}$,
is an equilibrium point of the metriplectic system, and thus, in particular, any
solution $u_e$ of~(\ref{eq:entropy-principle}) is necessarily an equilibrium
point \cite[sec. 4.1]{Bloch2013}. The converse, however, is not true, since all
constrained critical points are equilibria, not just the minima. In addition,
there can be equilibrium points of either~(\ref{eq:metriplectic-system})
or~(\ref{eq:metric-system}) that are not constrained critical points of
$\fun{S}$. One example is given in Section~\ref{sec:metr-double-brackets} below.
  
The fact that the set of equilibrium points of a metriplectic system can, in
general, be (much) larger than the set $\mathfrak{C}_{\fun{H}_0}$ of constrained
critical points of $\fun{S}$ can be an obstruction to convergence of an orbit
$u(t)$ to a solution of~(\ref{eq:entropy-principle}), and this is important in
some (but not all) applications.
In the next section, we review a few physically relevant equilibrium problems
and discuss their relation to variational principles of the
form~(\ref{eq:entropy-principle}). The main result of the paper is the
construction of appropriate metric brackets that relax a given initial condition
to a solution of such equilibrium problems.

\subsection{Examples of equilibrium problems}
\label{sec:testcases}

In this section, we review the examples of equilibrium problems that we shall
use as test cases for metriplectic relaxation. 
All considered test problems are mathematically ill-posed, because they admit
multiple solutions. In some cases, the ill-posedness can be mitigated by
adding additional physics constraints. We shall also discuss the variational
principles for the considered equilibrium problems.

\subsubsection{Reduced Euler  equations}
\label{sec:Euler-problem}

We begin with the Euler equations reduced to two dimensions 
\cite[p.488 and references therein]{Morrison1998}, which is the simplest of a
hierarchy of models including the reduced MHD model \cite{Yoshida2016}.
Let $x = (x_1,x_2)$ be Cartesian coordinates in a bounded domain
$\Omega \subset \R^2$ with a sufficiently regular boundary $\partial \Omega$.
In $\Omega$, we consider an incompressible flow 
$U = (U_1, U_2) = \big(\partial_2 \phi, -\partial_1 \phi \big)$, given in
terms of a stream function $\phi(x)$  with $\partial_i = \partial/\partial x_i$.
Then $\div U = 0$ and from the definition of the scalar vorticity
$\omega \coloneqq \partial_1 U_2 - \partial_2 U_1$ one obtains the Poisson
equation 
\begin{equation}
  \label{eq:Poisson-eq}
  -\Delta \phi = \omega \text{ in $\Omega$}\,, \quad
  \phi = 0 \text{ on } \partial \Omega\,,
\end{equation}
where $\Delta$ is the Laplace operator in $\R^2$. Vice versa, given $\omega$,
we can solve~(\ref{eq:Poisson-eq}) for $\phi$ and reconstruct the flow $U$.
Hence, the incompressible Euler equations in two-dimensions
\begin{equation*}
  \partial_t U + U \cdot \nabla U = -\nabla p, \quad
  \div U = 0, 
\end{equation*}
with $p$ being the pressure field, amount to an evolution equation for the
scalar vorticity $\omega$, 
\begin{equation*}
  \left\{
  \begin{aligned}
    \partial_t \omega + [\omega,\phi] &= 0,  && \text{ in $\Omega$}, \\
    -\Delta \phi - \omega &= 0, && \text{ in $\Omega$}, \\
    \phi &= 0,  && \text{ on $\partial \Omega$},
    \end{aligned}
  \right.
\end{equation*}
where $[\omega,\phi] \coloneqq \partial_1 \omega\,  \partial_2\phi -
\partial_2 \omega\,  \partial_1 \phi$ 
is the canonical Poisson bracket in  $\R^2$. This model is referred to as the 
reduced Euler equations.

The phase space $V$ of the reduced Euler equations is the space of vorticity
fields, i.e., $u = \omega$, and $\phi = -\Delta^{-1}_{\Omega,0}\, \omega$ is
regarded as a function of $\omega$, given by the inverse of the Laplacian on
$\Omega$ with homogeneous Dirichlet boundary conditions.

The equilibrium problem for the reduced Euler equations then reads 
\begin{equation}
  \label{eq:Euler-equilibrium}
  \left\{
  \begin{aligned}
    [\omega,\phi] &= 0,  && \text{ in $\Omega$}, \\
    -\Delta \phi - \omega &= 0, && \text{ in $\Omega$}, \\
    \phi &= 0,  && \text{ on $\partial \Omega$}.
  \end{aligned}
  \right.
\end{equation}
Problem~(\ref{eq:Euler-equilibrium}) admits many solutions and therefore is
mathematically ill-posed. One can in fact construct a large class of solutions
upon noticing that $[\omega,\phi] = 0$ implies that $\omega$ is constant on the
isolines (contours) of the potential $\phi$. The contours may have many
connected components, each one with a possibly different topology (i.e.,
homeomorphic to a different model space) and the constant value of $\omega$ on
different connected components may be different.
Given a function $f \in C^1(\R)$, we may set $\omega = \lambda f(\phi)$
with a normalization factor $\lambda \in \R$ to be determined; 
in this way we assign the same value of $\omega$ to all connected components of
the same contour of $\phi$. This is a special case which is considered here for
sake of simplicity. Then we consider the problem: find $(\phi,\lambda)$, with
$\lambda \not= 0$, such that 
\begin{equation}
  \label{eq:Euler-semi-linear-eq}
  \left\{
  \begin{aligned}
    -\Delta \phi &= \lambda f(\phi),  && \text{ in $\Omega$}, \\
    \phi &= 0,  && \text{ on $\partial \Omega$}.
  \end{aligned}
  \right.
\end{equation}
Any solution $(\phi,\lambda)$ of~(\ref{eq:Euler-semi-linear-eq})
yields a solution $\omega = \lambda f(\phi)$ of the equilibrium
problem~(\ref{eq:Euler-equilibrium}). The case $\lambda = 0$ leads 
to the trivial equilibrium $\omega = 0$ and it is not considered.

Problem~(\ref{eq:Euler-semi-linear-eq}) is a
``eigenvalue problem'' for a semilinear elliptic equation.
If $f'(y) \leq 0$ (respectively, $f'(y)\geq 0$) for all $y$, the equation has
a unique solution for any $\lambda \geq 0$ (respectively, $\lambda \leq 0$)
\cite{Taylor3}. In the other cases, the solution may not exist for all
$\lambda$; if a solution exists, uniqueness is not guaranteed, e.g., for
degenerate eigenvalues. For instance, when $f(y) = y$,
problem~(\ref{eq:Euler-semi-linear-eq}) reduces to the standard eigenvalue
problem for the Laplace operator with Dirichlet boundary conditions; then, we
have discrete, possibly degenerate, positive eigenvalues
$\lambda_n>0$, $n \in \N$, each with a corresponding finite set of
eigenfunctions $\phi_{n,k}$ depending on the multiplicity of $\lambda_n$.
For $\lambda \leq 0$, the trivial solution $\phi=0$ is the unique solution. 

This, in particular, shows that problem~(\ref{eq:Euler-equilibrium}) is
ill-posed, since there is a rich set of solutions for each choice of $f$, and
many choices of $f$ are possible. In order to mitigate the nonuniqueness
problem, in practice the function $f$, which will be referred to as the
\emph{equilibrium profile}, is prescribed, and among the solutions
of~(\ref{eq:Euler-semi-linear-eq}), the one with the lowest $\lambda$ is
considered. The reformulation of the equilibrium
problem~(\ref{eq:Euler-equilibrium}) into the nonlinear eigenvalue
problem~(\ref{eq:Euler-semi-linear-eq}) with fixed $f$ is good enough in
practice. There are efficient iterative algorithms \cite{Takeda1991} for the
solution of this type of eigenvalue problem with the lowest $\lambda$.  

An alternative reformulation of the equilibrium
problem~(\ref{eq:Euler-equilibrium}) is possible, based on a variational
principle of the form~(\ref{eq:entropy-principle}). For $\fun{H}_0 > 0$ and
$s \in C^2(\R)$ satisfying $s''(y) \not= 0$, $y \in \R$, let us consider the
problem 
\begin{equation}
  \label{eq:VP-Euler}
  \min \{\fun{S}(\omega) \colon \fun{H}(\omega) = \fun{H}_0 \}.
\end{equation}
where
\begin{equation}
  \label{eq:Euler-S-H}
  \fun{S}(\omega) = \int_\Omega s(\omega) dx, \quad  
  \fun{H}(\omega) = \frac{1}{2} \int_\Omega |\nabla \phi|^2 dx,
\end{equation}
are the entropy and Hamiltonian functions, respectively, with $\phi$ depending
on $\omega$ via the Poisson equation~(\ref{eq:Poisson-eq}). The condition on the
entropy profile $s(y)$ implies that $s'(y)$ is a strictly monotonic function,
either decreasing or increasing. In general, $s(y)$ will be chosen \emph{ad hoc}
and does not necessarily have a physical meaning. 
The Hamiltonian $\fun{H}$ is the kinetic energy of the incompressible fluid,
since $|\nabla \phi|^2 = |U|^2$, where $U$ is the flow velocity. 

Problem~(\ref{eq:entropy-principle-2}) in this case reads:
Find $(\omega,\lambda)$, such that
\begin{equation}
  \label{eq:complete-relaxation-vorticity2d}
  \left\{
  \begin{aligned}
    & s'(\omega) - \lambda \phi = 0, \quad \fun{H}(\omega) = \fun{H}_0, \\
    & \text{with $\phi$ solution of~(\ref{eq:Poisson-eq}).}
  \end{aligned}
  \right.
\end{equation}
The set $\mathfrak{C}_{\fun{H}_0}$ of constrained critical points,
cf.\ Section~\ref{sec:metriplectic}, amounts to 
\begin{equation*}
  \mathfrak{C}_{\fun{H}_0} =
  \{ \omega \suchthat (\omega, \lambda)
  \text{ solves~(\ref{eq:complete-relaxation-vorticity2d})
    for some $\lambda \in \R$} \}.
\end{equation*}
All elements of the set $\mathfrak{C}_{\fun{H}_0}$ are solutions of the original
equilibrium problem~(\ref{eq:Euler-equilibrium}). In fact, for
$\lambda \not=0$, we have 
$[\omega, \phi] = \lambda^{-1} [\omega, s'(\omega)]=0$.
The case $\lambda = 0$ is somewhat special, since $s'(\omega) = 0$ with
$s''(y) \not= 0$ implies that $\omega(x) = \omega_c(x) = y_c = $ constant,
where $y_c$ is the unique zero of $s'(y)$; the corresponding
potential $\phi_c$ is given by the solution of problem~(\ref{eq:Poisson-eq})
with constant right-hand side. There is therefore only one solution for
$\lambda = 0$ and this carries the energy
$\frac{1}{2} \|\nabla \phi_c\|^2_{L^2} = \fun{H}_c$ and it is always an
equilibrium, since $\omega_c$ is constant. If $\fun{H}_0 = \fun{H}_c$,  this
solution belongs to $\mathfrak{C}_{\fun{H}_0}$, otherwise $\lambda = 0$ is
not a possible value for the Lagrange multiplier. 

Under the hypothesis $s''(y) \not= 0$, $s'(y)$ is monotonic, and thus an
invertible function of $y \in \R$.
Problem~(\ref{eq:complete-relaxation-vorticity2d}) is
related to the eigenvalue problem~(\ref{eq:Euler-semi-linear-eq}) with
equilibrium profile given by $f(y) = (s')^{-1}(y)$. Precisely, if
$(\omega,\lambda)$ is a solution of~(\ref{eq:complete-relaxation-vorticity2d})
with $\lambda \not=0$, we can define $\tilde{\omega} = \lambda \omega$ and
$\tilde{\phi} = \lambda \phi$, and obtain
from~(\ref{eq:complete-relaxation-vorticity2d}) 
\begin{equation*}
  -\Delta \tilde{\phi} = \lambda \omega = \lambda (s')^{-1}(\lambda \phi) =
  \lambda f(\tilde{\phi}),
\end{equation*}
which shows that $\tilde{\phi}$ solves~(\ref{eq:Euler-semi-linear-eq}) with the
eigenvalue being the same as the Lagrange multiplier $\lambda$.

Among all these equilibria, the minimization in~(\ref{eq:VP-Euler}) selects
those with minimum entropy, thus mitigating the nonuniqueness problem, as
shown in the following special case. As in Yoshida and Mahajan
\cite{Yoshida2002}, we choose $s(y) = y^2/2$, hence $s'(y)=y$ and solutions
of~(\ref{eq:complete-relaxation-vorticity2d}) must necessarily solve the 
eigenvalue problem for the Laplace operator on $\Omega$ with homogeneous
Dirichlet boundary conditions. From standard theory \cite{Taylor1}, we know
that there is an orthonormal basis $\{\phi_{j,k}\}$ in $L^2(\Omega)$ of
eigenfunctions, $-\Delta \phi_{j,k} = \lambda_j \phi_{j,k}$, labeled by
$j \in \N_0$ with $k = 1,\ldots, d_j$ counting the multiplicity and $d_j$
being the dimension of the eigenspace corresponding to the eigenvalue
$\lambda_j > 0$. Then, the set $\mathfrak{C}_{\fun{H}_0}$ comprises all and only
the vorticity fields of the form 
\begin{equation*}
  \omega_j = \lambda_j \phi_j, \quad
  \phi_j = \sum_{k=1}^{d_j} a_{j,k} \phi_{j,k},
\end{equation*}
for any $j \in \N_0$ and $a_j = (a_{j,k}) \in \R^{d_j}$ satisfying the energy
constraint
\begin{equation*}
  \lambda_j a_j^2 = 2 \fun{H}_0.
\end{equation*}
The energy constraint fixes the length of the vector $a_j$ but
not its direction. The function $\fun{S}$ restricted to
$\mathfrak{C}_{\fun{H}_0}$ is given by    
\begin{equation*}
  \fun{S}(\omega_j) = \fun{H}_0 \lambda_j,
\end{equation*}
which is bounded from below, since the spectrum of $-\Delta$ is bounded from
below and $\fun{H}_0$ is a constant. The solutions of~(\ref{eq:VP-Euler})
correspond to the eigenfunctions with minimum eigenvalue (the ground
states).
Usually the eigenspace corresponding to the lowest eigenvalue is one-dimensional
hence we have two solutions that differ only by the sign of the vorticity. In
this example, the variational principle~(\ref{eq:VP-Euler}) picks the solution
of~(\ref{eq:Euler-semi-linear-eq}) with the lowest $\lambda$.  

In this paper, we use metriplectic dynamics in order to solve the
variational principle~(\ref{eq:VP-Euler}). For this particular application it
is essential that the orbit of the chosen metriplectic dynamical system
relaxes completely to a constrained entropy minimum.

\subsubsection{Axisymmetric MHD equilibria}
\label{sec:GS-problem}

A similar equilibrium problem arises from axisymmetric ideal magnetohydrodynamic
(MHD) equilibria of electrically conducting fluids. For MHD, the general
equilibrium condition with zero flow amounts to
\cite{Freidberg2014,Bruno1996,Kruskal1958a,Grad1958,Grad1964}
\begin{equation}
  \label{eq:ideal-mhd}
  J \times B = c\nabla p, \quad \curl B = 4\pi J /c, \quad \div B = 0,
\end{equation}
where $B$ and $J$ are vector fields and $p$ is a scalar field. Physically $B$
and $J$ are the magnetic field and the electric current density, respectively,
while $p$ is the fluid pressure. As noted above, Gaussian units are used with
$c$ being the speed of light in free space. The first equation expresses the
force balance between the Lorentz force $J \times B/c$ and the 
pressure gradient $\nabla p$. The force balance implies the necessary conditions
\begin{equation}
  \label{eq:pressure-constraints}
  B \cdot \nabla p = 0, \quad
  J \cdot \nabla p = 0,
\end{equation}
that is, $p$ is constant on the field lines of both $B$ and $J$.

For axisymmetric solutions, i.e., solutions that have rotational symmetry around
an axis, we introduce cylindrical coordinates $(r,\varphi,z)$ around the
symmetry axis $z$, and from $\div B = 0$ it follows that \cite{Freidberg2014},  
\begin{align*}
  B&= \chi \nabla \varphi + \nabla \psi \times \nabla \varphi\,, 
  \ncr 
   4\pi J/c = \curl B &=
  - \Delta^* \psi \nabla \varphi + \nabla \chi \times\nabla\varphi,
\end{align*}
where $\Delta^* = r [\partial_r(r^{-1} \partial_r)] + \partial_z^2$ is
a linear elliptic second-order differential operator in $(r,z)$ coordinates, the
Grad-Shafranov operator, whereas $\chi(r,z)$, $p(r,z)$, and $\psi(r,z)$ are
real-valued scalar functions. The operator $\nabla$ is the full
three-dimensional gradient. Then axisymmetric equilibria with zero flow must
satisfy the conditions  
\begin{align*}
  &- \Delta^*\psi \nabla \psi - \chi \nabla\chi - 4\pi r^2 \nabla p
  = [\psi, \chi] r \nabla \varphi,
  \\
  &  [\psi,p] = 0,
  \quad  [\chi,p] = 0,
\end{align*}
where the brackets $[\cdot,\cdot]$ are the canonical Poisson brackets in the
$(r,z)$-plane, e.g.,
$[\chi,\psi] = r \nabla\varphi \cdot (\nabla \psi \times \nabla \chi) =
\partial_r \chi \partial_z \psi - \partial_r \psi \partial_z \chi$. 
The first equilibrium condition expresses the force balance
of (\ref{eq:ideal-mhd}), while the latter
two follow from the necessary conditions~(\ref{eq:pressure-constraints}),
respectively. With homogeneous Dirichlet boundary conditions, we can
formulate the problem
\begin{equation}
    \label{eq:Grad-Shafranov-conditions}
    \left\{
    \begin{aligned}
      u \nabla \psi - \chi \nabla\chi - 4\pi r^2 \nabla p &= 0,
      && \text{ in $\Omega$}, \\
      [\psi, p] = 0, \quad [\psi, \chi] &= 0, && \text{ in $\Omega$},\\
      -\Delta^* \psi -u &= 0, && \text{ in $\Omega$}, \\
      \psi &= 0,  && \text{ on $\partial \Omega$},
    \end{aligned}
    \right.
  \end{equation}
where $\Omega \subset \R_+ \times \R$ is a domain in the $(r,z)$ plane
satisfying $r > 0$ in the closure $\ol{\Omega}$ (so that $\Omega$ is bounded
away from the singularity of cylindrical coordinates at $r=0$).
The auxiliary variable $u$ is related to the $\varphi$ component of the current
density, since 
$J_\varphi = J \cdot \nabla \varphi / |\nabla \varphi|
= -c \Delta^* \psi /(4\pi r) = cu /(4\pi r)$.
Without specifying other constraints this problem is ill-posed in the same way
as problem~(\ref{eq:Euler-equilibrium}): the vanishing of the two Poisson
brackets in~(\ref{eq:Grad-Shafranov-conditions}) implies that, if $\nabla \psi$,
$\nabla p$, and $\nabla \chi$ are all nonzero, both $\chi$ and $p$ are constant
on the level sets of $\psi$. However, the functional relation between $\chi$,
$p$ and $\psi$ is undetermined. Fortunately, providing such information is
straightforward. If we prescribe $\chi = \sqrt{\lambda} F(\psi)$ and
$p = \lambda G(\psi)$ for given functions $F,G \in C^1(\R)$ and $\lambda > 0$,
we obtain $u = \lambda [(F^2/2)'(\psi) + 4\pi r^2 G'(\psi)] = \lambda f(r,\psi)$
and the equilibrium problem reduces to
\begin{equation}
  \label{eq:Grad-Shafranov-equation}
  \left\{
  \begin{aligned}
    - \Delta^*\psi &= \lambda f(r,\psi), && \text{ in $\Omega$}, \\
    \psi &= 0, && \text{ on $\partial \Omega$},
  \end{aligned}
  \right.
\end{equation}
which is referred to as the Grad-Shafranov equation, and is the analog
of~(\ref{eq:Euler-semi-linear-eq}).  
(We remark that in realistic applications proper care must be taken to
assign physically meaningful values of $\chi$ and $p$ to different connected
components of the $\psi$-contours.) 
Equation (\ref{eq:Grad-Shafranov-equation}) is an ``eigenvalue problem'' for a 
semilinear elliptic equation and thus the same remarks about the well-posedness
of Eq.~(\ref{eq:Euler-semi-linear-eq}) are valid here. Solving this
eigenvalue problem is the standard way of computing axisymmetric MHD equilibria
in tokamaks \cite{Takeda1991, LoDestro1994, Pataki2013}. As for the Euler
equations, the Grad-Shafranov problem can be considered solved, and it is used
in this work as a benchmark problem.  

As for the case of reduced Euler equilibria, solutions of the axisymmetric
MHD equilibrium conditions~(\ref{eq:Grad-Shafranov-conditions}) can be
characterized by a variational principle of the
form~(\ref{eq:entropy-principle}).  
For a state variable we choose $u(r,z) = (4\pi/c) r J_\varphi(r,z)$ defined
over the bounded domain $\Omega \subset \R_+\times \R$, with $r>0$ on
$\ol{\Omega}$, and we consider the measure $d\mu = r^{-1} drdz$  on $\Omega$.
We assume that the profiles $F,G \in C^2(\R)$ are given so that
$f(r,y) \coloneqq (F^2/2)'(y) + 4\pi r^2 G'(y)$ satisfies
$\partial_y f(r,y) \not= 0$.
Therefore the map $y \mapsto f(r,y)$ is monotonic and has an inverse
$y \mapsto \partial_y s(r,y)$ for any fixed $r$. After integration in $y$
we find a function $s \in C^2(\R_+ \times \R)$, with
$\partial_y^2 s(r,y) \not= 0$ and such that
$\partial_y s(r,\cdot)^{-1} = f(r,\cdot)$. Then we consider the problem
\begin{equation}
  \label{eq:VP-GS}
  \min \{\fun{S}(u) \colon \fun{H}(u) = \fun{H}_0 \},
\end{equation}
with entropy and Hamiltonian
\begin{equation}
  \label{eq:GradShafranov-S-H}
  \fun{S}(u) = \int_\Omega s(r,u) d\mu, \quad
  \fun{H}(u) = \frac{1}{2} \int_\Omega |\nabla_{r,z}\psi|^2  d\mu,
\end{equation}
where $\psi$ is regarded as a function of $u$ given by the solution of the
linear elliptic problem
\begin{equation}
  \label{eq:GradShafranov-psi}
  -\Delta^* \psi = u, \quad
  \psi|_{\partial \Omega} = 0.
\end{equation}
Here $\fun{H}$ amounts to the magnetic energy stored in the
poloidal component $\nabla \psi \times \nabla \varphi$ of the magnetic field.

The functional derivatives are defined with respect to the $L^2$-product
with the measure $\mu$ on $\Omega$, so that
\begin{equation*}
  \frac{\delta \fun{S}(u)}{\delta u} = \partial_y s(r,u), \quad
  \frac{\delta \fun{H}(u)}{\delta u} = \psi.
\end{equation*}
Equation (\ref{eq:entropy-principle-2}) gives
\begin{equation}
  \label{eq:GradShafranov-complete-relaxation}
  \left\{
  \begin{aligned}
    & \partial_y s(r,u) - \lambda \psi = 0, \quad
    \fun{H}(u) = \fun{H}_0, \\
    & \text{with $\psi$ solution of~(\ref{eq:GradShafranov-psi}). }
  \end{aligned}
  \right.
\end{equation}
The set $\mathfrak{C}_{\fun{H}_0}$ of constrained critical points is then
  \begin{equation*}
    \mathfrak{C}_{\fun{H}_0} = \{ u \colon (u, \lambda)
    \text{ solves~(\ref{eq:GradShafranov-complete-relaxation}) for some
      $\lambda>0$}\},
  \end{equation*}
and among the elements of $\mathfrak{C}_{\fun{H}_0}$, those with minimum entropy
are solutions of the entropy principle~(\ref{eq:entropy-principle}).  

Under the assumption on $s(r,y)$, each element $u \in \mathfrak{C}_{\fun{H}_0}$
with $\lambda \not= 0$ corresponds to an axisymmetric equilibrium. In order to
see this, we have to find the fields $p$ and $\chi$ corresponding to $u$.
Given the profiles $F$ and $G$ from which $s$ has been derived, let us define
\begin{equation*}
  \chi = \sqrt{\lambda} F(\lambda \psi), \quad
  p = \lambda G(\lambda \psi),
\end{equation*}
together with the scaled variables $\tilde{u} = \lambda u$ and
$\tilde{\psi} = \lambda \psi$. One can check that $\tilde{u},\tilde{\psi}$
together with $\chi$ and $p$ given above solve
conditions~(\ref{eq:Grad-Shafranov-conditions}). We also have that
$\tilde{\psi}$ is a solution of the Grad-Shafranov equation. 

We remark that in general, the variational principle~(\ref{eq:VP-GS}) can be
formulated for generic profiles $s(r,y)$ and as long as
$\partial_y^2 s(r,y) \not= 0$,
we can still find $f(r,\cdot) = \partial_y s(r,\cdot)^{-1}$ and a
correspondence between~(\ref{eq:GradShafranov-complete-relaxation})
and~(\ref{eq:Grad-Shafranov-equation}). However, for general profiles, $f(r,y)$
cannot be written in terms of $F(y)$ and $G(y)$, since $f$ may not be a
quadratic function of $r$; therefore, one cannot always find $\chi$ and $p$
and thus an equilibrium in the sense of~(\ref{eq:Grad-Shafranov-conditions}). 

As in the case of the reduced Euler equations, for the application of
metriplectic dynamics to the solution of~(\ref{eq:VP-GS}) it  is  essential that
the orbits of the metriplectic system  completely relax to a constrained
entropy minimum.

\subsubsection{Beltrami fields}
\label{sec:beltrami-problem}

Let $\Omega \subset \R^3$ be a bounded simply connected domain, with
sufficiently regular boundary $\partial \Omega$, and let
$n\colon \partial \Omega \to \R^3$ be the outward unit normal to $\Omega$. A
vector field $B \colon \Omega \to \R^3$ is called a \emph{Beltrami field}
(also known as \emph{nonlinear} or \emph{weak} Beltrami field) if it satisfies
the Beltrami conditions 
\begin{equation*}
    (\curl B) \times B = 0, \quad \div B = 0,
  \end{equation*}
which, with the addition of the natural homogeneous boundary condition for a
divergence-free field, lead to \cite{Boulmezaoud2000, Amari2009}
\begin{equation}
  \label{eq:nonlinear-Beltrami}
  \left\{
  \begin{aligned}
    (\curl B) \times B = 0, \quad \div B &= 0, \quad &&\text{in } \Omega, \\
    n \cdot B &= 0, \quad &&\text{on } \partial \Omega.
  \end{aligned}
  \right.
\end{equation}
Solutions of~(\ref{eq:nonlinear-Beltrami}) satisfy the ideal MHD equilibrium
conditions~(\ref{eq:ideal-mhd}) with constant pressure, hence Beltrami fields
are force-free MHD equilibria \cite{Wiegelmann2012}. The force-free condition
can be equivalently rewritten as $\curl B = f B$ for a real scalar function $f$,
referred to as the proportionality factor. From the divergence-free condition,
$\div B = 0$, it follows that $B \cdot \nabla f = 0$, i.e., if $B$ is a smooth
nonlinear Beltrami field corresponding to a smooth non-constant scalar
multiplier $f$, then $f$ is a first integral of $B$. Enciso and Peralta-Salas
have shown that the Beltrami conditions constrain the field in such a way that
nontrivial solutions with a sufficiently regular proportionality factor $f$
exist only if $f$ satisfies a very restrictive condition, so that smooth
nonlinear Beltrami fields are ``rare'' \cite{Enciso2016}. On the other hand, one
can search for solutions with low regularity requirements, e.g.,
$B \in H^1(\Omega)^3$, the Sobolev space of fields $B : \Omega \to \R^3$ such
that $B \in L^2$ and $\nabla B \in L^2$ componentwise. With the nonhomogeneous
boundary condition $n \cdot B = g$ on $\partial \Omega$, an existence result has
been proven for $B \in H^1(\Omega)^3$ and $f \in L^\infty(\Omega)$ by using a
fixed-point argument \cite{Boulmezaoud2000}, but on a simply connected domain
this solution reduces to $B=0$ if $g=0$. For the purposes of this work, we are 
interested in a weaker formulation of the Beltrami condition, which will
become relevant in Section~\ref{sec:app-Beltrami}. For this formulation, we
need to introduce the Sobolev space $H_0(\div,\Omega)$ of $L^2$ vector fields
$w$ with $\div w$ in $L^2$ and $w \cdot n = 0$ on $\partial \Omega$. We also
need the space $H(\curl,\Omega)$ of $L^2$ vector fields $w$ with $\curl w$ in
$L^2$, and its subspace $H_0(\curl,\Omega)$ of vector fields
$w \in H(\curl,\Omega)$ satisfying the homogeneous boundary condition
$w \times n = 0$ on $\partial \Omega$. Specifically, we consider the problem
of finding $B \in H_0(\div,\Omega) \cap H(\curl,\Omega)$, such that
\begin{equation*}
  (\curl B) \times B = 0, \quad \div B = 0.
\end{equation*}
In this formulation the current $J = (4\pi/c) \curl B$ is only required to be
in $L^2$; thus, it can have singularities as long as they are
squared-integrable. Since the space $H_0(\div,\Omega) \cap H(\curl,\Omega)$ is
not convenient for the purposes of a finite element discretization, in our
numerical experiment in Section~\ref{sec:app-Beltrami}, we consider an even
weaker formulation: find  $B \in H_0(\div,\Omega)$, such that 
\begin{equation*}
  j \times H = 0, \quad \div B = 0,
\end{equation*}
where the current $j$ (different than $J= (4\pi/c) \curl B$)
and $H$ are the unique elements in
$H_0(\curl,\Omega)$ such that $(j, k)_{L^2} = (B, \curl k)_{L^2}$ and
$(H,G)_{L^2} = (B,G)_{L^2}$ for any $k,G \in H_0(\curl,\Omega)$.
This formulation, in principle allows for even
stronger current singularities. We are not aware of any existence result for
either of these two formulations of the Beltrami problem. While looking for
solutions with low regularity might appear physically obscure, there are two
reasons to consider them. The first is that existence of smooth equilibria is
an issue, and equilibria with singular currents are acceptable in some
applications as discussed in Section~\ref{ssec:Oequil} in the context of the
Grad's conjecture. The second reason is that these spaces of functions are
natural for modern numerical methods in MHD
\cite[and references therein]{Hu2017, Hu2021}. 
  
A special class of Beltrami fields is given by the solutions of the eigenvalue
problem for the $\curl$ operator: find $(B,\lambda)$, $\lambda \in \R$, such
that 
\begin{equation}
  \label{eq:linear-Beltrami}
  \left\{
  \begin{aligned}
    \curl B &= \lambda B, \quad &&\text{ in $\Omega$,} \\
    n \cdot B &= 0 &&\text{ on $\partial \Omega$.}
  \end{aligned}
  \right.
\end{equation}
The eigenfunctions of this problem will be referred to as \emph{linear} Beltrami
fields and they are necessarily divergence-free. Since there is a countable
family of eigenvalues of the $\curl$ operator, each corresponding to a
finite-dimensional space of eigenfunctions \cite{Yoshida1990a,Boulmezaoud1999},
both problem~(\ref{eq:nonlinear-Beltrami}) and~(\ref{eq:linear-Beltrami}) are
mathematically ill-posed due to nonuniqueness of the solution (although
for (\ref{eq:linear-Beltrami}), this is standard, since it is an eigenvalue
problem). 

Linear Beltrami fields can be characterized by a variational principle of the
form~(\ref{eq:entropy-principle}), which has been proposed by Woltjer
\cite{Woltjer1958a} and later applied to self-organized states in fusion plasmas
by Taylor \cite{Taylor1974}. This variational principle is also central in
multi-region relaxed MHD \cite{Dewar2008, Hudson2012}. Generalizations of
  Woltjer's variational principle have been proposed by Dixon and co-workers,
  including the free-boundary case \cite{Dixon1989}.

For the formulation of the variational principle, let $\fun{H}_0 \in \R$, and
let $V$ be the space of $L^2$ vector fields $u = B$ on $\Omega$ such
$\div B = 0$ and $n \cdot B = 0$ on $\partial\Omega$. Then we consider the
problem 
\begin{equation}
  \label{eq:linear-Beltrami-VP}
  \min \{\fun{S}(B) \colon \fun{H}(B) = \fun{H}_0\},
\end{equation}
with the entropy and Hamiltonian functions given by
\begin{equation}
  \label{eq:Beltrami-S-H}
  \fun{S}(B) = \frac{1}{2} \int_\Omega |B|^2 dx, \quad
  \fun{H}(B) = \frac{1}{2}H_m(B) \coloneqq \frac{1}{2} \int_\Omega A \cdot B dx,
\end{equation}
where $A$ is the vector potential for $B$ defined as the unique solution of the
problem 
\begin{equation}
  \label{eq:problem-A}
  \left\{
  \begin{aligned}
    \curl A &= B, && \text{ in $\Omega$,} \\
    \div A &= 0, && \text{ in $\Omega$,} \\
    n \times A &= 0, && \text{ on $\partial \Omega$}.
  \end{aligned}
  \right.
\end{equation}
The necessary condition for solutions of~(\ref{eq:linear-Beltrami-VP}) reads
\begin{equation}
  \label{eq:Beltrami-complete-relaxation}
  \left\{
  \begin{aligned}
    & B - \lambda A = 0, \quad \fun{H}(B) = \fun{H}_0, \\
    & \text{with $A$ given by~(\ref{eq:problem-A}),}
  \end{aligned}
  \right.
\end{equation}
where $\lambda \in \R$ is the Lagrange multiplier. The set of constrained
critical points of the entropy function is
\begin{equation*}
  \mathfrak{C}_{\fun{H}_0} = \{ B \colon (B,\lambda)
  \text{ solves~(\ref{eq:Beltrami-complete-relaxation}) for some }
  \lambda \in \R\}.
\end{equation*}
From~(\ref{eq:problem-A}) and~(\ref{eq:Beltrami-complete-relaxation}), one has
that each constrained critical point $B \in \mathfrak{C}_{\fun{H}_0}$ is a
linear Beltrami field since $\curl B = \lambda \curl A = \lambda B$,  and the
  same holds true for the corresponding potential, $\curl A = \lambda A$. 

The formal analysis of the variational problem proceeds as in the example of the
reduced Euler  equations with a linear profile $s'(y)$. 
Condition~(\ref{eq:Beltrami-complete-relaxation}) implies that $B \in V$
satisfies the eigenvalue problem~(\ref{eq:linear-Beltrami}). From the theory
of this eigenvalue problem \cite{Yoshida1990a}, we know that there is an
orthonormal basis of eigenfunctions corresponding to
the eigenvalue $\lambda =\lambda_j \in \R \setminus \{0\}$ labeled by
$j \in \N_0$ and with $k = 1, \ldots, d_j$ counting the multiplicity. 
As in the case of the Euler equations, the set $\mathfrak{C}_{\fun{H}_0}$
comprises all and only the vector fields $B_j = \lambda_j A_j$ with
 $A_j = \sum_k a_{j,k} A_{j,k}$, 
  $a_j = (a_{j,k})_k \in \R^{d_j}$, and that satisfy the constraint
$\fun{H}(B)=\fun{H}_0$, which is $2\fun{H}_0 \lambda_j = |a_j|^2$. Different 
from the case of the reduced Euler  equations, both $\fun{H}_0$ and the
eigenvalues of the $\curl$ operator can be negative. Since $|a_j|^2 > 0$, if
$\fun{H}_0 < 0$ (resp. $\fun{H}_0>0$), only the eigenfunctions $B_j$ with
$\lambda_j < 0$ (resp. $\lambda_j > 0$) belong to $\mathfrak{C}_{\fun{H}_0}$.
The entropy function restricted to $\mathfrak{C}_{\fun{H}_0}$ amounts to 
\begin{equation*}
  \fun{S}(B_j) = \fun{H}_0 \lambda_j > 0,
\end{equation*}
and it is always positive even if $\fun{H}_0$ and $\lambda_j$ can be negative.
It is possible to show that the minimum is attained \cite{Laurence1991}.
Hence the solutions of~(\ref{eq:linear-Beltrami-VP}) are the ground states for
the curl operator.  

The variational principle~(\ref{eq:linear-Beltrami-VP}) asks for 
divergence-free fields $B$ with minimum energy $(1/2)\|B\|_{L^2}^2 = \fun{S}(B)$
subject to the constraint that magnetic helicity $H_m(B) = 2 \fun{H}(B)$ is
held constant. Since magnetic helicity is a global constraint,
i.e. $\fun{H}(B)$ takes values in $\R$, $\lambda$ is a constant in $\R$ and
thus the variational principle selects \emph{linear} Beltrami fields. Magnetic
helicity $H_m$ is related to the topology of the field lines of $B$
\cite{Moffatt1969}. Specifically magnetic helicity is the average asymptotic
linking number of the field lines \cite{Arnold2014, Vogel2003}. Nonlinear
Beltrami fields~(\ref{eq:nonlinear-Beltrami}) on the other hand can be
obtained from a variational principle with a much stronger constraint, i.e.,
Beltrami fields are energy minima constrained to configurations of the field
$B$ that are continuous deformation of a prescribed field $B_0$. This
constraint preserves the topology of the field lines, and thus magnetic
helicity as well \cite{Kendall1960}. An overview of this variational principle
is given in \ref{sec:VP} for sake of completeness. 

In this work, we discuss metriplectic dynamical systems on the space $V$
of divergence-free fields. We address in particular convergence of an
orbit to \emph{nonlinear} Beltrami fields. This is an example 
for which complete relaxation of the orbit is not a desired property, since
complete relaxation would leads to linear Beltrami fields. We shall address
instead brackets that have a class of invariants much richer than just the
Hamiltonian.   

From a computational point of view, the direct numerical solution for
Beltrami fields is possible by a variety of techniques
\cite{Amari2009, Hudson2012, Malhotra2019}. Here, we shall not attempt to
compare the performance of these methods with that of the metriplectic
relaxation. Our aim is rather to study the convergence of metriplectic systems
on a physically relevant problem. We note however that relaxation methods for
force-free equilibria are common in several applications \cite{Wiegelmann2012}
and our study eventually aims at improving the rate of convergence of such
methods.

\section{On the relaxation of metriplectic systems}
\label{sec:remarks-relax-equil}

In this section, we present some remarks and mathematical results pertaining to
equilibrium points of metriplectic systems, their stability, and sufficient
conditions for convergence of nearby orbits. 
  
As described above, metriplectic dynamical systems dissipate entropy at constant
energy, cf.\ (\ref{eq:metric-system-properties}). It is therefore central to
ask whether a solution $u(t)$ of a metriplectic system, defined for
$t \in [0,+\infty)$, has a limit for $t \to +\infty$, and whether the limit,
when it exists, is a minimum of entropy on the surface of constant
Hamiltonian $\fun{H}(u) = \fun{H}_0 = \fun{H}(u_0)$, $u_0 = u(0)$ being the
initial state. This is the variational principle~(\ref{eq:entropy-principle}).
For some applications, e.g., the equilibrium problems introduced in
Sections~\ref{sec:Euler-problem} and~\ref{sec:GS-problem}, complete relaxation
of the orbit is essential. Recall, this means that the solution of the
metriplectic system has a limit as $t \to +\infty$ and the limit is a solution
of~(\ref{eq:entropy-principle}). In general for applications it is also useful
to know the rate of convergence to the limit.  
  
In this section we address implications of general metriplectic structure,
i.e.,  properties~(\ref{eq:metric-system-properties}), for the complete
relaxation of the orbit. A first implication is that, since $\fun{S}$ is
monotonically nonincreasing along an orbit, it is a candidate for a Lyapunov
function \cite{Hirsch2013, Henry1981, Temam1998}. This argument is standard for
nondegenerate gradient flows, but adaptation is needed for metriplectic systems,
which have degeneracy. In Section~\ref{sec:finite-dim-Lyapunov} we first recall
some standard arguments for Lyapunov stability, and then discuss their
adaptation to metriplectic systems. Then, in Section~\ref{sec:finite-dim-PL},
we consider another classical tool in the theory of nondegenerate gradient 
flows, the Polyak--{\L}ojasiewicz condition. We develop the
details in the finite-dimensional setting, but make some comments on extension
to infinite dimensions in Section~\ref{sec:infinite-dim}.

\subsection{Finite-dimensional systems: Lyapunov stability}
\label{sec:finite-dim-Lyapunov} 

Let us start by recalling the Lyapunov stability theorem for finite-dimensional
dynamical systems. Consider a generic vector field
$X\colon \mathcal{Z} \to \R^n$ on a domain $\mathcal{Z} \subseteq \R^n$,
$n \in \N$. We only assume that $X$  is locally Lipschitz continuous (locally
Lipschitz for short), that is, every point $z_0 \in \mathcal{Z}$ has a
neighborhood $\mathcal{U}_{z_0}$ where   
\begin{equation*}
  \big|X(z) - X(z')\big| \leq L_{z_0} \big|z-z'\big|, \quad
  z,z' \in \mathcal{U}_{z_0},
\end{equation*}
for a constant $L_{z_0} > 0$, possibly depending on $z_0$, thus we have 
existence and uniqueness for the ordinary differential equation
system $dz/dt = X(z)$. 

Let $z_* \in \mathcal{Z}$ be an equilibrium point, i.e., $X(z_*) = 0$.
A continuous function $\fun{L} \colon  \mathcal{O} \to \R$ defined on an open
subset $\mathcal{O} \subseteq \mathcal{Z}$ containing $z_*$, and differentiable
in $\mathcal{O} \setminus \{z_*\}$ is a Lyapunov function for $X$, if 
\begin{align*}
  \label{eq:L1}\tag{L1}
  X(z) \cdot \nabla \fun{L}(z) \leq 0, \quad
  &z \in \mathcal{O}\setminus \{z_*\};\\
  \label{eq:L2}\tag{L2}
  \text{$\fun{L}(z_*) = 0$ and $\fun{L}(z) > 0$,} \quad &z \not= z_*.
\end{align*}
Such an $\fun{L}$  is a \emph{strict} Lyapunov function for $X$, if it is a
Lyapunov function and satisfies 
\begin{equation*}
  \label{eq:L3}\tag{L3}
  X(z) \cdot \nabla \fun{L}(z) < 0, \quad z \in \mathcal{O}\setminus \{z_*\}.
\end{equation*}
Condition~(\ref{eq:L2}) in particular implies that $z_*$ is an \emph{isolated}
minimum of $\fun{L}$ in $\mathcal{O}$.

The Lyapunov stability theorem 
\cite{Hirsch2013, Marsden2001, Wiggins2003, Moretti2023} states that, if a
Lyapunov function exists, then $z_* \in \mathcal{O}$ is a \emph{stable}
equilibrium point. By definition this means that for any $\varepsilon > 0$
there is $\delta > 0$, such that $|z_0 - z_*| < \delta$ implies
$\big|z(t)-z_*\big| < \varepsilon$ for all $t \geq 0$, with $z(t)$ 
being an integral curve of $X$ with initial condition $z(0) = z_0$.
In addition, if the Lyapunov function is strict, $z_*$ is an
\emph{asymptotically stable} equilibrium point, that is, $z_*$ is a stable
equilibrium point, in the sense defined above, and $\delta>0$ can be chosen so
that $\lim_{t\to +\infty} z(t) = z_*$, for all initial conditions $z_0$
satisfying $|z_0-z_*| < \delta$. 

Although the  Lyapunov stability theorem is well known
\cite{Hirsch2013, Marsden2001, Wiggins2003, Moretti2023}, 
we recall the proof for sake of completeness, since the other results in
Section~\ref{sec:remarks-relax-equil} rely on the same ideas. The various
arguments available in the literature differ essentially only in the final
step in the proof of asymptotic stability. Here we follow 
Moretti \cite{Moretti2023}, which we find particularly clear. Recall that
$B_r(z) = \{z' \in \R^n \colon |z'-z| < r\}$ denotes the open ball of
radius $r > 0$ in $\R^n$.

\begin{theorem}[Lyapunov stability]
  \label{th:Lyapunov-finite-dim}
  Let $X \colon \mathcal{Z} \to \R^n$ be a locally Lipschitz vector field,
  $z_* \in \mathcal{Z}$ an equilibrium point of $X$,
  $\mathcal{O} \subseteq \mathcal{Z}$ an open subset containing $z_*$, and
  $\fun{L} \colon \mathcal{O} \to \R$ continuous in $\mathcal{O}$ and
  differentiable in $\mathcal{O} \setminus \{z_*\}$.
  \begin{itemize}
  \item[(i)] If $\fun{L}$ is a Lyapunov function for $X$, then $z_*$ is a
    stable equilibrium point. 
  \item[(ii)] If $\fun{L}$ is a strict Lyapunov function for $X$, then $z_*$
    is an asymptotically stable equilibrium point.  
  \end{itemize}
\end{theorem}

\begin{proof} 
  (i) Stability. Since $\mathcal{O}$ is open, we can choose
  $\varepsilon > 0$ so small that the ball $B_\varepsilon(z_*)$ is
  contained in the neighborhood $\mathcal{O}$. On the boundary
  $\partial B_\varepsilon$, $\fun{L}(z) > 0$ because of~(\ref{eq:L2}). Let
  $\fun{L}_\varepsilon \coloneqq
  \min_{z \in \partial B_\varepsilon(z_*)} \fun{L}(z)$.
  The minimum exists since $\partial B_\varepsilon(z_*)$ is compact.
  Let us now choose $\delta > 0$ so small that 
  $\fun{L}(z) < \fun{L}_\varepsilon$ for $z \in B_\delta(z_*)$.
  This is possible since $\fun{L}$ is continuous and $\fun{L}(z_*) = 0$.
  For any $z_0 \in B_\delta(z_*)$, let $z(t)$, $t \in [-\tau_0, +\tau_0]$ be
  an integral curve  of the vector field $X$ with initial condition
  $z(0) = z_0$. Assumption~(\ref{eq:L1}) implies that,
  \begin{equation*}
    \fun{L}\big(z(t)\big) \leq \fun{L}(z_0) < \fun{L}_\varepsilon,\;
    \text{for all $t \in [0, \tau_0]$.}
  \end{equation*}
  The function $t \mapsto |z(t) - z_*|$ is continuous; hence, if there is a
  time $t_o \in [0, \tau_0]$  at which $|z(t_o) - z_*| \geq \varepsilon$,
  the intermediate value theorem implies that there is $t_e \in (0,t_o]$ at
  which $|z(t_e) - z_*| = \varepsilon$ and thus
  $\fun{L}\big(z(t_e)\big) \geq \fun{L}_\varepsilon$, and this is a
  contradiction. Therefore $|z(t) - z_*| < \varepsilon$ for all
  $t \in [0, \tau_0]$. Since the integral curve stays in a bounded subdomain,
  the solution can be extended to the interval $[-\tau_0, +\infty)$ and
  $|z(t) - z_*| < \varepsilon$ for all $t \in [0, +\infty)$.
      
  (ii) Asymptotic stability. If $\fun{L}$ is a strict Lyapunov function, then
  in particular, it is a Lyapunov function and thus $z_*$ is a stable
  equilibrium point. We can choose $\varepsilon_0 > 0$ and a $\delta_0 > 0$ such
  that any integral curve $z(t)$ with $z(0) \in B_{\delta_0}(z_*)$ stays in
  $B_{\varepsilon_0}(z_*)$. 

  We want to show that for any $z_0 \in B_{\delta_0}(z_*)$ and for any
  $\varepsilon \in (0, \varepsilon_0)$ there is a time $T_\varepsilon > 0$ such
  that the integral curve $z(t)$ with initial condition $z(0) = z_0$ satisfies
  $|z(t) - z_*| < \varepsilon$ for all $t > T_\varepsilon$.

  Uniqueness of the orbit passing through a given point implies that,
  if $z_0 \not= z_*$, then $|z(t) - z_*| > 0$ (since $z(t) = z_*$ is a
  solution) and thus $d \fun{L}\big(z(t)\big) / dt < 0$ for all $t > 0$.  
  Therefore function $t \mapsto \fun{L}\big(z(t)\big)$ is strictly
  monotonically decreasing, and it is bounded form below, hence it has a limit
  \begin{equation*}
    \fun{L}\big(z(t)\big) \to \ell = \inf_{t\geq 0} \fun{L}\big(z(t)\big)
    \geq 0,
  \end{equation*}
  The limit must be $\ell = 0$. If not, $\fun{L}\big(z(t)\big) \geq \ell > 0$
  for all $t > 0$, and continuity of $\fun{L}$ implies that there a radius
  $r \in (0, \varepsilon_0)$ such that $|z(t) - z_*| > r$. This leads to a
  contradiction, since,  if $z(t)$ stays in the compact region
  $\mathscr{R} = \{z \colon r \leq |z-z_*| \leq \varepsilon\}$ for all
  $t \geq 0$, then with
  $-M = \max_{z \in \mathscr{R}} X(z) \cdot \nabla \fun{L}(z)<0$, 
  \begin{equation*}
    \fun{L}\big(z(t)\big) = \fun{L}(z_0) + \int_0^t
    X\big(z(s)\big) \cdot \nabla \fun{L}\big(z(s)\big) ds
    \leq \fun{L}(z_0) - Mt.
  \end{equation*}
  For $t > \fun{L}(z_0)/M > 0$, we have $\fun{L}\big(z(t)\big) < 0$, which is
  impossible. Hence, the limit must be $\ell = 0$, that is, for every
  $\lambda > 0$ there is $T_\lambda > 0$ such that
  $\fun{L}\big(z(t)\big) < \lambda$ for $t > T_\lambda$. 
  
  We claim that for any $\varepsilon > 0$ we can find
  $\lambda = \lambda_\varepsilon>0$ such that
  \begin{equation*}
    A_\lambda \coloneqq \{ z \in B_{\varepsilon_0}(z_*) \colon
    \fun{L}(z) < \lambda \},
  \end{equation*}
  is contained in the ball $B_{\varepsilon}(z_*)$, i.e.
  $A_\lambda \subset B_\varepsilon(z_*)$. If this is the case, corresponding
  to $\lambda_\varepsilon$, there is a time $T_\varepsilon>0$ such that,
  for all $t > T_\varepsilon$,
  $z(t) \in A_{\lambda_\varepsilon} \subset B_\varepsilon(z_*)$, which is the
  thesis. Therefore it remains to prove the claim. Let us assume that the
  claim is false, i.e. there is a value $\varepsilon_* > 0$ such that, for all
  $\lambda$, there is at least one point $\tilde{z}_\lambda$ that satisfies
  the conditions $\fun{L}(\tilde{z}_\lambda) < \lambda$ and 
  $\varepsilon_* \leq |\tilde{z}_\lambda - z_*| \leq \varepsilon_0$.
  Upon choosing $\lambda = 1/n$ for $n \in \N$ we obtain a sequence
  $z_n = \tilde{z}_{1/n}$, which belongs to a compact set. Hence there is a
  converging subsequence $z_{n_k} \to \tilde{z}_*$, and
  $\varepsilon_* \leq |\tilde{z}_* - z_*| \leq \varepsilon_0$ so that
  necessarily $\fun{L}(\tilde{z}_*) > 0$. On the other hand, we have
  $\fun{L}(z_{n_k}) < 1 / n_k \to 0$ and, by continuity of $\fun{L}$,
  $\fun{L}(\tilde{z}_*) = 0$, which is a contradiction.  
\end{proof}

We want to apply Theorem~\ref{th:Lyapunov-finite-dim} to metriplectic vector
fields of the following form:  
\begin{equation}
  \label{eq:X-metripl}
  X(z) = J(z)\nabla \fun{H}(z) - K(z) \nabla \fun{S}(z).
\end{equation}
For comparison, we also address the case of a standard nondegenerate gradient
flow
\begin{equation}
  \label{eq:X-grad-flow}
  X(z) = - \nabla \fun{S}(z),
\end{equation}
with entropy $\fun{S} \in C^2(\mathcal{Z})$.

First, recall that in the case of nondegenerate gradient
flows~(\ref{eq:X-grad-flow}), with $\fun{S} \in C^2(\mathcal{Z})$,
the function $\fun{L}(z)=\fun{S}(z)-\fun{S}(z_*)$ is a strict Lyapunov function
in a neighborhood of any \emph{isolated}, local minimum $z_*$ of $\fun{S}$
\cite[Proposition~15.0.2]{Wiggins2003}. 
Therefore, the Lyapunov stability theorem implies that any
integral curve of the gradient flow with initial condition near an isolated
entropy minimum converges to that entropy minimum for $t \to + \infty$
(asymptotic stability). This result is a \emph{local version} of the property
we called complete relaxation.

Now, consider the case of a metriplectic vector field~(\ref{eq:X-metripl}).
If $z_* \in \mathcal{Z}$ is an equilibrium point, the function
$\fun{L}(z) = \fun{S}(z) - \fun{S}(z_*)$ satisfies condition~(\ref{eq:L1}),
because of the general properties of metriplectic systems, 
cf.~(\ref{eq:metric-system-properties}). If the entropy function
$\fun{S}$ has an isolated minimum at $z_* \in \mathcal{Z}$,  $\fun{L}(z)$ is a
Lyapunov function of the system in a neighborhood $\mathcal{O}$ of $z_*$.
Therefore $z_*$ is a stable equilibrium. However, an orbit $z(t)$ can converge
to $z_*$ only if the initial condition $z_0 = z(0)$ and the local entropy
minimum $z_*$ belong to the same energy isosurface, i.e., 
$\fun{H}(z_0) = \fun{H}(z_*)$. This is a necessary condition that follows from
the continuity of $\fun{H}$: if $z(t) \to z_*$, passing to the limit in
$\fun{H}(z_0) = \fun{H}(z(t))$ yields $\fun{H}(z_0) = \fun{H}(z_*)$. This
observation leads to the following conclusion: 

\begin{proposition}
  \label{th:no-strict-Lyapunov}
  Let $X \colon  \mathcal{Z} \to \R^n$ be a locally Lipschitz vector field on a
  domain $\mathcal{Z} \subseteq \R^n$, $z_* \in \mathcal{Z}$ an equilibrium
  point of $X$, and $\fun{H}, \fun{S} \in C^1(\mathcal{Z})$ such that
  \begin{enumerate}
  \item $X(z) \cdot \nabla \fun{H}(z) = 0$ and
    $X(z) \cdot \nabla \fun{S}(z) \leq 0$, $z \in \mathcal{Z}$;
  \item $\nabla \fun{H}(z_*) \not= 0$.
  \end{enumerate}
  Then, if $\fun{L} = \fun{S} - \fun{S}(z_*)$ satisfies~(\ref{eq:L2}) in a
  neighborhood $\mathcal{O} \subseteq \mathcal{Z}$ of $z_*$, there is at least
  one point $z' \in \mathcal{O} \setminus \{z_*\}$ such that 
  $X(z') \cdot \nabla \fun{S}(z') = 0$. 
\end{proposition}

\begin{proof} 
  By contradiction, let us assume that $X \cdot \nabla \fun{S} < 0$ in
  $\mathcal{O} \setminus \{z_*\}$. Then hypothesis 1 implies that
  $\fun{L}(z) = \fun{S}(z) - \fun{S}(z_*)$ satisfies the conditions for a
  strict Lyapunov function for $X$. For the Lyapunov stability theorem, there
  exists $\delta' >0$ such that for any $z_0$ with $|z_0-z_*| < \delta'$
  the orbit $z(t)$ of the dynamical system $dz/dt = X(z)$ with initial
  condition $z(0) = z_0$ exists for all $t \geq 0$ and  $z(t) \to z_*$ as
  $t \to +\infty$. Hypothesis 1 also implies that $\fun{H}$ is a constant of
  motion and it is continuous, therefore $\fun{H}(z_0) = \fun{H}(z_*)$.
    
  From hypothesis 2 and the continuity of the derivative $\nabla \fun{H}$ we
  can find a ball of radius $\delta'' > 0$ around $z_*$ where
  $\nabla \fun{H} \not= 0$.

  We choose $\delta < \min \{\delta',\delta''\}$. In the ball of radius
  $\delta$ centered at $z_*$ there is at least one point $z_0$ such that
  $\fun{H}(z_0) \not= \fun{H}(z_*)$. If not, then $\fun{H}$ is constant in the
  ball, and this is not possible since $\nabla \fun{H} \not= 0$. On the other
  hand, since $|z_0-z_*| < \delta < \delta'$ we must have
  $\fun{H}(z_0)=\fun{H}(z_*)$, which is a contradiction.
\end{proof}

For metriplectic systems, hypothesis 1 is verified,
cf.~(\ref{eq:metric-system-properties}). Hypothesis 2 holds away from critical
points of $\fun{H}$, which are usually isolated. Therefore, this proposition
implies that the entropy $\fun{S}$ of a metriplectic system in most cases
(specifically under hypothesis 2) cannot be used as a strict Lyapunov function,
since it is not strictly decaying everywhere in $\mathcal{O} \setminus \{z_*\}$.
Without a strict Lyapunov function, asymptotic stability of local entropy minima
does not follows directly from the usual Lyapunov theorem. This is in stark
contrast with the case of pure gradient flows~(\ref{eq:X-grad-flow}) discussed
above. 

Asymptotic stability with Lyapunov functions that are not strictly
dissipated have been addressed by De Salle and Lefschetz \cite{LaSalle1961},
who showed that any integral curve, defined for $t\geq 0$, in a bounded strict
sublevel set of the Lyapunov function must approach the largest invariant set
contained in the region where $X(z) \cdot \nabla \fun{L}(z) = 0$. More specific
results for systems with a conserved energy and a dissipated entropy were
considered  by Beretta  \cite{Beretta1986} with applications to quantum
thermodynamics. Here we apply similar ideas to metriplectic systems. 
Suppose a locally Lipschitz vector field $X \colon  \mathcal{Z} \to \R^n$ on a
domain $\mathcal{Z} \subseteq \R^n$ has $k$ constants of motion 
$\fun{I}^1, \ldots, \fun{I}^k \in C^\infty(\mathcal{Z})$, for $1 \leq k < n$,
which means $X(z) \cdot \nabla \fun{I}^\alpha(z) = 0$,
for $\alpha \in \{1,\ldots,k\}$. The functions $\fun{I}^\alpha$ are independent
at $z \in \mathcal{Z}$ if  
\begin{equation}
  \label{eq:lin-indep-C}
  \nabla \fun{I}^1(z) ,\ldots, \nabla \fun{I}^k(z)
  \text{ are linearly independent in $\R^k$.}
\end{equation}
It is convenient to define the function
$\fun{I}\coloneqq\big(\fun{I}^1,\ldots,\fun{I}^k\big)\in
C^\infty(\mathcal{Z},\R^k)$.
Then~(\ref{eq:lin-indep-C}) is equivalent to $\rank \nabla \fun{I}(z) = k$. 
  
If the constants of motion are independent at an equilibrium point
$z_*\in \mathcal{Z}$, where  $X(z_*) = 0$, then the local submersion theorem
\cite[Theorem~2.5.13]{Marsden2001} allows us to find a neighborhood 
$\mathcal{U}$ of $z_*$ where the level sets 
\begin{equation*}
  \mathcal{U}_\eta \coloneqq \{z \in \mathcal{U}  \colon \fun{I}(z) = \eta\},
\end{equation*}
are closed submanifolds of $\mathcal{U}$, for $\eta$ is a neighborhood of
$\eta_* \coloneqq \fun{I}(z_*)$. Since $\fun{I}^\alpha$ are constants of
motion, $X$ is tangent to $\mathcal{U}_\eta$, so that the dynamical system can
be reduced locally to each submanifold $\mathcal{U}_\eta$. This leads to the
following straightforward result, which is essentially a finite-dimensional
version of Theorem~2 in \cite{Beretta1986}.   

\begin{proposition}
  \label{th:Lyapunov-constrained}
  Let $X \colon  \mathcal{Z} \to \R^n$ be a locally Lipschitz vector
  field on a domain $\mathcal{Z} \subseteq \R^n$,
  $\fun{I} = (\fun{I}^1, \ldots, \fun{I}^k) \in C^\infty(\mathcal{Z},\R^k)$
  be constants of motion with $1 \leq k < n$, $z_* \in \mathcal{Z}$  be an
  equilibrium point of $X$, $\rank \nabla \fun{I}(z_*) = k$,
  $\eta_* \coloneqq \fun{I}(z_*)$, and let 
  $\mathcal{U}_{\eta_*} = \{z \in \mathcal{U}  \colon \fun{I}(z) = \eta_*\}$,
  where $\mathcal{U}$ is the neighborhood of $z_*$ given by the local
  submersion theorem. If there is $\fun{L} \in C^1(\mathcal{U})$ satisfying
  \begin{align*}
    \label{eq:CL1}\tag{L1${}^\prime$}
    & X(z) \cdot \nabla \fun{L}(z) \leq 0, 
    && z \in \mathcal{U},\\
    \label{eq:CL2}\tag{L2${}^\prime$}
    & \text{$\fun{L}(z_*) = 0$ and $\fun{L}(z) > 0$,} 
    && z \in \mathcal{U}_{\eta_*} \setminus \{z_*\}, 
  \end{align*}
  then for any sufficiently small $\varepsilon > 0$, there is a $\delta > 0$
  such that the integral curve  $z(t)$ of $X$ with initial condition
  $z(0) = z_0 \in B_\delta(z_*) \cap \mathcal{U}_{\eta_*}$ is defined for all
  $t \geq 0$, and $z(t) \in B_{\varepsilon}(z_*) \cap \mathcal{U}_{\eta_*}$.
  If in addition 
  \begin{equation*}
    \label{eq:CL3}\tag{L3${}^\prime$}
    X(z) \cdot \nabla \fun{L}(z) < 0, \quad
    z \in \mathcal{U}_{\eta_*} \setminus \{z_*\},
  \end{equation*}    
  then $z(t) \to z_*$ for $t \to + \infty$.
\end{proposition}

\begin{proof} 
  Since $\rank \nabla \fun{I}(z_*) = k$, the local
  submersion theorem \cite{Marsden2001} allows us to find open subsets
  $\mathcal{N} \subseteq \fun{I}(\mathcal{Z})$ and
  $\mathcal{Z}^\prime \subseteq \R^{n-k} \cong \ker \nabla \fun{I}(z_*)$,
  together with local coordinates
  $(\eta,\zeta) \in \mathcal{N} \times \mathcal{Z}^\prime$, defined 
  in an open, connected subset $\mathcal{U} \subseteq \mathcal{Z}$
  containing the point $z_*$, such that the inverse coordinate map
  $\varphi \colon \mathcal{N} \times \mathcal{Z}^\prime \to \mathcal{U}$ is a
  $C^\infty$-diffeomorphism and satisfies
  \begin{equation*}
    \fun{I}\big(\varphi(\eta,\zeta)\big) = \eta.
  \end{equation*}
  Therefore, the sets
  $\mathcal{U}_{\eta} = \{z \in \mathcal{U} \colon \fun{I}(z)=\eta \}$,
  with $\eta \in \mathcal{N}$, are closed submanifolds parameterized by
  $\zeta \mapsto \varphi(\eta,\zeta)$. Let us denote by
  $(\eta_*, \zeta_*) \in \mathcal{N} \times \mathcal{Z}^\prime$ the point
  corresponding to $z_* = \varphi(\eta_*,\zeta_*)$.
  
  In this local coordinate system the integral curves of the vector field $X$
  in $\mathcal{U}$ solve 
  \begin{equation*}
    \frac{d\eta}{dt} = X \cdot \nabla \eta = 0, \quad
    \frac{d\zeta}{dt} = X_\eta (\zeta)
    \coloneqq (X \cdot \nabla \zeta) \circ \varphi(\eta,\zeta),
  \end{equation*}
  since $X \cdot \nabla \eta = X \cdot \nabla \fun{I} = 0$.
  The field $X_\eta$ defines a tangent vector on $\mathcal{U}_\eta$ for any
  $\eta \in \mathcal{N}$. Since $\nabla \zeta$ is $C^\infty$, one can check
  that $X_\eta$ is a locally Lipschitz continuous function of $\zeta$. Then,
  with $\fun{L}_\eta \coloneqq \fun{L} \circ \varphi(\eta,\cdot)$, one has
  \begin{equation*}
    (X \cdot \nabla \fun{L}) \circ \varphi 
    = X_\eta^a \frac{\partial \fun{L}_\eta }{\partial \zeta^a},
  \end{equation*}
  where the sum over $a \in \{1, \ldots,n-k\}$ is implied. Hence
  hypotheses~(\ref{eq:CL1})-(\ref{eq:CL3}) are equivalent to Lyapunov
  conditions~(\ref{eq:L1})-(\ref{eq:L3}) for the function $\fun{L}_{\eta_*}$
  and for the dynamical system $d\zeta /dt = X_{\eta_*}$ on
  $\mathcal{U}_{\eta_*}$, and the claim follows from
  Theorem~\ref{th:Lyapunov-finite-dim}.  
\end{proof}

Proposition~\ref{th:Lyapunov-constrained} establishes stability and asymptotic
stability for orbits with initial conditions  in $\mathcal{U}_{\eta_*}$, which
being  a lower dimensional set has zero Lebesgue measure in the phase space.
This  issue was addressed in \cite{Beretta1986} by assuming  
$z_*$ is part of a continuous family of equilibria.

At last, we address metriplectic vector fields. We know that there is at least
one constant of motion, the Hamiltonian $\fun{H}$. More generally, let $X$ be a
metriplectic vector field with $k$ constants of motion
$\fun{I}^1,\ldots, \fun{I}^k$, $1 \leq k < n$, and for definiteness
$\fun{I}^k = \fun{H}$. Unlike the case of
Proposition~\ref{th:Lyapunov-constrained}, we assume
that~(\ref{eq:lin-indep-C}) holds in a whole subdomain $\mathcal{U}$. 
Under these conditions, the submersion theorem
\cite[Theorem~3.5.4]{Marsden2001} establishes that the set
\begin{equation}
  \label{eq:Ueta}
  \mathcal{U}_\eta \coloneqq \{z \in \mathcal{U} \colon \fun{I}(z) = \eta\},
\end{equation}
is a closed submanifold of $\mathcal{U}$ for any
$\eta \in \fun{I}(\mathcal{U})$, and the restriction of the entropy to such
submanifolds, i.e., $\fun{S}|_{\mathcal{U}_\eta}$, is smooth.

We make the following hypothesis on the entropy in $\mathcal{U}$: 
\begin{equation}
  \label{eq:U}
  \mathcal{U}    \   \text{is  bounded,}\  
  \ol{\mathcal{U}} \subset \mathcal{Z},  \ 
  \exists \, z_m \in \mathcal{U}\,,  \text{ where}\  \fun{S}(z_m) <
  \inf \{\fun{S}(z) \colon z \in \partial \mathcal{U}\}.
\end{equation}
This ensures that there is a minimum of the entropy in the interior of the
subdomain $\mathcal{U}$.
  
\begin{lemma}
  \label{th:lemma}
  Let $X\colon  \mathcal{Z} \to \R^d$ be a vector field on a domain
  $\mathcal{Z} \subseteq \R^d$, and let $\mathcal{U} \subset \mathcal{Z}$
  and $\fun{S} \in C^1(\mathcal{Z})$ satisfy~(\ref{eq:U}). If in the subdomain 
  $\mathcal{U}$, $X$ is locally Lipschitz and $X \cdot \nabla \fun{S} \leq 0$,
  then there is an open, non-empty subset $\mathcal{O} \subseteq \mathcal{U}$
  such that, for any $z_0 \in \mathcal{O}$ the integral curve $z(t)$ of $X$
  with initial condition $z(0) = z_0$ can be prolonged to the interval
  $t \in [0,+\infty)$ and $z(t) \in \mathcal{O}$ for all $t\geq 0$. 
\end{lemma}

\begin{proof}
  Per hypotheses $\fun{S} \in C^1(\mathcal{Z})$, therefore
  $\fun{S} \in C^1(\ol{\mathcal{U}})$, and we have 
  $\fun{S}_b\coloneqq\inf_{\partial\mathcal{U}}\fun{S}>\fun{S}(z_m)>-\infty$.
  The set 
  \begin{equation*}
    \mathcal{O} \coloneqq \{z \in \mathcal{U} \colon
    \fun{S}(z) < \fun{S}_b \}
  \end{equation*}
  is non-empty, since $z_m \in \mathcal{O}$, and open, since $\fun{S}(z)$ is
  continuous.
  Let $z: (\tau_1, \tau_2) \to \mathcal{U}$, $\tau_1 < 0 < \tau_2$, be the
  maximal solution of the initial value problem
  \begin{equation*}
    dz/dt = X(z), \quad z(0) = z_0 \in \mathcal{O}.
  \end{equation*}
  The solution exists given that $X$ is locally Lipschitz continuous in
  $\mathcal{U}$. Since $X \cdot \nabla \mathcal{S} \leq 0$, the function
  $t \mapsto \mathcal{S}\big(z(t)\big)$ is $C^1$ and non-increasing,
  hence $\fun{S}\big(z(t)\big) \leq \fun{S}(z_0) < \fun{S}_b$, and thus
  $z(t) \in \mathcal{O}$ for all $t \in [0,\tau_2)$.
    
  We show that $\tau_2 = +\infty$. With this aim we rely on a standard
  argument from the theory of ordinary differential equations, which we report
  in full for sake of completeness. If $\tau_2$ is finite, let
  $\{t_n\}_{n\in\N}$ be a sequence, $t_n \in (\tau_1,\tau_2)$, and
  $t_n \to \tau_2$ as $n \to +\infty$. Then 
  \begin{equation*}
    \big|z(t_n) - z(t_m)\big| \leq \max_{z \in \ol{U}}|X(z)| |t_n-t_m|,
  \end{equation*}
  for all $m,n \in \N$. Since $\{t_n\}$ is a convergent sequence,
  $\{z(t_n)\} \subset \mathcal{O}$ is a Cauchy sequence and must have a limit
  $\bar{z} \in \ol{\mathcal{O}}$ as $n \to \infty$. Passing to the limit in
  the inequality $\fun{S}\big(z(t_n)\big) \leq \fun{S}(z_0) < \fun{S}_b$
  yields $\fun{S}(\bar{z}) \leq \fun{S}(z_0) < \fun{S}_b$, hence
  $\bar{z} \in \mathcal{O}$. We can use the limit point $\bar{z}$ as an
  initial condition, and we can extend the solution $z(t)$ beyond $\tau_2$
  contradicting the fact that $z(t)$ is the maximal solution.
  Therefore $\tau_2=+\infty$ as claimed.
\end{proof}

The set $\mathcal{O}$ is given by
\begin{equation*}
  \mathcal{O} = \big\{ z \in \mathcal{U} \colon \;
  \fun{S}(z) < \inf_{\partial \mathcal{U}} \fun{S} \big\}.
\end{equation*}
This lemma shows that, if $\fun{S}$ and $\mathcal{U}$ satisfy~(\ref{eq:U}) and
$\fun{S}$ is non-increasing along the integral curves of a vector field $X$,
then we can find a positively invariant subset $\mathcal{O}$. This lemma can
be applied directly to both standard gradient flows and metriplectic systems.

\begin{proposition}
  \label{th:mod-Lyapunov}
  Let $X = J \nabla \fun{H} - K \nabla \fun{S}$ be a metriplectic vector field
  on a domain $\mathcal{Z} \subseteq \R^n$, with $1 \leq k < n$ constants of
  motion
  $\fun{I}=(\fun{I}^1, \ldots, \fun{I}^k) \in C^\infty(\mathcal{Z},\R^k)$,
  $\rank \nabla \fun{I} = k$ in a subdomain $\mathcal{U}$
  satisfying~(\ref{eq:U}), and let $\mathcal{O}$ be the open set given by 
  Lemma~\ref{th:lemma}. If in addition it holds that
  \begin{align*}
    \label{eq:MpL2}\tag{L2${}^{\prime\prime}$}
    &\begin{aligned}
       & \text{$\forall \eta \in \fun{I}(\mathcal{O})$, on  
         $\,\mathcal{U}_\eta \coloneqq
         \{z \in \mathcal{U} \colon \fun{I}(z) = \eta \}$,} \ 
       \text{$\fun{S}|_{\mathcal{U}_\eta}$ has a unique critical point
         $z_\eta$,}
       \\&
       \text{which is a  strict local minimum,}  
     \end{aligned}
    \intertext{and}
    \label{eq:MpL3}\tag{L3${}^{\prime\prime}$}
    & \ker K = \spn\{\nabla \fun{I}^1,\ldots, \nabla\fun{I}^k\}
    \; \text{ in }\; \ol{\mathcal{U}}.
  \end{align*}
  then for any integral curve $z(t)$ of $X$ with $z(0) = z_0 \in \mathcal{O}$,
  $\lim_{t \to +\infty} z(t) = z_\eta \in \mathcal{O}$,
  where $\eta = \fun{I}(z_0)$.
\end{proposition}

\begin{proof}
  First we show that $z_\eta \in \mathcal{U}_\eta \cap \mathcal{O}$.
  Given that $\eta \in \fun{I}(\mathcal{O})$, the set
  $\mathcal{U}_\eta \cap \mathcal{O}$ is non-empty.
  If $z_\eta \not \in \mathcal{U}_\eta \cap \mathcal{O}$,
  there cannot be any extremum of $\fun{S}|_{\mathcal{U}_\eta}$ in
  $\mathcal{U}_\eta \cap \mathcal{O}$ since $z_\eta$ is the only critical point
  in $\mathcal{U}_\eta$. This implies that both the minimum and the maximum of
  $\fun{S}$ restricted to $\mathcal{U}_\eta \cap \ol{\mathcal{O}}$ are attained
  on the boundary $\mathcal{U}_\eta \cap \partial \mathcal{O}$, where we have
  $\fun{S} = \fun{S}_b = \inf\{\fun{S}(z) \colon z \in \partial \mathcal{U}\}$.
  We deduce that $S|_{\mathcal{U}_\eta}$ is constant and equal to $\fun{S}_b$
  in $\mathcal{U}_\eta \cap \ol{\mathcal{O}}$, but per definition, points in
  $\mathcal{O}$ satisfy $\fun{S} < \fun{S}_b$. Therefore, the only critical
  point $z_\eta$ of $\fun{S}|_{\mathcal{U}_\eta}$ must be in
  $\mathcal{U}_\eta \cap \mathcal{O}$.

  The function $\fun{S}|_{\mathcal{U}_\eta}$ is smooth and its critical points
  necessarily satisfy the Lagrange condition 
  \begin{equation}
    \label{eq:constrained-critical-points}
    \nabla \fun{S}(z) = \sum_{\alpha =1}^k \lambda_\alpha \nabla
    \fun{I}^\alpha(z), \quad \fun{I}(z) = \eta,
  \end{equation}
  for $\lambda_\alpha \in \R$ \cite[Theorem~3.5.27]{Marsden2001}.
  Hypothesis~(\ref{eq:MpL2}), states, in particular, that this can only
  happen at the point $z_\eta$. 
  
  Condition~(\ref{eq:MpL2}) also requires $z_\eta$ to be a strict local
  minimum of $\fun{S}|_{\mathcal{U}_{\eta}}$. Then, $z_\eta$ must be an
  equilibrium point of $X$. In order to see this, let us consider the integral
  line $\hat{z}_\eta(t) \in \mathcal{O}$ of the vector field $X$ with initial
  condition $\hat{z}_\eta(0) = z_\eta$. Given that $z_\eta \in \mathcal{O}$,
  $\hat{z}_\eta(t)$ is defined for $t \geq 0$ (Lemma~\ref{th:lemma}).
  Since $\fun{I}$ is a constant of motion,
  $\hat{z}_\eta(t) \in \mathcal{U}_\eta$, and continuity implies that for any
  $\varepsilon > 0$, there is $\delta >0$ such that $|t| \leq \delta$ implies
  $|\hat{z}_\eta(t) - z_\eta| < \varepsilon$. If $\hat{z}_\eta$ is not
  identically equal to $z_\eta$ when $t \in [0,\delta)$, then at some
  $t' \in (0, \delta)$, $\hat{z}_\eta(t') \not = z_\eta$ and
  $\fun{S}\big(\hat{z}_\eta(t')\big) > \fun{S}(z_\eta)$ since
  $z_\eta$ is an entropy local minimum. On the other hand, dissipation of
  entropy requires
  $\fun{S}\big(\hat{z}_\eta(t')\big) \leq \fun{S}\big(\hat{z}_\eta(0)\big)
  = \fun{S}(z_\eta)$.
  Therefore, it must be $\hat{z}_\eta(t) = z_\eta$ identically for 
  $t \in [0,\delta)$, which implies $X(z_\eta) = 0$ and $z_\eta$ is an
  equilibrium. 
    
  Given $z_0 \in \mathcal{O}$, the integral curve of $X$ with initial condition
  $z(0) = z_0$ exists for all $t \geq 0$ (Lemma~\ref{th:lemma}).
  Let us define the function
  \begin{equation*}
    \fun{L}(z) = \fun{S}(z) - \fun{S}(z_\eta),
  \end{equation*}
  where $\eta = \fun{I}(z_0)$. Then $\nabla\fun{L}(z) = \nabla\fun{S}(z)$ and
  \begin{equation}
    \label{eq:L-decay}
    X(z) \cdot \nabla \fun{L}(z) =
    - \nabla \fun{S}(z) \cdot K(z) \fun{S}(z) \leq 0,
  \end{equation}
  so that $\fun{L}$ satisfies condition~(\ref{eq:CL1}) of
  Proposition~\ref{th:Lyapunov-constrained} in the open set $\mathcal{O}$
  defined above.
  
  Since $z_\eta$ is a local minimum of $\fun{S}|_{\mathcal{U}_\eta}$ there is
  a neighborhood of $z_\eta$ on $\mathcal{U}_\eta$ where
  $\fun{S}(z)>\fun{S}(z_\eta)$ for $z \not = z_\eta$.
  But by definition of the submanifold topology,
  this neighborhood has the form $\mathcal{U}_\eta \cap \mathcal{V}$
  with $\mathcal{V}$ an open subset of $\mathcal{U}$. Then,
  condition~(\ref{eq:CL2}) holds true with $\mathcal{U}$ replaced by
  $\mathcal{V}$.   
  
  At last we show that condition~(\ref{eq:CL3}) is also satisfied. If not, one
  could find a point $\tilde{z}_\eta \in \mathcal{U}_\eta$ such that
  $\tilde{z}_\eta \not = z_\eta$ and 
  $\nabla \fun{S}(\tilde{z}_\eta) \in \ker K(\tilde{z}_\eta)$, but then
  assumption~(\ref{eq:MpL3}) implies that $\tilde{z}_\eta$ satisfies the
  Lagrange conditions~(\ref{eq:constrained-critical-points}) and this
  contradicts the uniqueness of the critical point $z_\eta$.
  Therefore, inequality~(\ref{eq:L-decay}) must be strict, which is
  condition~(\ref{eq:CL3}).  
        
  In summary, the point $z_* = z_\eta$, $\eta = \fun{I}(z_0)$, and the
  function $\fun{L}(z) = \fun{S}(z) - \fun{S}(z_\eta)$ on the neighborhood
  $\mathcal{V}$ satisfy all the hypotheses of
  Proposition~\ref{th:Lyapunov-constrained}, including~(\ref{eq:CL3}).
  Therefore, there is a $\delta>0$ such that any integral curve
  of $X$ with initial condition in $\mathcal{U}_\eta \cap B_\delta(z_\eta)$
  converges to $z_\eta$ as $t \to +\infty$.
    
  We claim that, for any integral curve, $z(t) \in \mathcal{O}$, $t \geq 0$ and
  $z(0) = z_0 \in \mathcal{U}_\eta \cap \mathcal{O}$, there must be a time
  $t_\delta$ such that
  $z(t_\delta) \in \mathcal{U}_\eta \cap B_\delta(z_\eta)$.
  If this is the case, then necessarily $z(t) \to z_\eta$, which is the
  thesis. Therefore it remains to prove this last claim.
  
  The sets $\ol{\mathcal{U}_\eta}$ and
  $\ol{\mathcal{U}_\eta} \setminus B_\delta(z_\eta)$ are
  compact in $\R^n$. In view of conditions~(\ref{eq:MpL2}) and~(\ref{eq:MpL3}),
  on the set $\ol{\mathcal{U}_\eta} \setminus B_\delta(z_\eta)$ we have
  $X \cdot \nabla \fun{S} < 0$.
  Let $-M \coloneqq \sup \{X(z) \cdot \nabla \fun{S}(z) \colon z \in
  \ol{\mathcal{U}_\eta} \setminus B_\delta(z_\eta)\}$.  
  Then $M>0$ since the supremum is attained, and if the integral curve $z(t)$ is
  such that $z(t) \in \ol{\mathcal{U}_\eta} \setminus B_\delta(z_\eta)$
  for all $t \geq 0$, 
  \begin{equation*}
    \fun{S}\big(z(t)\big) = \fun{S}\big(z_0\big) + \int_0^t
    (X \cdot \nabla \fun{S})\big(z(s)\big) ds \leq
    \fun{S}\big(z_0\big) - M t.
  \end{equation*}
  This leads to a contradiction as in the proof of the classical Lyapunov
  stability theorem. Hence there must be a time $t_\delta$ at which
  $z(t_\delta) \in \mathcal{U}_\eta \cap B_\delta(z_\eta)$. 
\end{proof}

Condition~(\ref{eq:MpL2}) ensures that~(\ref{eq:CL2}) holds for the function
$\fun{L}(z) = \fun{S}(z) - \fun{S}(z_\eta)$ on $\mathcal{U}_\eta$. The
uniqueness of the critical point of $\fun{S}$ together with~(\ref{eq:MpL3})
ensures that entropy is strictly dissipated by requiring that the metric
bracket generated by the tensor $K$ is ``specifically degenerate'', by which we
mean  it preserves the invariants $\fun{I}^\alpha$ only. Recall we use the
terminology ``minimally degenerate'' to mean that the only degeneracy is that
associated with $\fun{H}$. 

We remark that the limit point $z_\eta \in \mathcal{U}_\eta$ is the unique
solution of the optimization problem  
\begin{equation*}
  \min \{\fun{S}(z) \colon z \in \mathcal{O},\;
  \fun{I}(z) = \fun{I}(z_0) \};
\end{equation*}
hence Proposition~\ref{th:mod-Lyapunov} establishes a local version of the
desired complete-relaxation property for all orbits
starting in the set $\mathcal{O}$.

\subsection{Finite-dimensional systems: the Polyak--{\L}ojasiewicz inequality} 
\label{sec:finite-dim-PL}

Nondegenerate gradient flows and gradient descent methods have been studied
under the assumptions that the entropy function $\fun{S}$ satisfies the
classical Polyak--{\L}ojasiewicz (PL) condition
\cite{Polyak1963,Lojasiewicz1984} and $\nabla \fun{S}$ is Lipschitz continuous,
the latter being a natural hypothesis. 

A function $\fun{S} \in C^1(\mathcal{Z})$ satisfies the PL condition in a
non-empty subset $Q \subseteq \mathcal{Z}$ if
$\fun{S}_* \coloneqq \inf \{\fun{S}(z) \colon z \in Q \} > -\infty$ and
there exists a constant $\kappa > 0$ such that
\begin{equation*}
  \label{eq:PL}\tag{PL}
  \frac{1}{\kappa} \big|\nabla \fun{S}(z)\big|^2 \geq
  \fun{S}(z) - \fun{S}_*, \quad z \in Q.
\end{equation*}
If $\fun{S}$ satisfies the \ref{eq:PL} inequality, the associated gradient
flow has the following properties:

\begin{itemize}
    
\item If $z(t) \in Q$, $t \in [0, +\infty)$, is an integral curve of
  the gradient flow~(\ref{eq:X-grad-flow}), then $\fun{S}\big(z(t)\big)$
  converges exponentially to $\fun{S}_*$ as $t \to + \infty$.
  This follows from
  \begin{equation*}
    \frac{d}{dt} \big[\fun{S}\big(z(t)\big) - \fun{S}_*\big] =
    - \big|\nabla \fun{S}\big(z(t)\big)\big|^2
    \leq - \kappa \big[\fun{S}\big(z(t)\big) - \fun{S}_* \big],
  \end{equation*}
  which implies
  \begin{equation*}
    \big[\fun{S}\big(z(t)\big) - \fun{S}_* \big] \leq
    \big[\fun{S}\big(z_0\big) - \fun{S}_*\big] e^{-\kappa t},
  \end{equation*}
  where $z_0 \in Q$ is the initial condition at $t=0$. The exponential
  rate of convergence is given by the constant $\kappa$ in the inequality.
  Unfortunately convergence of the entropy values alone does not directly
  imply convergence of the orbit $z(t)$ to a limit for $t \to +\infty$
  (cf.\  the counterexample of  Palis and de Melo \cite{Palis1982}).
    
\item All equilibrium points of the gradient flow~(\ref{eq:X-grad-flow})
  are global minima of $\fun{S}$ (not just critical points).
  In fact, if $z_e \in Q$ is an equilibrium 
  point, $\nabla \fun{S}(z_e) = 0$ and the \ref{eq:PL} condition implies 
  $0 = |\nabla \fun{S}(z_e)|^2 \geq \kappa \big(\fun{S}(z_e) - \fun{S}_*\big)
  \geq 0$,  hence $\fun{S}(z_e) = \fun{S}_*$. 
  
\end{itemize}

Any strongly convex function satisfies condition~(\ref{eq:PL}),
cf.\  the short proof reported below, and thus are included as a special case.
In general, (\ref{eq:PL}) is weaker than strong convexity. In fact, it has
been shown that the classical \ref{eq:PL} condition is weaker than several
other conditions introduced in order to address the convergence of gradient
descent methods \cite{Karimi2016}. In addition, no assumption is made on the
global minimum of  $\fun{S}$, which, in particular, does not need to be an
isolated point. For instance, the function $\fun{S}(z) = (|z|^2-1)^2$ satisfies
the \ref{eq:PL} condition with $\kappa = 16 r_0^2$ in the domain
$S = \{z \colon |z| > r_0\}$ for any $r_0 \in (0,1/2)$, and attains
its minimum on the sphere $|z|=1$; hence there is no isolated minimum.  On the
other hand, this function does not satisfy the \ref{eq:PL} condition on  the
whole space $\R^n$, because of the critical point at $z=0$, which is a local
maximum.  

Polyak \cite{Polyak1963} under the additional (natural) hypothesis that the 
vector field $X = -\nabla \fun{S}$ is Lipschitz continuous (\emph{not} just
locally Lipschitz), established convergence of the gradient flow trajectories
to the global entropy minimum. Specifically, Polyak's result in the notation
used here amounts to the following.  

\begin{theorem}[Polyak 1963, Theorem~9 \cite{Polyak1963}]
  \label{th:Polyak}
  Let $z_0 \in \mathcal{Z}$, $\rho, \kappa, L > 0$, and
  $\fun{S} \in C^1(\mathcal{Z})$ be such that~(\ref{eq:PL}) is satisfied and
  $\nabla \fun{S}$ is Lipschitz continuous in the closed ball
  $\ol{B_\rho(z_0)} \subset \mathcal{Z}$ with Lipschitz constant $L$, and
  $\gamma = \sqrt{8L \varphi_0}/(\rho \kappa) \leq 1$, where
  $\varphi_0 = \fun{S}(z_0) - \fun{S}_*$. Then there is
  $z_* \in \ol{B_{\gamma \rho}(z_0)}$ such that the integral
  curve $z(t)$ of the gradient flow~(\ref{eq:X-grad-flow}) with $z(0)=z_0$
  is defined for all $t \geq 0$ and
  $\big|z(t) - z_* \big| \leq \gamma \rho e^{-\kappa t/2}$.
\end{theorem}

We shall now give the details concerning the Polyak--{\L}ojasiewicz condition.
Although well-known, we recall for convenience of the reader the short proof
of the fact that a strongly convex entropy satisfies inequality~(\ref{eq:PL}). 

\begin{proof}[Proof: Strongly convex imply \ref{eq:PL}]
  A strongly convex function $\fun{S} \in C^1(\mathcal{Z})$,
  $\mathcal{Z} \subseteq \R^n$ with parameter $\alpha>0$ satisfies
  \begin{equation*}
    \fun{S}(z') \geq \fun{S}(z) + (z'-z) \cdot \nabla \fun{S}(z)
    + \frac{\alpha}{2} |z'-z|^2,
  \end{equation*}
  for any $z,z' \in \mathcal{Z}$. The right-hand side is bounded from below by
  its infimum over $z'$, which is attained at $z'=z-\alpha^{-1}\nabla\fun{S}(z)$
  and thus
  \begin{equation*}
    \fun{S}(z') \geq \fun{S}(z) -
    \frac{1}{2\alpha} \big|\nabla\fun{S}(z)\big|^2.
  \end{equation*}
  This inequality holds for any $z'$, and thus implies that
  $\fun{S}_* = \inf \{\fun{S}(z) \colon z \in \mathcal{Z}\} > -\infty$.
  Taking the infimum over $z'$ yields
  \begin{equation*}
    \fun{S}_* \geq \fun{S}(z) - \frac{1}{2\alpha} \big|\nabla\fun{S}(z)\big|^2.
  \end{equation*}
  This can be rearranged to give the \ref{eq:PL} inequality
  with constant $\kappa = 2\alpha$.
\end{proof} 

In Polyak's original formulation of the theorem, the size $\rho$ of the
ball is determined by the condition $\gamma \leq 1$ in terms of the entropy at
the initial position $z_0$. For instance, $\fun{S}(z) = z^2/2$ satisfies the
hypothesis with $\kappa = 2$ and $L=1$; given any $z_0\in\R^n$, $\gamma\leq 1$
is equivalent to $\rho \geq |z_0|$, hence $\ol{B_\rho(z_0)}$ is large enough to
contain the unique global minimum $z_*=0$. We give a slightly different
statement of Polyak's result.  
  
\begin{proposition}
  \label{th:grad-flow-PL}
  Let $\fun{S} \in C^1(\mathcal{Z})$, $X = -\nabla \fun{S}$ Lipschitz
  continuous in a a subdomain $\mathcal{U} \subset \mathcal{Z}$
  satisfying~(\ref{eq:U}), and let $\mathcal{O}$ be the open, non-empty subset
  given by Lemma~\ref{th:lemma}. If $\fun{S}$ satisfies~(\ref{eq:PL}) in
  $\mathcal{U}$, then for any integral curve $z(t)$ of $X$ with initial
  condition $z(0) = z_0 \in \mathcal{O}$ there is $z_* \in \mathcal{O}$ and a
  constant $\theta > 0$ depending on $z_0$ such that $\fun{S}(z_*) = \fun{S}_*$,
  \begin{equation*}
    \big|z(t) - z_* \big| \leq \theta e^{-\kappa t/2}\,, \quad 
    \big[\fun{S}\big(z(t)\big) - \fun{S}_* \big] 
    \leq
    \big[\fun{S}\big(z_0\big) - \fun{S}_*\big] e^{-\kappa t}\,,\quad
    \text{for}\ t \geq 0\,.
    \end{equation*}
\end{proposition}

Next we give a proof of Proposition~\ref{th:grad-flow-PL}, which is based on
the original argument due to Polyak \cite{Polyak1963}. The only difference
consists in the use of Lemma~\ref{th:lemma} to establish the existence of a
global solution. We give the proof in details since similar ideas are then
needed below for the metriplectic case.

\begin{proof}[Proof of Proposition~\ref{th:grad-flow-PL}.]
  Lemma~\ref{th:lemma} ensures that for any $z_0 \in \mathcal{O}$ there is an
  integral curve $z(t) \in \mathcal{O}$ of the gradient flow
  $X = -\nabla\fun{S}$ passing through $z(0) = z_0$ and this is defined for all
  $t\geq 0$.  As a consequence of~(\ref{eq:PL}),
  $\varphi(t) = \fun{S}\big(z(t)\big) - \fun{S}_*$ decays to zero
  exponentially, i.e., $\varphi(t) \leq \varphi_0 e^{-\kappa t}$,
  with $\varphi_0 = \varphi(0)$.
  
  If there is a finite time $0 \leq \bar{t} < +\infty$ at which
  $\varphi(\bar{t}) = 0$, then $z(\bar{t})$ is a minimum for the entropy
  $\fun{S}$ and thus $\nabla \fun{S}\big(z(\bar{t})\big) = 0$ so that
  $z(\bar{t})$ is an equilibrium point. We deduce $z(t) = z(\bar{t}) = z_0$
  for all $t\geq 0$ and the thesis holds true with $z_* = z_0$.
  
  As for the non-trivial case $\varphi(t) > 0$ for all $t\geq 0$, we
  distinguish two key steps.
    
  \noindent
  \emph{Step 1.} For any $t_1,t_2 \geq 0$, $t_1 < t_2$, we have
  \begin{equation*}
    \fun{S}\big(z(t_1)\big) - \fun{S}\big(z(t_2)\big) =
    \int_{t_1}^{t_2} \big|\nabla \fun{S}\big(z(t)\big)\big|^2 dt.
  \end{equation*}
  Following Polyak, we estimate the right-hand side from below. First the
  triangular inequality gives
  \begin{equation*}
    \big|\nabla \fun{S}\big(z(t)\big) \big| \geq
    \Big|
    \big|\nabla \fun{S}\big(z(t_1)\big) \big| -
    \big|\nabla \fun{S}\big(z(t)\big) - \nabla \fun{S}\big(z(t_1)\big)\big|
    \Big|.
  \end{equation*}
  The second term on the right-hand side is bounded by
  \begin{equation*}
    \big| \nabla \fun{S}\big(z(t)\big) - \nabla \fun{S}\big(z(t_1)\big)\big|
    \leq L \big|z(t) - z(t_1)\big|,
  \end{equation*}
  where $L > 0$ is the Lipschitz constant of $\nabla \fun{S}$. In addition,
  \begin{align*}
    \frac{d}{dt} \big|z(t) - z(t_1)\big| &\leq \Big|\frac{d}{dt}
    \big|z(t) - z(t_1)\big| \Big| \leq \big|
    \nabla \fun{S}\big(z(t)\big)\big| \\
    & \leq \big|\nabla \fun{S}\big(z(t_1)\big)\big| + \big|
    \nabla \fun{S}\big(z(t)\big) - \nabla \fun{S}\big(z(t_1)\big)\big| \\
    & \leq \big|\nabla \fun{S}\big(z(t_1)\big) \big| + L
    \big|z(t) - z(t_1)\big|.
  \end{align*}
  Gr{\"o}nwall's inequality then gives
  \begin{equation*}
    \big|z(t) - z(t_1)\big| \leq \frac{1}{L}
    \big| \nabla \fun{S}\big(z(t_1)\big) \big| \big[e^{L(t-t_1)} - 1 \big],
  \end{equation*}
  and thus
  \begin{equation*}
    \big|\nabla \fun{S}\big(z(t)\big) - \nabla \fun{S}\big(z(t_1)\big)\big|
    \leq
    \big| \nabla \fun{S}\big(z(t_1)\big) \big| \big[e^{L(t-t_1)} - 1 \big].
  \end{equation*}
  If $L (t_2 - t_1) \leq \log 2$, the term on the right-hand side is
  $\leq \big| \nabla \fun{S}\big(z(t_1)\big) \big|$ and thus 
  \begin{equation*}
    \big|\nabla \fun{S}\big(z(t)\big) \big| \geq
    \big|\nabla \fun{S}\big(z(t_1)\big) \big| \big[2 - e^{L (t-t_1)}\big].
  \end{equation*}
  With $t_2 = t_1 + (\log 2)/L$ and $t_1$ arbitrary, we obtain
  \begin{equation*}
    \fun{S}\big(z(t_1)\big) - \fun{S}\big(z(t_2)\big) \geq
    \frac{1}{\alpha L} \big|\nabla \fun{S}\big(z(t)\big) \big|^2,
  \end{equation*}
  where
  \begin{equation*}
    \frac{1}{\alpha} \coloneqq
    \int_0^{\log 2} \big[2 - e^{s}\big]^2 ds
  \end{equation*}
  is a positive numerical constant. At last, we use the fact that
  $\fun{S}\big(z(t)\big)$ is non-increasing and $t_1 < t_2$ so that
  \begin{align*}
    \fun{S}\big(z(t_1)\big) - \fun{S}\big(z(t_2)\big) &=
    \big[\fun{S}\big(z(t_1)\big) - \fun{S}_*\big] -
    \big[\fun{S}\big(z(t_2)\big) - \fun{S}_*\big] \\
    &\leq
    \big[\fun{S}\big(z(t_1)\big) - \fun{S}_*\big],
  \end{align*}
  and deduce that, for any $t_1 \geq 0$,
  \begin{equation}
    \label{eq:PL-grad-estimate}
    \big|\nabla \fun{S}\big(z(t_1)\big) \big|^2 \leq \alpha L \big[
      \fun{S}\big(z(t_1)\big) - \fun{S}_*\big] = \alpha L \varphi(t_1).
  \end{equation}
  We have established that, under the hypotheses, $\nabla \fun{S}$ along the
  orbit is controlled by the entropy decay.
  
  \noindent
  \emph{Step 2.} Given arbitrary points in time $t_1, t_2 \geq 0$,
  $t_1 < t_2$,
  \begin{equation*}
    \big|z(t_1) - z(t_2)\big| \leq \int_{t_1}^{t_2} \big|\nabla
    \fun{S}\big(z(t)\big)\big| dt,
  \end{equation*}
  and~(\ref{eq:PL-grad-estimate}) together with 
  $\varphi(t) \leq \varphi_0 e^{-\kappa t}$ yields
  \begin{equation}
    \label{eq:PL-Cauchy}
    \big|z(t_1) - z(t_2)\big| \leq \frac{\sqrt{4\alpha L \varphi_0}}{\kappa}
    \big[e^{-\kappa t_1 / 2} - e^{-\kappa t_2/2}\big].
  \end{equation}
  This inequality can be used to show that for any sequence
  $\{t_n\}_{n \in \N}$ with $t_n \to +\infty$ as $n \to +\infty$, $z(t_n)$ is
  a Cauchy sequence on $\mathcal{O}$. Therefore there is a point
  $z_* \in \ol{\mathcal{O}}$ such that $z(t) \to z_*$ as $t \to +\infty$.
  Continuity of $\fun{S}$ implies $\fun{S}(z_*) = \fun{S}_*$ as claimed.
  At last, $z_* \in \mathcal{O}$, for, if not, then
  $z_* \in \partial \mathcal{O}$ and
  $\fun{S}(z_*) > \fun{S}(z_0) \geq \fun{S}(z_*)$, which is a contradiction
  (the first inequality is strict). Passing to the limit $t_2 \to +\infty$ in
  inequality~(\ref{eq:PL-Cauchy}) yields
  \begin{equation*}
    \big|z(t_1) - z_*\big| \leq \theta e^{-\kappa t_1 / 2},
  \end{equation*}
  with $\theta = \sqrt{4 \alpha L \varphi_0} / \kappa$ and any $t_1 \geq 0$.
\end{proof}

We now generalize this result to the case of finite-dimensional metriplectic
systems. Under the same conditions as in Proposition~\ref{th:mod-Lyapunov}, a
simple generalization of~(\ref{eq:PL}) for metriplectic vector fields reads:
for any $\eta \in \fun{I}(\mathcal{U})$,
$\inf \{\fun{S}(z) \colon z \in \mathcal{U}_\eta\} \eqqcolon \fun{S}_\eta > -
\infty$ and there is a constant $\kappa_\eta > 0$, depending on $\eta$,
such that  
\begin{equation*}
  \label{eq:MpPL}\tag{PL${}^\prime$}
  \frac{1}{\kappa_\eta} \nabla \fun{S}(z) \cdot K(z) \nabla \fun{S}(z) \geq
  \fun{S}(z) - \fun{S}_\eta, \quad z \in \mathcal{U}_\eta.
\end{equation*}
  
Differently from~(\ref{eq:PL}), which is a condition on $\fun{S}$,
inequality~(\ref{eq:MpPL}) involves both the entropy function and the bracket 
$(\fun{S},\fun{S}) = \nabla \fun{S} \cdot K \nabla \fun{S}$, which gives the
entropy decay rate. Similarly to~(\ref{eq:PL}), if inequality~(\ref{eq:MpPL})
holds, then the metriplectic vector field has the following properties:

\begin{itemize}

\item If $z(t) \in \mathcal{U}_\eta$, $t \in [0,+\infty)$, is an integral curve
  of $X$ with initial condition $z_0 = z(0)$, then
  $\fun{S}\big(z(t)\big) \to \fun{S}_\eta$ as $t \to +\infty$, with exponential
  convergence. In fact, (\ref{eq:MpPL}) implies
  \begin{equation*}
    \big[\fun{S}\big(z(t)\big) - \fun{S}_\eta \big] \leq
    \big[\fun{S}\big(z_0\big) - \fun{S}_\eta \big] e^{-\kappa t},
  \end{equation*}
  which follows as in the case of standard gradient flows.
    
\item If $z_e \in \mathcal{U}_\eta$ is an equilibrium point, then it is
  necessarily a global minimum of $\fun{S}|_{\mathcal{U}_\eta}$ that is,
  $\fun{S}(z_e) = \fun{S}_\eta$.
  This follows from the fact that at an equilibrium point necessarily
  $\nabla \fun{S}(z_e) \cdot K(z_e) \nabla \fun{S}(z_e) = 0$,
  and~(\ref{eq:MpPL}) implies $\fun{S}(z_e) = \fun{S}_\eta$.  
  
\end{itemize}

We can now state the analog of Proposition~\ref{th:grad-flow-PL} for the case
of metriplectic vector fields. However, we consider only the dissipative part,
that is, a metriplectic vector field of the form~(\ref{eq:metric-system}),
without the symplectic part, and make an additional assumption that is
sufficient to ensure that the orthogonal projection onto $\ker K(z)$ is smooth
in $z$. Specifically, we assume that
\begin{equation}
  \label{eq:K-eigenvalues}
  \begin{aligned}
    &\text{there is a constant $r > 0$, such that, for any
      $z \in \mathcal{Z}$,} \\
    &\text{zero is the only eigenvalue of $K(z)$ in the interval $[-r,r]$.}
  \end{aligned}
\end{equation}
Hypothesis~(\ref{eq:K-eigenvalues}) implies that, for any $z$, only
one eigenvalue of $K(z)$, the zero eigenvalue, belongs in the disk of
radius $r>0$ centered at zero in the complex plane.

Concerning the application of metriplectic dynamics to the calculation 
of equilibria of fluids and plasmas, the restriction to systems of the
form~(\ref{eq:metric-system}) is not a significant limitation, since we often
construct the relaxation method from the metric bracket only, as we do in the
examples of Sections \ref{sec:simple}, \ref{sec:coll-like-metr},
and~\ref{sec:diff-like-metr} below. In general however, it could be convenient
to account for the ideal dynamics of the considered system. In that case
Polyak's argument fails since the entropy gradient alone is not sufficient to
control the time derivative of $|z(t)-z(t_1)|$ in the first step of the proof.
We are not aware of any generalization of the \ref{eq:PL} inequality to
completely general metriplectic fields.

\begin{proposition}
  \label{th:MpPL}
  Let $X = - K \nabla \fun{S}$ be a vector field of the
  form~(\ref{eq:metric-system}) on a domain $\mathcal{Z} \subseteq \R^n$, 
  with $K$ satisfying~(\ref{eq:K-eigenvalues}) and let
  $\fun{I}=(\fun{I}^1, \ldots, \fun{I}^k) \in C^\infty(\mathcal{Z},\R^k)$,
  $1 \leq k < n$, be such that $K \nabla \fun{I}^\alpha = 0$.
  Assume that $\rank \nabla \fun{I} = k$ in a subdomain $\mathcal{U}$
  satisfying~(\ref{eq:U}) and let $\mathcal{O}$ be the open set given by
  Lemma~\ref{th:lemma}. If~(\ref{eq:MpPL}) holds in $\mathcal{U}$, then for
  any integral curve $z(t)$ of $X$ with initial condition
  $z(0) = z_0 \in \mathcal{O}$  there is a point $z_\eta \in \mathcal{O}$ and
  a constant $\theta_\eta$ depending on $z_0$, such that
  $\eta = \fun{I}(z_\eta) = \fun{I}(z_0)$, $\fun{S}(z_\eta) = \fun{S}_\eta$
  and
  \begin{equation*}
    \big|z(t) - z_\eta \big| \leq \theta_\eta e^{-\kappa_\eta t/2}\,,\quad 
    \big[\fun{S}\big(z(t)\big) - \fun{S}_\eta \big]
    \leq
    \big[\fun{S}\big(z_0\big) - \fun{S}_\eta\big] e^{-\kappa_\eta t}\,,\quad
    \text{for}\ t\geq 0\,.
  \end{equation*}
\end{proposition}

The proof of Proposition~\ref{th:grad-flow-PL} can be adapted to this case.
The key point is replacing $\nabla \fun{S}$ with $(I-\pi_0) \nabla \fun{S}$,
where $\pi_0$ is the orthogonal projector onto the $\ker K$.

\begin{proof}[Proof of Proposition~\ref{th:MpPL}.]
  Per hypothesis, both $K$ and $\fun{S}$ are smooth on $\mathcal{Z}$ and thus
  $X$ is Lipschitz continuous on any bounded subset and in particular on
  $\mathcal{U}$. Then Lemma~\ref{th:lemma} gives an open subset of
  $\mathcal{O} \subseteq \mathcal{U}$ such that for any point
  $z_0 \in \mathcal{O}$ there is an integral curve $z(t) \in \mathcal{O}$ of
  $X$ through the point $z_0 = z(0)$, defined for all $t \geq 0$. Along the
  orbit $\fun{I}^\alpha\big( z(t) \big) = \fun{I}^\alpha(z_0) = \eta^\alpha$,
  $\alpha \in \{1,\ldots,k\}$, and thus
  $z(t) \in \mathcal{U}_\eta \cap \mathcal{O}$.
  
  At a point $z_\eta \in \mathcal{U}_\eta$ where
  $\fun{S}(z_\eta) = \fun{S}_\eta$ the function $\fun{S}|_{\mathcal{U}_\eta}$
  attains its minimum and thus $\fun{S}$ must
  satisfy~(\ref{eq:constrained-critical-points}). It follows from the
  assumption $K\nabla \fun{I}^\alpha = 0$ that $X(z_\eta) = 0$. Therefore, if
  there is $\bar{t} \geq 0$ such that
  $\fun{S}\big(z(\bar{t})\big) = \fun{S}_\eta$,
  then $X\big(z(\bar{t})\big) = 0$ and $z(t) = z(\bar{t}) = z_0$ for all
  $t\geq 0$. In this case the statement of the proposition is trivially true.
  
  Let us now consider the non-trivial case
  $\fun{S}\big(z(t)\big) > \fun{S}_\eta$, $t\geq 0$. We shall follow the same
  two steps as in Polyak original proof, with the necessary changes to account
  for the degeneracy of the metriplectic flow. This requires a preliminary
  step in which we establish the needed properties of the matrix $K$.
  
  \emph{Step 0.} For any point $z \in \mathcal{Z}$, let $K_i(z)>0$ be the
  $i$-th non-zero eigenvalue of $K(z)$ and $\pi_i(z)$ the orthogonal
  projector on $\ker \big(K(z) - K_i(z) I)$, i.e. on the eigenspace of $K(z)$
  corresponding to the eigenvalue $K_i(z)$. Then
  \begin{equation*}
    K(z) = \sum_i K_i(z) \pi_i(z).
  \end{equation*}
  Since $K(z)$ is symmetric any eigenvalue (including zero) is semisimple
  \cite[Appendix 3.I]{Rauch2012}. We assumed that the closed disk of radius
  $r>0$ centered in $0 \in \C$ contains only the zero eigenvalue of $K(z)$ for
  all $z \in \mathcal{Z}$, hence  $K_i(z) > r$. It follows that
  \cite[Theorem~3.I.1]{Rauch2012} the orthogonal projector $\pi_0(z)$ onto
  $\ker K(z)$ is a smooth function of $z \in \mathcal{Z}$.
  In addition we have that, for any vector $Z \in \R^n$,
  \begin{align*}
    Z \cdot K(z) Z &= \sum_i K_i(z) Z \cdot \pi_i(z) Z \\
    &\geq r \sum_i Z \cdot \pi_i(z) Z =r Z \cdot (I- \pi_0) Z.
  \end{align*}
  If $Z = \nabla \fun{S}(z)$, we deduce
  \begin{equation}
    \label{eq:K-lower}
    \nabla\fun{S}(z) \cdot K(z) \nabla \fun{S}(z)
    \geq r \big|\big(I-\pi_0(z)\big)\nabla\fun{S}(z)\big|^2.
  \end{equation}
  We also have $X =- K \nabla \fun{S} = -K(I-\pi_0)\nabla\fun{S}$, and
  $|X(z)| \leq \|K(z)\|_F \big|\big(I-\pi_0(z)\big) \nabla\fun{S} (z)\big|$,
  where $\|K(z)\|_F$ is the Frobenius norm, which is a continuous function of
  $z$. For any $z$ in the compact set $\ol{\mathcal{U}}$, 
  $\|K(z)\|_F \leq R = \max\{\|K(z')\|_F \colon z' \in \ol{\mathcal{U}}\}$,
  so that
  \begin{equation}
    \label{eq:K-upper}
    \big|X(z)\big| \leq R \big|\big(I-\pi_0(z)\big)\nabla \fun{S}(z)\big|.
  \end{equation}
  Since $\pi_0$ is smooth, $Y(z) = \big(I-\pi_0(z)\big)\nabla \fun{S}(z)$ is
  smooth. This is the component of $\nabla \fun{S}$ orthogonal to $\ker K$.
  
  From~(\ref{eq:K-upper}) we can also deduce that $Y\big(z(t)\big)\not =0$
  for, if not, then~(\ref{eq:MpPL}) implies
  $\fun{S}\big(z(t)\big) = \fun{S}_\eta$, which is the trivial case.
  
  We shall show that Polyak's argument can be repeated with $Y$ instead of
  $\nabla \fun{S}$. 
  
  \emph{Step 1.} Upon using~(\ref{eq:K-lower}), given $0 \leq t_1 < t_2$,
  \begin{equation*}
    \fun{S}\big(z(t_1)\big) - \fun{S}\big(z(t_2)\big) \geq
    r \int_{t_1}^{t_2} \big|Y\big(z(t)\big)\big|^2 dt.
  \end{equation*}
  On the other end, (\ref{eq:K-upper}) yields
  \begin{equation*}
    \frac{d}{dt} \big|z(t) - z(t_1)\big| \leq \big|X\big(z(t)\big)\big|
    \leq R \big|Y\big(z(t)\big)\big|.
  \end{equation*}
  Since $Y \in C^\infty(\mathcal{Z},\R^n)$, it is in particular Lipschitz
  continuous on $\mathcal{O}$. Let $L > 0$ be the Lipschitz constant of $Y$.
  The same argument of step 1 in the proof of
  Proposition~\ref{th:grad-flow-PL} can be repeated leading to 
  \begin{equation}
    \label{eq:Y-estimate}
    \big|Y\big(z(t)\big)\big|^2 \leq \frac{\alpha R L}{r} \big[
      \fun{S}\big(z(t)\big) - \fun{S}\big],
  \end{equation}
  where $\alpha$ is the same constant defined in the proof of
  Proposition~\ref{th:grad-flow-PL}.
  
  \emph{Step 2.} We have already shown that
  \begin{equation*}
    \varphi_\eta (t) \leq \varphi_{0,\eta} e^{-\kappa_\eta t},
  \end{equation*}
  where $\varphi_\eta(t) = \fun{S}\big(z(t)\big) - \fun{S}_\eta$, and
  $\varphi_{0,\eta} = \varphi_\eta(0)$. For any $0 \leq t_1 < t_2$,
  inequality~(\ref{eq:Y-estimate}) gives 
  \begin{align*}
    \big|z(t_1) - z(t_2)\big| &\leq R \int_{t_1}^{t_2}
    \big|Y\big(z(t)\big)\big| dt \\
    &\leq R \sqrt{\frac{\alpha R L}{r}} \int_{t_1}^{t_2}
    \sqrt{\varphi_\eta(t)} dt \\
    &\leq \frac{1}{\kappa_\eta} \sqrt{\frac{4\alpha R^3 L \varphi_{0,\eta}}{r}}
    \big[e^{-\kappa_\eta t_1/2} - e^{-\kappa_\eta t_2/2}\big].
  \end{align*}
  It follows that $z(t)$ has a limit $z_\eta \in \ol{\mathcal{O}}$ for
  $t \to +\infty$ and the limit satisfies
  $\fun{S}(z_\eta)=\lim_{t\to +\infty}\fun{S}\big(z(t)\big)=\fun{S}_\eta$,
  hence $z_\eta \in \mathcal{O}$. In addition,
  $\fun{I}(z_\eta) = \lim_{t\to +\infty}\fun{I}\big(z(t)\big) = \fun{I}(z_0)$,
  hence $z_\eta \in \mathcal{U}_\eta \cap \mathcal{O}$.
  At last, passing to the limit $t_2 \to +\infty$ yields
  \begin{equation*}
    \big|z(t_1) - z_\eta \big| \leq \theta_\eta e^{-\kappa_\eta t_1/2},
  \end{equation*}
  with constant
  \begin{equation*}
    \theta_\eta = \frac{1}{\kappa_\eta}
    \sqrt{\frac{4\alpha R^3 L \varphi_{0,\eta}}{r}}.
  \end{equation*}
  This is the claimed inequality.
\end{proof}

This exponential convergence result becomes less useful when the constant
$\theta_\eta$ is large, which can happen when either $r$ or $\kappa_\eta$
are small. In the former case, an eigenvalue of $K$ becomes small at least in
some region of the domain. In the latter case, the metric bracket is small
even where entropy is far from the minimum.

This result follows from a minimal modification of Polyak's argument for
standard gradient flows. 
One should note that here $\nabla \fun{I}^\alpha$ are assumed to be in the
kernel of $K$, and this assumption is consistent with the
requirement~(\ref{eq:metric-system-compatibility}) for $\fun{H}$. 
  
Unlike the results of Section~\ref{sec:finite-dim-Lyapunov},
convergence results based on Polyak inequalities do not require the uniqueness
of the minimum entropy state. On the other hand, the precise point on the set
of minima at which each orbit converges depends on the specific orbit. We
illustrate inequality~(\ref{eq:MpPL}) with a few examples. 
  
\begin{example}
  \label{ex:1}

  Let the phase space be $\mathcal{Z} = \R^n$, $n \in \N$ and $n \geq 2$,
  with coordinates $z = (z^i)_{i=1}^n$, $s = (s_i)_{i=1}^n$ and
  $h = (h_i)_{i=1}^n \in \R^n$ be covariant vectors, 
  $K = (K^{ij})$ be a symmetric, positive-semidefinite,
  contravariant tensor on $\R^n$ such that $K^{ij}h_j = 0$, and let
  $\sigma = (\sigma_{ij})_{ij}$ be a symmetric positive definite, covariant
  tensor. We define the dissipative part of the metriplectic vector field
  $X = - K \nabla \fun{S}$, with Hamiltonian $\fun{H}(z) = h_i z^i$, and entropy
  $\fun{S}(z) = s_i z^i + \frac{1}{2} \sigma_{ij} z^i z^j$.
  We further assume that the null space of $K$ coincides with the line spanned
  by $h$, i.e., $K\omega = 0$ implies $\omega = \lambda h$ for some
  $\lambda \in \R$. Then the metric bracket defined by $K$ is
  ``minimally degenerate'', in the sense defined above.
    
  We claim that this metriplectic system satisfies condition~(\ref{eq:MpPL}),
  and thus Proposition~\ref{th:MpPL} applies.  
    
  In order to show this, let us first observe that the change of variables
  $z \mapsto \tilde{z} = \sigma^{1/2} z$ transforms the system into an analogous
  one with $\sigma_{ij}$ replaced by $\delta_{ij}$ and with $h$ and $s$ replaced
  by $\sigma^{-1/2} h$ and $\sigma^{-1/2} s$,  respectively. Hence it is enough
  to discuss the case $\sigma_{ij} = \delta_{ij}$. We can also assume
  $|h|^2 = 1$, because the normalization of $h$ only changes the value of the
  Hamiltonian but not its isosurfaces. Then, for any $\eta \in \R$,
  $\mathcal{U}_\eta = \{z \in \R^n \colon \fun{H}(z) = \eta \}$ is the plane
  given by  $h_i z^i = \eta$.   A point $z \in \mathcal{U}_\eta$ can be written
  as $z = \eta h + z_\perp$, with $z_\perp = z - (h \cdot z)h$.
  Given $\eta \in \R$, we can use Lagrange multipliers  to compute the
  constrained entropy minima $\fun{S}_\eta$: we search for
  $(\lambda, z)$ such that (with $\partial_i = \partial/\partial z^i$)
  $\partial_i \fun{S}(z) = \lambda \partial_i \fun{H}(z)$ with
  $\fun{H}(z) = \eta$, which is equivalent to 
  \begin{equation*}
    \left\{
    \begin{aligned}
      z &= \lambda h - s, \\
      h \cdot z &= \eta.
    \end{aligned}\right.
  \end{equation*}
  The solution $(\lambda_\eta, z_\eta)$ is readily found,
  \begin{equation*}
    \lambda_\eta = \eta + h \cdot s, \quad
    z_\eta =  \lambda_\eta h - s = \eta h - s_\perp,
  \end{equation*}
  where $s_\perp = s-(h\cdot s)h$. Therefore, if $z \in \mathcal{U}_\eta$,
  \begin{equation*}
    \fun{S}(z) = s \cdot (\eta h + z_\perp) + \frac{1}{2} |\eta h + z_\perp|^2 
    = \fun{S}(z_\eta) + \frac{1}{2}|z_\perp + s_\perp|^2.
  \end{equation*}
  On the other hand, since $Kh=0$, for any $z \in \mathcal{U}_\eta$, 
  \begin{equation*}
    \nabla \fun{S}(z) \cdot K \nabla \fun{S}(z) = (z+s) \cdot K (z+s) 
    = (z_\perp + s_\perp) \cdot K (z_\perp + s_\perp) 
    \geq K_1 |z_\perp + s_\perp|^2,
  \end{equation*}
  where $K_1 > 0$ is the smallest eigenvalue of $K$ restricted to
  $(\ker K)^\perp$, the orthogonal of its kernel. We deduce
  \begin{equation*}
    \nabla \fun{S}(z) \cdot K \nabla \fun{S}(z) \geq 2 K_1
    \big[\fun{S}(z) - \fun{S}_\eta\big], \quad z \in \mathcal{U}_\eta,
  \end{equation*}
  where $\fun{S}_\eta = \fun{S}(z_\eta)$ is the constrained minimum of the
  entropy. This is condition~(\ref{eq:MpPL}) with $\kappa_\eta=2K_1$.
  
  We remark that, at least in this case, the condition on the kernel of $K$
  being ``minimal'' is crucial for the modified PL condition. In fact, if there
  is a vector $h' \in \R^n$, orthogonal to $h$, and such that $K h' = 0$, then
  $(z_\perp + s_\perp) \cdot K (z_\perp + s_\perp) = 0$ for any non-zero 
  $z_\perp + s_\perp \propto h'$.

  This example is simple enough that an analytical solution of the integral
  curves of $X$ can be obtained. In fact, the equation for the new variable
  $y = s + z$ amounts to the linear system $dy/dt = -Ky$ with initial
  condition $y_0 = s + z_0$. If $\fun{H}(z_0) = h \cdot z_0 = \eta$, we must
  have $ h \cdot y_0 = \eta + h \cdot s$. Upon representing $y$ on the basis
  of the unit eigenvectors $\{e_i\}_{i=0}^{n-1}$ of $K$, with $e_0 = h$ being
  the eigenvector that corresponds to the zero eigenvalue, we obtain
  \begin{equation*}
    y(t) = (h \cdot y_0) h + \sum_{i \geq 1} e^{-t\lambda_i(K)}
    c_i e_i,
  \end{equation*}
  where $\lambda_i(K) > 0$ are the positive eigenvalues of $K$ and
  $c_i = e_i \cdot y_0$. We deduce
  \begin{equation*}
    |z(t) - z_\eta| \leq |z_0 - z_\eta| e^{-t K_1},
  \end{equation*}
  with $K_1 = \min_{i\geq1} \{\lambda_i(K)\}$.

  Hence, in this case we have exponential convergence to the equilibrium point,
  with convergence rate being half of the constant in~(\ref{eq:MpPL}). 
\end{example}

\begin{example}
  \label{ex:2}

  With the same metric bracket and Hamiltonian as in Example~\ref{ex:1}, let
  us consider the entropy function
  \begin{equation*}
    \fun{S}(z) = \frac{|z|^2}{1+|z|^2}, \quad z \in \mathcal{Z}=\R^n.
  \end{equation*}
  As before, this entropy is rotationally symmetric, with a global minimum at
  $z=0$, but it is not a convex function.  

  Since $\fun{H}$ is the same as in Example~\ref{ex:1}, $z\in\mathcal{U}_\eta$
  if and only if $z = \eta h + z_\perp$ and we compute
  \begin{equation*}
    \nabla \fun{S}(z) \cdot K \nabla \fun{S}(z)
    = 4 \frac{z_\perp \cdot K z_\perp}{(1+|z|^2)^4}
    \geq 4 K_1 \frac{|z_\perp|^2}{(1+|z|^2)^4},
  \end{equation*}
  where $K_1$ is defined in Example~\ref{ex:1}. We can use Lagrange
  multipliers in order to compute minima of the entropy constrained to
  $\mathcal{U}_\eta$ with the result that there is a unique minimum at
  $z_\eta = \eta h$ and
  \begin{equation*}
    \fun{S}_\eta = \fun{S}(z_\eta) = \frac{\eta^2}{1+\eta^2}.
  \end{equation*}
  Then, we compute
  \begin{equation*}
    \fun{S}(z) - \fun{S}_\eta =
    \frac{|z_\perp|^2}{(1+\eta^2)(1+|z|^2)},
  \end{equation*}
  from which we deduce
  \begin{equation*}
    \frac{|z_\perp|^2}{(1+|z|^2)^4} = \frac{1+\eta^2}{(1+|z|^2)^3}
    \big[\fun{S}(z) - \fun{S}_\eta\big], \quad z \in \mathcal{U}_\eta,
  \end{equation*}
  hence, for any $R>0$, on the ball $|z| < R$, we have
  \begin{equation*}
    \nabla \fun{S}(z) \cdot K \nabla \fun{S}(z) \geq \kappa_\eta \big[
      \fun{S}(z) - \fun{S}_\eta\big], \quad
    z \in \mathcal{U}_\eta \cap B_R(0),
  \end{equation*}
  with constant $\kappa_\eta = 4 K_1 (1+\eta^2) / (1+R^2)^3$.
  Therefore the modified PL condition is satisfied
  on balls of arbitrary large radius $R$, even though the entropy is
  not convex, but the constant as a function of the radius $R$ is not
  uniformly bounded away from zero.
\end{example}
  
\begin{example}
  \label{ex:3}
    
  As a last example, we consider a strongly nonlinear case with a bracket
  built from an orthogonal projection onto the hyper-plane perpendicular to
  the gradient of the Hamiltonian. This particular metric structure will play
  a key role in the following, even in the infinite-dimensional cases in fluid
  and plasma dynamics.
  
  Given $s \in \R^n$, on the open half-space
  $\mathcal{Z} = \{z \in \R^n \colon z \cdot s < 0\}$, 
  let us consider the field $X(z) = - K(z) \nabla \fun{S}(z)$, with
  \begin{equation*}
    K(z) \coloneqq |z|^2 I - z \otimes z, \quad
    \fun{H}(z) \coloneqq \tfrac{1}{2} |z|^2, \quad \text{and}\quad 
    \fun{S}(z) \coloneqq s \cdot z\,.
  \end{equation*}
  Since $K(z)$ is proportional to the projector onto the subspace normal to
  $\nabla \fun{H}(z)$, we have $K(z) \nabla \fun{H}(z) = 0$ and $K(z)$ is
  a symmetric positive semidefinite tensor; hence,  $X$ is metriplectic with a
  trivial symplectic part. We also stress that $K(z)$ is minimally degenerate
  since $K(z)$ is strictly positive definite on the subspace normal to
  $\nabla \fun{H}(z)$.
  
  The constant-energy surfaces are spheres, for any $\eta > 0$,
  \begin{equation*}
    z \in \mathcal{U}_\eta \iff z = \sqrt{2\eta} \zeta, \quad
    \zeta \in S^{n-1}, \quad \zeta \cdot s < 0,
  \end{equation*}
  where points $\zeta$ on the $(n-1)$-dimensional sphere $S^{n-1}$ are
  identified with unit vectors in $\R^n$. The entropy restricted to
  $\mathcal{U}_\eta$ amounts to
  $\fun{S}|_{\mathcal{U}_\eta} (\zeta) = \sqrt{2\eta} s \cdot \zeta$
  and the minimum $\fun{S}_\eta = - \sqrt{2\eta s^2}$ is attained at
  $\zeta = - s/|s|$. The same result is of course obtained by means of
  Lagrange multipliers that lead to the system
  \begin{equation*}
    \left\{
    \begin{aligned}
      s &= \lambda z, \\
      \tfrac{1}{2} |z|^2 &= \eta,
    \end{aligned}
    \right. \quad \text{and} \quad z \cdot s < 0. 
  \end{equation*}
  Then we compute, for $z = \sqrt{2\eta} \zeta \in \mathcal{U}_\eta$,
  \begin{align*}
    \nabla \fun{S}(z) \cdot K(z) \nabla \fun{S}(z)
    &= |z|^2 |s|^2 - (z \cdot s)^2  
    = \big(\sqrt{2\eta} |s| - \sqrt{2\eta} s \cdot \zeta \big)
    \big(\sqrt{2\eta} |s| + \sqrt{2\eta} s \cdot \zeta \big) \\
    &\geq \sqrt{2\eta} |s| \big[\fun{S}(z) - \fun{S}_\eta\big],
  \end{align*}
  which is inequality~(\ref{eq:MpPL}).
  
  It should be noted that, if one drops the condition $z \cdot s <0$, so that
  $\mathcal{Z} = \R^n$, then the metric system cannot satisfy~(\ref{eq:MpPL})
  since $S|_{\mathcal{U}_\eta}$ has two critical points: a minimum $z_\eta^-$
  with $z_\eta^-\cdot s<0$ and a maximum $z_\eta^+$ with $z_\eta^+\cdot s > 0$.
  The metric bracket 
  $(\fun{S}, \fun{S}) = \nabla \fun{S} \cdot K \nabla \fun{S}$ vanishes at
  both points, but $\fun{S}(z_\eta^+) - \fun{S}_\eta > 0$.
\end{example}

\subsection{Infinite-dimensional systems: tentative generalizations}  
\label{sec:infinite-dim}

A version of the Lyapunov stability theorem, valid for the case of
infinite-dimensional systems, is available under suitable coercivity assumptions
on the Lyapunov function \cite{Marsden2001}. Such assumptions are needed to
compensate for the lack of compactness. For instance, the closed unit ball is
not compact in a Banach space. Compactness of closed and bounded sets in
$\R^n$ is used repeatedly in the classical proofs in finite dimensions 
(cf.~Sections \ref{sec:remarks-relax-equil} and \ref{sec:finite-dim-PL}). 
Analogously, the \eqref{eq:PL} condition can be extended
to infinite-dimensional systems. A more difficult point is the existence of a
global-in-time solution to the equation defining the dynamical system, under
reasonable hypothesis \cite{Temam1998}. In the infinite-dimensional setting,
this means proving the existence of a global solution for highly nonlinear
partial differential equations, which is often difficult and requires special
treatment for each individual case. Nonetheless, under the assumption that a
global-in-time solution exists, one can think of extending
Propositions~\ref{th:mod-Lyapunov} and~\ref{th:MpPL} to infinite dimensions, but
we leave the details for future work. Here we merely state the
infinite-dimensional version of
condition~(\ref{eq:MpL3}) and inequality~(\ref{eq:MpPL}).

Consider a metriplectic system on a Banach space $V$ as introduced in
Section~\ref{sec:metriplectic}. We assume that this system has a finite
(for simplicity) family of constants of motion $\fun{I} \in C^\infty(V,\R^k)$,
that satisfies the hypotheses of the submersion theorem
\cite[Theorem~3.5.4]{Marsden2001} in an open set $\mathcal{U} \subseteq V$.
In particular, the operator $D\fun{I}(u)$ must be surjective with split kernel
for any $u \in \mathcal{U}$.
Then $\mathcal{U}_\eta = \{u \in V \colon \fun{I}(u) = \eta\}$ are closed
submanifolds of $\mathcal{U}$ for any $\eta \in \fun{I}(\mathcal{U})$, as in
the finite-dimensional case. Since we consider systems of the
form~(\ref{eq:metric-system}), satisfying in particular
condition~(\ref{eq:metric-system-compatibility}), there is at least one
invariant, namely the Hamiltonian, and thus we have $k \geq 1$. Then 
condition~(\ref{eq:MpL3}) can be generalized by 
\begin{equation}
  \label{eq:minimal-degenerate}
  (\fun{F}, \fun{F})(u) = 0
  \;\iff\; 
  D \fun{F}(u) = \sum_\alpha \lambda_\alpha D\fun{I}^\alpha(u),
\end{equation}
for some constant $\lambda_\alpha \in \R$. This means that if, for a given
function $\fun{F}$, the bracket $(\fun{F},\fun{F})$ vanishes at a point $u_0$,
then $u_0$ must be a critical point of $\fun{F}$ restricted to the manifold
$\fun{I}(u) = \fun{I}(u_0) = $ constant. We referred to brackets with this
property as specifically degenerate brackets. If the only invariant is the
Hamiltonian, then we called them minimally degenerate.

The equivalent of condition~(\ref{eq:MpPL}) reads:  
$\fun{S}_\eta \coloneqq \inf \{\fun{S} \colon z \in \mathcal{U}_\eta\}>-\infty$
and there exists a constant $\kappa_\eta>0$ depending on $\eta$, such that
\begin{equation*}
  \label{eq:MpPL-V}\tag{PL${}^{\prime\prime}$}
  \frac{1}{\kappa_\eta} (\fun{S}, \fun{S}) \geq \fun{S} - \fun{S}_\eta,
  \quad \text{on $\mathcal{U}_\eta$}.
\end{equation*}
If inequality~(\ref{eq:MpPL-V}) is fulfilled, the exponential convergence of the
entropy follows as in the finite-dimensional case.
Also, $(\fun{S},\fun{S})(u) = 0$ on $\mathcal{U}_\eta$ only if $u$ is a global
minimum of $\fun{S}$ restricted to $\mathcal{U}_\eta$,
i.e. $\fun{S}(u) = \fun{S}_\eta$. 

A necessary condition for~(\ref{eq:MpPL-V}) can be stated for the special
class of specifically degenerate metric brackets, i.e., when
(\ref{eq:minimal-degenerate}) is satisfied.
Then condition~(\ref{eq:MpPL-V}) is satisfied \emph{only if} critical points of
$\fun{S}|_{\mathcal{U}_\eta}$ are global minima. In fact, if
$u \in \mathcal{U}_\eta$ is a critical point of $\fun{S}|_{\mathcal{U}_\eta}$,
the theory of Lagrange multiplier \cite[Theorem~3.5.27]{Marsden2001} gives
$D \fun{S}(u) = \sum_\alpha \lambda_\alpha D\fun{I}^\alpha(u)$, 
hence $(\fun{S},\fun{S})(u) = 0$. But if $\fun{S}(u) > \fun{S}_\eta$,
inequality~(\ref{eq:MpPL-V}) is violated. 
  
Beyond these preliminary  considerations, the mathematical analysis
of~(\ref{eq:MpPL-V}) exceeds the scope of this paper. We conclude with 
an example of a infinite-dimensional metric bracket that is specifically
degenerate and satisfies~(\ref{eq:MpPL-V}).  We shall give
physically relevant examples in Section~\ref{sec:simple}.

\begin{example}
  \label{ex:diffusion}
  
  In this example (cf.\ \cite{pjmU24}) we proceed formally.
  On the space $V$ of smooth functions from $\T \to \R$, where the torus
  $\T \coloneqq \R/2\pi\Z$ is identified with the interval $\Omega = [0,2\pi]$
  with periodic boundary conditions, we consider the metric bracket given by 
  \begin{equation}
    \label{eq:heateq-bracket}
    \big(\fun{F},\fun{G}) = \int_0^{2\pi}
    \Big(\frac{\delta \fun{F}(u)}{\delta u} \Big)^\prime
    \Big(\frac{\delta \fun{G}(u)}{\delta u} \Big)^\prime dx,
  \end{equation}
  where $v'(x) = dv(x)/dx$ denotes the derivative of $v : \T \to \R$, and the
  functional derivatives are computed with respect to the standard $L^2$ product
  (cf.~Section~\ref{sec:metriplectic}). We assume that the functions
  $\fun{F}(u)$ and $\fun{G}(u)$ are regular enough for their functional
  derivative to exist and be sufficiently smooth. The Hamiltonian and entropy
  functions are given by  
  \begin{equation*}
    \fun{H}(u) = \int_0^{2\pi}\! u(x) dx \quad \text{and}\quad
    \fun{S}(u) = \frac{1}{2} \int_0^{2\pi}\! |u(x)|^2 dx
    = \frac{1}{2} \|u\|^2_{L^2(\Omega)}.
  \end{equation*}
  Condition~(\ref{eq:metric-system-compatibility}) is satisfied since
  $\delta \fun{H}(u)/\delta u = 1$. After integration by parts, the strong
  form of~(\ref{eq:metric-system-equation}) amounts to the heat equation
  \begin{equation}
    \label{eq:heateq}
    \left\{
    \begin{aligned}
      \partial_t u &= \partial_x^2 u, &&
      (t,x) \in [0,+\infty) \times [0,2\pi],\\
      u(t,0) &= u(t,2\pi), && t\in[0,+\infty), \\
      u(0,x) &= u_0(x), && x \in [0,2\pi].
    \end{aligned}
    \right.
  \end{equation}
  First we show that~(\ref{eq:heateq-bracket}) is a minimally degenerate
  bracket, i.e., the null space is spanned by $\delta \fun{H}(u)/\delta u$.
  In fact, $(\fun{F},\fun{F}) = 0$ implies $(\delta \fun{F}(u)/\delta u)'=0$ and
  thus $\delta\fun{F}(u)/\delta u = \lambda$ where $\lambda \in \R$ is constant.
  Hence, $\delta \fun{F}(u)/\delta u = \lambda \delta \fun{H}(u)/\delta u$,
  which proves property~(\ref{eq:minimal-degenerate}), with the Hamiltonian
  being the only invariant.

  The manifolds of constant Hamiltonian,
  \begin{equation*}
    \mathcal{U}_\eta = \{u : \fun{H}(u) = \eta \in \R\},
  \end{equation*}
  consist of functions with the same average over $[0,2\pi]$. They are affine
  spaces, rather than generic manifolds. The critical points of entropy
  restricted to the constant-Hamiltonian spaces are determined by 
  \begin{equation*}
    u(x) = \lambda, \quad \int_0^{2\pi} u(x)dx = \eta.
  \end{equation*}
  Therefore, for any $\eta \in \R$  there is only one critical point, that is,
  the constant function
  \begin{equation*}
    u_\eta(x) = \eta / (2\pi).
  \end{equation*}
  The entropy of $u_\eta$ is $\fun{S}_\eta = \fun{S}(u_\eta) = \eta^2/(4\pi)$.
  The Fourier series representation,
  \begin{equation*}
    u(x) = \sum_{n \in \Z} u_n e^{i n x}, \quad u_n \in \C,
  \end{equation*}
  yields that $\fun{H}(u) = \eta$ if and only if $u_0 = \eta / (2\pi)$, hence
  \begin{equation*}
    \fun{S}(u) - \fun{S}_\eta = \pi \sum_{n \not = 0} |u_n|^2 \geq 0, \quad
    u \in \mathcal{U}_\eta,
  \end{equation*}
  with equality only if $u=u_\eta$. This shows that $u_\eta$ is a global
  minimum of $\fun{S}$ restricted to $\mathcal{U}_\eta$. Upon using again the
  Fourier series representation, one finds
  \begin{equation*}
    \big(\fun{S}, \fun{S}\big)(u) = \|u'\|^2_{L^2(\Omega)} = 2 \pi
    \sum_{n \not = 0} n^2 |u_n|^2 
    \geq 2 \pi \sum_{n\not=0} |u_n|^2,
  \end{equation*}
  and
  \begin{equation*}
    \big(\fun{S}, \fun{S}\big)(u) \geq 2 \big[\fun{S}(u) - \fun{S}_\eta\big],
    \quad u \in \mathcal{U}_\eta,
  \end{equation*}
  which is condition~(\ref{eq:MpPL-V}) with $\kappa_\eta = 2$.  
  
  In fact, the solution of~(\ref{eq:heateq}) can be readily written in
  terms of a Fourier series as 
  \begin{equation*}
    u(t,x) = \sum_{n \in \Z} e^{inx-n^2 t} u_{0,n},
  \end{equation*}
  where $u_{0,n} \in \C$ are the Fourier coefficients of the initial condition.
  With $\eta = \fun{H}(u_0) = 2 \pi u_{0,0}$, we deduce
  \begin{equation*}
    \|u(t) - u_\eta \|_{L^2(\Omega)} \leq e^{-t}
    \|u(0) - u_\eta \|_{L^2(\Omega)},
  \end{equation*}
  which shows exponential relaxation toward the entropy minimum on
  $\mathcal{U}_\eta$, with the energy $\eta$ being determined by the initial
  condition.

  In conclusion, this metric system satisfies the generalized
  Polyak--{\L}ojasiewicz condition~(\ref{eq:MpPL-V}) and all orbits completely
  relax exponentially to a solution of~(\ref{eq:entropy-principle}) with
  exponential convergence rate given by $\kappa_\eta/2$, where $\kappa_\eta$
  is the constant in~(\ref{eq:MpPL-V}). We observe that the convergence rate is
  the same as the one predicted in Proposition~\ref{th:MpPL} for
  finite-dimensional systems.
\end{example}

\section{Two examples: metric double brackets and projectors} 
\label{sec:simple}

In this section, we discuss two special cases of metriplectic systems of the
form~(\ref{eq:metric-system}), which we call metric double bracket and
projector bracket systems. We shall see that for the metric double brackets
treated in Section~\ref{sec:metr-double-brackets}, the dissipation mechanism
does not completely relax the state of the system (in the sense of
Section~\ref{sec:remarks-relax-equil}), while for the projector brackets  of
Section~\ref{sec:projector-based-metric-bracket} complete relaxation is
achieved. We shall discuss and compare the properties of these two metric
brackets on the basis of the insights gained in
Section~\ref{sec:remarks-relax-equil}.   

We select two benchmark equilibrium problems, and we attempt to construct a
relaxation method to solve them by using the two considered metric brackets.
The benchmark problems are: the reduced Euler equations  
(cf.~Section~\ref{sec:Euler-problem}), but with Dirichlet boundary conditions
replaced by periodic boundary conditions, and an analytically solvable
model derived from the  reduced Euler  equations.  In both cases periodic
boundary conditions give rise to an additional invariant other than the
Hamiltonian, and this allows us to examine cases where specifically
degenerate brackets are not minimally degenerate. We shall return to
the original problem with Dirichlet boundary conditions later. 

Therefore in both cases, the domain is $\T^2 \coloneqq (\R / 2\pi \Z)^2$
with coordinates $x = (x_1,x_2)$, and the phase space $V$ is the space of
smooth functions $v \colon \T^2 \to \R$. As usual, $\T^2$ is identified with the
square $\Omega = [0,2\pi]^2$ with periodic boundary conditions. 
On such a periodic domain, the scalar vorticity $\omega = - \Delta \phi$
(cf. Section~\ref{sec:Euler-problem} for the definitions) must have zero 
average, i.e.,
\begin{equation*}
  \omega_\Omega \coloneqq \frac{1}{4\pi^2} \int_{\Omega} \omega(x) dx = 0.
\end{equation*}
We use systematically the subscript $\Omega$ to denote the average over the
domain $\Omega$. We choose the whole space $V$ as the phase space and define
the vorticity by 
\begin{equation*}
  \omega = u - u_\Omega,
\end{equation*}
but we impose the additional constraint
\begin{equation*}
  \fun{M}(u) = \int_\Omega u(x) dx = 4 \pi^2 u_\Omega = \fun{M}_0 \in \R,
\end{equation*}
as well as energy conservation $\fun{H}(u) = \fun{H}_0 \in \R$.
Given $u \in V$, and thus $\omega$, the Poisson equation determines the stream
function $\phi$ modulo a constant, which we set to zero;  hence 
\eqref{eq:Poisson-eq}  is replaced by   
\begin{equation}
  \label{eq:Poisson-eq-periodic}
  -\Delta \phi = u - u_\Omega, \quad \phi_\Omega = 0.
\end{equation}
(Equivalently, we could have chosen the phase space to be the subspace
of functions satisfying $u_\Omega=0$ and $u=\omega$.)

Both the considered benchmark problems can be formulated as variational
problems: given a Hamiltonian function $\fun{H}$, and a regular value
$\eta = (\fun{M}_0, \fun{H}_0)$ for the two invariants
$\fun{I} = (\fun{I}^1, \fun{I}^2) \coloneqq (\fun{M},\fun{H})$, find
\begin{subequations}
  \label{eq:test-problems}
    \begin{equation}
      \min \{\fun{S}(u) \colon \fun{I}(u) = \eta\},
    \end{equation}
    with entropy 
    \begin{equation}
      \fun{S}(u) = \frac{1}{2} \int_\Omega \omega^2 dx
      = \frac{1}{2} \|u-u_\Omega\|^2_{L^2(\Omega)}.
      \label{Som2}
    \end{equation}
\end{subequations}
The two test cases differ by the choice of the Hamiltonian. In summary, 
\begin{itemize}
  
\item Analytical test case: Given a function $h\colon \T^2 \to \R$, let
  \begin{equation}
    \label{eq:analytical_H}
    \fun{H}(u) = \int_\Omega h\, \omega\, dx = (h - h_\Omega, u)_{L^2(\Omega)}.
  \end{equation}
  The solutions of (\ref{eq:test-problems}) with~(\ref{eq:analytical_H}) are
  equilibria of the linear advection equation
  \begin{equation*}
    \partial_t u + [h,u] = 0,
  \end{equation*}
  with $[f,g] = \partial_1 f \partial_2 g - \partial_1 g \partial_2 f$.
  
\item Reduced Euler test case: We again use $\fun{S}$ as in \eqref{Som2}, but now
  \begin{equation}
    \label{eq:Euler_H_periodic}
    \fun{H}(u) = \frac{1}{2} \int_\Omega |\nabla \phi|^2 dx =
    \frac{1}{2} (\phi, u)_{L^2(\Omega)},
  \end{equation}
  i.e., we assume \eqref{eq:Euler-S-H} with $s(y)=y^2/2$, and we assume
  $\phi$ depends  on $u$ via \eqref{eq:Poisson-eq-periodic}.  
  Solutions of~(\ref{eq:test-problems}) with \eqref{eq:Euler_H_periodic} are
  equilibria of the reduced Euler  equations on the flat torus $\T^2$. 

\end{itemize}
In both cases, problem~(\ref{eq:test-problems}) can be solved analytically.
In order to compute the solutions, we first find the set of critical points of
entropy restricted to $\mathcal{U}_\eta = \{u \colon \fun{I}(u)=\eta\}$, i.e.,
\begin{equation}
  \label{eq:Ceta}
  \mathfrak{C}_\eta \coloneqq \{ u \colon\;
  D\fun{S}(u) = \sum_\alpha \lambda_\alpha D\fun{I}^\alpha(u), \;
  \fun{I}(u) = \eta
  \}.
\end{equation}
Then we find the minimum of $\fun{S}$ on $\mathfrak{C}_\eta$. We summarize
here the results, which  will be used to assess the
metriplectic relaxation methods. 

\smallskip

\noindent{\it Solution for the analytical test case:} For the case
of Eq.~(\ref{eq:analytical_H}), the set $\mathfrak{C}_\eta$ is given by
\begin{equation}
  \label{eq:complete-relaxation-analytic}
  \begin{aligned}
    &u - u_\Omega = \lambda_1 + \lambda_2 (h - h_\Omega), \quad
    \fun{I}(u) = \eta.
  \end{aligned}
\end{equation}
Then $\lambda_1 = 0$ and $u_\Omega = \fun{M}_0/(4\pi^2)$. Upon
multiplying by $h-h_\Omega$ and integrating over $\Omega$, we deduce
\begin{equation*}
  \fun{H}_0 = \lambda_2 \|h-h_\Omega\|^2_{L^2(\Omega)},
\end{equation*}
from which we can compute $\lambda_2$. Hence, the set $\mathfrak{C}_\eta$
contains the following single point: 
\begin{equation}
  \label{eq:u-eta_analytic}
  u_\eta = \frac{\fun{M}_0}{4\pi^2} +
  \frac{\fun{H}_0}{\|h-h_\Omega\|^2_{L^2(\Omega)}} (h-h_\Omega) 
\end{equation}
and the value of entropy on this  unique critical point is 
\begin{equation}
  \label{eq:Seta-analytic}
  \fun{S}_\eta = \min \{\fun{S}(u) \colon \fun{I}(u)=\eta\} =
  \frac{\fun{H}_0^2/2}{\|h-h_\Omega\|^2_{L^2(\Omega)}}.
\end{equation}
One can check that this is the minimum of the entropy on $\mathcal{U}_\eta$.

\smallskip

\noindent{\it Solution for the reduced Euler equations:} For the
case of  \eqref{eq:Euler_H_periodic}, elements of $\mathfrak{C}_\eta$
satisfy
\begin{equation}
  \label{eq:complete-relaxation-vorticity2d-periodic}
    \left\{
    \begin{aligned}
      &u - u_\Omega = \lambda_1 + \lambda_2 \phi,
      \quad \fun{I}(u) = \eta, \\
      &\text{with $\phi$ solution of~(\ref{eq:Poisson-eq-periodic}).}
    \end{aligned}
    \right.
\end{equation}
Since $\phi_\Omega = 0$, we have $\lambda_1 = 0$ and
$u_\Omega = \fun{M}_0/(4\pi^2)$, as before. Then $(\phi,\lambda_2)$ must be a
solution of the eigenvalue problem  
\begin{equation*}
  -\Delta \phi = \lambda_2 \phi, \qquad \phi_\Omega=0,
\end{equation*}
which is readily solved in terms of Fourier series. We find that the set
$\mathfrak{C}_\eta$ consists of vorticity fields
$u - u_\Omega = \omega = \lambda_2 \phi$,
with $\phi$ being an eigenfunction of $-\Delta$ corresponding to the
eigenvalue $\lambda_2>0$, and with norm
$\|\phi\|_{L^2(\Omega)}^2 = 2\fun{H}_0/\lambda_2$. Then the entropy evaluated
on the constrained critical points amounts to $\lambda_2\fun{H}_0$. 
It follows that the entropy minimum on $\mathcal{U}_\eta$ corresponds to the
lowest non-trivial eigenvalue, which is $\lambda_2 = 1$. The corresponding
stream function must be of the form 
\begin{equation*}
  \phi(x) = a_1 \cos(x_1 + \theta_1) + a_2 \cos(x_2 + \theta_2),
\end{equation*}
with arbitrary phase shifts $\theta_1$, $\theta_2$, and with coefficient
$a_1$, $a_2$ determined by the condition
$\|\phi\|_{L^2(\Omega)}^2=2\fun{H}_0$. Since $\lambda_2 = 1$, this
amounts to $a_1^2 + a_2^2 = \fun{H}_0/\pi^2$. Thus,
\begin{equation}
  \label{eq:u-eta_Euler_periodic}
  \omega(x) = \phi(x) = \frac{\sqrt{\fun{H}_0}}{\pi} \big[
    \cos \theta_0 \cos(x_1+\theta_1) + \sin \theta_0 \cos(x_2+\theta_2)\big],
\end{equation}
with arbitrary phases $\theta_0$, $\theta_1$, and $\theta_2$ in $[0,2\pi)$.
  
From the analytical solution, \eqref{eq:u-eta_Euler_periodic}, we deduce
that the entropy minimum constrained to $\fun{I}(u)=\eta$ is not attained at
an isolated point, but on a family of points parameterized by three
phases. The constrained minimum value of the entropy is given by 
\begin{equation}
  \label{eq:Seta-Euler_periodic}
  \fun{S}_\eta = \min \{\fun{S}(u) \colon \fun{I}(u)=\eta\} = \fun{H}_0,
\end{equation}
since $\lambda_2 = 1$.

\subsection{Metric double brackets}
\label{sec:metr-double-brackets}

The first example is given by the metric double bracket
\cite{Bloch2013,Gay-Balmaz2013,Gay-Balmaz2014} (not to be confused with the
double bracket of Flierl and Morrison \cite{Flierl2011} discussed in the
introduction). We recall the general definition first, but quickly restrict the
discussion to the examples. 

In general, metric double brackets originate from a Lie algebra. 
If $\mathfrak{g}$ is a Lie algebra with Lie brackets $[\cdot,\cdot]$,
let the vector space $V$ be its dual, $V = \mathfrak{g}^*$.
The functional derivative of a function $f \in C^\infty(\mathfrak{g}^*)$ is
computed with respect to the duality pairing between $\mathfrak{g}^*$ and
$\mathfrak{g}$, that is, $\delta f(u) /\delta u \in \mathfrak{g}$ is the unique
element of $\mathfrak{g}$ such that
$Df(u) v = \langle \delta f(u) /\delta u, v\rangle$ for all
$v \in V = \mathfrak{g}^*$. Under these conditions, it is well-known
\cite{Marsden2001} that the Lie bracket in $\mathfrak{g}$ induces two Poisson
brackets in $C^\infty(\mathfrak{g}^*)$, namely,
$\{f,g\}_\pm = \pm \langle u, [\tfrac{\delta f(u)}{\delta u},
  \tfrac{\delta g(u)}{\delta u} ]\rangle$, so that both
$(\mathfrak{g}^*,\{\cdot,\cdot\}_\pm)$ are Poisson manifolds. 
However, if in addition $\mathfrak{g}$ is equipped with a positive definite
bilinear form $\gamma : \mathfrak{g} \times \mathfrak{g} \to \R$, on the space
of smooth functions $C^\infty(\mathfrak{g}^*)$, we can also define the symmetric
bracket
\begin{equation*}
  (f,g) = \gamma\Big( \Big[
    \frac{\delta f(u)}{\delta u},
    \frac{\delta h(u)}{\delta u}\Big], \Big[
    \frac{\delta g(u)}{\delta u},
    \frac{\delta h(u)}{\delta u}\Big]
  \Big),
\end{equation*}
for any fixed function $h \in C^\infty(\mathfrak{g}^*)$. One can readily check
that this is a metric bracket preserving the Hamiltonian function $h$.

Formally at least, this construction can be extended to infinite-dimensional
systems. Let $V$ and $W$ be Banach spaces, with a nondegenerate duality pairing
$\langle \cdot, \cdot \rangle_{V \times W} : V \times W \to \R$, and let $W$ be
equipped with (i) a symmetric positive definite bilinear form
$\gamma : W \times W \to \R$ and (ii) a bilinear antisymmetric operation
$[\cdot,\cdot] : W \times W \to W$. (For the purpose of defining the metric
bracket we do not need to require $[\cdot,\cdot]$ to be a Lie bracket, that is,
we can relax the Jacobi identity.)
Then, given a fixed Hamiltonian $\fun{H} \in C^\infty(V)$,
we can construct the bilinear form 
$(\cdot, \cdot) : C^\infty(V) \times C^\infty(V) \to C^\infty(V)$ given by
\cite[Eq. (2.9)]{Gay-Balmaz2013}
\begin{equation}
  \label{eq:gmdb}
  (\fun{F}, \fun{G}) \coloneqq \gamma\big(
  \Big[\frac{\delta \fun{F}}{\delta u}, \frac{\delta\fun{H}}{\delta u}\Big],
  \Big[\frac{\delta \fun{G}}{\delta u}, \frac{\delta\fun{H}}{\delta u}\Big]
  \big),
\end{equation}
where the functional derivative are evaluated with respect to the duality
pairing between $V$ and $W$, and thus are elements of $W$. We remark that
when $W = V'$ is the (topological) dual of $V$, that is, the space of continuous
linear functionals on $V$, $\delta \fun{F}(u)/\delta u$ exists and it is equal
to $D\fun{F}(u)$ for any $\fun{F} \in C^1(V)$ and for all $u \in V$. In general,
however, $\delta \fun{F}(u)/\delta u$ does not always exists for all $\fun{F}$.

As an example, let $V=W$ be the space of smooth functions from $\T^d \to \R$,
identified with functions over $\Omega = [0,2\pi]^d$ with  periodic boundary
conditions. If $[u,v]_J = \nabla u \cdot J \nabla v$ is a Poisson bracket on
$\R^d$ (not necessarily canonical), $\fun{H}(u)$ is a given Hamiltonian 
function, and $\gamma$ is given by the standard product in $L^2(\Omega)$,
then~(\ref{eq:gmdb}) reduces to 
\begin{equation}
  \label{eq:paired-bracket-L2}
  (\fun{F},\fun{G}) \coloneqq \int_{\Omega}
  \Big[\frac{\delta \fun{F}}{\delta u},\frac{\delta \fun{H}}{\delta u} \Big]_{J}
  \Big[\frac{\delta \fun{G}}{\delta u},\frac{\delta \fun{H}}{\delta u} \Big]_{J}
  dx.
\end{equation}
In the following, let $d=2$,
$x=(x_1,x_2)\in \Omega = [0,2\pi]^2 \subset \R^2$ with periodic boundary
conditions, and for $J$, we choose the canonical Poisson tensor in $\R^2$, so
that $[\cdot,\cdot]_J = [\cdot,\cdot]$ is the canonical bracket defined after
Eq.~(\ref{eq:Euler-equilibrium}).

We construct a relaxation method based on this bracket for the solution of
the two test problems introduced at the beginning of this section. In both
cases we consider a field $u \in V$ evolving from an initial condition
$u_0$ according to Eqs.~(\ref{eq:metric-system}) and with bracket given
by~(\ref{eq:paired-bracket-L2}).

We start from problem~(\ref{eq:test-problems}) with~(\ref{eq:analytical_H})
for the linear advection equation. With those choices of bracket, Hamiltonian
and entropy, Eq.~(\ref{eq:metric-system-equation}) amounts to   
\begin{equation*}
  \int_\Omega \frac{\delta \fun{F}}{\delta u} \frac{\partial u}{\partial t} dx
  = -\int_\Omega \big[\frac{\delta \fun{F}}{\delta u}, h \big]
  \big[u, h\big]dx,
\end{equation*}
for all $\fun{F}$. This can be viewed as the weak form of the evolution
equation, with $\delta \fun{F}/\delta u$ being the test function. After
integration by parts, one obtains the evolution equation for $u$ in strong form,
namely, 
\begin{equation*}
  \partial_t u = \big[h,[h,u]\big].
\end{equation*}
We observe that $[u,h] = \div (X_h u) = X_h \cdot \nabla u$, where
$X_h = \transpose{(\partial_2 h, -\partial_1 h)}$ is the Hamiltonian vector
field generated by $h$ with canonical Poisson bracket in $\R^2$. Hence
\begin{equation}
  \label{eq:parallel-diffusion}
  \partial_t u = \div(X_h \otimes X_h \nabla u),
\end{equation}
which shows that this particular combination of metric bracket and entropy
describes anisotropic diffusion, parallel to the field lines of $X_h$, or
equivalently along the contours of the function $h$.

The Cauchy problem associated to~(\ref{eq:parallel-diffusion}) with periodic
boundary conditions can be solved analytically, at least under suitable
conditions. We consider the equation in a region of the domain where the
contours of $h$ are closed simple curves and $\nabla h \not=0$ (as relevant in
the example discussed below). Let $\xi$ be a function such that 
\begin{equation*}
  [\xi, h] = 1,
\end{equation*}
in the considered subdomain. Since $[\xi,h] \not=0$, the pair of functions
$(\xi,h)$ defines a local coordinate system with inverse Jacobian determinant
$\mathcal{J}^{-1} = [\xi,h] = 1$ and such that 
$X^1_h \coloneqq X_h \cdot \nabla \xi = [\xi,h] = 1$ and
$X^2_h \coloneqq X_h \cdot \nabla h = [h,h] = 0$.
Then, in these coordinates the contravariant components $X^i_h$, $i=1,2$,
of the vector field $X_h$ as well as the Jacobian determinant $\mathcal{J}$
are constant. In order to compute $\xi$, we observe that, along a contour we
must have $dx/ds = X_h/|X_h|$ where $s$ is the arclength (with the Euclidean
metric, $ds^2 = dx_1^2 + dx_2^2$), because the field $X_h$ is tangent to $h=$
constant contours. Hence
\begin{equation*}
  \frac{d\xi}{ds} = \frac{dx}{ds} \cdot \nabla \xi
  = \frac{X_h \cdot \nabla \xi}{|X_h|} = \frac{1}{|\nabla h|},
\end{equation*}
so that $d\xi = |\nabla h|^{-1} ds$. With some abuse of notation,
let $x(\xi,h)$ be the coordinate map $(\xi,h) \mapsto x$. Then we also
have, with $h$ fixed,
\begin{equation*}
  \frac{\partial x(\xi,h)}{\partial \xi} =
  \frac{d x}{d s} \frac{d s}{d \xi} =
  X_h\big(x(\xi,h)\big),
\end{equation*}
hence the coordinate map $x(\xi,h)$ is essentially related to the flow of the
vector field $X_h$. This conclusion also follows from the fact that
$\partial_\xi x$ is a vector of the covariant basis, hence it must hold that
$\partial_\xi x \cdot \nabla h = 0$ and $\partial_\xi x \cdot \nabla \xi = 1$,
which imply $\partial_\xi x = X_h$.

Equation~(\ref{eq:parallel-diffusion}) in the coordinates $(\xi,h)$ takes the
form of a heat equation, 
\begin{equation*}
  \partial_t \tilde{u} - \partial_\xi^2 \tilde{u} = 0,
\end{equation*}
where the new unknown is given by $u(t,x) = \tilde{u}(t,\xi,h)$. Since, per
assumption, the contours of $h$ are closed, the function $\tilde{u}$ must be
periodic in $\xi$ with period possibly depending on $h$, i.e., there is
$\ell_h$ such that $\tilde{u}(t,\xi+\ell_h,h) = \tilde{u}(t,\xi,h)$.
The period $\ell_h$ is given by the variation of $\xi$ over a full loop around
the considered contour of $h$, that is, 
\begin{equation*}
  \ell_h = \int_{C_h} d\xi = \int_{C_h} \frac{ds}{|\nabla h|},
\end{equation*}
where $C_h$ is the considered contour of $h$. This gives a way to compute
$\ell_h$ from $|\nabla h|$ on a contour $C_h$.

We rescale the variable $\xi$ to an angle $\vartheta \in [0,2\pi]$, i.e.
$\vartheta = 2\pi\xi/\ell_h$. In terms of the angle $\vartheta$,
equation~(\ref{eq:parallel-diffusion}) amounts to
\begin{equation*}
  \partial_t v - \kappa_h \partial_\vartheta^2 v = 0,
\end{equation*}
where $\kappa_h = (2\pi/\ell_h)^2$, and the rescaled unknown is given by
$u(t,x) = \tilde{u}(t,\ell_h\vartheta/(2\pi),h) = v(t,\vartheta,h)$.
This is the classic heat equation on $[0,2\pi]$ with periodic boundary
conditions, and it can be readily solved by Fourier series. The solution is
\begin{equation*}
  v(t,\vartheta,h) = \sum_{n \in \Z} \hat{v}_n(0)
  e^{-n^2 \kappa_h t + in\vartheta},
\end{equation*}
where $\hat{v}_n(0)$ are the Fourier coefficients of the initial condition
for $v$, which is explicitly given by $v(0,\theta,h) = u_0\big(x(\xi,h)\big)$.
Each Fourier mode with $n \not=0$ decays exponentially with exponential decay
time given by $1/(n^2\kappa_h)$. The relaxation time is identified with the
decay time of the slowest modes ($n = \pm 1$), 
\begin{equation}
  \label{eq:relaxation-time}
  \tau_h = 1/ \kappa_h = (\ell_h / 2\pi)^2, \qquad
  \ell_h = \int_{C_h} d\xi = \int_{C_h} \frac{ds}{|\nabla h|}.
\end{equation}
This expression allows us to estimate numerically the relaxation time from a
sample of points on the considered contour of $h$, which can be obtained
by integrating the ordinary differential equation for the flow of $X_h$. 
The limit for $t \to +\infty$ of the solution exists and is equal to the average
of the initial condition on the contours of $h$. Explicitly this can be computed
as 
\begin{equation*}
  u_\infty(h) \coloneqq \hat{v}_0(0) = \frac{1}{\ell_h} \int_0^{\ell_h}
  u_0\big(x(\xi,h)\big) d\xi
  = \frac{1}{\ell_h} \int_{C_h} \frac{u_0 ds}{|\nabla h|}.
\end{equation*}
Therefore, in general the limit of the solution retains some
information of the initial condition, while the completely relaxed
solutions~(\ref{eq:u-eta_analytic}) only depends on the energy
$\fun{H}_0 = \fun{H}(u_0)$ of the initial condition $u_0$. This
implies that for generic initial conditions, the limit of the solution of a
metric dynamical system is in general not a solution of the variational
principle~(\ref{eq:test-problems}).

\begin{figure}
  \centering
  \includegraphics[width=\columnwidth]{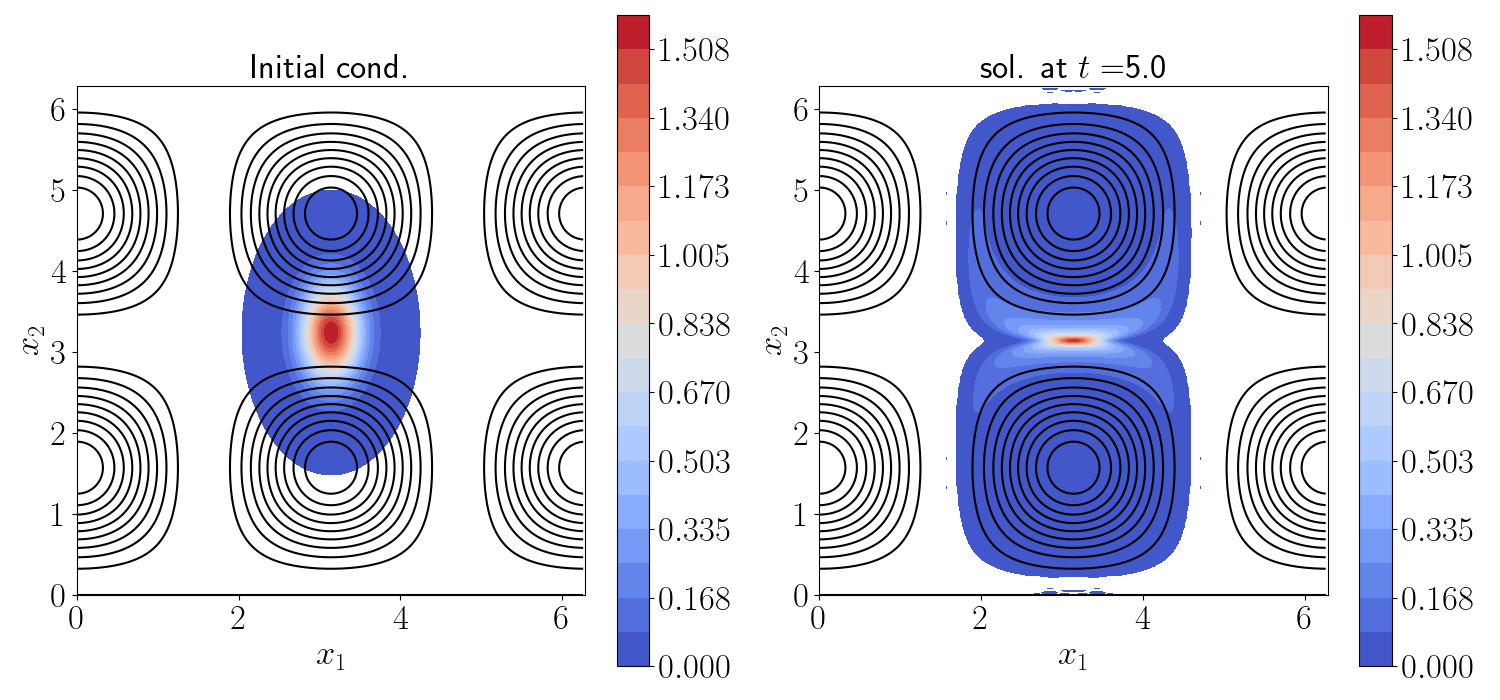}
  \includegraphics[scale=0.35]{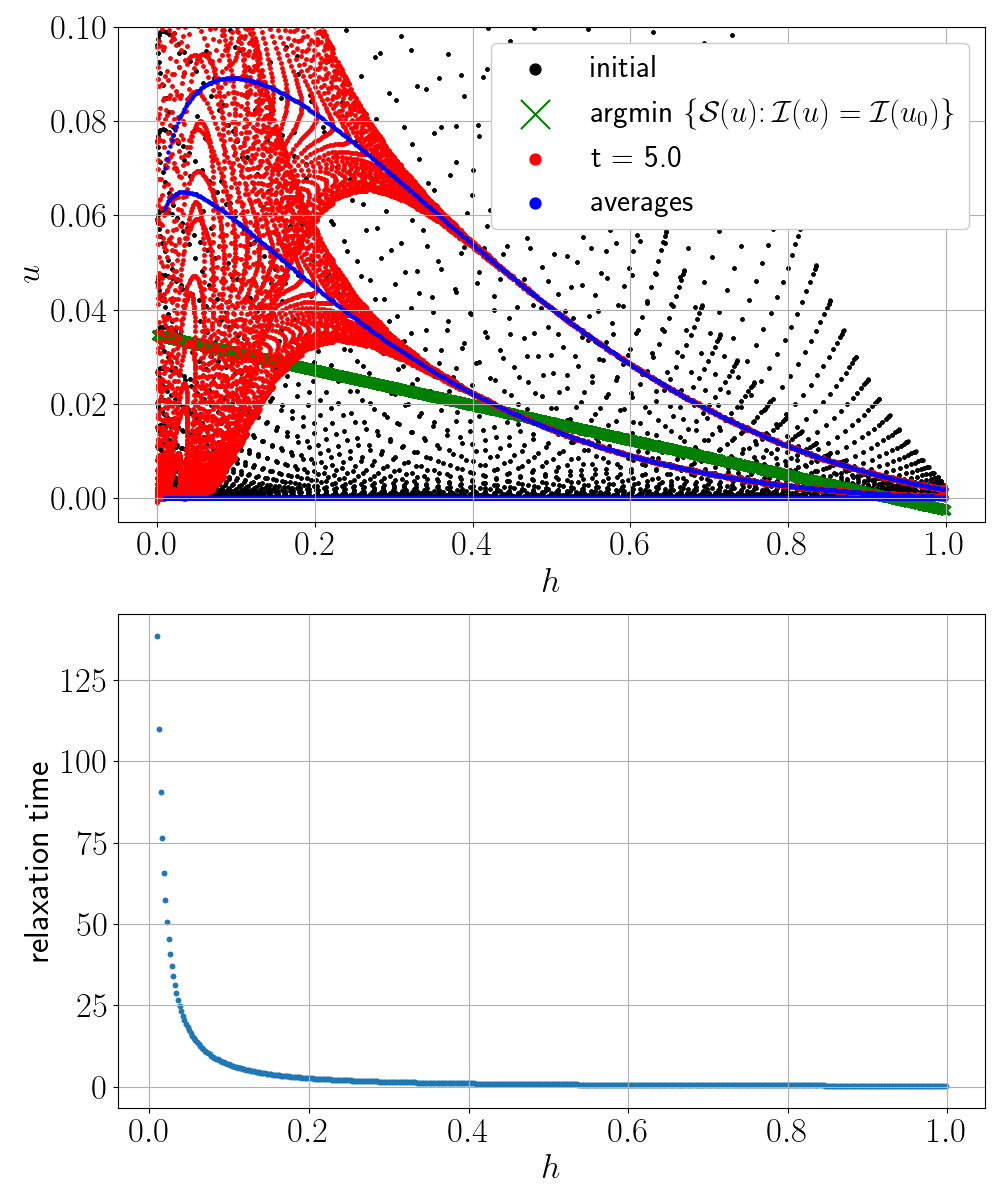}
  \caption{\label{fig:1} Example of solution of
    Eq.~(\ref{eq:parallel-diffusion}) with Hamiltonian~(\ref{eq:islands-h}). 
    Upper panels: initial condition and final state, compared to the contours of
    $h$ (black circular curves). Middle panel: visualization of the functional
    relation between $h$ and $u$ obtained by plotting the points
    $(h_{ij},u_{ij})$, with $h_{ij}$ and $u_{ij}$ being the values of $h$ and
    $u$, at the node $(i,j)$ of the computational grid.
    Lower panel: relaxation time $\tau_h$, computed from
    Eq.~(\ref{eq:relaxation-time}) on the contours of the two central (full)
    islands, as a function of $h$.
    (For clarity, in the color maps we display the solution
    $u$ only where $u \geq 10^{-4}$.)}  
\end{figure}

Figure \ref{fig:1} shows the result of the numerical solution
of~(\ref{eq:parallel-diffusion}) with Hamiltonian
\begin{equation}
  \label{eq:islands-h}
  h(x) =\cos^2(x_1) \sin^2(x_2).
\end{equation}
The contours of $h$ form a periodic array of islands, cf. the black contours in
the upper panels of Fig.~\ref{fig:1}. In each island, $h$ takes the same values.
The initial condition, represented as a color map, is an anisotropic Gaussian
centered between two islands, i.e., $u_0(x) = u_G(x)$ with 
\begin{equation}
  \label{eq:initial_gaussian}
  u_G(x) = \frac{1}{N} \exp\Big[ - \frac{(x_1-x_{0,1})^2}{w_1^2}
    - \frac{(x_2-x_{0,2})^2}{w_2^2}\Big],
\end{equation}
with center $x_0 = (\pi, \pi+0.1)$, $w_1 = 0.25$, $w_2 = 0.4$, and
$N = 2\pi w_1 w_2$. The solution is obtained with a standard spectral method
with Fourier basis, on a $256 \times 256$ uniform grid. The time integrator is
the standard 4th order explicit Runge-Kutta method with time step
$\Delta t = 10^{-4}$. The parallel diffusion
equation~(\ref{eq:parallel-diffusion}) tends to equalize the solution on the
contours of $h$, but the dynamics at the boundary of the islands is very slow:
the Hamiltonian vector field $X_h$ at the boundary of the islands is zero and
the solution remains constant on those boundary contours
(referred to as separatrices). 
The color map of the relaxed state is shown in Fig.~\ref{fig:1} upper panel,
while the middle panel represents the functional relation between $h$ and the
solution $u$ by marking on the $h$-$u$ plane a point $(h_{ij},u_{ij})$ for
each grid node $x_{ij}$. At the initial time, (black markers) there is no
relation between the values of $u$ and those of $h$, showing that the initial
condition~(\ref{eq:initial_gaussian}) is far from an equilibrium. As the
solution evolves, all values of $u$ sampled on the same contour of $h$ tend to a
common value, but the ``condensation'' of points on a line is slower for $h$
small, that is, near the separatrices. Blue markers show the average of the
initial condition on each contour of $h$: one can see that the solution tends to
the averages as predicted by the analytical solution.
In the limit $t \to +\infty$ the relation between $h$ and $u$ is multi-valued
with a countable set of branches, one for each island. In Fig.~\ref{fig:1} one
can distinguish the upper branch (larger values of $u$), corresponding to the
island that contains the maximum of the initial condition. (In order to separate
the two branches the center $x_0$ of the Gaussian has been shifted up in the
direction $x_2$.) A second branch with lower values of $u$ corresponds to the   
neighboring island. All the other islands do not overlap with the initial
condition significantly and therefore appear as a line of points $u=0$ for all
$h$. In Fig.~\ref{fig:1} the analytical solution~(\ref{eq:u-eta_analytic}) for
the completely relaxed state is also shown (green crosses), and it is clearly
different from the obtained equilibrium. Fig.~\ref{fig:1}, lower panel, shows
the relaxation time $\tau_h$ computed according to
Eq.~(\ref{eq:relaxation-time}). This result confirms that the relaxation becomes
progressively slower as $h$ approaches $h=0$, that is, near the separatrices of
the islands. In the limit $h \to 0$ we have $\tau_h \to +\infty$ consistently
with the fact that $\nabla h = 0$ and $X_h = 0$ on the separatrices, cf. the
denominator in Eq.~(\ref{eq:relaxation-time}).

We now consider problem~(\ref{eq:test-problems})
with~(\ref{eq:Euler_H_periodic}) for the Euler equations. For this choice of
entropy and Hamiltonian, Eq.~(\ref{eq:metric-system-equation}) 
with~(\ref{eq:paired-bracket-L2}) amounts to the anisotropic diffusion
equation (\ref{eq:parallel-diffusion}), but with $h$ replaced by $\phi$, which
depends on the state variable $u$, and thus on the vorticity
  $\omega = u - u_\Omega$, via equation~(\ref{eq:Poisson-eq-periodic}).
Hence the problem is nonlinear with a cubic nonlinearity, and in general no
analytical solution is known (to the best of our knowledge). 

Figure \ref{fig:2} shows an example of relaxation of an initially anisotropic
vortex. The initial state is again given by~(\ref{eq:initial_gaussian}), but now
with $x_0=(\pi,\pi)$, $w_1=0.3$, $w_2=1.0$, and $N=1$, on $[0,2\pi]^2$ with a
uniform mesh of $256 \times 256$ nodes. The time integrator is the standard
4th-order explicit Runge-Kutta method with time step $\Delta t = 10^{-3}$. The
solution relaxes to a symmetric vortex, which is an equilibrium of the Euler
equations. During the evolution, the Hamiltonian is constant and the entropy is
monotonically dissipated, consistently with~(\ref{eq:metric-system-properties}).
However, the entropy appears to converge to a value that is higher than its
constrained minimum $\fun{S}_\eta$, given in Eq.~(\ref{eq:Seta-Euler_periodic}),
and indicated by the thick horizontal line in Fig.~\ref{fig:2}.
The fact that the final state is (a numerical approximation of) an equilibrium
of the Euler equations can be deduced from the plot of the final state: the
solution for the vorticity $\omega$ appears to be  constant on the contours of
the potential $\phi$. A more quantitative indication is provided in
Fig.~\ref{fig:3}, where the relation between $\phi$ and $\omega$ is
represented. At the initial time $t=0$, there is no functional relation between
$\omega$ and $\phi$: the values $(\phi,\omega)$ on the computational grid do not
belong to a curve. As the state relaxes, the scatter of points is reduced, and
at the final time, one finds a clear functional relation. 

\begin{figure}
  \centering
  \includegraphics[width=\columnwidth]{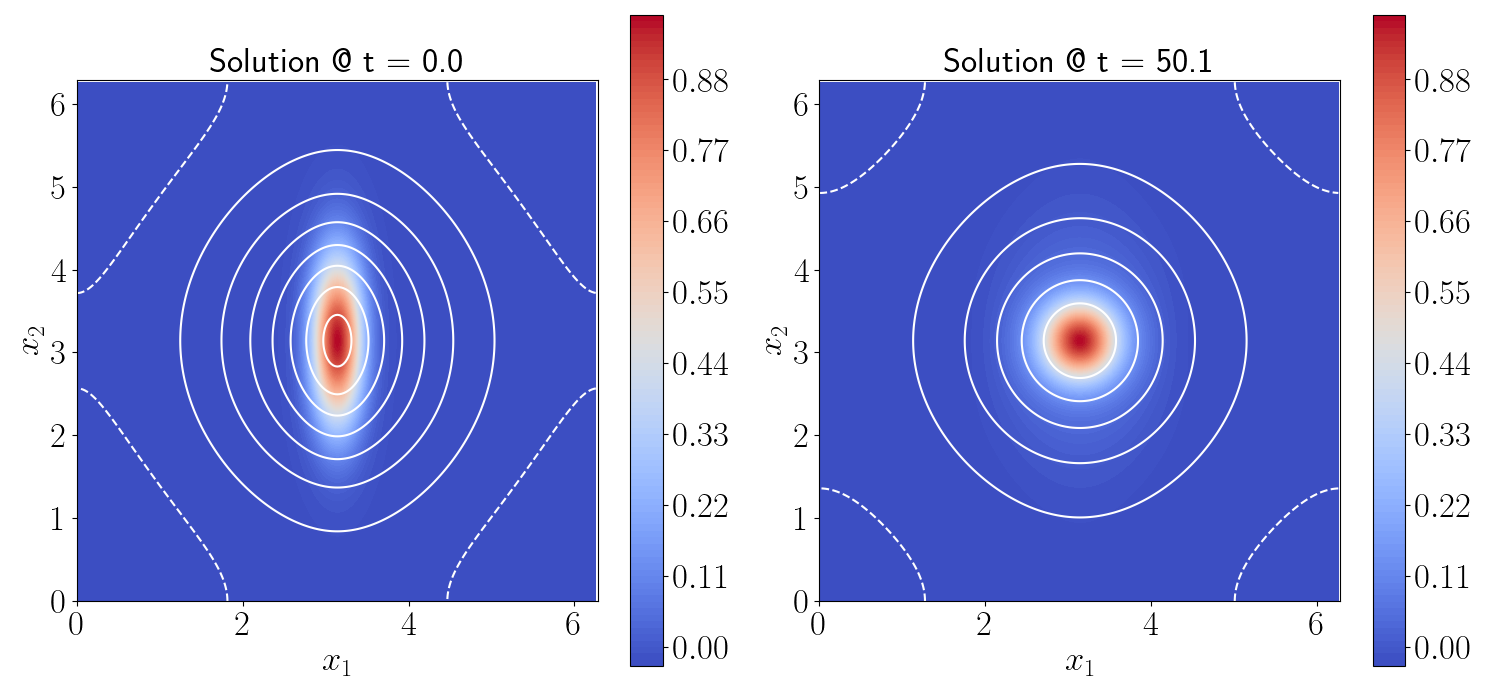}
  \includegraphics[width=\columnwidth]{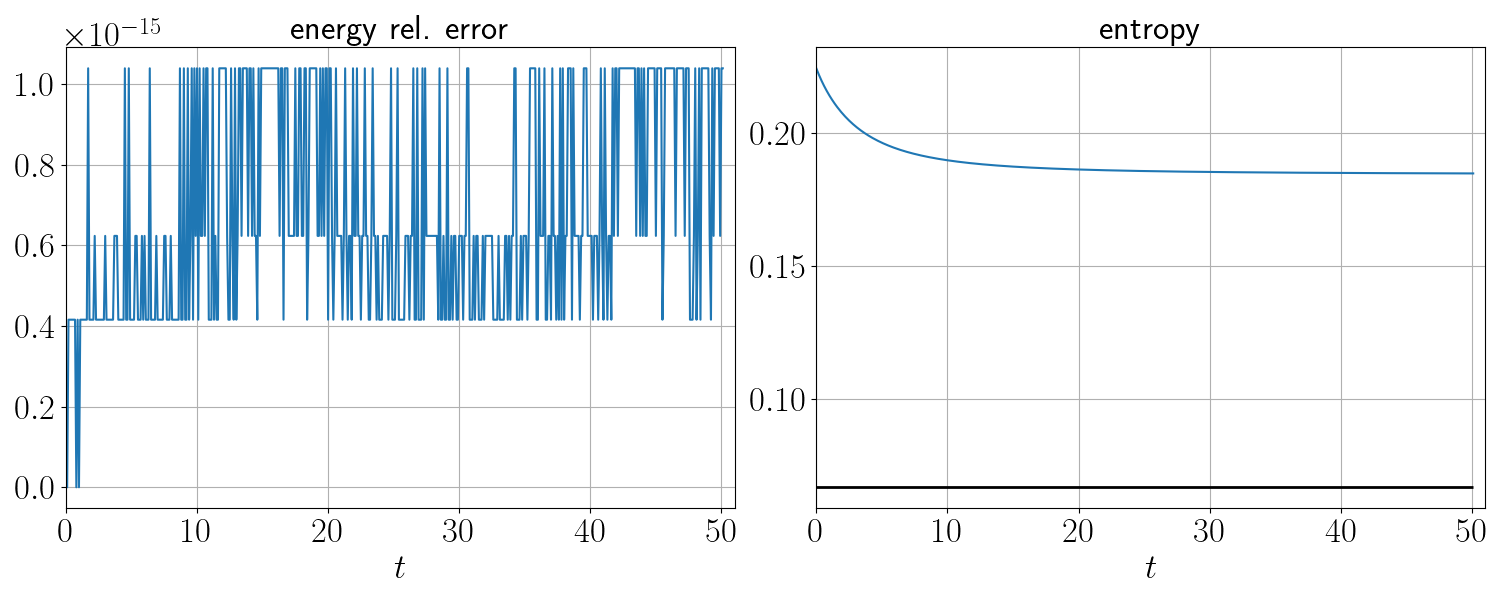}
  \caption{\label{fig:2} Metriplectic relaxation of a vortex toward an
    equilibrium of the reduced Euler equations,
    using~(\ref{eq:paired-bracket-L2}) with (\ref{eq:Euler-S-H}) and
    $s(y) = y^2/2$.
    Top row: initial and final state of the system; 
    the color scheme represents the vorticity $\omega = u - u_\Omega$; white
    lines represent the contours of the potential $\phi$.
    Bottom row: relative error of the Hamiltonian and the value of entropy
    during the evolution. The thick horizontal line
    indicates the constrained entropy minimum $\fun{S}_\eta = \fun{H}_0$, 
    cf.~Eq.~(\ref{eq:Seta-Euler_periodic}). } 
\end{figure}

\begin{figure}
  \centering
  \includegraphics[scale=0.4]{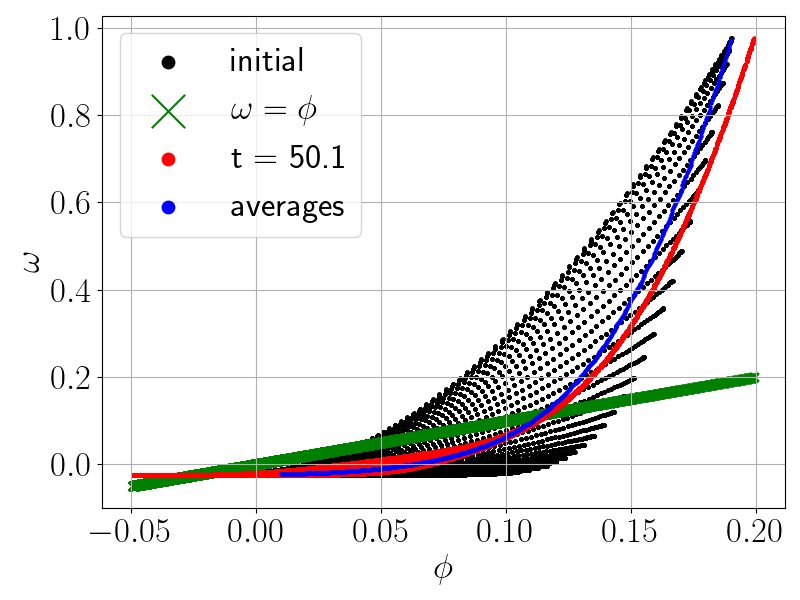}
  \caption{\label{fig:3} Visualization of the relation between the potential
    $\phi$ and the vorticity $\omega$, for the initial and final state of the
    calculation. The green crosses mark the linear relation for a minimum
    entropy state, Eq.~(\ref{eq:u-eta_Euler_periodic}).
    The data marked ``averages'' represent the average of the initial condition
    on the contours of the corresponding potential $\phi$.} 
\end{figure}

For a completely relaxed state, we expect a linear relation between the
potential and the vorticity, Eq.~(\ref{eq:u-eta_Euler_periodic}), and this is
indicated by green crosses in Fig.~\ref{fig:3}.
The obtained relationship is however very different, proving that the dynamical
system reaches an equilibrium that does not satisfy the variational principle of
minimum constrained entropy~(\ref{eq:entropy-principle}). From the numerical
experiments in Fig.~\ref{fig:3}, one can see that the relaxed state is close
to the average of the initial condition on the contours of the initial
potential. The average of the initial condition gives the exact long-time limit
of the solution in the case of the linear problem (\ref{eq:parallel-diffusion}),
and it is not expected to give an accurate prediction of the relaxed state in
general. Yet we observe that the solution converges to a state close to the
average of the initial condition.  

We conclude that the relaxation mechanism of~(\ref{eq:paired-bracket-L2}) fails
to capture the linear profile encoded in the choice of the entropy function. The
relaxed state is an equilibrium of the reduced Euler equations, but
corresponding to a profile that differs from the target one and that depends
on the initial condition in a complicated way.

We can try to understand the behavior of this metric system in terms of the
ideas put forward in Section~\ref{sec:infinite-dim}.
Specifically, for both the analytical case of Fig.~\ref{fig:1} and the
reduced Euler case of Figs.~\ref{fig:2} and~\ref{fig:3}, we shall show that
the metric bracket is not specifically degenerate, and the generalization of
the PL inequality, equation~(\ref{eq:MpPL-V}), is not satisfied. We recall
that in both cases the bracket is given by the metric double bracket defined
in equation~(\ref{eq:paired-bracket-L2}), but the Hamiltonian is different,
and thus the null space of the bracket is different. Therefore we treat
the two cases separately even though there are some similarities.  
  
\begin{itemize}
    
\item Analytical test case.
  We begin by showing that the bracket is \emph{not} specifically
  degenerate. With mass $\fun{I}^1 = \fun{M}$ and the energy
  $\fun{I}^2 = \fun{H}$ as the only two invariants, we want to show that
  $(\fun{F},\fun{F})(u) = 0$ at a point $u$ does not imply 
  $\delta \fun{F}(u)/\delta u = \lambda_1 + \lambda_2 h$. With this aim, we
  observe that
  \begin{equation*}
    (\fun{F}, \fun{F})(u) = 0 \iff
    \Big[\frac{\delta \fun{F}(u)}{\delta u},
      \frac{\delta \fun{H}(u)}{\delta u} \Big] = 
    \Big[ \frac{\delta \fun{F}(u)}{\delta u}, h \Big] = 0.
  \end{equation*}
  This condition is satisfied at any point $u$ for functions of the form
  \begin{equation*}
    \fun{F}(u) = (f(h), u)_{L^2(\Omega)},
  \end{equation*}
  for any sufficiently regular function $f :\R \to \R$. We see that the
  condition $\delta \fun{F}(u)/\delta u = \lambda_1 + \lambda_2 h$,
  corresponds to the special case $f(h) = \lambda_1 + \lambda_2 h$. Therefore
  there are functions $\fun{F}$ for which $(\fun{F}, \fun{F})(u) = 0$, but 
  $\delta \fun{F}(u)/\delta u \not= \lambda_1 + \lambda_2 h$, and thus the
  bracket is \emph{not} specifically degenerate, in the sense of
  Eq.~(\ref{eq:minimal-degenerate}).  
  
  As for inequality~(\ref{eq:MpPL-V}), we have that
  \begin{equation*}
    \big(\fun{S}, \fun{S}\big)(u) = 0 \iff
    \Big[\frac{\delta \fun{S}(u)}{\delta u},
      \frac{\delta \fun{H}(u)}{\delta u} \Big] = [u, h] = 0.
  \end{equation*}
  Therefore any phase-space point of the form $u = f(h)$, with
  $f :\R \to \R$ a sufficiently regular function, is a zero of the
  bracket $(\fun{S}, \fun{S})$. Constrained entropy minima, on the other hand,
  are affine functions $u - u_\Omega = \lambda_1 + \lambda_2 h$ with specific
  values of the multiplier $\lambda_1$ and $\lambda_2$ (the exact formula has
  been given in equation~(\ref{eq:u-eta_analytic}) but it is not needed here);
  hence condition (\ref{eq:MpPL-V}) is false.

  Therefore, neither one of the conditions of Section~\ref{sec:infinite-dim}
  holds true in this case.
  In fact the analytical solution and the numerical experiment show that the
  relaxation method finds a point of the (rather large) set
  $\{u \colon (\fun{S},\fun{S})(u) = 0\}$, instead of the unique entropy
  minimum~(\ref{eq:u-eta_analytic}). 
  
\item Reduced Euler test case.
  We first show that the bracket is not specifically degenerate. With this aim
  we construct a similar counterexample to the one used in the analytical case
  above, the only difference being that now
  $\delta \fun{H}(u)/\delta u = \phi$ is related to $u$ via the Poisson
  equation $-\Delta \phi = u - u_\Omega$, Eq.~(\ref{eq:Poisson-eq-periodic}).
  We consider a point $u$ in phase space given by a solution of the problem
  \begin{equation*}
    -\Delta \phi = u - u_\Omega, \qquad u = f(\phi),
  \end{equation*}
  for a given smooth function $f:\R \to \R$. The existence of nontrivial
  solutions is guaranteed for a large class for functions $f$ \cite{Denny2016}.
  Then, for any smooth function $g : \R \to \R$, and 
  for $\fun{F}(u) = \int_\Omega g(u)dx$, we have that
  \begin{equation*}
    (\fun{F}, \fun{F})(u) = \int_\Omega \big[g'(u), \phi \big] dx
    = \int_\Omega \big[g' \circ f(\phi), \phi\big] dx = 0,
  \end{equation*}
  where $u$ is the phase-space point defined above. On the other hand,
  $\delta \fun{F}(u) / \delta u = g'\circ f (\phi)$ in general is not a linear
  combination of the derivatives $\frac{\delta \fun{M}(u)}{\delta u}$ and
  $\frac{\delta \fun{H}(u)}{\delta u}$ of the two invariants, i.e., 
  \begin{equation*}
    \frac{\delta \fun{F}(u)}{\delta u} \not=
    \lambda_1 \frac{\delta \fun{M}(u)}{\delta u} +
    \lambda_2 \frac{\delta \fun{H}(u)}{\delta u} = \lambda_1 + \lambda_2 \phi,
  \end{equation*}
  therefore the bracket is not specifically degenerate.
    
  As for condition (\ref{eq:MpPL-V}), we observe that the entropy $\fun{S}$ is
  a special case of the class of functions discussed above, corresponding the
  the choice $g(u) = u^2 / 2$, and, if $u$ is the same phase-space point used
  above, we have $(\fun{S},\fun{S})(u) = 0$, but 
  $\delta \fun{S}(u)/\delta u = f(\phi)$, which shows that in general $u$
  is not a constrained entropy minimum, since for a constrained entropy
  minimum we should have $f(\phi) = \phi$, 
  cf. Eq.~(\ref{eq:complete-relaxation-vorticity2d-periodic}).

\end{itemize}

While we have rigorous results in the finite-dimensional case only, these
observations show that, at least in these two cases, failure to relax the
system to a (local) constrained entropy minimum occurs for a bracket that is 
neither specifically degenerate, nor satisfies the generalized PL inequality.
This supports the idea that the convergence results obtained in
Section~\ref{sec:finite-dim-PL} for finite-dimensional systems may also hold
in general.
For comparison, below in section~\ref{sec:projector-based-metric-bracket},
we shall discuss a bracket that is specifically degenerate, and 
for this bracket we observe complete relaxation.

It is worth noting that in both the analytical case and the reduced Euler case
the bracket defined in Eq.~(\ref{eq:paired-bracket-L2}) does find a valid
equilibrium of the system, but this equilibrium, in general, is not a
constrained critical point of the entropy function, and in particular it
cannot be a local constrained minimum of entropy.
This implies that the bracket (\ref{eq:paired-bracket-L2}) could not be used
to solve problems like the Grad-Shafranov equation as the resulting
equilibrium would not be consistent with the imposed profiles that are encoded
in the entropy function, cf. Section~\ref{sec:GS-problem}. Yet they can be
useful in another way as we shall see below for the Beltrami fields.

\subsection{Projector-based metric bracket}
\label{sec:projector-based-metric-bracket}

We address now a construction of metric brackets based on $L^2$-orthogonal
projectors, which are patterned after that  given for finite-dimensional systems
in \cite{Morrison1986}. As before let $\Omega = [0,2\pi]^d$ and $V$ be the space
of functions on $\R^d$, $2\pi$-periodic in each direction. Given a Hamiltonian
function $\fun{H}$, the $L^2$ orthogonal projector onto the direction of
$\delta \fun{H}(u)/\delta u$ is   
\begin{equation}
  \label{eq:L2-projector}
  \begin{aligned}
    \Pi_{\fun{H}}(u) v &\coloneqq
    v - c(u,v) \frac{\delta \fun{H}(u)}{\delta u},\\
    c(u,v) &\coloneqq \Big\|
    \frac{\delta \fun{H}(u)}{\delta u} \Big\|_{L^2(\Omega)}^{-2}
    \Big(\frac{\delta \fun{H}(u)}{\delta u}, v \Big)_{L^2(\Omega)}.
  \end{aligned}
\end{equation}
With the projector, let us define
\begin{equation}
  \label{eq:projector-brackets}
  (\fun{F},\fun{G}) \coloneqq \Big(\frac{\delta \fun{F}}{\delta u},
  \Pi_{\fun{H}} \frac{\delta \fun{G}}{\delta u} \Big)_{L^2(\Omega)}.
\end{equation}
We claim that the symmetric bi-linear form~(\ref{eq:projector-brackets})
satisfies (\ref{eq:metric-system-compatibility}) and the Leibniz identity. Since
$\Pi_{\fun{H}}$ is a projector $\Pi_{\fun{H}} (\delta \fun{H} / \delta u) =0$,
hence $(\fun{F},\fun{H}) = 0$ for all $\fun{F}$; in addition,
\begin{equation*}
  (\fun{F},\fun{F}) = 
  \Big(\frac{\delta \fun{F}}{\delta u}, \Pi_{\fun{H}}
  \frac{\delta \fun{F}}{\delta u} \Big)_{L^2(\Omega)} \geq 0,
\end{equation*}
since projectors are symmetric and nonnegative definite. The Leibniz identity is
straightforward.  

We utilize bracket~(\ref{eq:projector-brackets}) in (\ref{eq:metric-system})
in order to obtain a relaxation method for the variational
problems~(\ref{eq:test-problems}). For the case of the linear advection
equation, energy (\ref{eq:analytical_H}) yields the evolution equation
\begin{equation*}
  \partial_t u = -\left[u - u_\Omega 
    - \frac{\fun{H}(u)}{\|h-h_\Omega\|^2_{L^2(\Omega)}} (h-h_\Omega) \right],
\end{equation*}
where $u_\Omega = \fun{M}(u)/(4\pi^2)$. Since both $\fun{M}$ and $\fun{H}$ are
constants of motion, the affine transformation
$u \mapsto w = u - \fun{M}_0/(4\pi^2) - \big[\fun{H}_0
  /\|h-h_\Omega\|^2_{L^2(\Omega)} \big] (h-h_\Omega)$
with $\fun{M}_0 = \fun{M}(u_0)$ and $\fun{H}_0 = \fun{H}(u_0)$, $u_0$ being
the initial condition, transforms the equation into $\partial_t w = - w$, which
leads to the analytical solution
\begin{equation*}
  u(t,\cdot) = \left[ \frac{\fun{M}_0}{4\pi^2} +
    \frac{\fun{H}_0}{\|h-h_\Omega\|^2_{L^2(\Omega)}} (h-h_\Omega)\right]
  (1-e^{-t}) + u_0 e^{-t}. 
\end{equation*}
The term in square brackets is exactly the unique entropy
minimum~(\ref{eq:u-eta_analytic}) on the manifold $\mathcal{U}_\eta$ with
$\eta= (\fun{M}_0,\fun{H}_0)$ being determined by the initial
condition. Therefore for any initial condition $u_0$, this metriplectic system
relaxes completely to the accessible entropy minimum and with exponential
convergence rate. This is the desired behavior. From the analytical solution
one can see how the initial condition is quickly ``forgotten'', leaving only
the fully relaxed state.

Let us now move to the test case of the reduced Euler equations.
Bracket~(\ref{eq:projector-brackets}) with
Hamiltonian~(\ref{eq:Euler_H_periodic}) and Eq.~(\ref{eq:metric-system})
yields the evolution equation
\begin{equation*}
  \partial_t u = - \left[u - u_\Omega
    - \frac{\fun{H}(u)}{\|\phi\|^2_{L^2(\Omega)}} \phi\right],
\end{equation*}
with $\phi$ depending on $u$ via the Poisson
equation~(\ref{eq:Poisson-eq-periodic}). The evolution of $u$ is therefore
governed by a nonlinear integral operator with a nonpolynomial nonlinearity. 
We begin by considering a numerical experiment. Figure \ref{fig:4} shows the
initial and final state of a solution of the initial value problem for this
equation, and Fig.~\ref{fig:5} gives the representation of the functional
relation between the potential $\phi$ and the vorticity $\omega = u - u_\Omega$.
The initial condition as well as the numerical method and the numerical
parameters (grid size and time steps) are the same as in Fig.~\ref{fig:2}.

Figure \ref{fig:4} shows the relative error in energy conservation and the
entropy as a function of time. The thick horizontal line denotes the minimum
entropy value in Eq.~(\ref{eq:Seta-Euler_periodic}). This time the minimum
entropy value is quickly reached by the system.

From the scatter plot in Fig.~\ref{fig:5} one can see that the final state is
characterized by a linear relation between $\phi$ and $\omega$. This suggests
that the projector-based metric bracket~(\ref{eq:projector-brackets}) relaxes
the state of the system completely. This nice property comes at the price of a
larger dissipation of entropy, cf.\  Fig.~\ref{fig:4}. As a consequence,
the vorticity of the final state is significantly lower than in the initial
condition. The scatter plot in Fig.~\ref{fig:5} also shows the average of the
initial condition on the contours of the initial potential (which is the same as
in Fig.~\ref{fig:3}); in this case, the relaxed state bears little or no
similarity to the initial condition.

In order to check if the relaxed vorticity is in agreement with the
analytical solution~(\ref{eq:u-eta_Euler_periodic}), we have computed the
best fit of the analytical solution~(\ref{eq:u-eta_Euler_periodic})
to the final state in Fig.~\ref{fig:4}, varying the three phases $\theta_0$,
$\theta_1$, and $\theta_2$. The difference between the best fit and the
relaxed state gives an estimate of the distance of the latter from the entropy
minimum and it is shown in Fig~\ref{fig:5}, right-hand-side panel. 

\begin{figure}[t]
  \centering
  \includegraphics[width=\columnwidth]{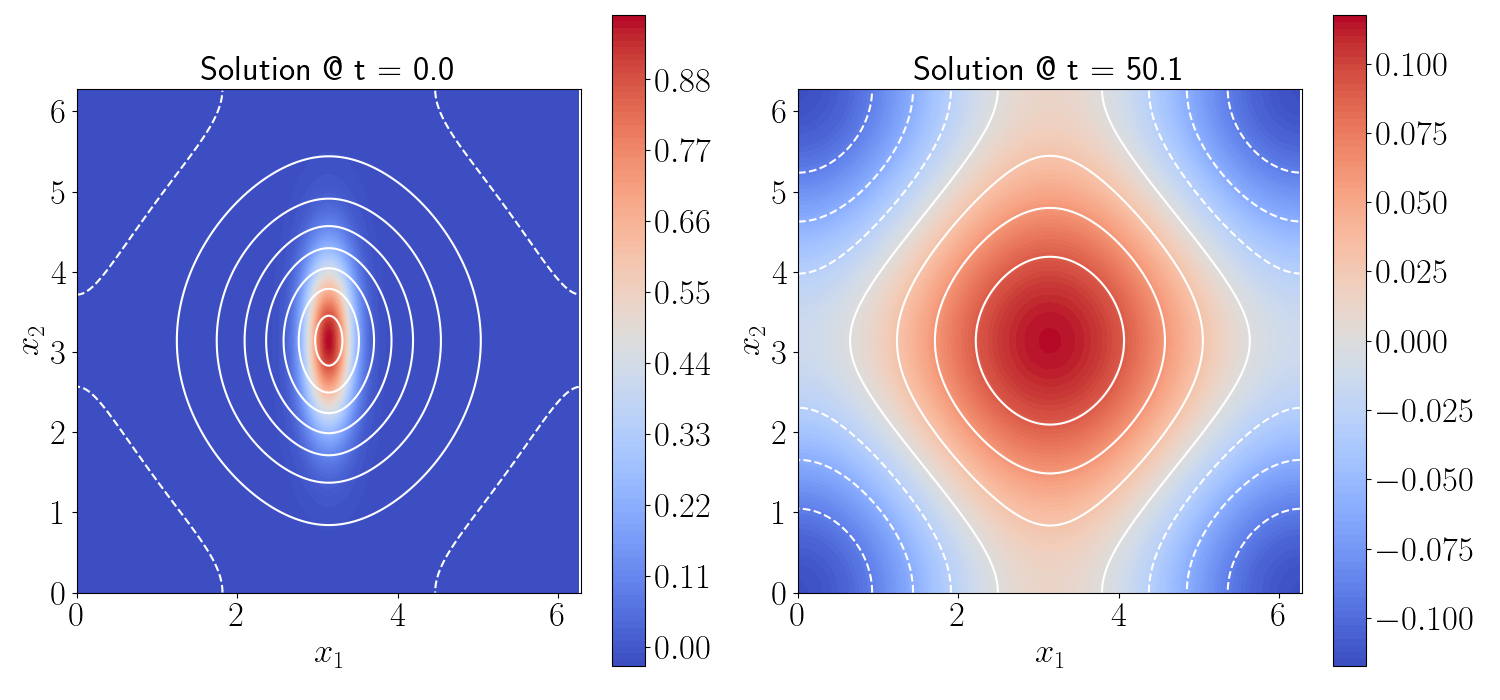}
  \includegraphics[width=\columnwidth]{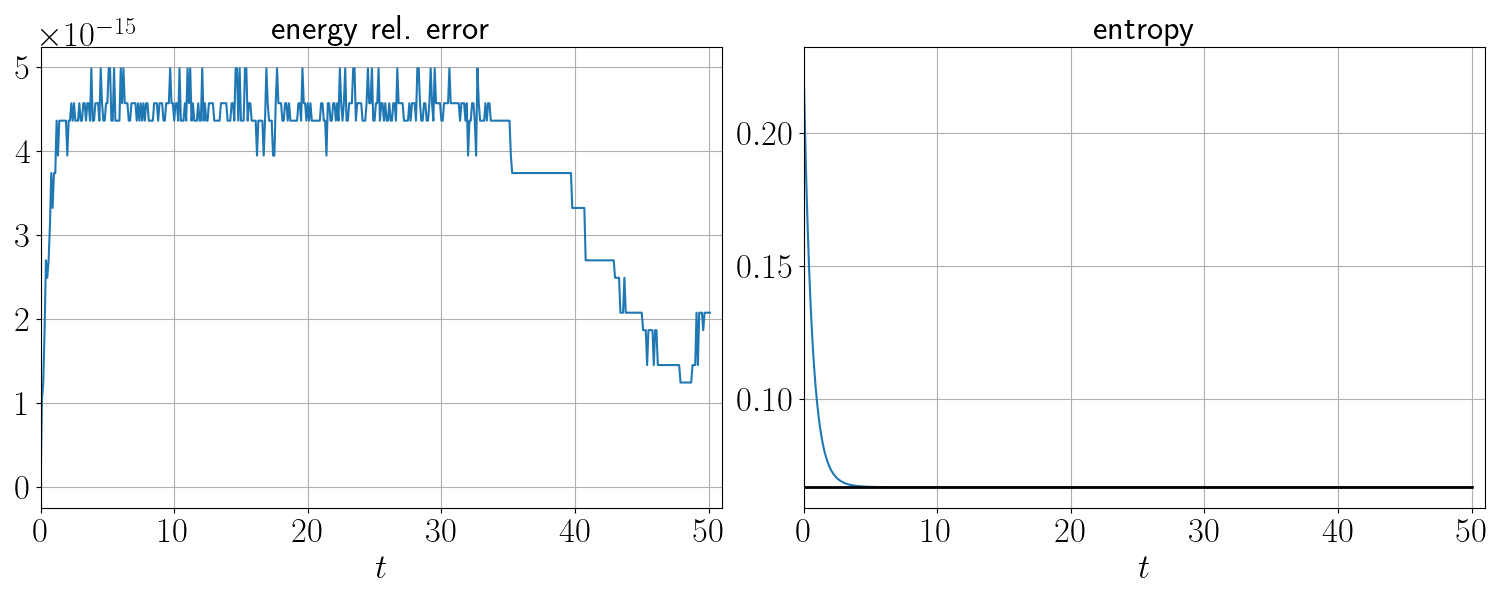}
  \caption{\label{fig:4} The same as in Fig.~\ref{fig:2} but for the metric
    bracket~(\ref{eq:projector-brackets}).}
\end{figure}

\begin{figure}[t]
  \centering
  \includegraphics[scale=0.32]{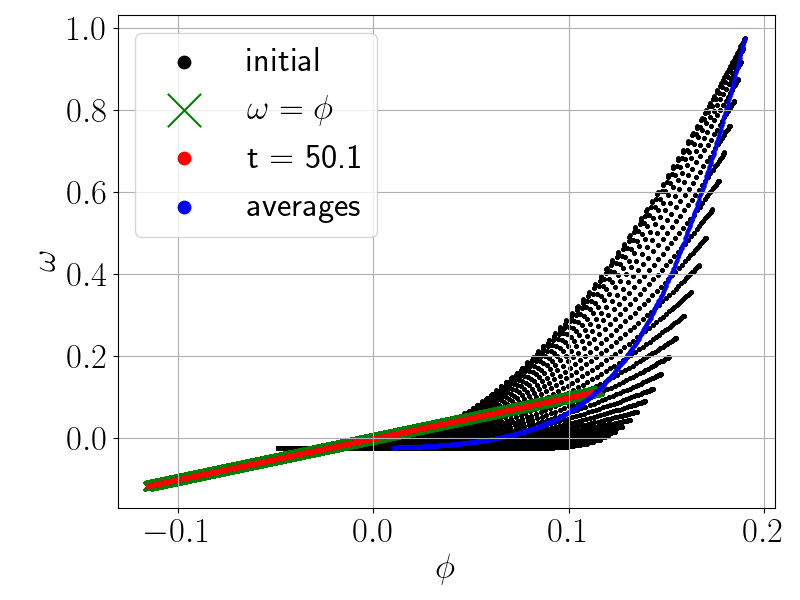}
  \includegraphics[scale=0.33]{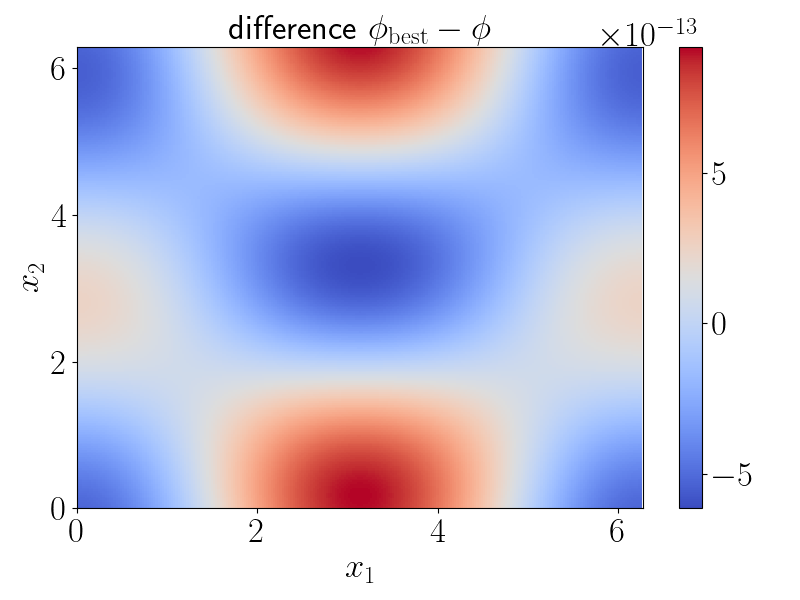}
  \caption{\label{fig:5} Left-hand-side panel: the relation between 
    $\phi$ and $\omega = u - u_\Omega$ for the case of Fig.~\ref{fig:4}.
    Right-hand-side panel: difference between the best fit of the exact
    solution~(\ref{eq:u-eta_Euler_periodic}) and the relaxed state.}
\end{figure}

In the terminology of Section~\ref{sec:infinite-dim}, we claim that
bracket~(\ref{eq:projector-brackets}) is \emph{minimally degenerate}. In
fact, for any function $\fun{F}$, we have $(\fun{F},\fun{F}) = 0$ if and only
if $\delta \fun{F}(u) /\delta u \in \ker \Pi_{\fun{H}}(u)$, or equivalently
\begin{equation*}
  \frac{\delta \fun{F}(u)}{\delta u} = \lambda \phi,
\end{equation*}
which is condition~(\ref{eq:minimal-degenerate}) with $\lambda_\alpha \not=0$
only if $\fun{I}^\alpha = \fun{H}$.  

If $\fun{F} = \fun{S}$, this condition is satisfied at any point of the set
$\mathfrak{C}_\eta$ defined in equation~(\ref{eq:Ceta}), i.e., for any
constrained critical point of the entropy, not just at the minimum. It follows
that the generalization (\ref{eq:MpPL-V}) of the \ref{eq:PL} condition cannot
be true on the whole phase space. Nonetheless,
with~(\ref{eq:paired-bracket-L2}) we have 
\begin{equation}
  (\fun{S}, \fun{S})(u) = 2 \fun{S}(u) -
  \frac{4 \fun{H}_0^2}{\|\phi\|^2_{L^2(\Omega)}}.
\end{equation}
We claim that for any $u$ such that $\fun{H}(u) = \fun{H}_0$, the following
inequalities hold true:
\begin{subequations}
  \label{eq:theoretical-limits}
  \begin{align}
    \label{eq:tb1}
    2 \fun{S}(u) - \frac{4 \fun{H}_0^2}{\|\phi\|^2_{L^2(\Omega)}} &\geq 0,\\
    \label{eq:tb2}
    \fun{S} &\geq \fun{H}_0, \\
    \label{eq:tb3}
    \|\phi\|^2_{L^2(\Omega)} &\leq  \|\nabla \phi\|_{L^2(\Omega)}^2 =
    2 \fun{H}_0.
  \end{align}
\end{subequations}
Specifically, Eq.~(\ref{eq:tb1}) follows from the Cauchy-Schwarz inequality
(which also ensure the positivity of the projector). Equation (\ref{eq:tb2}) is
a consequence of (\ref{eq:Seta-Euler_periodic}), while  (\ref{eq:tb3}) is
a Poincar\'e inequality on $\T^2$, obtained from the solution of the Poisson
problem~(\ref{eq:Poisson-eq-periodic}) via Fourier series. Inequalities
(\ref{eq:theoretical-limits}) imply that the evolution of the solution of the
metriplectic system must be such that
$\big(\fun{S}(u), 1/\|\phi\|^2_{L^2(\Omega)}\big) \in \R_+^2$ remains within
the cone
\begin{equation*}
  \fun{S}(u) \geq \fun{H}_0 = \fun{S}_\eta, \quad
  \frac{1}{2\fun{H}_0} \leq \frac{1}{\|\phi\|^2_{L^2(\Omega)}} \leq
  \frac{1}{2\fun{H}_0} + \frac{\fun{S}(u) - \fun{S}_\eta}{2 \fun{H}_0^2}.
\end{equation*}
We observe that the set $\mathfrak{C}_\eta$ is contained in the ``upper
boundary'' of this cone, i.e., it  satisfies 
\begin{equation*}
  \frac{1}{\|\phi\|^2_{L^2(\Omega)}} =
  \frac{1}{2\fun{H}_0} + \frac{\fun{S}(u) - \fun{S}_\eta}{2 \fun{H}_0^2}.
\end{equation*}
For any $a \in (0,1)$, and for any $u \in \mathcal{U}_\eta$ such that
\begin{equation}
  \label{eq:shrunk-cone}
  \frac{1}{\|\phi\|^2_{L^2(\Omega)}} =
  \frac{1}{2\fun{H}_0} + (1-a) \frac{\fun{S}(u)-\fun{S}_\eta}{2 \fun{H}_0^2},
\end{equation}
we have
\begin{equation*}
  \big(\fun{S},\fun{S}\big) \geq 2 a \big[\fun{S}(u) - \fun{S}_\eta\big],
\end{equation*}
which is inequality~(\ref{eq:MpPL-V}) with constant $\kappa_\eta=2a$. Hence, if
the solution stays within the shrunk cone~(\ref{eq:shrunk-cone}), we expect
exponential convergence with exponent related the angle of the
cone~(\ref{eq:shrunk-cone}). 

\begin{figure}[t]
  \centering
  \includegraphics[scale=0.33]{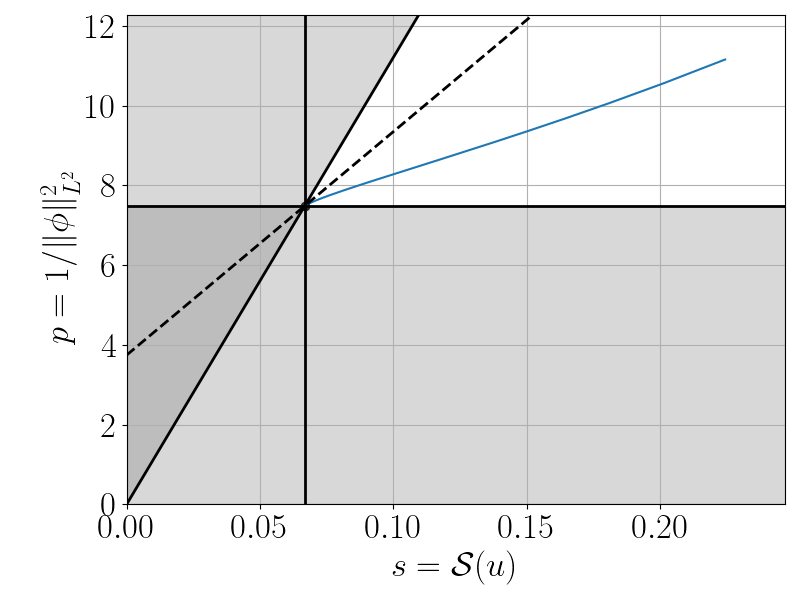}
  \includegraphics[scale=0.32]{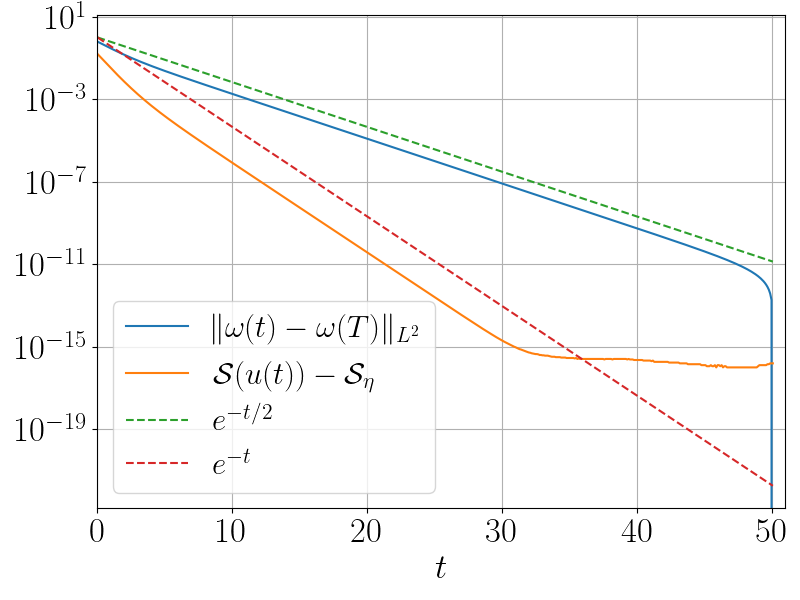}
  \caption{\label{fig:6} Left-hand-side panel: trajectory in the plane
    $(\fun{S}(u), 1/\|\phi\|_{L^2(\Omega)})$ for the case of Fig.~\ref{fig:4};
    the shaded area indicates the region of the plane excluded by
    inequalities~(\ref{eq:theoretical-limits}). The dashed line indicates the
    upper boundary of the shrunk cone (\ref{eq:shrunk-cone}) with $a=1/2$.
    The minimum entropy state corresponds to the vertex of the cone.
    Right-hand-side panel: evolution of the norm of the distance of $\omega(t)$
    from the relaxed state, using the vorticity $\omega$ at the last point in
    time $t=T$ as an approximation of the latter, and the ``excess entropy''
    $\fun{S}\big(u(t)\big)-\fun{S}_\eta$. The semi-log scale shows exponential
    relaxation of entropy with exponential rate $\approx 1$. This is consistent
    with the inequality~(\ref{eq:MpPL-V}). The fact that $\omega$ has relaxation
    rate $\approx 1/2$ is a consequence of the simple choice of the entropy
    function, cf.~Eq.~(\ref{eq:test-problems}).} 
\end{figure}

\begin{figure}[t]
  \centering
  \includegraphics[scale=0.33]{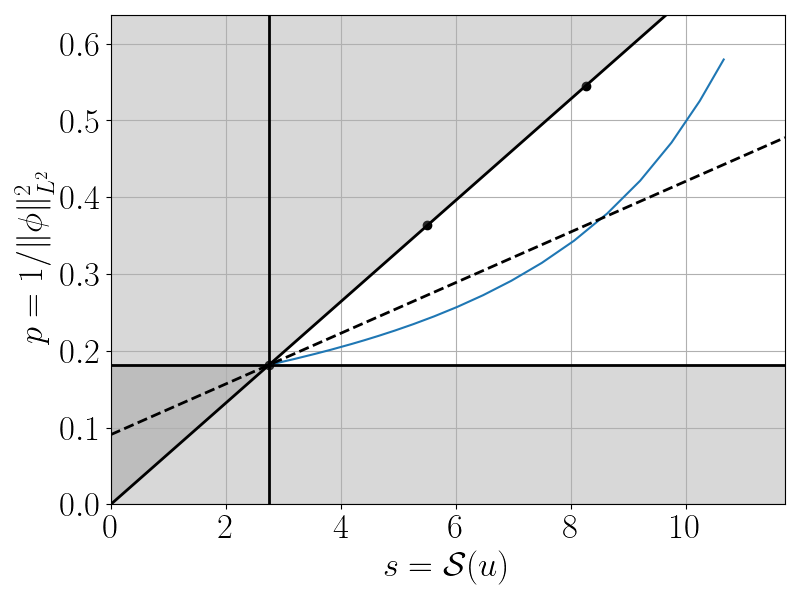}
  \includegraphics[scale=0.32]{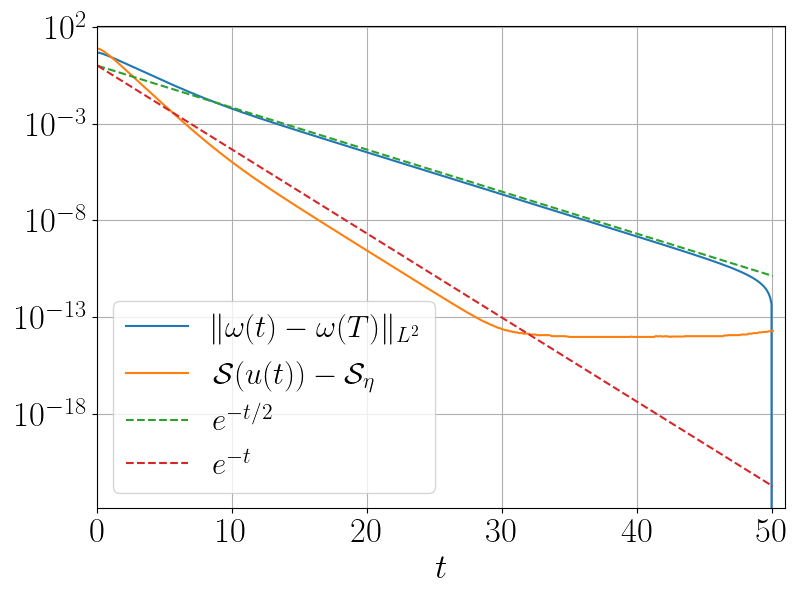}
  \caption{\label{fig:7} The same as in Fig.~\ref{fig:6}, but for the initial
    condition (\ref{eq:ic-projector2}). The black dots marked on the boundary of
    the cone correspond to the points of the set $\mathfrak{C}_\eta$.}
\end{figure}

Figure~\ref{fig:6} shows the trajectory of the solution in the plane
$\big(\fun{S}(u), 1/\|\phi\|^2_{L^2(\Omega)} \big)$ for the case of
Fig.~\ref{fig:4}. The solution indeed stays within the cone defined
by~(\ref{eq:theoretical-limits}). The figure also demonstrates exponential
convergence of both entropy and vorticity with exponential rate consistent
with the generalized \ref{eq:PL} condition. We note that, cf.~Fig.~\ref{fig:4},
the system traverses most of the trajectory in Fig.~\ref{fig:6} quickly, and
reaches the vertex of the cone already at $t \approx 5$.

Figure~\ref{fig:7} shows the same trajectory for a different initial condition,
viz.,
\begin{equation}
  \label{eq:ic-projector2}
  u_0(x) = \cos(2x_2) + u_G(x),
\end{equation}
where $u_G$ is the Gaussian defined in~(\ref{eq:initial_gaussian}) with
$1/N = 1.8$, $x_0 = (\pi, 3\pi/2)$, $w_1 = 0.3$, and $w_2 = 1$. The cosine
terms shift the initial condition  closer to the boundary. As a
consequence the initial entropy relaxation rate is slower, but approaches
$\approx 1$ as the trajectory approaches the vertex of the cone.

\section{Collision-like metric brackets}
\label{sec:coll-like-metr}

The specific structure of the metric bracket for the Landau collision operator
was introduced and generalized  in \cite{Morrison1984,Morrison1986}. Here we
propose a further generalization that we use  as a ``template'' for the 
construction of relaxation methods. This generalized bracket will be referred
to as the collision-like metric bracket \cite{Bressan2023}, as it originates
from Morrison's bracket for the Landau collision operator. The resulting
evolution equation is integro-differential. We demonstrate the use of
collision-like brackets for the solution of the variational principles in
Eqs.~(\ref{eq:VP-Euler}) and~(\ref{eq:VP-GS}) for equilibria of the reduced
Euler equations and axisymmetric MHD, respectively. In both these applications
the specification of \emph{equilibrium profiles}, i.e., the relation between the
state variable $u$ and the corresponding potential (either $\phi$ for the Euler
equations or $\psi$ for axisymmetric MHD), is essential. In the variational
principle, the profile is encoded in the choice of the entropy;  therefore, for
the result of metriplectic relaxation to be consistent with the imposed profile,
the metric bracket must \emph{completely relax} the state of the system, in the
sense made precise in Section~\ref{sec:remarks-relax-equil}. Preliminary,
numerical results on these two problems have been reported in the proceedings of
the joint Varenna-Lausanne workshop on ``The Theory of Fusion Plasmas''
\cite{Bressan2018} and in the Ph.D.\  thesis by Bressan \cite{Bressan2023}.

\subsection{General construction of collision-like brackets}
\label{sec:c-general}

We introduce the class of collision-like metric brackets in a fairly abstract
way, but we make no attempt to give a mathematically rigorous definition, except
for a few basic considerations. The construction given here is somewhat more
general than the one proposed earlier \cite{Bressan2023, Bressan2018}.

As in Section~\ref{sec:metriplectic}, the phase space $V$ is a space of
sufficiently regular functions $v : \Omega \to \R^N$ over a bounded domain 
$\Omega \subset \R^d$ with $N,d \in \N$. We always assume 
$V \subseteq L^2(\Omega,\mu;\R^N) \eqqcolon W$, and the functional derivative of
a function $\fun{F} \in C^1(V)$ is computed with respect to the standard
inner product in $W$, hence $\delta \fun{F}(u)/\delta u \in W$, when it exists.

For the construction of the bracket, we consider a \emph{bounded} domain
$\mathcal{O} \subset \R^n$ equipped with a measure $\nu$, and the space
$\tilde{W} \coloneqq L^2(\mathcal{O},\nu;\R^{\tilde{N}})$ with
$n,\tilde{N} \in \N$. For our purposes it is sufficient to assume that
$d\nu(z) = \tilde{m}(z)dz$ with $\tilde{m} \in C^\infty(\ol{\mathcal{O}})$,
$dz$ being the Lebesgue measure on $\mathcal{O}$. Next, we choose a (possibly
unbounded) linear operator   
\begin{equation*}
  P : W \to \tilde{W}, \qquad \dom(P) = \Phi,
\end{equation*}
where the domain $\Phi$ is a subspace of $W$ with a finer topology, so that
$V \subseteq \Phi \subseteq W \subseteq \Phi'$, with $\Phi'$ being the dual of
$\Phi$ (the space of continuous linear functionals on $\Phi$), and with
continuous inclusions. Then, $W$ has the structure of a rigged Hilbert space, 
with the finer space $\Phi$ containing the phase space $V$. In this way, both
quadratic functionals like $\fun{F}(u) = \int_\Omega u^2 d\mu$ and linear
functionals like $\fun{F}(u) = \int_\Omega w \cdot u d\mu$ for $w \in \Phi$ are
such that $\delta \fun{F}(u) / \delta u \in \Phi$. In addition, given a
Hamiltonian $\fun{H}$ such that $\delta \fun{H}(u)/\delta u \in \Phi$, we choose
\begin{equation*}
  \TT : V \to \mathcal{B}(\tilde{W}),
\end{equation*}
with values in the space $\mathcal{B}(\tilde{W})$ of bounded linear operators
from $\tilde{W}$ into $\tilde{W}$, such that $T(u)$ is symmetric, positive
semidefinite, and
\begin{equation}
  \label{eq:T-energy-cond}
  \TT(u) P \frac{\delta \fun{H}(u)}{\delta u} = 0.
\end{equation}
In terms of the operator $P$ and the function $\TT$, collision-like brackets
that preserve $\fun{H}$ are defined by 
\begin{equation}
  \label{eq:gen-clb}
  (\fun{F},\fun{G}) \coloneqq \int_{\mathcal{O}}
  P \frac{\delta \fun{F}(u)}{\delta u} \cdot \TT(u)
  P \frac{\delta \fun{G}(u)}{\delta u}\, d\nu.
\end{equation}
Symmetry and positive semidefiniteness of the bracket is ensured by the fact
that $\TT(u)$ is symmetric and positive semidefinite, while
Eq.~(\ref{eq:T-energy-cond}) implies $(\fun{F},\fun{H}) = 0$ for any $\fun{F}$.
Hence, Eq.~(\ref{eq:gen-clb}) defines a metric bracket on the class of functions
$\fun{F}$ such that $\delta \fun{F}(u)/\delta u \in \Phi$. In most cases,
$\TT(u)$ is the operator of multiplication by a function $\TT(u;x)$ with values
in the space of real symmetric, positive semidefinite
$\tilde{N} \times \tilde{N}$ matrices.

Upon introducing the dual operator $P' \colon  \tilde{W} \to \Phi'$ defined by
\begin{equation}
  \label{eq:P-dual}
  \big\langle w, P'\tilde{w} \big\rangle = \int_{\mathcal{O}} \tilde{w} \cdot
  Pw \,d\nu, \qquad \text{for all } w \in \Phi,
\end{equation}
where $\langle \cdot, \cdot \rangle \colon  \Phi \times \Phi' \to \R$ is the
duality pairing between $\Phi$ and $\Phi'$, bracket~(\ref{eq:gen-clb}) can be
equivalently written as
\begin{equation*}
  (\fun{F},\fun{G}) = \Big\langle
  \frac{\delta \fun{F}(u)}{\delta u}, P'\Big[
    \TT(u) P \frac{\delta \fun{G}(u)}{\delta u}
    \Big]\Big\rangle.
\end{equation*}
If $u \in C^1([0,T],V)$ is a trajectory in $V$ and $\fun{F} \in C^1(V)$
has a functional derivative $\delta \fun{F}(u)/\delta u$ in $\Phi$, then
$t \mapsto \fun{F}\big(u(t)\big)$ is differentiable and 
\begin{equation*}
  \frac{d}{dt} \fun{F} \big( u(t) \big) = \Big\langle
  \frac{\delta \fun{F}\big(u(t)\big)}{\delta u}, \partial_t u\Big\rangle.
\end{equation*}
Therefore, given an entropy $\fun{S}$ such that
$\delta \fun{S}(u) / \delta u \in \Phi$, Eq.~(\ref{eq:metric-system-equation})
for $u(t)$ amounts to
\begin{equation}
  \label{eq:clb-ev}
  \partial_t u = - P'\Big[\TT(u) P \frac{\delta \fun{S}(u)}{\delta u} \Big]
  \quad \text{in } \Phi'.
\end{equation}
This construction is summarized in Fig.~\ref{fig:8}. 

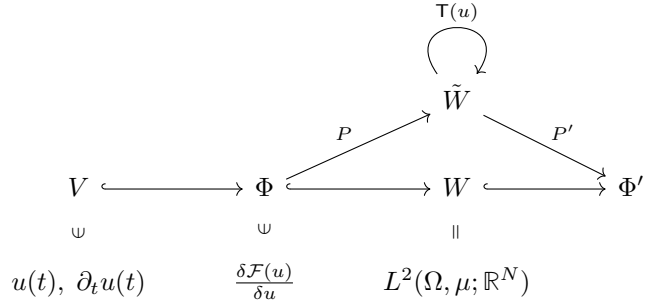
\begin{figure}[t]
  \begin{center}
    \begin{tikzcd}
      && {\tilde{W}} \\
      V & \Phi & W & {\Phi'} \\
      {u(t),\; \partial_t u(t)} & {\frac{\delta \fun{F}(u)}{\delta u}} &
      {L^2(\Omega,\mu;\R^N)} & {}
      \arrow["{\TT(u)}", from=1-3, to=1-3, loop, in=55, out=125, distance=10mm]
      \arrow["{P'}", from=1-3, to=2-4]
      \arrow[hook, from=2-1, to=2-2]
      \arrow["P", from=2-2, to=1-3]
      \arrow[hook, from=2-2, to=2-3]
      \arrow[hook, from=2-3, to=2-4]
      \arrow["\vin"{description}, draw=none, from=3-1, to=2-1]
      \arrow["\vin"{description}, draw=none, from=3-2, to=2-2]
      \arrow["\veq"{description}, draw=none, from=3-3, to=2-3]
    \end{tikzcd}
  \end{center}
  \caption{\label{fig:8} Construction of the operators $P$ and $P'$ in
    Eqs.~(\ref{eq:gen-clb}) and (\ref{eq:clb-ev}).}
\end{figure}

The operator $P$ maps an $\R^N$-valued
function over the $d$-dimensional domain $\Omega$ into an
$\R^{\tilde{N}}$-valued function over the $n$-dimensional domain $\mathcal{O}$, 
and we are particularly interested in the case $n>d$, $\tilde{N} \geq N$. As we
shall see, increasing the number of dimensions has some advantages that,
however, come at the price of a higher computational cost. The function $\TT$
will be referred to as the \emph{kernel of the bracket} and, in general, it
depends on both $P$ and $\fun{H}$, because of~(\ref{eq:T-energy-cond}), but for
simplicity, this dependence is not explicitly indicated in the notation.  

There are of course many ways to choose $\TT$ such that
condition~(\ref{eq:T-energy-cond}) is satisfied and we give two particularly 
relevant examples below. Among the various choices, we are interested in those
that satisfy:  
\begin{subequations}
  \label{eq:MDC}
  \begin{align}
    \label{eq:PI}
    \|w\|_W \leq C_{P} \|Pw\|_{\tilde{W}}, \qquad &w \in \dom(P), \\
    \label{eq:KC}
    \tilde{w} \in \ker \TT(u) \cap \rng P \iff
    \tilde{w} = \lambda P \frac{\delta \fun{H}(u)}{\delta u}, \quad
    &\text{for $\lambda \in \R$ constant,}
  \end{align}
\end{subequations}
for a constant $C_P$. Here $\ker \TT(u)$ and $\rng P$ denote the null space of
the operator $\TT(u)$ and the range of the operator $P$, respectively, and they
are both subspaces of $\tilde{W}$. When $P = \nabla$, Eq.~(\ref{eq:PI}) is the
Poincar\'e inequality \cite{Evans1998}.  

Formally at least, conditions~(\ref{eq:MDC}) imply that the bracket defined in
Eq.~(\ref{eq:gen-clb}) is \emph{minimally degenerate} in the sense of
Eq.~(\ref{eq:minimal-degenerate}). In fact, $(\fun{F},\fun{F}) = 0$ is
equivalent to  
\begin{equation*}
  P \frac{\delta \fun{F}(u)}{\delta u} \in \ker \TT(u) \cap \rng(P),
\end{equation*}
and condition~(\ref{eq:KC}) implies that there is a constant $\lambda$ such that
\begin{equation*}
  P \frac{\delta \fun{F}(u)}{\delta u} =
  \lambda P \frac{\delta \fun{H}(u)}{\delta u}.
\end{equation*}
Then condition~(\ref{eq:PI}) yields
\begin{equation*}
  \Big\|\frac{\delta \fun{F}(u)}{\delta u} -
  \lambda \frac{\delta \fun{H}(u)}{\delta u} \Big\|_W \leq
  C_P \Big\| P \Big[\frac{\delta \fun{F}(u)}{\delta u} -
    \lambda \frac{\delta \fun{H}(u)}{\delta u}\Big] \Big\|_{\tilde{W}} = 0,
\end{equation*}
which implies condition~(\ref{eq:minimal-degenerate}) with
$\lambda_\alpha \not = 0$ only if $\fun{I}^\alpha = \fun{H}$.

It follows that metric brackets of the form~(\ref{eq:gen-clb}), with the
defining operator $P$ and the kernel $\TT$ satisfying (\ref{eq:MDC}) could be
used to construct a relaxation method for variational problems of the
form~(\ref{eq:entropy-principle}).  

We now give examples of (\ref{eq:gen-clb}) for which conditions~(\ref{eq:MDC})
are, at least formally, satisfied.

\begin{example}
  \label{ex:Landau}
  Morrison's brackets for the Landau collision operator \cite{Morrison1984} and
  its generalization \cite{Morrison1986} can be obtained as  special cases of
  (\ref{eq:gen-clb}). This motivates our choice of the name
  \emph{collision-like}. 
  
  \smallskip

  The configuration space $V$ is the space of particle distribution functions
  $u(x) = f(\xx,\vv)$, where $\xx$ is the spatial position and $\vv$ is the
  velocity of the particles. We assume that $\xx \in \Omega_{\xx}$ and
  $\vv \in \Omega_{\vv}$, with both domains
  $\Omega_{\xx},\; \Omega_{\vv} \subset \R^D$ being bounded, $D \geq 2$, and
  $f(\xx,\cdot)$ compactly supported in $\Omega_{\vv}$, i.e., the velocity-space
  domain is large enough to contain all particle velocities.
  Then, $\Omega = \Omega_{\xx} \times \Omega_{\vv}$, $d=2D$, and $N=1$,
  since $f$ is a scalar field. The measure on $\Omega$ is the Lebesgue measure
  $d\mu(x) = d\xx d\vv$.
  
  We choose $\mathcal{O} = \Omega \times \Omega_{\vv}$, $n = 3D$, 
  $d\nu(\xx, \vv, \vv') = d\xx d\vv d\vv'$,
  \begin{equation*}
    Pg(\xx,\vv,\vv') = \nabla_{\vv} g(\xx,\vv) - \nabla_{\vv'} g(\xx,\vv'),
  \end{equation*}
  with $\dom P$ given by the functions $g \in L^2(\Omega)$ such that
  $\nabla_\vv g \in L^2(\Omega;\R^D)$;
  hence $Pg(\xx,\vv,\vv') \in \R^{\tilde{N}}$ with $\tilde{N} = D$. The
  Hamiltonian is
  \begin{equation*}
    \fun{H}(f) = \int_{\Omega} \big[\frac{1}{2} m \vv^2 + \mathrm{V}(\xx)\big]
    f(\xx,\vv) d\xx d\vv,
  \end{equation*}
  where $m$ is the mass of the considered particle species and $\mathrm{V}$ is a
  potential energy. We have $P (\delta \fun{H}(f)/\delta f) = m (\vv - \vv')$.
  Therefore, a possible choice of the operator $\TT(u) = \TT(f)$ is the
  multiplication by the matrix-valued kernel
  \begin{equation}
    \label{eq:Landau-kernel}
    \TT_{\mathrm{L}}(f; \xx, \vv, \vv') = \frac{\nu_c}{2}
    M\big(f(\xx, \vv)\big) M\big(f(\xx, \vv')\big)
    U_{\mathrm{L}}(\vv-\vv'),
  \end{equation}
  where $\nu_c >0$ is a constant collision frequency, $M \colon \R_+ \to \R_+$
  is arbitrary, and 
  \begin{equation*}
    U_{\mathrm{L}}(\vv-\vv') \coloneqq \frac{1}{|\vv-\vv'|}
    \Big(I - \frac{(\vv-\vv') \otimes (\vv-\vv')}{|\vv-\vv'|^2}\Big).
  \end{equation*}
  Condition~(\ref{eq:T-energy-cond}) holds, and bracket~(\ref{eq:gen-clb}) 
  reduces to  
  \begin{multline}
    \label{eq:Landau-bracket}
    (\fun{F},\fun{G}) = \frac{\nu_c}{2} \int_{\mathcal{O}}
    M\big(f(\xx, \vv)\big) M\big(f(\xx, \vv')\big) \Big[
      \nabla_{\vv} \frac{\delta \fun{F}(f)}{\delta f} (\xx,\vv)
      - \nabla_{\vv'} \frac{\delta \fun{F}(f)}{\delta f} (\xx,\vv')
      \Big] \\
    \cdot U_{\mathrm{L}}(\vv-\vv') \Big[
      \nabla_{\vv} \frac{\delta \fun{G}(f)}{\delta f} (\xx,\vv)
      - \nabla_{\vv'} \frac{\delta \fun{G}(f)}{\delta f} (\xx,\vv')
      \Big] d\xx d\vv d\vv' ,
  \end{multline}
  which is the bracket of Eq.~(44) in  Morrison's paper~\cite{Morrison1986}.
  One can notice that in Eq.~(\ref{eq:Landau-kernel}) the orthogonal projection
  onto the direction of
  \begin{equation*}
    m(\vv-\vv') = \nabla_{\vv} \frac{\delta \fun{H}(f)}{\delta f}(\xx, \vv) -
    \nabla_{\vv'} \frac{\delta \fun{H}(f)}{\delta f}(\xx, \vv')
  \end{equation*}
  ensures property~(\ref{eq:T-energy-cond}). In most examples considered in this
  paper, the kernel of the bracket is constructed from a similar projector.
  We also note that, in the construction given here, there is no need to use
  a singular distributional kernel, cf. the Dirac's delta in Eq.~(45) of
   \cite{Morrison1986}.

  The operator $P'$ defined in Eq.~(\ref{eq:P-dual}) acting on a function
  $\tilde{w}(z) = \tilde{g}(\xx,\vv,\vv')$ can be formally computed after
  integration by parts, with the result that 
  \begin{equation*}
    P'\tilde{g}(\xx,\vv) = - \div_{\vv} \Big[
      \int_{\Omega_{\vv}} \big( \tilde{g}(\xx,\vv,\mathsf{w})
      - \tilde{g}(\xx,\mathsf{w},\vv) \big)d\mathsf{w},
      \Big],
  \end{equation*}
  and the evolution equation (\ref{eq:clb-ev}) takes the form
  \begin{multline*}
    \partial_t f = \div_{\vv} \Big[ \nu_c \int_{\Omega_{\vv}}
      M\big(f(\xx, \vv)\big) M\big(f(\xx, \mathsf{w})\big) \\
      \times\  U_{\mathrm{L}}(\vv-\mathsf{w}) \Big(
      \nabla_{\vv}
      \frac{\delta \fun{S}(f)}{\delta f}(\xx,\vv) -
      \nabla_{\mathsf{w}}
      \frac{\delta \fun{S}(f)}{\delta f}(\xx,\mathsf{w}) \Big)
      d\mathsf{w} \Big],
  \end{multline*}
  which reduces to the Landau operator for Coulomb collisions when
  $\fun{S}(f) = \int f \log f d\xx d\vv$ and $M(f) = f$.
  
  As for conditions~(\ref{eq:MDC}), given $g(\xx,\vv)$ in the domain of the
  operator $P$, we have 
  \begin{equation*}
    \int_{\Omega_{\vv}} \nabla_{\vv} g(\xx,\vv) d\vv = 0,
  \end{equation*}
  provided that
  $g(\xx,\cdot)$ vanishes near the boundary of $\Omega_{\vv}$. Then,
  \begin{equation*}
    \|Pg\|_{\tilde{W}}^2 = 2 |\Omega_{\vv}| \int_{\Omega}
    \big|\nabla_{\vv} g(\xx,\vv)\big|^2 d\xx d\vv,
  \end{equation*}
  where $|\Omega_{\vv}| = \int_{\Omega_{\vv}}d\vv'$. The
  minimum of $\|Pg\|_{\tilde{W}}^2$ subject to the constraint $\|g\|_W = 1$ is
  attained for
  \begin{equation*}
    -\Delta_{\vv} g = \lambda_0 g, \qquad
    \int_{\Omega} |g(\xx,\vv)|^2 d\xx d\vv = 1,
  \end{equation*}
  where $\lambda_0$ is the minimum eigenvalue of the Laplace operator
  $\Delta_{\vv}$ on the velocity domain $\Omega_{\vv}$ with
  homogeneous Dirichlet boundary conditions. Hence
  \begin{equation*}
    \|Pg\|_{\tilde{W}}^2 = 2 |\Omega_{\vv}| \int_{\Omega}
    \big|\nabla_{\vv} g(\xx,\vv)\big|^2 d\xx d\vv
    \geq 2 |\Omega_{\vv}| \lambda_0 \|g\|^2_W,
  \end{equation*}
  which shows (formally) that condition~(\ref{eq:PI}) holds.
  We shall use this type of argument to study the brackets considered
  below, for which assuming homogeneous Dirichlet boundary conditions is
  natural. For the Landau collision operator, however, there is no physical
  reason to assume that $g = 0$ on the boundary of $\Omega_\vv$.
  If we drop this unphysical requirement, then $P$ no longer satisfies
  condition~(\ref{eq:PI}) as it has a nontrivial null space, $\ker P$, given
  by the functions $g \in \dom P$ such that
  \begin{equation*}
    \nabla_\vv g(\xx,\vv) - \nabla_{\vv'} g(\xx,\vv') = 0.
  \end{equation*}
  More specifically, the velocity gradient of an element of $\ker P$ is
  necessarily constant in velocity $\vv$ for almost all $\xx$, and thus 
  \begin{equation*}
    g \in \ker P \iff g(\xx,\vv) = a(\xx) + b(\xx) \cdot \vv,
  \end{equation*}
  where $a : \Omega_\xx \to \R$ and $b : \Omega_\xx \to \R^D$ are arbitrary
  functions.
  As for condition~(\ref{eq:KC}), a function $\tilde{g}(\xx,\vv,\vv')$ belongs
  to $\ker \TT_{\mathrm{L}}$ when
  \begin{equation*}
    \tilde{g}(\xx,\vv,\vv') = \Lambda(\xx,\vv,\vv') \big(\vv - \vv'\big),
  \end{equation*}
  almost everywhere (a.e.) in $\mathcal{O}$,
  where $\Lambda(\xx,\vv,\vv') \in \R$. If, in addition,
  $\tilde{g} \in \rng (P)$, we must have
  \begin{equation*}
    \nabla_{\vv} g(\xx,\vv) - \nabla_{\vv'} g(\xx,\vv') =
    \Lambda(\xx,\vv,\vv') \big(\vv - \vv'\big),
    \quad \text{a.e. in } \mathcal{O}.
  \end{equation*}
  If $g$ and $\Lambda$ satisfying this condition exist, then necessarily
  $\Lambda(\xx,\vv,\vv') = \Lambda(\xx,\vv',\vv)$ and, upon fixing a point
  $\vv' = \mathsf{a} \in \Omega_{\vv}$, 
  \begin{equation*}
    \nabla_\vv g(\xx,\vv) = \nabla_{\vv'} g(\xx,\mathsf{a}) +
    \Lambda(\xx,\vv,\mathsf{a}) (\vv - \mathsf{a}),
  \end{equation*}
  and
  \begin{align*}
    \nabla_{\vv} g(\xx,\vv) - \nabla_{\vv'} g(\xx,\vv') &=
    \Lambda(\xx,\vv,\mathsf{a}) (\vv - \mathsf{a})
    -\Lambda(\xx,\vv',\mathsf{a}) (\vv' - \mathsf{a}) \\
    &  = \Lambda(\xx,\vv,\vv') \big((\vv - \mathsf{a}) - (\vv'-\mathsf{a})).
  \end{align*}
  Then we must have
  \begin{equation*}
    \big[\Lambda(\xx,\vv,\vv') - \Lambda(\xx,\vv,\mathsf{a})\big]
    (\vv - \mathsf{a}) -
    \big[\Lambda(\xx,\vv,\vv') - \Lambda(\xx,\mathsf{a},\vv')\big]
    (\vv' - \mathsf{a}) = 0.
  \end{equation*}
  The two vectors $\vv-\mathsf{a}$ and $\vv'-\mathsf{a}$ in $\R^D$ are linearly
  dependent only if they are proportional, which can only happen for a set of
  points $(\xx,\vv,\vv')$ of measure zero in $\mathcal{O}$. Hence, we deduce
  that $\Lambda$ must satisfy the necessary conditions
  \begin{equation*}
    \left\{
    \begin{aligned}
      \Lambda(\xx,\vv,\vv') - \Lambda(\xx,\vv,\mathsf{a}) &= 0, \\
      \Lambda(\xx,\vv,\vv') - \Lambda(\xx,\mathsf{a},\vv') &= 0,
    \end{aligned}
    \right.
    \quad \text{ a.e. in } \mathcal{O},
  \end{equation*}
  and this must hold for almost any choice of the arbitrary point $\mathsf{a}$.
  This is possible only if there is a function $m c(\xx)$ of position $\xx$ only
  (we factor out the mass $m$ for convenience) such that
  \begin{equation*}
    \Lambda(\xx,\vv,\vv') = m c(\xx) = \text{ a.e. in } \mathcal{O}.
  \end{equation*}
  We deduce that a function $\tilde{g} \in \ker T(f) \cap \rng P$ must be of
  the form
  \begin{equation*}
    \tilde{g}(\xx,\vv,\vv') = m c(\xx) (\vv - \vv') =
    c(\xx) P \frac{\delta \fun{H}(f)}{\delta f}(\xx,\vv,\vv').
  \end{equation*}
  This in not exactly Eq.~(\ref{eq:KC}), since $c$ is not a constant on the
  whole extended domain $\mathcal{O}$, but only in velocity space.
  This residual dependence of $\xx$ should be expected since the collision
  operator only acts in velocity space, pointwise in $\xx$. Therefore,
  conditions~(\ref{eq:MDC}) are not satisfied for this bracket. We can however
  obtain a complete characterization of the null space of this bracket by
  using the results on $T(f)$ and $P$ obtained above. In fact, we have
  \begin{equation*}
    \big(\fun{F},\fun{F}\big)(f) = 0 \iff
    P \frac{\delta \fun{F}(f)}{\delta f} \in \ker T(f),
  \end{equation*}
  hence
  \begin{equation*}
    P\frac{\delta \fun{F}(f)}{\delta f}
    = mc(\xx) (\vv-\vv') = c(\xx) P\frac{\delta \fun{H}(f)}{\delta f},
  \end{equation*}
  and, since $c$ is independent on $(\vv,\vv')$, we have
  \begin{equation*}
    P\Big(\frac{\delta \fun{F}(f)}{\delta f} - c
    \frac{\delta \fun{H}(f)}{\delta f}\Big) = 0.
  \end{equation*}
  The elements of the null space of $P$ have been computed above, and can
  conclude that
  \begin{equation*}
    \big(\fun{F},\fun{F}\big)(f) = 0 \iff
    \frac{\delta \fun{F}(f)}{\delta f} = a \frac{\delta \fun{N}(f)}{\delta f}
    + b \cdot \frac{\delta \fun{P}(f)}{\delta f}
    + c \frac{\delta \fun{H}(f)}{\delta f},
  \end{equation*}
  where $a$, $b$, and $c$ are functions of $\xx$ only and
  \begin{equation*}
    \fun{N}(f) = \int_\Omega f(\xx,\vv)d\xx d\vv, \quad
    \fun{P}(f) = \int_\Omega f(\xx,\vv) \vv d\xx d\vv,
  \end{equation*}
  together with $\fun{H}(f)$ constitute the three collision invariants
  \cite{Lenard1960, Villani1999}, namely, the total particle number, momentum
  (per unit mass), and energy. In summary, the null space of Morrison's
  bracket is spanned by a linear combination of the derivatives of the three
  collision invariants, but with coefficients depending on $\xx$. Therefore
  the bracket is not minimally degenerate, since energy is not the only
  invariant, and it is not even specifically degenerate since the coefficients
  $a$, $b$, and $c$ are functions of space. This is due to the well known
  fact that the collision operator acts on velocity space only, and the
  relaxation of the full distribution function $f(\xx,\vv)$ in both space and
  velocity requires the interaction of the collision operator with the ideal
  phase-space transport dynamics \cite{Villani2007}. This bracket does become
  specifically degenerate with respect to the three invariants if we restrict
  the phase space to functions of velocity only.

\end{example}

\begin{example}
  \label{ex:proj-revisited}
  The bracket based on the $L^2$-orthogonal projector discussed in
  Section~\ref{sec:projector-based-metric-bracket} can be obtained as a special
  case of (\ref{eq:gen-clb}). In fact, we can utilize Eq.~(\ref{eq:gen-clb}) in
  order to generalize (\ref{eq:projector-brackets}) to the case of $\R^N$-valued
  fields.
  
  \smallskip

  With this aim, let $\Omega \subset \R^d$, $N \in \N$,
  $\mathcal{O} = \Omega \times \Omega$, and $\tilde{N} = 2N$, that is, we double
  both the dimension of the domain and of the field. Then,
  with $\Phi = W = L^2(\Omega,\mu;\R^N) = \Phi'$,
  $\tilde{W} = L^2(\mathcal{O},\nu;\R^{\tilde{N}})$, and
  $d\nu(x,x') = d\mu(x) d\mu(x')$, we define
  \begin{equation*}
    P : W \to \tilde{W}, \qquad
    Pw(x,x') = \begin{pmatrix}
      w(x) \\
      w(x')
    \end{pmatrix} \in \R^{2N}.
  \end{equation*}
  Given a Hamiltonian function $\fun{H}(u)$ with
  $\delta \fun{H}(u)/\delta u \in W$, we choose $\TT(u)$ to be the
  multiplication operator by the kernel 
  \begin{equation*}
    \TT(u;x,x') = \kappa(u) \begin{pmatrix}
      |h(x')|^2 I_N         &   -h(x) \otimes h(x') \\
      - h(x') \otimes h(x)  &          |h(x)|^2 I_N
    \end{pmatrix},
  \end{equation*}
  where $I_N$ is the $N \times N$ identity block, and
  $h = \delta \fun{H}(u)/\delta u$ is a short-hand notation for the functional
  derivative ($h$ may still depend on $u$). One can check that
  condition~(\ref{eq:T-energy-cond}) is satisfied. Inserting these choices into 
  Eq.~(\ref{eq:gen-clb}) yields
  \begin{align*}
    (\fun{F},\fun{G}) &= 2 \kappa(u) \int_{\mathcal{O}}
    \frac{\delta \fun{F}(u)}{\delta u} (x) \cdot \Big[
      \Big| \frac{\delta \fun{H}(u)}{\delta u} (x') \Big|^2
      \frac{\delta \fun{G}(u)}{\delta u} (x) \\
      & \hspace{40mm}
      -\frac{\delta \fun{H}(u)}{\delta u} (x) \Big(
      \frac{\delta \fun{H}(u)}{\delta u} (x') \cdot
      \frac{\delta \fun{G}(u)}{\delta u} (x') \Big)\Big] d\mu(x) d\mu(x') \\
    &= 2 \kappa(u) \Big\|\frac{\delta \fun{H}(u)}{\delta u} \Big\|^2_{L^2}
    \int_{\Omega} \frac{\delta \fun{F}(u)}{\delta u} (x)
    \cdot \Pi_{\fun{H}} \frac{\delta \fun{G}(u)}{\delta u}(x) d\mu(x),
  \end{align*}
  where $\Pi_{\fun{H}}$ is the $L^2$-orthogonal projector defined in
  Eq.~(\ref{eq:L2-projector}), but generalized to the case of $\R^N$-valued
  fields. If $2\kappa$ is set to the inverse of $\|h\|^2_{L^2}$, this reduces to
  (\ref{eq:projector-brackets}) when $N=1$.
  
  Condition~(\ref{eq:PI}) follows from
  \begin{equation*}
    \|Pw\|^2_{L^2(\mathcal{O},\nu;\R^{2N})} =
    \int_{\Omega\times\Omega} \big(|w(x)|^2 + |w(x')|^2 \big) d\mu(x) d\mu(x')
    = 2 \mu(\Omega) \|w\|^2_{L^2(\Omega,\mu;\R^N)},
  \end{equation*}
  where $\mu(\Omega) = \int_{\Omega} d\mu$. As for condition~(\ref{eq:KC}),
  if $h \not= 0$, $\tilde{w} \in \ker \TT(u) \cap \rng(P)$ implies that there is
  $w \in W$ such that $\tilde{w} = Pw$, and
  \begin{equation*}
    |h(x')|^2 w(x) - h(x) \big(h(x') \cdot w(x')\big) = 0 \in \R^N,
  \end{equation*}
  for almost all $(x,x')$. Upon multiplying by $h(x)$, we have
  \begin{equation*}
    |h(x')|^2 \big(h(x) \cdot w(x)\big) = |h(x)|^2 \big(h(x') \cdot w(x')\big).
  \end{equation*}
  This is equivalent to saying that 
  $\big(-|h(x')|^2, |h(x)|^2 \big)$ and
  $\big(h(x)\cdot w(x), h(x') \cdot w(x')\big)$ are orthogonal in $\R^2$, or
  that 
  \begin{equation*}
    \begin{pmatrix}
      h(x)\cdot w(x)\\ h(x') \cdot w(x')
    \end{pmatrix}
    = \Lambda(x,x')
    \begin{pmatrix}
      |h(x)|^2 \\ |h(x')|^2
    \end{pmatrix}
    ,
  \end{equation*}
  for a function $\Lambda$, which in general depends on $(x,x')$ and this holds
  for almost all $(x,x')$. We deduce that $\Lambda$ must be constant almost
  everywhere in $\mathcal{O}$, i.e. $\Lambda(x,x') = \lambda$ for almost all
  $(x,x')$ and thus
  \begin{equation*}
    \lambda |h(x)|^2 = h(x) \cdot w(x) \iff w(x) = \lambda h(x),
  \end{equation*}
  which implies condition~(\ref{eq:KC}). Hence, metriplectic brackets
  constructed by means of an $L^2$-orthogonal projection are minimally
  degenerate as already shown in
  Section~\ref{sec:projector-based-metric-bracket} for the scalar case ($N=1$). 
\end{example}

\begin{example}
  \label{ex:old-clb}
  As in example~\ref{ex:proj-revisited},
  let $\mathcal{O} = \Omega \times \Omega$, $n = 2d$, with coordinates
  $z = (x,x') \in \Omega \times \Omega$, and $d\nu(x,x') = d\mu(x) d\mu(x')$.
  The operator $P$ is given by
  \begin{equation*}
    Pw(x,x') = P_L w(x,x') = Lw(x) - Lw(x'),
  \end{equation*}
  where $L : W \to L^2(\Omega,\mu; \R^{\tilde{N}})$ is possibly an unbounded
  linear operator with domain $\dom(L) = \Phi \subseteq W$ and taking values in
  the space of $\R^{\tilde{N}}$-valued, squared-integral functions, with
  $\tilde{N} \geq 2$. In general, we allow $\tilde{N}$ to be different from $N$.
  Specific cases will be considered in Secs.~\ref{sec:c-div-grad}
  and~\ref{sec:c-curl-curl} below. 
  
  \smallskip 

  As for the kernel, a way to satisfy condition~(\ref{eq:T-energy-cond})
  utilizes the matrix
  \begin{equation}
    \label{eq:Q}
    Q_k(z) \coloneqq |z|^2 I_k - z \otimes z, \qquad z \in \R^k,
  \end{equation}
  where $k \in \N$ and $I_k$ is the $k \times k$ identity block.
  Specifically, let $\TT(u)$ be the multiplication operator by the matrix
  \begin{equation}
    \label{eq:T-L}
    \TT(u;x,x') = \frac{1}{2} \kappa(u;x,x') Q_{\tilde{N}} \Big(
    P_L \frac{\delta \fun{H}(u)}{\delta u} (x,x')
    \Big),
  \end{equation}
  where $\kappa(u;\cdot,\cdot)$ is a positive scalar weight function satisfying the
  symmetry condition $\kappa(u;x,x') = \kappa(u;x',x)$. (This choice of the
  kernel is trivial if $\tilde{N}=1$, which is excluded.) With the foregoing
  choices, Eq.~(\ref{eq:gen-clb}) becomes 
  \begin{multline}
    \label{eq:cb-clb1}
    (\fun{F},\fun{G}) \coloneqq \frac{1}{2} \int_\Omega \int_\Omega
    \kappa(u;x,x') \Big[P_L \frac{\delta \fun{F}(u)}{\delta u}(x,x') \Big] \\
    \cdot Q_{\tilde{N}} \Big(
    P_L \frac{\delta \fun{H}(u)}{\delta u}(x,x')\Big)
    \Big[P_L \frac{\delta \fun{G}(u)}{\delta u}(x,x')\Big] d\mu(x) d\mu(x').
  \end{multline}
  This is the form of collision-like bracket as originally formulated by
  Bressan et al. \cite{Bressan2023,Bressan2018}. The dual operator $P'$ can be
  determined in terms of the dual of $L$, which is the linear operator
  $L' : L^2(\Omega,\mu;\R^{\tilde{N}}) \to \Phi'$ defined for all
  $\varpi \in L^2(\Omega,\mu;\R^{\tilde{N}})$ by
  \begin{equation}
    \label{eq:L-dual}
    \langle w, L' \varpi \rangle = \int_{\Omega} \varpi(x) \cdot Lw(x) d\mu(x),
    \qquad \text{for all } w \in \Phi,
  \end{equation}
  and we find
  \begin{align*}
    \int_{\mathcal{O}} \tilde{w}(x,x') \cdot P_Lw(x,x') d\nu(x,x') &=
    \int_\Omega Lw(x) \int_\Omega \big[\tilde{w}(x,x') - \tilde{w}(x',x)\big]
    d\mu(x') d\mu(x) \\
    &= \Big\langle w, L'\Big[
      \int_\Omega \big[\tilde{w}(\cdot,x') - \tilde{w}(x',\cdot)\big] d\mu(x')
      \Big] \Big\rangle \\
    &= \big\langle w, P_L' \tilde{w}\big\rangle.
  \end{align*}
  In terms of the operator $L'$, the evolution equation~(\ref{eq:clb-ev}) reads
  \begin{equation}
    \label{eq:old-clb-ev}
    \partial_t u = -L' \Big[ \mathbb{D}(u) L\frac{\delta \fun{S}(u)}{\delta u}
      - \mathbb{F}\Big(u,\frac{\delta \fun{S}(u)}{\delta u}\Big)\Big],
  \end{equation}
  where we have defined
  \begin{subequations}
    \label{eq:D-F}
    \begin{align}
      \label{eq:D}
      \mathbb{D}(u;x)  &\coloneqq \int_\Omega
      \kappa(u;x,x') Q_{\tilde{N}} \Big(
      P_L \frac{\delta \fun{H}(u)}{\delta u}(x,x')\Big) d\mu(x'),\\
      \label{eq:F}
      \mathbb{F}(u,v;x) &\coloneqq \int_\Omega \kappa(u;x,x')
      Q_{\tilde{N}} \Big(P_L \frac{\delta \fun{H}(u)}{\delta u}(x,x')\Big)
      Lv(x') d\mu(x').
    \end{align}
  \end{subequations}  
  Eq.~(\ref{eq:old-clb-ev}) has the same structure as the Landau operator for
  Coulomb collisions discussed in Example~\ref{ex:Landau}, with $\nabla_\vv$
  being replaced by $L$.  

  At last, we check conditions~(\ref{eq:MDC}). Since $d\mu(x) = m(x) dx$, with
  $m$ continuous on the closed domain $\ol{\Omega}$, and
  $d\nu(x,x') = d\mu(x) d\mu(x')$, we have
  \begin{equation*}
    \|Pw\|^2_{\tilde{W}} = \int_{\mathcal{O}} \big|L w(x) - L w(x')\big|^2
    d\nu(x,x') \geq m_0^2
    \int_{\mathcal{O}} \big|L w(x) - L w(x')\big|^2 dx dx',
  \end{equation*}
  where $m_0 = \min \{m(x) \colon x \in \ol{\Omega}\}$. For
  condition~(\ref{eq:PI}) to hold, it is therefore sufficient that $L$ satisfies
  \begin{equation}
    \label{eq:PI-L}
    \|w\|_W \leq C_L \|Lw\|_{L^2(\Omega)}, \qquad
    \int_\Omega Lw(x) dx = 0,
  \end{equation}
  so that
  \begin{equation*}
    \int_{\mathcal{O}} \big|L w(x) - L w(x')\big|^2 dx dx'
    = 2 |\Omega| \int_\Omega \big|Lw(x)|^2 dx - 2 \Big|
    \int_\Omega Lw(x)dx\Big|^2\geq 2 |\Omega| C_L^{-2} \|w\|^2_W,
  \end{equation*}
  which give Eq.~(\ref{eq:PI}). All considered cases in the applications below
  satisfy condition~(\ref{eq:PI-L}).

  As for the kernel of the bracket, we argue as in the case of
  Example~\ref{ex:Landau}. With the kernel given in Eq.~(\ref{eq:T-L}),
  a function $\tilde{w} \in \tilde{W}$ belongs to $\ker \TT(u) \cap \rng(P)$
  only if there is a function $w \in W$ and
  $\Lambda : \Omega \times \Omega \to \R$, such that
  \begin{equation*}
    Lw(x) - Lw(x') = \Lambda(x,x') \big[Lh(x) - Lh(x')\big],
  \end{equation*}
  where $h = \delta \fun{H}(u) / \delta u$ and necessarily
  $\Lambda(x,x') = \Lambda(x',x)$ for $x \not= x'$.
  If we fix an arbitrary point
  $x' = a \in \Omega$, we obtain an explicit expression for $Lw(x)$,
  \begin{equation*}
    Lw(x) = Lw(a) + \Lambda(x,a) X_a(x), \qquad X_a(x) \coloneqq Lh(x) - Lh(a).
  \end{equation*}
  Therefore
  \begin{align*}
    Lw(x) - Lw(x') &= \Lambda(x,a) X_a(x) - \Lambda(a,x') X_a(x') \\
    &= \Lambda(x,x') \big[X_a(x) - X_a(x')\big],
  \end{align*}
  or equivalently
  \begin{equation*}
  \big[\Lambda(x,x') - \Lambda(x,a)\big] X_a(x) -
  \big[\Lambda(x,x') - \Lambda(a,x')\big] X_a(x') = 0.
  \end{equation*}
  Under the assumption that $X_a(x)$ and $X_a(x')$ are linearly independent for
  almost all $(x,x') \in \Omega \times \Omega = \mathcal{O}$, we deduce
  $\Lambda(x,x') = \lambda =$ constant as in Example~\ref{ex:Landau}. Hence, if
  the operator $L$ satisfies (\ref{eq:PI-L}) and the Hamiltonian is such that
  $X_a(x)$ and $X_a(x')$ are linearly independent almost everywhere in
  $\Omega \times \Omega$, the bracket~(\ref{eq:cb-clb1}) is minimally
  degenerate. 
\end{example}

In general we observe that the brackets of the form discussed in
Example~\ref{ex:old-clb} are all special cases of the metriplectic $4$-bracket
introduced recently by Morrison and Updike \cite{pjmU24}, since the kernel
depends quadratically on $\delta \fun{H} / \delta u$. Specifically, they can
be obtained from the metriplectic $4$-bracket constructed by utilizing the
Kulkarni-Nomizu product as discussed in Section~III.D.1 of Ref.~\cite{pjmU24}.

Example~\ref{ex:old-clb} will be used as a template for the construction of
dissipative operators for the solution of some of the variational problems
introduced in Section~\ref{sec:testcases}. Before addressing the applications,
we specialize bracket~(\ref{eq:cb-clb1}) for two relevant choices of the
operator $L$.

\subsection{Collision-like brackets based on div--grad operators}
\label{sec:c-div-grad}

In Eq.~(\ref{eq:cb-clb1}), let us consider the case of a scalar field ($N=1$)
and choose
\begin{equation*}
  L w = \nabla w,
\end{equation*}
with domain $\dom(L) = H^1_0(\Omega)$, the space of functions in $L^2(\Omega)$
with weak derivatives also in $L^2(\Omega)$ and satisfying homogeneous Dirichlet
boundary conditions on $\partial \Omega$. This operator satisfies condition
(\ref{eq:PI-L}), and $\tilde{N} = d$ (the inequality, in particular, is the
classical Poincar\'e inequality for $H_0^1$, cf.\ Theorem 3 in Section~5.6.1 of
Ref.~\cite{Evans1998}). The dual operator $L'$ can be obtained from
Eq.~(\ref{eq:L-dual}). Given $\varpi \in L^2(\Omega,\mu; \R^d)$, sufficiently
regular, we have  
\begin{equation*}
  \langle w, L' \varpi \rangle = \int_\Omega \varpi(x) \cdot \nabla w(x) d\mu(x)
  = - \int_\Omega w(x) \div_\mu \varpi(x) d\mu(x),
\end{equation*}
and thus $-L' \varpi = \div_\mu \varpi = \frac{1}{m} \sum_i \partial_{x_i}
\big[m \varpi^i\big]$ is the divergence operator associated to the volume form 
$d\mu (x) = m(x) dx$. For a generic $\varpi$, we write
$L' = -\widetilde{\div_\mu}$ to denote that the divergence is defined weakly. As
for the kernel, we choose  
\begin{equation*}
  \kappa (u;x,x') = M\big(x,u(x)\big) M\big(x',u(x')\big),
\end{equation*}
where $M(x,y) > 0$ is a positive function over $\Omega \times \R$, which can
be used to simplify the bracket \cite{Morrison1986} as in
Example~\ref{ex:Landau}. 

As for the entropy, we shall restrict to functions of the form
\begin{equation}
  \label{eq:entropy-2D}
  \fun{S}(u) = \int_\Omega s\big(x,u(x)\big) dx,
\end{equation}
where $s(x,y)$ is a given profile, convex in $y$, and possibly depending on
position $x$. Then 
\begin{equation*}
  \frac{\delta \fun{S}(u)}{\delta u} (x) = \partial_y s\big(x,u(x)\big),
\end{equation*}
and
\begin{equation*}
  \nabla \frac{\delta \fun{S}(u)}{\delta u} (x) =
  \partial_y^2 s\big(x,u(x)\big) \nabla u(x) +
  \partial_x\partial_y s\big(x,u(x)\big).
\end{equation*}
The arbitrary function $M$ can be chosen so that \cite{Morrison1986}
\begin{equation}
  \label{eq:M-condition}
  M(x,y) \partial_y^2s(x,y) = 1,
\end{equation}
under the condition $\partial_y^2 s(x,y) > 0$, consistently with the convexity
assumption for the profile. Eq.~(\ref{eq:M-condition}) introduces a dependence
of the kernel of the bracket on the entropy function. 
With this choice, Eq.~(\ref{eq:old-clb-ev}) becomes
\begin{equation}
  \label{eq:div-grad-ev}
  \partial_t u = \widetilde{\div_\mu} \Big[\mathbb{D}_s(u)
    \big(\nabla u + M(x,u)\partial_x\partial_y s(x,u)\big) -
    M(x,u) \mathbb{F}_s(u, \nabla u)\Big],
\end{equation}
and
\begin{align*}
  \mathbb{D}_s(u;x) &=
  \int_\Omega Q_2\Big(\nabla \frac{\delta \fun{H}(u)}{\delta u}(x) -
  \nabla \frac{\delta \fun{H}(u)}{\delta u}(x')\Big)
  M\big(x',u(x')\big)d\mu(x'), \\
  \mathbb{F}_s(u;x) &= 
  \int_\Omega Q_2\Big(\nabla \frac{\delta \fun{H}(u)}{\delta u}(x) -
  \nabla \frac{\delta \fun{H}(u)}{\delta u}(x')\Big) \\
  & \hspace{10mm}
  \big[\nabla u(x') + M\big(x',u(x')\big)
    \partial_x\partial_y s\big(x',u(x')\big) \big] d\mu(x').
\end{align*}
Eq.~(\ref{eq:div-grad-ev}) still depends on the choice of the measure
$\mu$, the entropy profile $s$, and the Hamiltonian function $\fun{H}$. This
choices will be specified below separately for the cases of the reduced Euler
and Grad-Shafranov equations in two-dimensions ($d = \tilde{N} = 2$).

\subsection{Collision-like brackets based on curl--curl operators}  
\label{sec:c-curl-curl}

With magnetic fields in mind, we consider an example of
bracket~(\ref{eq:gen-clb}) for divergence-free fields. Then, $\Omega$ is
bounded, simply connected domain in $\R^3$  with sufficiently regular connected
boundary, $d = N = 3$, and $d\mu(x) = dx$. We choose
\begin{equation*}
  L w = \curl w,
\end{equation*}
with domain
\begin{equation*}
  \Phi = \{w \in H(\curl,\Omega) \cap H(\div,\Omega) \colon
  \text{ $\div w = 0$ in $\Omega$, $n \times w = 0$ on $\partial \Omega$}\},
\end{equation*}
where $H(\curl,\Omega)$ and $H(\div,\Omega)$ are the spaces of vector fields
$w \in L^2(\Omega;\R^3)$ such that $\curl w \in L^2(\Omega;\R^3)$ and
$\div w \in L^2(\Omega)$, respectively. Specifically, $w \in \Phi = \dom(L)$ is
a divergence-free field with zero tangential component $n \times w$ on
the boundary, and $n : \partial \Omega \to \R^3$ is the outward unit normal on
$\partial \Omega$. It follows that condition~(\ref{eq:PI-L}) is satisfied: the
inequality amounts to the Poincar\'e inequality for divergence-free fields on a
simply connected domain with connected boundary (cf. Corollary 3.51 in
Ref.\cite{Monk2003}), while
\begin{equation*}
  \int_\Omega \curl w dx = \int_{\partial \Omega} (n \times w) d\sigma = 0,
\end{equation*}
where $d\sigma$ is the surface element on $\partial \Omega$.

If $\varpi \in L^2(\Omega;\R^3)$ is sufficiently regular, Eq.~(\ref{eq:L-dual})
and integration by parts gives
\begin{equation*}
  \langle w, L'\varpi \rangle = \int_\Omega \varpi(x) \cdot \curl w(x) dx
  = \int_\Omega w(x) \cdot \curl \varpi(x) dx,
\end{equation*}
hence $L' \varpi = \curl \varpi$, while for a generic $\varpi$,
$L' = \widetilde{\curl}$ is the weak curl operator.

For the kernel~(\ref{eq:T-L}), we choose $\kappa = 1$, and the corresponding
evolution equation can be written as 
\begin{equation}
  \label{eq:curl-curl-ev}
  \partial_t u = -\widetilde{\curl} \Big[
    \mathbb{D}_v(u) \curl \frac{\delta \fun{S}(u)}{\delta u}
    - \mathbb{F}_v\Big(u, \frac{\delta \fun{S}(u)}{\delta u}\Big) \Big],
\end{equation}
where
\begin{align*}
  \mathbb{D}_v(u;x) &= \int_\Omega Q_3\Big(
  \curl \frac{\delta \fun{H}(u)}{\delta u}(x) -
  \curl \frac{\delta \fun{H}(u)}{\delta u}(x') \Big) dx', \\
  \mathbb{F}_v(u,v;x) &= \int_\Omega Q_3\Big(
  \curl \frac{\delta \fun{H}(u)}{\delta u}(x) -
  \curl \frac{\delta \fun{H}(u)}{\delta u}(x') \Big)
  \curl v (x') dx'.
\end{align*}
A non-trivial factor $\kappa$ of the form used in Section~\ref{sec:c-div-grad}
can easily be accounted for \cite{Bressan2023}.

The evolution equation (\ref{eq:curl-curl-ev}) preserves the divergence-free
constraint. In fact, for any $\varphi \in H_0^1(\Omega)$ the function
$\fun{F}_\varphi(u) = \int_\Omega u \cdot \nabla\varphi dx
= -\int_\Omega \varphi \div u dx$ is a constant of motion.

\subsection{Application to the reduced Euler equations}
\label{sec:app-euler}

In order to demonstrate the properties of the collision-like brackets, we
construct a relaxation method for the solution of the variational
principle~(\ref{eq:VP-Euler}) for equilibria of the reduced Euler equations,
with entropy and Hamiltonian functions given in Eq.~(\ref{eq:Euler-S-H}). 

This is the same problem addressed in Section~\ref{sec:simple}, with the
difference here being we consider a bounded domain $\Omega$ with homogeneous
Dirichlet boundary conditions as discussed in Section~\ref{sec:Euler-problem}.
Specifically, the domain is the unit square $\ol{\Omega} = [0,1] \times [0,1]$. 

When, in Eq.~(\ref{eq:Euler-S-H}), $s(\omega) = \omega^2/2$, we known from
Section~\ref{sec:Euler-problem} that the solution of the variational principle
(\ref{eq:VP-Euler}) is related to the eigenfunction of the Laplacian operator on
$\Omega$ corresponding to the smallest eigenvalue. On the unit square with
homogeneous Dirichlet boundary conditions, the eigenfunctions and the
corresponding eigenvalues are given by  
\begin{equation*}
  \phi_{m,n}(x) = \mathcal{N}_{m,n} \sin(m \pi x_1) \sin(n \pi x_2), \qquad
  \lambda_{m,n} = \pi^2 (m^2 + n^2),
\end{equation*}
labeled by two non-zero integers $m,n \in \N$ and normalized by a non-zero
constant $\mathcal{N}_{m,n}$. The smallest eigenvalue corresponds to the
function $\phi_{1,1}$, and energy conservation allows us to determine the
normalization constant $\mathcal{N}_{1,1}$, that is 
\begin{equation*}
  2 \fun{H}_0 = \int_\Omega \phi \omega dx = \lambda_{1,1}
  \|\phi_{1,1}\|^2_{L^2(\Omega)} = \lambda_{1,1} \mathcal{N}_{1,1}^2 /4,
\end{equation*}
from which we deduce the unique solution of (\ref{eq:VP-Euler}) with this
entropy, that is
\begin{equation}
  \label{eq:euler-ref-lin}
  \begin{aligned}
    \phi(x) &= \big(2\sqrt{\fun{H}_0} / \pi\big) \sin(\pi x_1) \sin(\pi x_2), \\
    \omega(x) = \lambda_{1,1} \phi(x)
    &= 4\pi \sqrt{\fun{H}_0}  \sin(\pi x_1) \sin(\pi x_2),
  \end{aligned}
  \qquad
  \lambda_{1,1} = 2 \pi^2 \approx 19.7392.
\end{equation}
More general choices of the entropy density $s(\omega)$ lead to the nonlinear
eigenvalue problems of the form~(\ref{eq:complete-relaxation-vorticity2d}). Then
no analytical solution is known, but an estimate of the eigenvalue $\lambda$ can
be found from energy conservation in a similar way. Upon multiplying by $\omega$
Eq.~(\ref{eq:complete-relaxation-vorticity2d}) and integrating over $\Omega$, we
find  
\begin{equation*} 
  \int_\Omega \omega s'(\omega) dx = \lambda \int_\Omega \omega \phi dx =
  2 \lambda \fun{H}_0,
\end{equation*}
from which we can obtain $\lambda$, provided that the integral on the left-hand
side can be evaluated, e.g. from the numerical solution. Specifically, if
$s(\omega) = \omega \log \omega$ we have
\begin{equation}
  \label{eq:euler-ref-log-lambda}
  \lambda =\frac{1}{2 \fun{H}_0} \big(\fun{M}(\omega) + \fun{S}(\omega) \big),
\end{equation}
where, in particular, $\fun{M}(\omega) = \int_\Omega \omega dx$. These
analytical results can be used to assess the result of the proposed relaxation
method.

As for the construction of the method itself, we consider the collision-like
metriplectic system of Section~\ref{sec:c-div-grad}. The state variable
$u(t) \in V \subseteq \Phi$ is identified with vorticity, $u(t,x)=\omega(t,x)$,
and evolved according to Eq.~(\ref{eq:div-grad-ev}) with $\fun{S}$ and $\fun{H}$
given in Eq.~(\ref{eq:Euler-S-H}). 

The numerical scheme for the solution of Eq.~(\ref{eq:div-grad-ev}) is based on
$H^1$-conforming finite elements, i.e., the discrete solution $u_h(t)$ belongs
to the same space $H_0^1(\Omega)$ that contains the phase space $V$. The scheme
preserves the discrete Hamiltonian $\fun{H}(u_h)$ (modulo round-off errors) and
we use the Crank-Nicolson scheme in time, which preserves the property of
monotonic dissipation of entropy in the quadratic case
($s(\omega) = \omega^2/2$). More details on the derivation of the scheme can be
found in Ref.~\cite{Bressan2023}.

The obtained numerical method has been implemented in the finite-element library
FEniCS \cite{Alnaes2015, Logg2012}, in which the weak form of the operator in
Eq.~(\ref{eq:div-grad-ev}) can be directly specified by means of the
domain-specific language UFL \cite{Alnaes2014}, and discretized by the
finite-element component DOLFIN \cite{Logg2010,Logg2012a}. Among the various
tests \cite{Bressan2023}, here we discuss in details three cases only.

\paragraph*{Single vortex}

In the simplest example, the domain $\ol{\Omega} = [0,1]^2$ is discretized by a
uniform grid of $64 \times 64$ nodes. The vorticity field $u = \omega$ is
approximated in the space of second-order Lagrange elements that are available
in FEniCS \cite{Logg2012}. The entropy density is quadratic, i.e.
$s(\omega) = \omega^2/2$, hence the analytical solution of the variational
principle is given in Eq.~(\ref{eq:euler-ref-lin}). The initial condition is
\begin{equation}
  \label{eq:sv-ic}
  u_0(x) = u_G(x), \;
  \text{with $u_G$ defined in Eq.~(\ref{eq:initial_gaussian})},
\end{equation}
and with parameters $x_{0,1} = x_{0,2} = 1/2$, $w_1^2 = 0.01$, $w_2^2 = 0.07$,
and $N=1$. The numerical scheme provably preserves the Hamiltonian, which in
this case is the kinetic energy of the fluid, and dissipates monotonically the
entropy, independently of the magnitude of the time step. 

\begin{figure}
  \centering
  \includegraphics[scale=0.24]{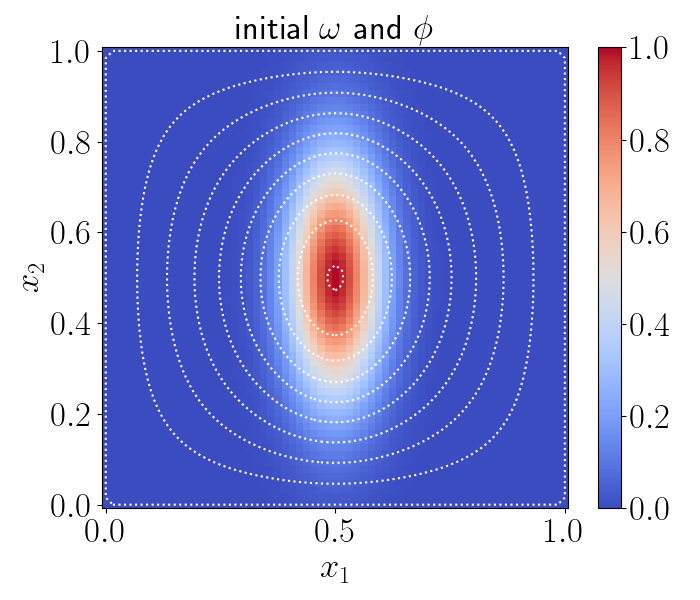}
  \includegraphics[scale=0.24]{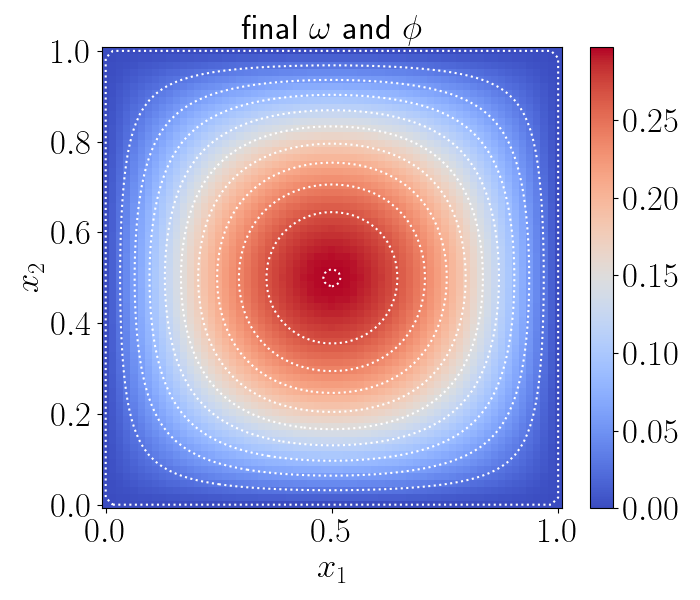}
  \includegraphics[scale=0.24]{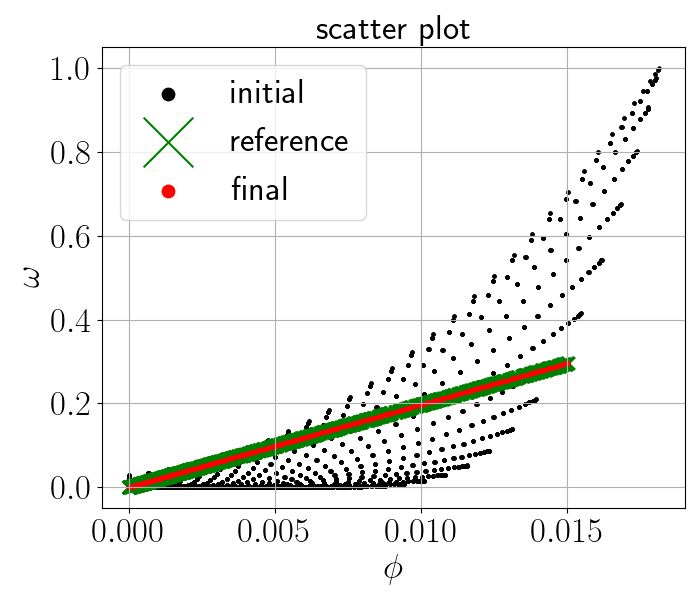}
  \caption{\label{fig:sv_sol} Relaxation of an initial vortex with initial
    vorticity given in Eq.~(\ref{eq:sv-ic}), according to the metriplectic
    system~(\ref{eq:div-grad-ev}) with the state variable $u$ being the fluid
    scalar vorticity $\omega$. The initial condition is shown in the
    left-hand-side panel: the color map represents the vorticity field $\omega$ 
    and the dashed (white) lines the contours of the corresponding streaming
    function $\phi$, cf. Eq.~(\ref{eq:Poisson-eq}). The relaxed state is
    displayed in the middle panel. The right-hand-side panel represents the
    functional relation between $\omega$ and $\phi$ for the initial condition
    (black dots), the final state (red bullets), and the expected linear
    relation $\omega = \lambda_{1,1} \phi$, cf.\ Eq.~(\ref{eq:euler-ref-lin}),
    plotted using the numerical solution for $\phi$ and the analytical value of
    $\lambda_{1,1}$ (green crosses).}  
\end{figure}

\begin{figure}
  \centering
  \includegraphics[scale=0.29]{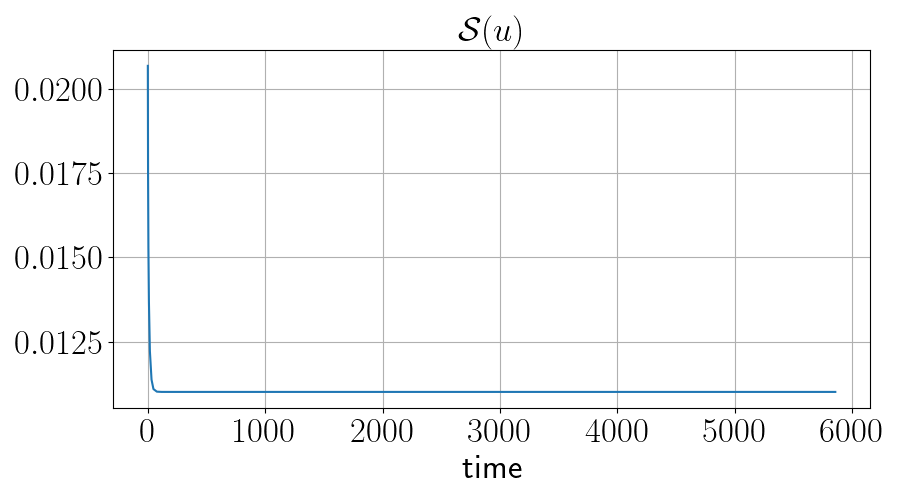}
  \includegraphics[scale=0.29]{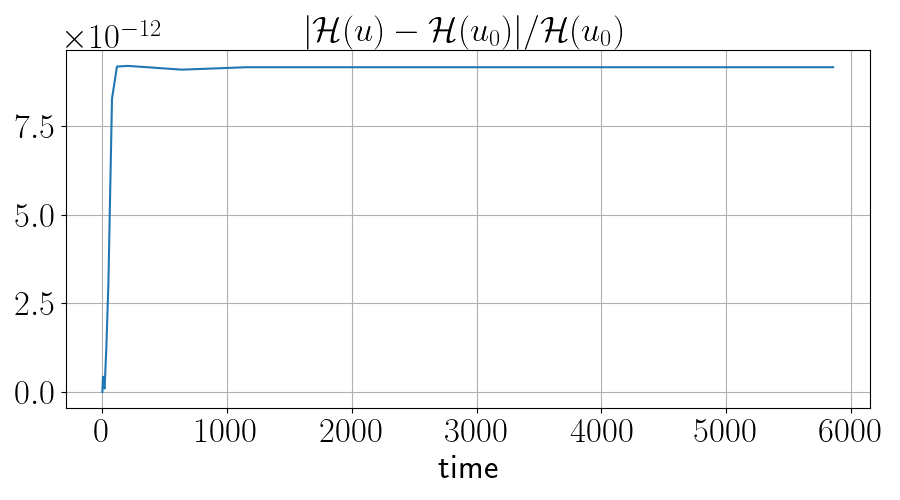}
  \caption{\label{fig:sv_diag} Evolution of entropy (left-hand-side panel) and
    of the variation of the Hamiltonian relative to its initial value
    (right-hand-side panel), for the case in Fig.~\ref{fig:sv_sol}.}  
\end{figure}

\begin{figure}
  \centering
  \includegraphics[scale=0.3]{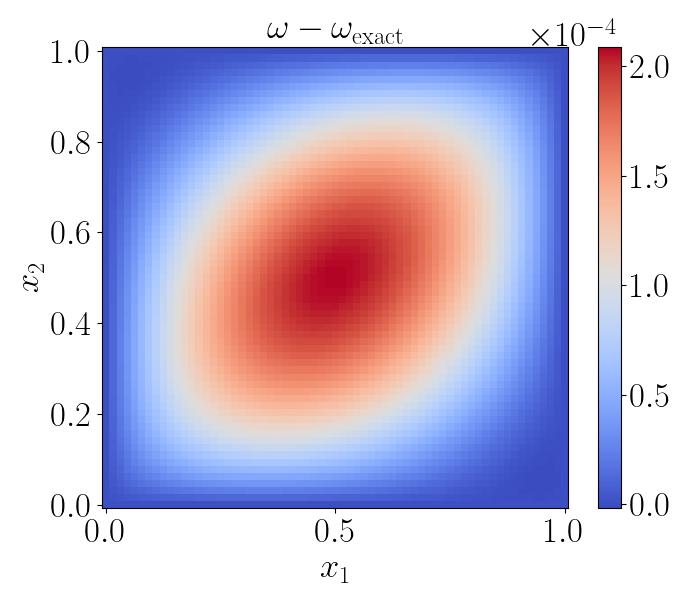}
  \includegraphics[scale=0.3]{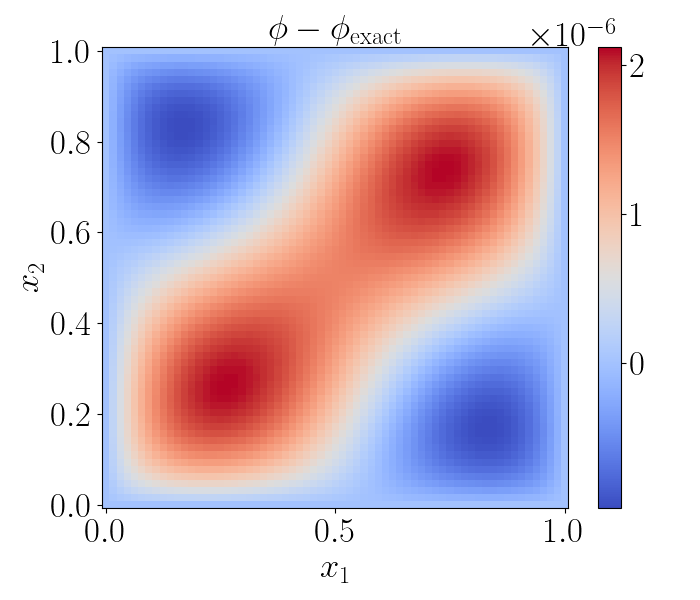}
  \caption{\label{fig:sv_err} Difference between the relaxed state and the
      exact analytical solution~(\ref{eq:euler-ref-lin}) of the variational
      principle (\ref{eq:VP-Euler}) for the case in Fig.~\ref{fig:sv_sol}. The
      difference of the vorticity fields is shown in the left-hand-side panel,
      while the difference of the corresponding potentials is shown in the
      right-had-side panel. For the evaluation of the analytical
      solution~(\ref{eq:euler-ref-lin}), the initial energy $\fun{H}_0$ has been
      computed numerically. }  
\end{figure}

Figure~\ref{fig:sv_sol} shows the initial condition, the final state, and the
``scatter plot'', which we use to identify functional relations between
different fields, cf. the analysis in Section~\ref{sec:simple}. Qualitatively,
we see that the initial condition is quite far from an equilibrium of the Euler
equations as the contours of the streaming function $\phi$ and those of
vorticity $\omega$ are misaligned. Metriplectic relaxation with collision-like
brackets yields a solution that appears to be an equilibrium, and from the
scatter plot (right-hand-side panel in Fig.~\ref{fig:sv_sol}), one can see that
the relaxed state is indeed an equilibrium characterized by the same linear
relation of the exact solution~(\ref{eq:euler-ref-lin}). Therefore, the
collision-like metric bracket appears to have completely relaxed the initial
solution in the sense discussed in Section~\ref{sec:remarks-relax-equil}. From
Fig.~\ref{fig:sv_diag}, one can verify energy conservation and entropy monotonic
dissipation.

For this test case, the exact solution of the variational problem
(\ref{eq:VP-Euler}) has been computed analytically,
Eq.~(\ref{eq:euler-ref-lin}), and we can evaluate the difference between the
the relaxed state of the metriplectic system and the solution of the
variational problem. Figure~\ref{fig:sv_err} shows that the relaxed state
appears to be close to the expected solution of the variational principle.

\paragraph*{Perturbed equilibrium}

We repeat the experiment of Fig.~\ref{fig:sv_sol} with an initial condition
close to a critical point of entropy restricted to the constant-Hamiltonian
surface. Specifically, the initial condition is
\begin{equation}
  \label{eq:pe-ic}
  u_0(x) = \sin(6\pi x_1)\sin(4\pi x_2) + u_G(x), \;
  \text{with $u_G$ defined in Eq.~(\ref{eq:initial_gaussian})},
\end{equation}
and with the same parameters as in Eq.~(\ref{eq:sv-ic}) except for $N=100$.

\begin{figure}
  \centering
  \includegraphics[scale=0.24]{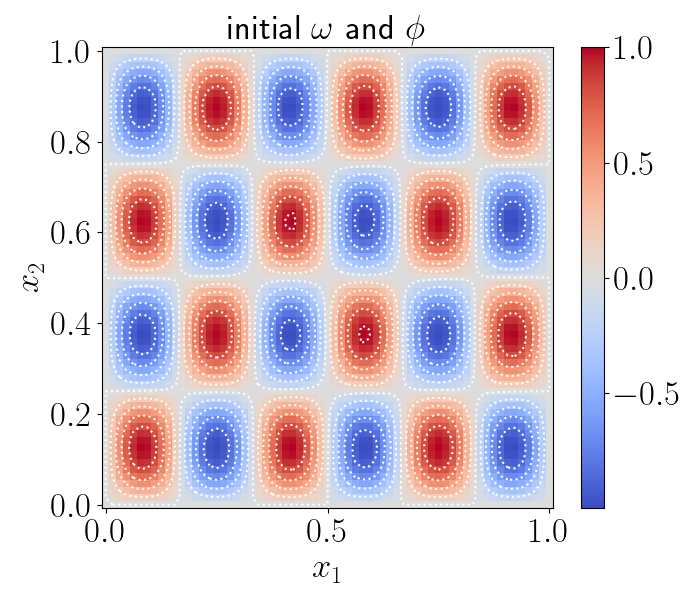}
  \includegraphics[scale=0.24]{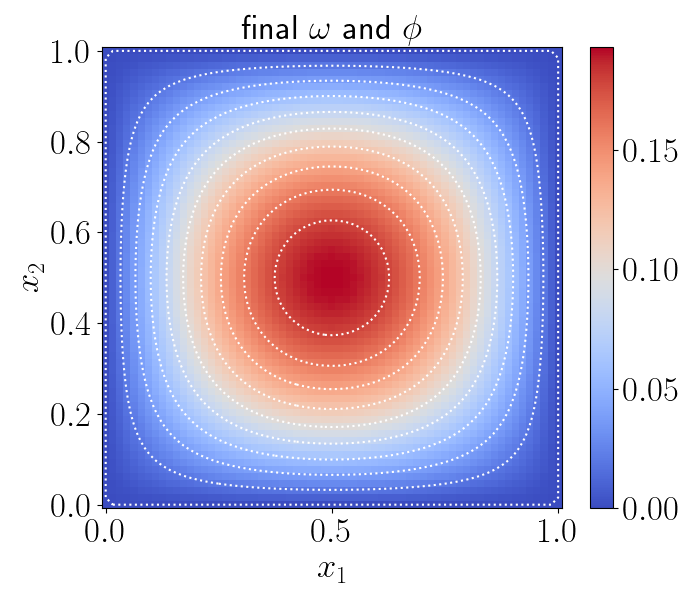}
  \includegraphics[scale=0.24]{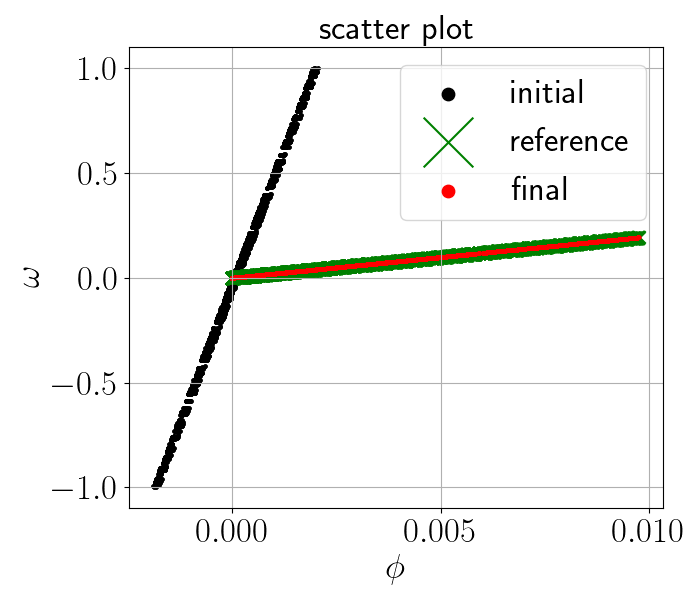}
  \caption{\label{fig:pe_sol} The same as Fig.~\ref{fig:sv_sol}, but for the
    initial condition~(\ref{eq:pe-ic}). } 
\end{figure}

\begin{figure}
  \centering
  \includegraphics[scale=0.29]{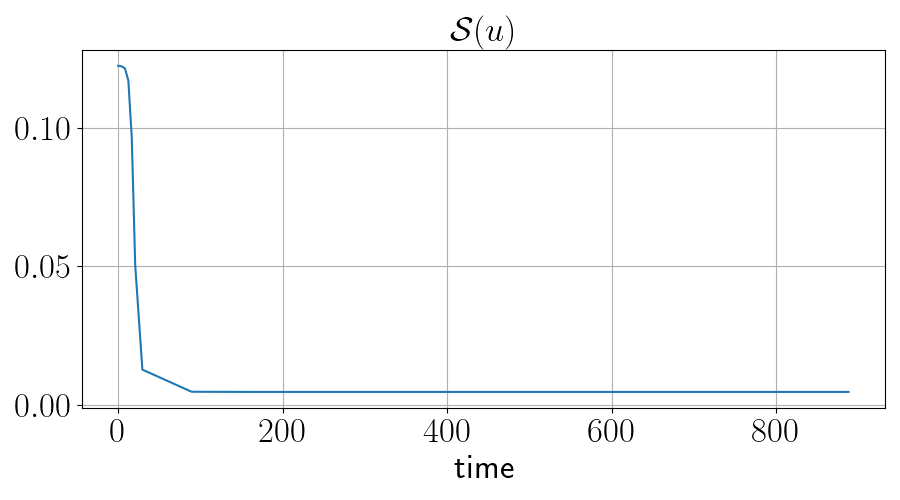}
  \includegraphics[scale=0.29]{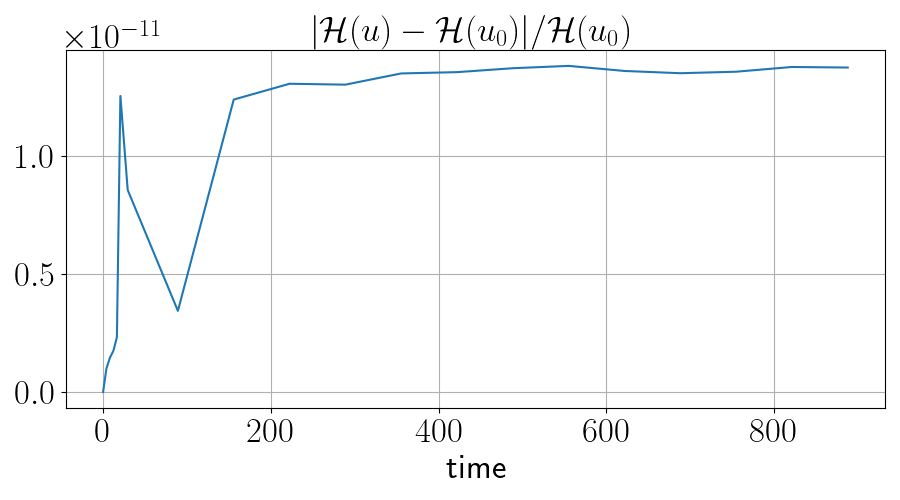}
  \caption{\label{fig:pe_diag} Evolution of entropy (left-hand-side panel) and
    of the variation of the Hamiltonian relative to its initial value
    (right-hand-side panel), for the case in Fig.~\ref{fig:pe_sol}.}  
\end{figure}

\begin{figure}
  \centering
  \includegraphics[scale=0.3]{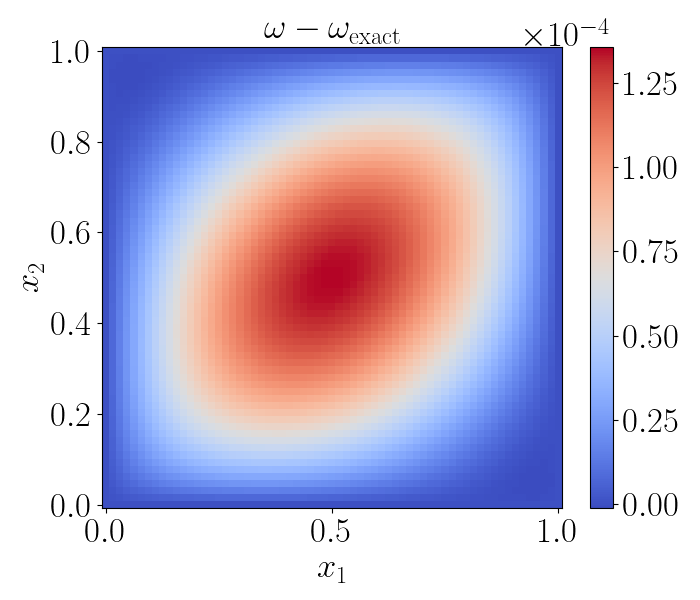}
  \includegraphics[scale=0.3]{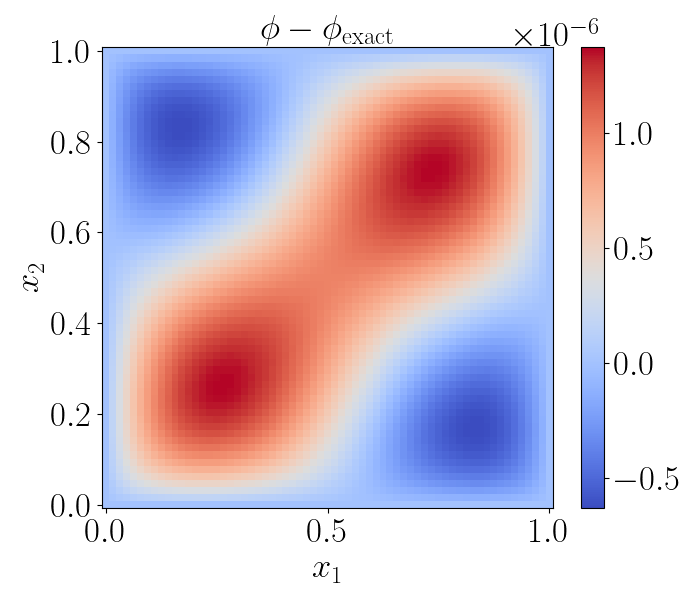}
  \caption{\label{fig:pe_err} The same as Fig.~\ref{fig:sv_err}, but for the
    initial condition~(\ref{eq:pe-ic}). } 
\end{figure}

Figure~\ref{fig:pe_sol} shows the initial condition, the final state, and
the usual scatter plot, which visualizes the relationship between $\omega$ and
$\phi$. The initial condition (Fig.~\ref{fig:pe_sol}, left-hand-side panel) is
basically an eigenfunction of the Laplace operator corresponding to a large
eigenvalue, the perturbation being hardly visible. This is confirmed by the
scatter plot (Fig.~\ref{fig:pe_sol} right-hand-side panel, black dots) in which
the initial state is concentrated on a straight line, with a small spread due to
the Gaussian perturbation. Therefore the initial condition is close to an
equilibrium. The relaxed state (Fig.~\ref{fig:pe_sol}, middle panel) is not
exactly the same as in Fig.~\ref{fig:sv_sol}, since the initial value
$\fun{H}_0$ of the Hamiltonian is different, but is it consistent with the exact
solution (\ref{eq:euler-ref-lin}) as shown in the scatter plot. Therefore, the
initial condition has been completely relaxed to a solution of variational
principle~(\ref{eq:VP-Euler}). In Fig.~\ref{fig:pe_diag}, one can see that the 
Hamiltonian function is preserved to machine accuracy, while entropy decays
monotonically as expected. However, initially entropy appears to remain
constant, due to the proximity of the initial condition to an equilibrium point.
The difference between the relaxed state and the analytical solution of the
variational principle~(\ref{eq:VP-Euler}) is shown in Fig.~\ref{fig:pe_err}.

\begin{figure}
  \centering
  \includegraphics[scale=0.24]{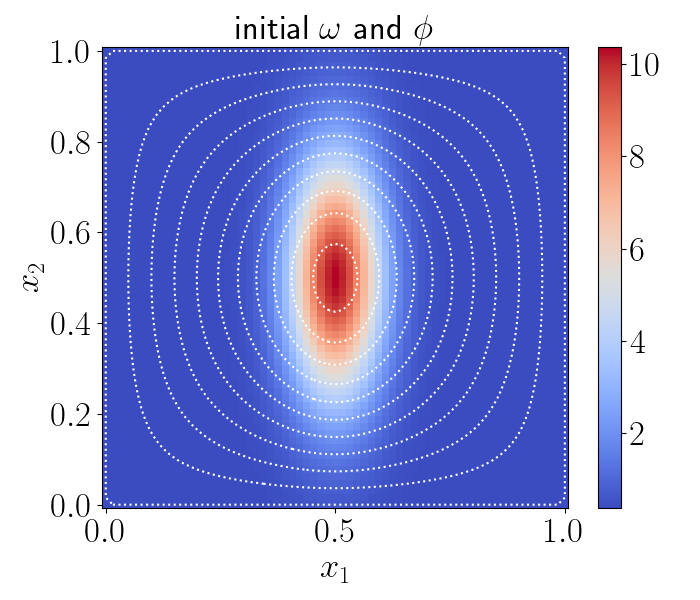}
  \includegraphics[scale=0.24]{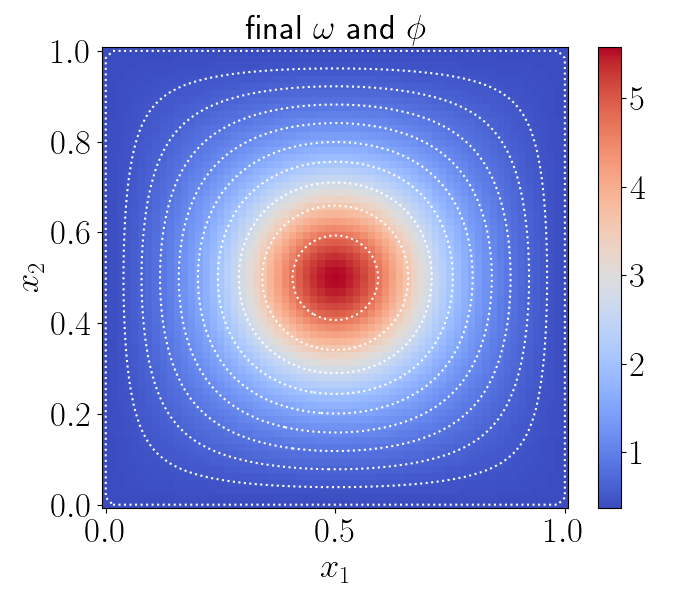}
  \includegraphics[scale=0.24]{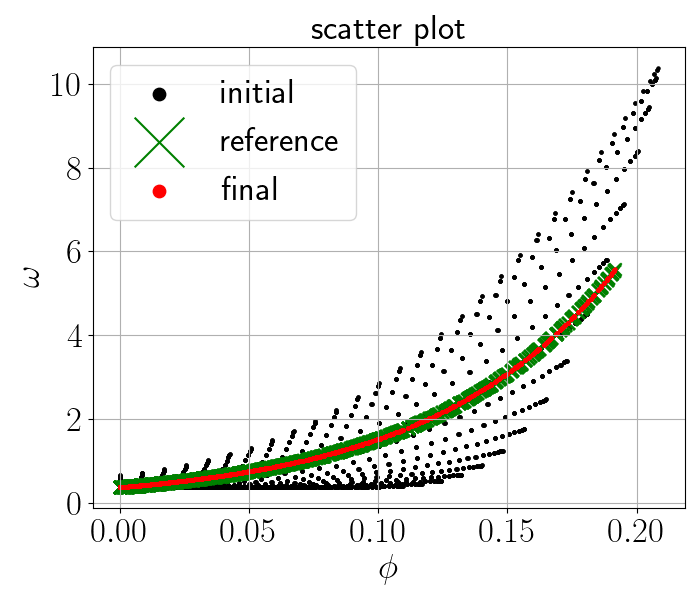}
  \caption{\label{fig:ge_sol} The same as Fig.~\ref{fig:sv_sol}, but with the
    entropy density $s(\omega) = \omega\log\omega$. For the reference solution
    (green crosses) we used Eq.~(\ref{eq:euler-ref-log}) with $\phi$ given by
    the numerical solution and $\lambda$ computed from
    Eq.~(\ref{eq:euler-ref-log-lambda}).}  
\end{figure}

\begin{figure}
  \centering
  \includegraphics[scale=0.29]{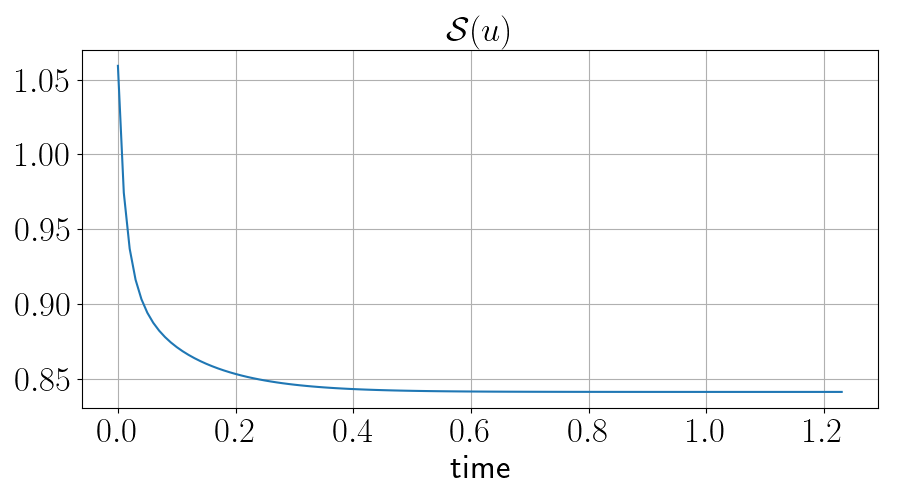}
  \includegraphics[scale=0.29]{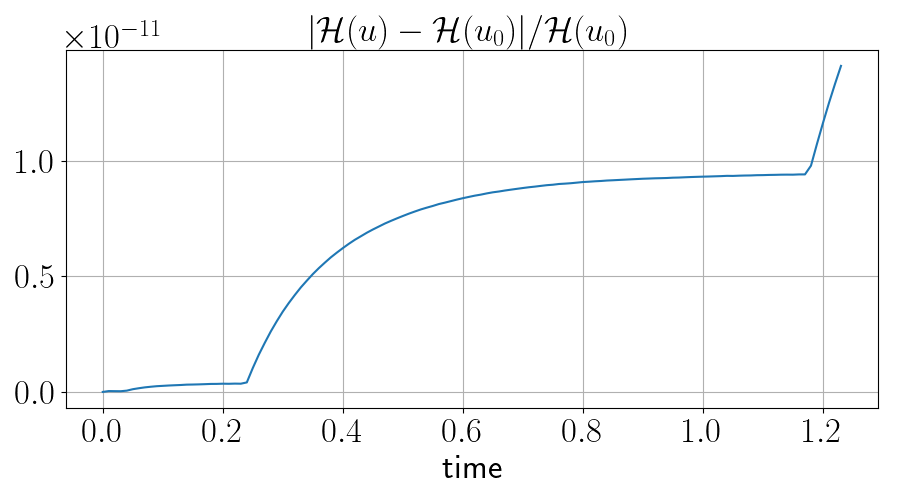}
  \caption{\label{fig:ge_diag} Evolution of entropy (left-hand-side panel) and
    of the variation of the Hamiltonian relative to its initial value
    (right-hand-side panel), for the case in Fig.~\ref{fig:ge_sol}.}  
\end{figure}

\paragraph*{Gibbs entropy density}

We consider now a more complicated entropy function, namely, the Gibbs entropy,
which is given by the entropy density $s(\omega) = \omega\log\omega$. In this
case the expected relaxed state is determined by,
cf. Eq.~(\ref{eq:complete-relaxation-vorticity2d}),
\begin{equation}
  \label{eq:euler-ref-log}
  1 + \log \omega = \lambda \phi \quad \iff \quad
  \omega = e^{\lambda \phi -1},
\end{equation}
where the eigenvalue $\lambda$ could be numerically estimated by means of
Eq.~(\ref{eq:euler-ref-log-lambda}). The initial condition is the same as the
one in Eq.~(\ref{eq:sv-ic}) except for the amplitude, which here is increased to
$1/N = 10$.

Figure~\ref{fig:ge_sol} shows the initial condition, the final state after
the relaxation and the usual scatter plot. In this case the solution of the
variational principle (green crosses) is computed using the numerically computed
values of $\fun{M}(\omega)$ and $\fun{S}(\omega)$ in
Eq.~(\ref{eq:euler-ref-log-lambda}). Again we see evidence of complete 
relaxation of the system toward the solution of the variational principle.
Figure~\ref{fig:ge_diag} shows the expected behavior of entropy and
Hamiltonian functions. It is worth noting that since the entropy is not
quadratic in $\omega$, the numerical scheme does not preserve the property of
monotonic entropy dissipation, hence sufficiently small time steps must be used
to ensure the the evolution of the system is approximated with sufficient accuracy.

\subsection{Application to Grad-Shafranov equilibria}
\label{sec:app-GS}

As a second example, we construct a relaxation method for axisymmetric MHD
equilibria, cf. Section~\ref{sec:GS-problem}. Essentially, this amounts to a
different iterative method for the solution of the Grad-Shafranov
equation. The metriplectic relaxation method ensures conservation of the
poloidal magnetic energy and monotonic dissipation of an ``ad hoc'' entropy, but
at a higher computational cost.  

Specifically, we construct a relaxation method for the variational principle
(\ref{eq:VP-GS}). On a bounded domain $\Omega$, satisfying
$\ol{\Omega} \subset \R_+ \times \R$ with coordinates $x = (x_1,x_2) = (r,z)$,
cf. Section~\ref{sec:GS-problem}, the state variable is a scalar field $u(t)$
proportional to the toroidal component of the plasma current,
$u(t,r,z) = (4\pi/c) r J_\varphi(t,r,z)$, and it is evolved in time according to
Eq.~(\ref{eq:div-grad-ev}) as in the case of the reduced Euler equations.

The entropy and Hamiltonian functions are chosen as in
Eq.~(\ref{eq:GradShafranov-S-H}), with entropy density and measure $\mu$ on
$\Omega$ given by
\begin{equation}
  \label{eq:mm-entropy}
  s(r,y) = \frac{1}{2} \frac{y^2}{C r^2 + D}, \qquad
  d\mu(r,z) = \frac{1}{r} dr dz,
\end{equation}
and the assumptions on the domain imply $r \geq r_0 > 0$ in $\ol{\Omega}$. Then,
the first condition in Eq.~(\ref{eq:GradShafranov-complete-relaxation}) defines
the toroidal current 
\begin{equation*}
  \frac{4\pi}{c} J_\varphi = \lambda \Big(Cr + \frac{D}{r}\Big) \psi,
\end{equation*}
which is the current of the equilibrium found by Herrnegger
\cite{Herrnegger1972} and Maschke \cite{Maschke1973},
cf. also Mc Carthy \cite[Section~II.B]{Carthy1999}.
Equivalently, since the state variable is $u = (4\pi/c) rJ_\varphi$, the
condition in Eq.~(\ref{eq:GradShafranov-complete-relaxation}) can be written as
\begin{equation}
  \label{eq:gs-ref}
  \frac{u}{Cr^2 +D} = \lambda \psi.
\end{equation}
In order to have a reference solution, we resort to the direct numerical
solution of the Grad-Shafranov equation by using the classical iterative scheme
\cite[pp. 22--23, Eqs.~(2.111) and (2.112)]{Takeda1991}. In the following
experiments $C = 0.6$ and $D = 0.2$.

\paragraph*{Rectangular domain}

We start with a rectangular domain $\ol{\Omega} = [1,7] \times [-9.5,+9.5]$,
with coordinates $x = (x_1,x_2) = (r,z)$, discretized by a uniform grid of
$64 \times 64$ nodes. The initial condition $u_0$ for the state variable is the
same as in Eq.~(\ref{eq:sv-ic}) with parameters
$x_{0,1} = r_0 = 4$, $x_{0,2} = z_0 = 0$, $w_1^2 = 0.5$,  $w_2^2 = 3.2$, and
$N=1$.
  
Figure~\ref{fig:gsr_sol} shows the results of this numerical experiment.
Instead of plotting the state variable $u$ directly, the color plot represents
the field $u / (Cr^2 + D)$, which should be proportional to $\psi$ if the system
reaches a state consistent with Eq.~(\ref{eq:gs-ref}). Qualitatively the results
are similar to those of Fig.~\ref{fig:sv_sol} for the reduced Euler equations:
the initial condition evolves toward an equilibrium state consistent with
Eq.~(\ref{eq:gs-ref}). The scatter plot in Fig.~\ref{fig:gsr_sol} shows that the
functional relation between $u / (Cr^2 + D)$ and the potential $\psi$ is
linear. The reference solution (green cross) is obtained computing the field
$u / (Cr^2 + D)$ from Eq.~(\ref{eq:gs-ref}), with $\psi$ being the numerical
solution and with the eigenvalue $\lambda = 0.030302$ being obtained from the
standard iterative solver of the Grad-Shafranov equation, which has also been
implemented in FEniCS. Figure~\ref{fig:gsr_diag} confirms the expected behavior
of the entropy and Hamiltonian functions. 

\begin{figure}
  \centering
  \includegraphics[scale=0.24]{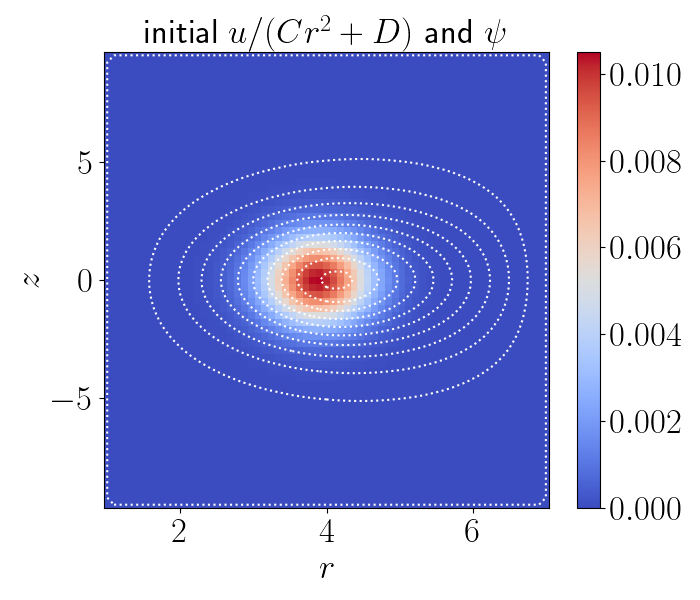}
  \includegraphics[scale=0.24]{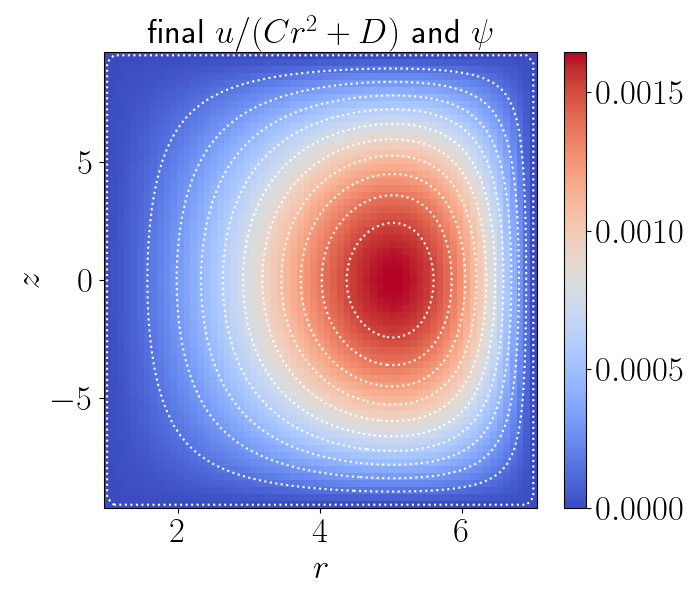}
  \includegraphics[scale=0.24]{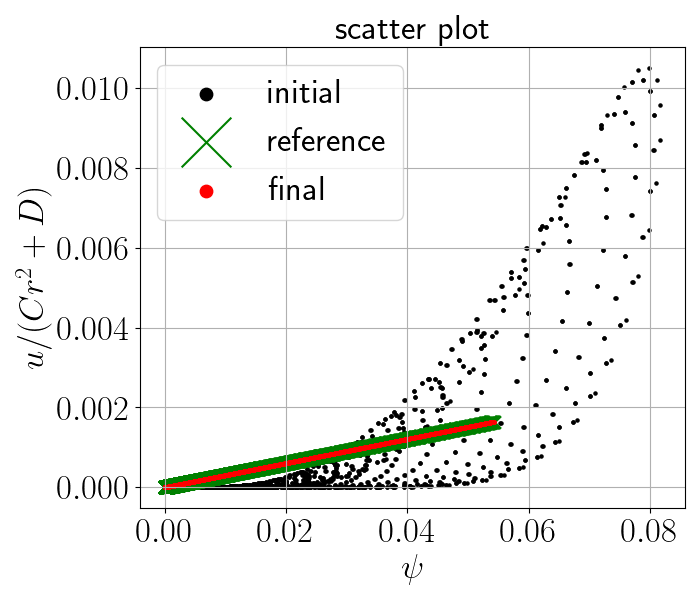}
  \caption{\label{fig:gsr_sol} Relaxation of an initial current
    $J_\varphi = (c/4\pi r) u_0$, with $u_0$ Gaussian, according to the
    evolution equation~(\ref{eq:div-grad-ev}) applied to the Grad-Shafranov
    problem~(\ref{eq:VP-GS}) with entropy~(\ref{eq:mm-entropy}) on a rectangular
    domain. The initial condition and the final state of the system are given in
    left-hand-side and middle panels, respectively, while the right-hand side
    panel shows the scatter plot, with the same color/symbol code as in
    Fig.~\ref{fig:sv_sol}. The color map represents the field $u/(Cr^2 + D)$,
    and the white contours are the flux function $\psi$, so that condition
    (\ref{eq:gs-ref}) is easily checked. Analogously the axes in the scatter
    plot refer to the values of the flux function $\psi$ and the field
    $u/(Cr^2 + D)$. The reference solution (green crosses) is computed using
    relation (\ref{eq:gs-ref}), with $\psi$ being given by the numerical
    solution and with the eigenvalue $\lambda$ computed by a standard
    Grad-Shafranov solver. Here, $C = 0.6$ and $D = 0.2$.} 
\end{figure}

\begin{figure}
  \centering
  \includegraphics[scale=0.29]{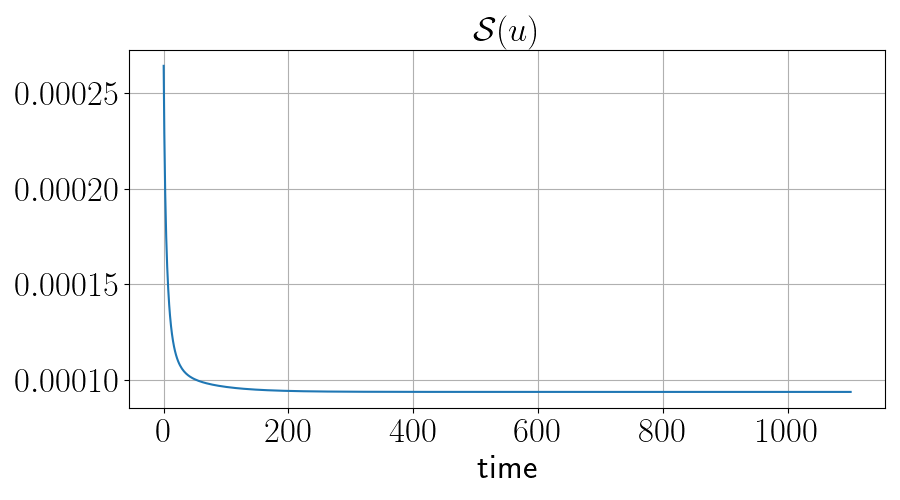}
  \includegraphics[scale=0.29]{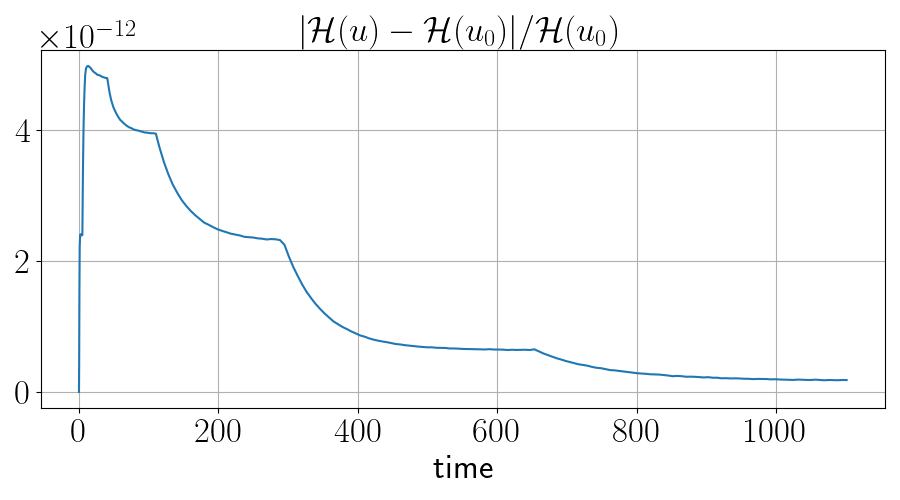}
  \caption{\label{fig:gsr_diag} Evolution of entropy (left-hand-side panel) and
    of the variation of the Hamiltonian relative to its initial value
    (right-hand-side panel), for the case in Fig.~\ref{fig:gsr_sol}.}  
\end{figure}

\paragraph*{Mapped domain}

At last, we give an example of the relaxation method for the Grad-Shafranov
equation on a non-trivial mapped domain with a smooth boundary. The domain
$\Omega$ is obtained by mapping the unit disk $\{z \in \C \colon |z|<1\}$ with
the map defined by
\begin{equation}
  \label{eq:mapping}
  \begin{aligned}
    r &= a \Big[b + \frac{1}{\varepsilon} \Big(1 -
    \sqrt{1 + \varepsilon (\varepsilon + 2 s \cos\theta) }\Big)\Big], \\
    z &= c \frac{e \xi s\sin\theta}{
      2-\sqrt{1 + \varepsilon(\varepsilon + 2 s \cos\theta)}},
  \end{aligned}
\end{equation}
where $z = s \exp(i\theta)$ is a point in the unit disk, and the parameters
are $e = 1.4$, $\varepsilon = 0.3$, $a = 4$, $b = 3$, and $c=6.3$, with
$\xi = 1 / \sqrt{1-\varepsilon^2 / 4}$. This map is a slightly modified version
of the one used by
Zoni and G\"ucl\"u \cite[Eq.~(3) and references therein]{Zoni2019}.
The initial condition is the same as for the rectangular domain, except for
$r_0 = 12$, $w_1^2 = 0.6$, and $w_2^2 = 6.0$.

The results for the mapped domain are shown in Fig.~\ref{fig:gsc_sol}. The color
map represents the field $u/(Cr^2+D)$, as in the previous case. The relaxed
state is again consistent with Eq.~(\ref{eq:gs-ref}), with the eigenvalue
$\lambda = 0.002599$ computed by a standard Grad-Shafranov solver.
Figure~\ref{fig:gsc_diag} shows the expected monotonic entropy dissipation
and the conservation of the Hamiltonian to machine precision.

\begin{figure}
  \centering
  \includegraphics[scale=0.24]{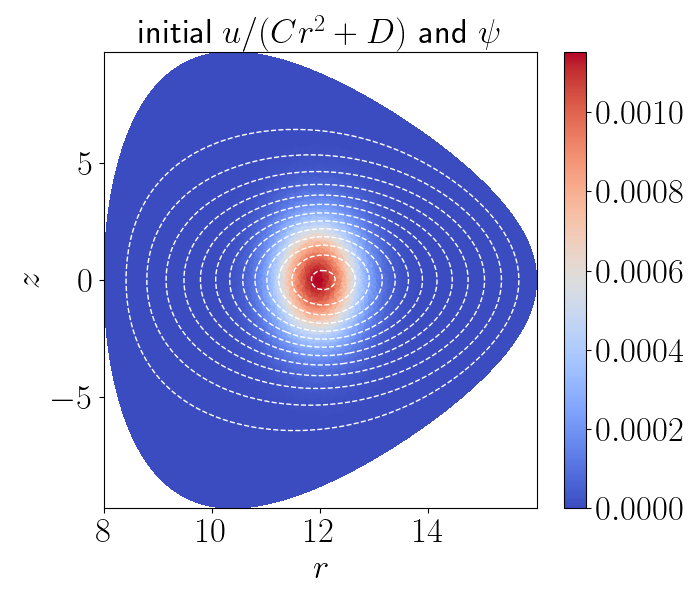}
  \includegraphics[scale=0.24]{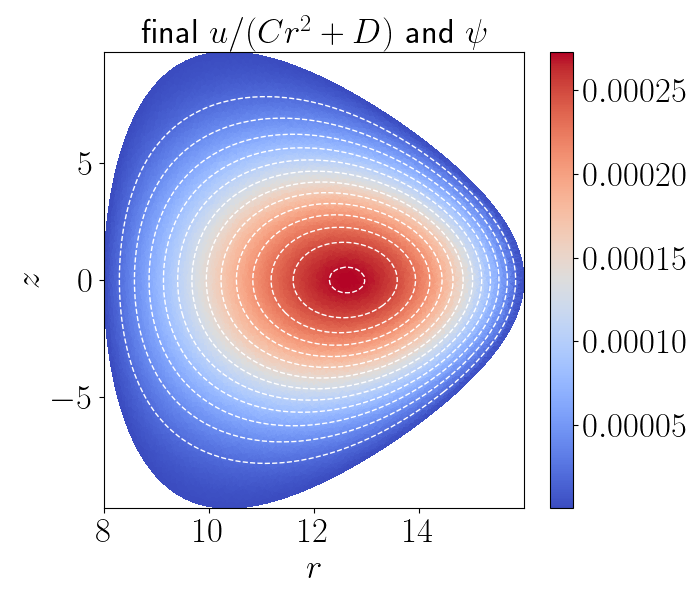}
  \includegraphics[scale=0.24]{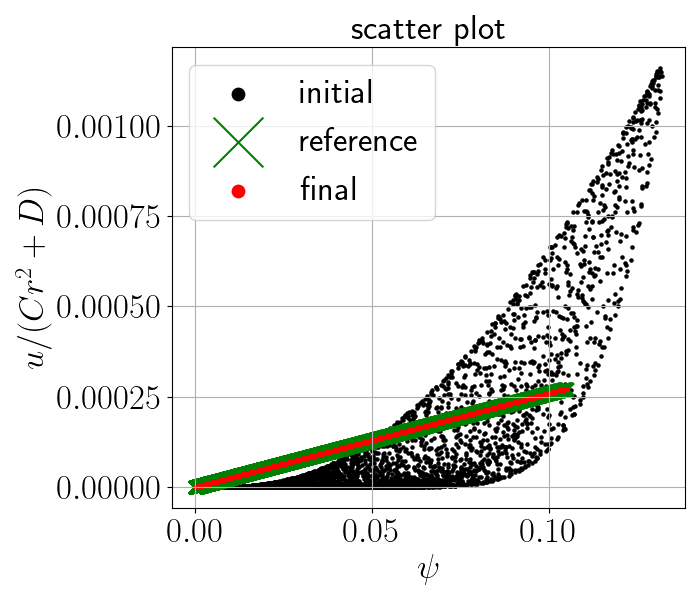}
  \caption{\label{fig:gsc_sol} The same as in Fig.~\ref{fig:gsr_sol}, but on the
    domain obtained mapping the unit disk with Eq.~(\ref{eq:mapping}).} 
\end{figure}

\begin{figure}
  \centering
  \includegraphics[scale=0.29]{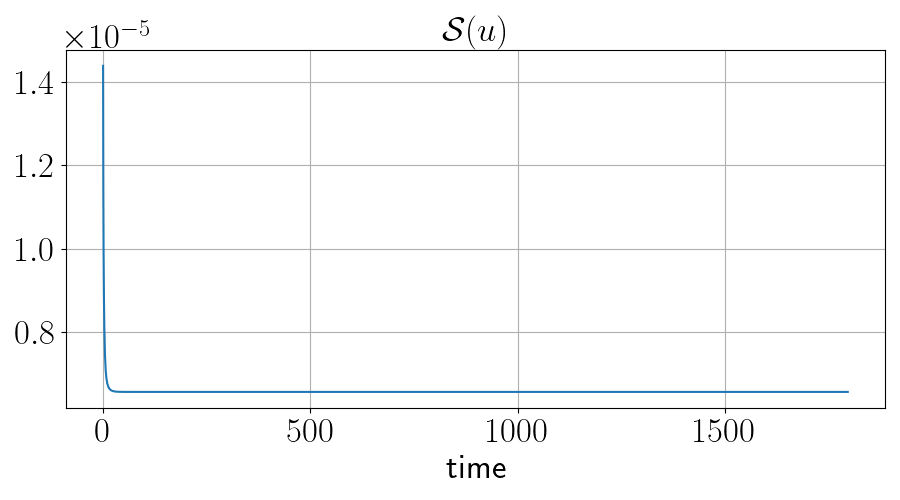}
  \includegraphics[scale=0.29]{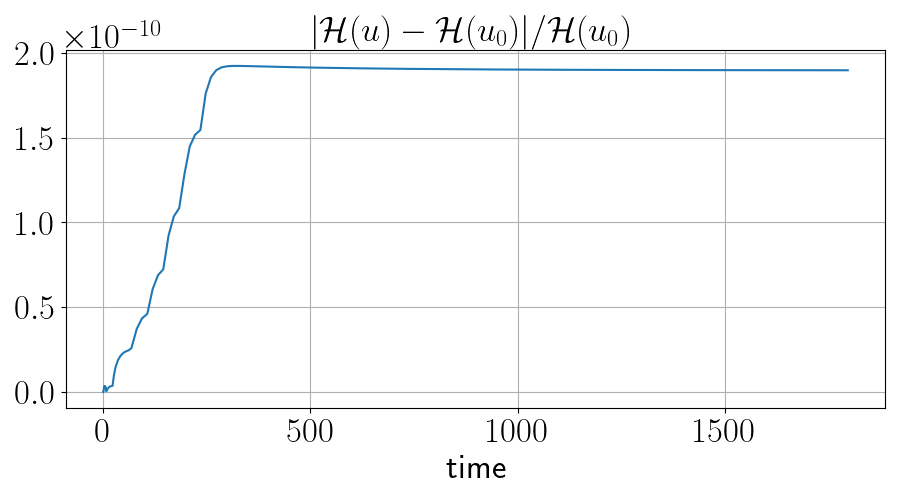}
  \caption{\label{fig:gsc_diag} Evolution of entropy (left-hand-side panel) and
    of the variation of the Hamiltonian relative to its initial value
    (right-hand-side panel), for the case in Fig.~\ref{fig:gsc_sol}. }  
\end{figure}

\section{Diffusion-like metric brackets}
\label{sec:diff-like-metr}

A feature of the general collision-like bracket is that the generalized
diffusion tensor and friction flux~(\ref{eq:D-F}) are \emph{nonlocal} functions
of the unknown $u$. Therefore, their evaluation requires an integration over the
whole domain $\Omega$. This poses an issue of computational complexity even
harder than that of the standard Landau collision operator, which is local in
half of the variables (cf.\ Example~\ref{ex:Landau}). In order to mitigate the
computational cost of a relaxation method based on these brackets, we have
studied the special case of brackets (\ref{eq:gen-clb}) that corresponds to
choosing $\mathcal{O} = \Omega$, i.e., we do not increase the size of the domain
($n = d$). We will refer to this case as a diffusion-like bracket, and reserve
the name ``collision-like'' for  the case $n > d$. As we shall see, in general,
one cannot expect complete relaxation (in the sense defined in
Section~\ref{sec:remarks-relax-equil}) from the diffusion-like brackets. The
metric double bracket~(\ref{eq:gmdb}) is a special case of a diffusion-like
bracket.

\subsection{General construction of the brackets}
\label{sec:d-general}

With the same setup of Section~\ref{sec:c-general}, which is summarized in
Fig.~\ref{fig:8}, let us consider the case $\mathcal{O} = \Omega$, $\nu = \mu$,
but allow $\tilde{N} \not = N$, so that in general $\tilde{W} \not = W$.

Then bracket~(\ref{eq:gen-clb}) reduces to
\begin{equation}
  \label{eq:gen-dlb}
  (\fun{F},\fun{G}) \coloneqq \int_\Omega
  P \frac{\delta \fun{F}(u)}{\delta u} \cdot \TT(u)
  P \frac{\delta \fun{G}(u)}{\delta u} d\mu,
\end{equation}
where $P : W \to \tilde{W}$ is a linear (possibly unbounded) operator, and
$\TT(u) \in \mathcal{B}(\tilde{W})$ is symmetric, positive semidefinite,
bounded linear operator that satisfies~(\ref{eq:T-energy-cond}). The evolution
equation generated by the bracket~(\ref{eq:gen-dlb}) is formally the same as
Eq.~(\ref{eq:clb-ev}), but we shall see in the examples that the dual operator
$P'$ does not involve any integral operator. 

Conditions~(\ref{eq:MDC}) are of course still sufficient conditions for brackets
of the form~(\ref{eq:gen-dlb}) to be minimally degenerate, but we shall see that
condition (\ref{eq:KC}) is usually not satisfied in this case.

\begin{example}
  \label{ex:mdb-rev}

  Brackets of the form~(\ref{eq:gmdb}), i.e., metriplectic double
  brackets acting on scalar fields, are special cases of diffusion-like
  brackets~(\ref{eq:gen-dlb}). In order to see this, let $N=1$, thus
  $W = L^2(\Omega,\mu)$, $\Phi = C^\infty(\ol{\Omega})$,  and
  \begin{equation*}
    Pw = \nabla w.
  \end{equation*}
  Hence, $\tilde{N} = d$ and $\tilde{W} = L^2(\Omega,\mu;\R^d)$ is the space of
  $L^2$ vector fields over $\Omega$. Given a function $J(x)$ of class
  $C^\infty(\ol{\Omega})$ taking values in the space of antisymmetric
  $d \times d$ matrices, we can define the antisymmetric bilinear operation   
  \begin{equation*}
    [w_1, w_2]_J = \nabla w_1(x) \cdot J(x) \nabla w_2(x).
  \end{equation*}
  As in Section~\ref{sec:metr-double-brackets}, it is not necessary that
  $[\cdot,\cdot]_J$ satisfies the Jacobi identity. In addition let us consider a
  symmetric, positive-definite bi-linear form $\gamma \colon W \times W \to \R$,
  together with the associated linear bounded, symmetric positive definite
  operator $\Gamma \colon  W \to W$, that is, cf.~\ref{sec:bilinear-forms},
  \begin{equation*}
    \gamma(w_1,w_2) = \int_\Omega  w_1(x) \cdot \Gamma w_2(x) d\mu(x),
  \end{equation*}
  for all $w_1, w_2 \in W$. In terms of $J$ and $\gamma$, we define the kernel
  \begin{equation*}
    \TT(u) = X_h(u) \circ \Gamma \circ \transpose{X_h(u)}, 
  \end{equation*}
  where $X_h(u)$ and $\transpose{X_h(u)}$ are the operators of multiplication by
  the vector fields 
  \begin{equation*}
    X_h(u;x) = J(x) \nabla h(x), \quad
    \transpose{X_h(u;x)} = -\nabla h(x) \cdot J(x),
  \end{equation*}
  respectively, and
  $h = \delta \fun{H}(u) / \delta u \in \Phi = C^\infty(\ol{\Omega})$.
  If $J$ is a Poisson tensor, i.e., $[\cdot,\cdot]_J$ satisfies the Jacobi
  identity, $X_h(u;\cdot)$ is the $d$-dimensional Hamiltonian vector field
  generated by the Hamiltonian function $h = \delta \fun{H}(u) / \delta u$.
  Since $X_h(u;\cdot) \in C^\infty(\ol{\Omega};\R^d)$, $\TT(u)$ maps $\tilde{W}$
  into itself. The operator $\Gamma$ is symmetric and positive definite, hence
  $\TT(u)$ is symmetric and positive semidefinite, and we have
  $\nabla h \in \ker \TT(u)$ as required by condition (\ref{eq:T-energy-cond}).
 
  With the foregoing choices of $P$ and $\TT(u)$, Eq.~(\ref{eq:gen-dlb}) becomes
  \begin{align*}
    \big(\fun{F},\fun{G}) &= \int_\Omega
    \Big( \nabla \frac{\delta \fun{F}(u)}{\delta u} \cdot
    J \nabla \frac{\delta \fun{H}(u)}{\delta u} \Big) \Gamma
    \Big( \nabla \frac{\delta \fun{G}(u)}{\delta u} \cdot
    J \nabla \frac{\delta \fun{H}(u)}{\delta u} \Big) d\mu \\
    &= \gamma\Big(
    \Big[\frac{\delta \fun{F}(u)}{\delta u},
      \frac{\delta \fun{H}(u)}{\delta u} \Big]_J,
    \Big[\frac{\delta \fun{G}(u)}{\delta u},
      \frac{\delta \fun{H}(u)}{\delta u} \Big]_J\Big),
  \end{align*}
  which is Eq.~(\ref{eq:gmdb}). As discussed in
  Section~\ref{sec:metr-double-brackets}, this bracket is in general not
  minimally degenerate. In fact, while condition~(\ref{eq:PI}) amounts to the
  Poincar\'e inequality and holds true on a bounded domain $\Omega$,
  condition (\ref{eq:KC}) fails since for any sufficiently regular function
  $f : \R \to \R$, any function of the form $\tilde{w} = \nabla f(h)$, with
  $h = \delta \fun{H}(u){\delta u}$, belongs to $\ker \TT(u) \cap \rng(P)$.
\end{example}

\subsection{Diffusion-like brackets based on div--grad operators}
\label{sec:d-div-grad}

We address the diffusion-like version of the bracket introduced in
Section~\ref{sec:c-div-grad}. For scalar fields ($N=1$) on a bounded domain
$\Omega \subset \R^d$, $d\geq 2$, let $P = \nabla$ with domain
$\Phi = H_0^1(\Omega)$. The kernel of the bracket is constructed from the matrix
$Q_d$, cf. Eq.~(\ref{eq:Q}),
\begin{equation*}
  \TT(u;x) = \kappa(u;x)
  Q_d\Big(\nabla \frac{\delta \fun{H}(u)}{\delta u} (x)\Big),
\end{equation*}
where $\kappa(u;x)$ is a positive scalar function, and $\TT(u)$ is defined as
the operator of multiplication by $\TT(u;\cdot)$. Then, bracket 
(\ref{eq:gen-dlb}) becomes \cite{Bressan2023}
\begin{equation}
  \label{eq:d-div-grad}
  (\fun{F},\fun{G}) = \int_\Omega \kappa(u)
  \nabla \frac{\delta \fun{F}(u)}{\delta u}
  \cdot Q_d\Big(\nabla \frac{\delta \fun{H}(u)}{\delta u} \Big)
  \nabla \frac{\delta \fun{G}(u)}{\delta u} d\mu,
\end{equation}
and the corresponding evolution equation is
\begin{equation*}
  \partial_t u = \widetilde{\div_\mu} \Big[\kappa(u)
    Q_d\Big(\nabla \frac{\delta \fun{H}(u)}{\delta u} \Big)
    \nabla \frac{\delta \fun{S}(u)}{\delta u}\Big] \quad
  \text{in } \Phi' = H^{-1}(\Omega).
\end{equation*}
This is a ``local version'' of Eq.~(\ref{eq:div-grad-ev}) which  justifies the
name ``diffusion-like'' for this bracket. Condition~(\ref{eq:KC}) fails in the
same way as in Example~\ref{ex:mdb-rev}.

It is worth noting that in two spatial dimensions, $d=2$, one has
\begin{equation*}
  Q_2(\nabla h) = X_h \otimes X_h = X_h \transpose{X_h}, \qquad
  h = \delta \fun{H}(u) / \delta u,
\end{equation*}
where $X_h = J_c \nabla h = (-\partial_2 h, \partial_1 h)$ is the canonical
Hamiltonian vector field generated by $h(x)$ and $\transpose{X_h}$ denotes its
transpose. This means that the diffusion-like bracket~(\ref{eq:d-div-grad}) in a
two dimensional domain amounts to the metric double bracket addressed in
Example~\ref{ex:mdb-rev}, with $\Gamma = I$ being the identity operator. For
$d>2$ the bracket~(\ref{eq:d-div-grad}) is however different from the metric
double brackets in Example~\ref{ex:mdb-rev}. In fact, if $X_h \not = 0$, the
null space of the matrix $Q_d$ is always one dimensional for any dimension $d$,
while the null space of $X_h \otimes X_h$ is $d-1$ dimensional, hence the two
matrices have the same null space only if $d=2$. 
  
Nonetheless, for $d \geq 3$, one can write the matrix $Q_d(\nabla h)$ in terms
of suitable pairing of two antisymmetric operations by using the identity
\begin{equation*}
  \frac{1}{(d-2)!} \sum_{i_1, \ldots, i_{d-2}}
  \epsilon_{i_1,\ldots,i_{d-2}, i, k} \epsilon_{i_1,\ldots,i_{d-2}, j, l} =
  \delta_{ij} \delta_{kl} - \delta_{il} \delta_{jk},
\end{equation*}
with $\epsilon_{i_1,\ldots,i_d}$ being the completely antisymmetric symbol.
We obtain
\begin{equation}
  \label{eq:Levi-Civita-identities}
  \begin{aligned}
    [Q_d(\nabla h)]_{ij} &= |\nabla h|^2 \delta_{ij} - \partial_i h \partial_j h
    = \sum_{kl} \big[\delta_{ij} \delta_{kl} - \delta_{il} \delta_{jk}\big]
    \partial_k h \partial_l h \\
    &= \frac{1}{(d-2)!} \sum_{i_1, \ldots, i_{d-2}} \sum_{k,l}
    \epsilon^{i_1,\ldots,i_{d-2}, i, k} \epsilon^{i_1,\ldots,i_{d-2}, j, l}
    \partial_k h \partial_l h,
  \end{aligned}
\end{equation}
where $\partial_j h = \partial h/\partial x_j$, and thus,
Eq.~(\ref{eq:d-div-grad}) can be written equivalently as
\begin{equation*}
  (\fun{F},\fun{G}) = \frac{1}{(d-2)!} \sum_\alpha
  \int_\Omega \kappa(u)\,  \mathcal{E}_d^{\alpha}(\nabla f, \nabla h)
 \,  \mathcal{E}_d^{\alpha}(\nabla g, \nabla h)\,  d\mu,
\end{equation*}
where $\alpha = (i_1,\ldots,i_{d-2})$ is a multi-index,
$f = \delta \fun{F}(u)/\delta u$, $g = \delta \fun{G}(u)/\delta u$, and
\begin{equation*}
  \mathcal{E}_d^{\alpha}(\nabla \varphi, \nabla \psi)
  = \sum_{i, k} \epsilon^{i_1,\ldots,i_{d-2}, i, k} \partial_i \varphi
  \partial_k \psi,  \quad \alpha = (i_1,\ldots,i_{d-2}).
\end{equation*}
In dimension $d = 3$, $\mathcal{E}_3$ defines a Lie bracket in $\R^3$. This is
the standard Lie algebra structure on $\R^3$ given by the cross product
arising in the case of rigid body rotation \cite{Morrison1986}. 

Yet another form of this bracket makes use of the Kulkarni-Nomizu (K-N) product
and the metriplectic $4$-bracket structure \cite{pjmU24}. Using the first
identity in~(\ref{eq:Levi-Civita-identities}), we can write
\begin{equation*}
  (\fun{F},\fun{G}) = \frac{1}{2} \sum_{i,j,k,l} \int_\Omega \kappa(u)
  (\delta \owedge \delta)_{ijkl}
  \Big[\partial_i \frac{\delta \fun{F}(u)}{\delta u}\Big]
  \Big[\partial_j \frac{\delta \fun{H}(u)}{\delta u}\Big]
  \Big[\partial_k \frac{\delta \fun{G}(u)}{\delta u}\Big]
  \Big[\partial_l \frac{\delta \fun{H}(u)}{\delta u}\Big]
  d\mu,
\end{equation*}
where
$(\delta\owedge\delta)_{ijkl}=2(\delta_{ik}\delta_{jl}-\delta_{il}\delta_{jk})$
is the K-N product of two identity tensors.

\subsection{Diffusion-like brackets based on curl--curl operators}
\label{sec:d-curl-curl}

As a last example, we address the diffusion-like version of the $\curl$-$\curl$
brackets of Section~\ref{sec:c-curl-curl}. We consider a vector field over a
bounded domain $\Omega \subset \R^3$ with Lebesgue measure $d\mu(x) = dx$, hence
$d=N=3$. We choose
\begin{equation*}
  Pw = \curl w,
\end{equation*}
so that $\tilde{N} = N = 3$ and $\tilde{W} = W$, with $\dom(P)$ given by
\begin{equation*}
  \Phi = \{w \in H(\curl,\Omega) \cap H(\div,\Omega) \colon
  \text{ $\div w = 0$ in $\Omega$, $n \cdot w = 0$ on $\partial \Omega$}\}.
\end{equation*}
This differs from the space $\Phi$ considered in Section~\ref{sec:c-curl-curl}
by the ``opposite'' choice of boundary conditions: the normal component is set
to zero instead of the tangential component. The Poincar\'e inequality for the
operator $\curl$ holds for this space as well \cite[Corollary 3.51]{Monk2003},
so that condition~(\ref{eq:PI}) holds true. (Here, we have the choice of the
boundary condition since we do not need to satisfy the second identity in
Eq.~(\ref{eq:PI-L}).)

As an example, let the kernel be once again constructed from $Q_3$, defined in
Eq.~(\ref{eq:Q}), 
\begin{equation*}
  \TT(u;x) = \kappa(u;x) Q_3\Big(\curl \frac{\delta \fun{H}(u)}{\delta u}\Big),
\end{equation*}
where $\kappa(u;x)$ is a positive function. Even though the operator $P = \curl$
with domain $\dom(P) = \Phi$ satisfies a Poincar\'e inequality, in general, the
kernel fails to satisfy condition (\ref{eq:KC}): a function
$\tilde{w} \in \ker \TT(u) \cap \rng(P) \subset \tilde{W} = L^2(\Omega;\R^3)$
must satisfy $\tilde{w} = \curl w$, with $w \in \Phi$ and
$\tilde{w}(x) = \Lambda(x) b(x)$, $b = \curl h$, $h = \delta\fun{H}/\delta u$,
and thus the pair $(\Lambda, w)$ must solve
\begin{equation*}
  \left\{
  \begin{aligned}
    \curl w &= \Lambda b, && \text{in } \Omega, \\
    \div w &= 0, && \text{in } \Omega,\\
    b \cdot \nabla \Lambda &= 0, && \text{in } \Omega, \\
    n \cdot w &= 0, && \text{on } \partial \Omega.
  \end{aligned}
  \right.
\end{equation*}
As a special case let $b$ be a nonlinear Beltrami field, i.e., a solution of
(\ref{eq:nonlinear-Beltrami}) such that $\curl b = \mu b$ with $\mu(x)$ not a
constant, then $\Lambda = \mu$ and $w = b$ is a solution that violates condition
(\ref{eq:KC}). 

With the foregoing choices, Eq.~(\ref{eq:gen-dlb}) amounts to
\begin{equation}
  \label{eq:d-curl-curl}
  (\fun{F},\fun{G}) = \int_\Omega \kappa(u)
  \curl \frac{\delta \fun{F}(u)}{\delta u}
  \cdot Q_3\Big(\curl \frac{\delta \fun{H}(u)}{\delta u} \Big)
  \curl \frac{\delta \fun{G}(u)}{\delta u} dx,
\end{equation}
and the corresponding evolution equation becomes
\begin{equation*}
  \partial_t u = -\widetilde{\curl} \Big[\kappa(u)
    Q_3\Big(\curl \frac{\delta \fun{H}(u)}{\delta u} \Big)
    \curl \frac{\delta \fun{S}(u)}{\delta u}\Big] \quad
  \text{in } \Phi' = H'(\curl,\Omega).
\end{equation*}
This bracket can be written in terms of the antisymmetric bilinear operator
\begin{equation*}
  \mathcal{E}_3(X,Y) = [X,Y]_{\R^3} \coloneqq X \times Y, \qquad X,Y \in \R^3,
\end{equation*}
which is the standard Lie bracket in $\R^3$. In fact Eq.~(\ref{eq:d-curl-curl})
can be shown to be a special case of the following
(cf. \cite{Morrison1986,Gay-Balmaz2014}): 
\begin{equation}
  \label{eq:d-curl-curl-2}
  (\fun{F},\fun{G}) = \int_\Omega 
  \Big[\curl \frac{\delta \fun{F}(u)}{\delta u}, 
    \curl \frac{\delta \fun{H}(u)}{\delta u} \Big]_{\R^3} \Gamma
  \Big[\curl \frac{\delta \fun{G}(u)}{\delta u}, 
    \curl \frac{\delta \fun{H}(u)}{\delta u} \Big]_{\R^3} dx,
\end{equation}
where $\Gamma(u) \in \mathcal{B}\big(L^2(\Omega;\R^3)\big)$ is a symmetric,
positive definite operator; Eq.~(\ref{eq:d-curl-curl}) is obtained for
$\Gamma(u) = \kappa(u)$, the multiplication operator by the function
$\kappa(u;\cdot)$. Equation (\ref{eq:d-curl-curl-2}) is a metric double
bracket of the form~(\ref{eq:gmdb}).
Applied to magnetic fields this gives a generalization of the relaxation method
of Chodura and Schl\"uter \cite{Chodura1981} with constant pressure, cf. also
Moffatt~\cite{Moffatt2021}. An explicit example will be briefly reported in
Section~\ref{sec:app-Beltrami} below.

\subsection{Application to nonlinear Beltrami fields} 
\label{sec:app-Beltrami}

So far we have focused on examples of equilibrium problems for which complete
relaxation of the solution is essential. We have shown that a metriplectic
relaxation method for such problems should be based on metric brackets that are
minimally degenerate (or specifically degenerate if more than one constraint is
considered), cf.\  Section~\ref{sec:remarks-relax-equil}. Diffusion-like
brackets do not appear to be appropriate for those problems.

For sake of completeness, we address an example of equilibrium problems that are
characterized as minima of a function subject to topological constraints. This
is the case of nonlinear Beltrami fields, for which the variational principle is
given in Lagrangian representation, cf. Section~\ref{sec:beltrami-problem} and
\ref{sec:VP}. Full three-dimensional MHD equilibria satisfy the same type of
Lagrangian variational principle. 

Because of their larger null space, metric double brackets of the
form~(\ref{eq:d-curl-curl-2}) allow us to obtain an evolution equation that
preserves the necessary constraints. To this end, we
identify the state variable $u(t)$ with a magnetic field $u(t,x) = B(t,x)$
on a simply connected, bounded domain $\Omega \subset \R^3$. More specifically,
we assume that $B(t) \in V \subset \Phi$, where $\Phi$ is the same space defined
in Section~\ref{sec:c-curl-curl}. The evolution equation for $B(t)$ is given by
Eq.~(\ref{eq:metric-system-equation}) and bracket (\ref{eq:d-curl-curl-2}),
with $\Gamma = I$, the identity operator, for simplicity, and with entropy and
Hamiltonian given in Eq.~(\ref{eq:Beltrami-S-H}). Therefore, if an orbit of this
metriplectic system completely relaxes, it would converge in time to a 
\emph{linear} Beltrami field, cf. Section~\ref{sec:beltrami-problem}.
In fact, this bracket has a much larger null space. The equilibrium
points, given by $B \in \Phi$ such that $(\fun{S},\fun{S})(B) = 0$, satisfy
the Beltrami condition $(\curl B) \times B = 0$, in the weak formulation
discussed in Section~\ref{sec:beltrami-problem}. 

The resulting evolution equation amounts to 
\begin{equation}
  \label{eq:CS-bracket-evolution}
  \partial_t B = \widetilde{\curl} \big[
    B \times \big(B \times \curl B \big) \big],
\end{equation}
where we have accounted for the identity
$\curl[\delta \fun{H}(B) /\delta B] = \curl A = B$. If $B$ is sufficiently
regular, we can replace $\widetilde{\curl}$ by $\curl$ and write
\begin{equation}
  \label{eq:CS-method}
  \left\{
  \begin{aligned}
    \partial_t B - \curl \big[V \times B\big] &= 0, && \text{ in } \Omega, \\
    V - (\curl B) \times B &=0, && \text{ in } \Omega, \\
    n \cdot B &=0, &&  \text{ on } \partial \Omega\,, 
  \end{aligned}
  \right.
\end{equation}
which shows that the magnetic field $B$ is advected by the flow of the effective
``velocity'' field $V$. Hence, so long as the solution remains smooth, the field
lines of $B$ are frozen into the flow (actually flux), i.e., they cannot change
their topological properties. This is a much stronger constraint than just
preservation of magnetic helicity $2\fun{H}(B)$. As anticipated,
Eq.~(\ref{eq:CS-method}) is exactly the relaxation method of Chodura and
Schl\"uter \cite{Chodura1981} with constant pressure. The method itself is
therefore not new. In solar physics this relaxation method is known as the
magneto-frictional method \cite{Yang1986,Klimchuk1992,Valori2007,Valori2010},
and it has been applied to the computation of force-free magnetic fields
in coronal active regions \cite{Wiegelmann2012}. The bracket formalism, however,
opens the way to possible generalizations by means of different choices of the
kernel. This possibility will be explored in future work.
Since this relaxation method is based on the MHD induction equation,
smoothness of the solution may be lost in a finite time due to the formation
of current sheets, as conjectured by Parker and discussed in
Section~\ref{ssec:Oequil}. In this work we allow for weak solutions. In fact,
Eq.~(\ref{eq:CS-bracket-evolution}) is reformulated with
$B(t) \in H_0(\div,\Omega)$ only. More precisely, we search for
$B \in C^1\big(([0,T]; H_0(\div,\Omega)\big)$ and auxiliary variables
$E, j, H \in C\big([0,T]; H_0(\curl,\Omega)\big)$ satisfying
\begin{equation}
  \label{eq:weak-CS-method}
  \left\{
  \begin{aligned}
    \partial_t B + \curl E &= 0,
    \qquad \text{in } H_0(\div,\Omega), \\
    (H, G)_{L^2}-(B, G)_{L^2} &= 0,
    \qquad \forall G \in H_0(\curl,\Omega), \\
    (j, k)_{L^2}-(B, \curl k)_{L^2} &= 0,
    \qquad \forall k \in H_0(\curl,\Omega), \\
    (E,F)_{L^2} - (H \times j, H \times F)_{L^2} &= 0,
    \qquad \forall F \in H_0(\curl,\Omega),
  \end{aligned}
  \right.
\end{equation}
pointwise in time, with $F, G, k \in H_0(\curl,\Omega)$ being test functions.
Faraday's equation is posed strongly as an identity in $H_0(\div,\Omega)$.
As a result the condition $\div B = 0$ is preserved. One can also show
directly that a solution of this system preserves magnetic helicity and
dissipate magnetic energy, that is, the properties of the bracket hold for
this reformulation. In particular, we observe that
\begin{equation*}
  \frac{1}{2}\frac{d}{dt} \int_\Omega |B|^2 dx =
  - \big\|j \times H\big\|_{L^2}^2,
\end{equation*}
and the equilibrium condition is $j \times H = 0$, which is the weak
formulation of the Beltrami condition anticipated in
Section~\ref{sec:beltrami-problem}.
Here we present a single numerical experiment obtained by a structure-preserving
numerical scheme \cite{Bressan2023}, which we derived by adapting the
finite-element exterior calculus (FEEC) scheme of Hu et al.\  \cite{Hu2021} for
incompressible MHD. The scheme provably preserves the Hamiltonian (magnetic
helicity), the constraint $\div B = 0$, and the monotonic behavior of entropy
(magnetic energy), but it does \emph{not} preserve the topology of the field
lines exactly. Similar work has been recently published by
He et al. \cite{He2025}. 
Previously, the magnetic relaxation problem has been dealt with by means of
Lagrangian \cite{Candelaresi2014} and finite difference \cite{Guo2016}
methods. More recently, a different kind of Lagrangian numerical scheme has been
proposed \cite{Padilla2022,Gross2023}, which is based on the discretization of
the domain in narrow flux tubes, each one being relaxed by a curve-shortening
flow in a modified metric. This interesting scheme therefore preserves the
topological properties of the field lines. Yet with the domain discretized by a
finite set of lines, the reconstruction of the magnetic field at arbitrary
points of the domain, needs to be addressed.  

Before describing the considered test case, let us address the role of
magnetic-helicity conservation. In a domain $\Omega$ where the Poincar\'e
inequality for the $\curl$ operator holds true, magnetic helicity
$H_m(B) = 2\fun{H}(B)$ provides a lower bound for magnetic energy.
In fact, one has \cite{Arnold1998}
\begin{equation}
  \label{eq:Hm-bound}
  \big|H_m(B)\big| \leq \|A\|_{L^2(\Omega)} \|B\|_{L^2(\Omega)} \leq
  C_P \|B\|_{L^2(\Omega)}^2.
\end{equation}
For an initial condition $B_0$ with $H_m(B_0) = 0$, it is possible that the
solution of~(\ref{eq:CS-method}) with the chosen boundary conditions
($B \cdot n = 0$ on $\partial \Omega$) relaxes to a trivial field, i.e.,
$|B(t)| \to 0$ for $t \to +\infty$, even if the topology of the field lines is
preserved. This is the case for the class of one-dimensional solutions
of~(\ref{eq:CS-method}), which are obtained, for instance, by assuming
\begin{equation*}
  B(t,x) = \begin{pmatrix}
    0 \\  0  \\  b(x_1)
  \end{pmatrix}
  = \curl
  \begin{pmatrix}
    0  \\  a(x_1)  \\   0
  \end{pmatrix},
\end{equation*}
where $a'(x_1) = b(x_1)$, $x = (x_1,x_2,x_3)$ and the field is constant in
$(x_2,x_3)$. We have $A \cdot B = 0$ and thus $H_m(B)=0$. Correspondingly,
equations~(\ref{eq:CS-method}) reduce to 
\begin{equation*}
  \partial_t b - (b^2 b')' = 0,
\end{equation*}
where a prime denotes spatial differentiation. This is a standard heat equation.
The steady states are solution to $b^2 b' =$ constant, which gives
$b(x_1) = (c_1 + c_2 x_1)^{1/3}$, with $c_1, c_2$ being integration constants.
For instance homogeneous boundary conditions for $b$ on an interval yields the
unique solution $b(x_1) = 0$. Magnetic relaxation in one dimension has been
recently considered by Yeates \cite{Yeates2022} and compared to the
corresponding full MHD relaxation, thus exposing the limitations of the
magneto-frictional method. 

It is therefore meaningful to consider initial conditions with non-trivial
magnetic helicity $\fun{H}(B) = \frac{1}{2} H_m(B) \not= 0$. We construct such
an initial condition from the vector potential 
\begin{equation*}
  \tilde{A}(x) \coloneqq a \begin{pmatrix}
    (n/\sqrt{m^2+n^2}) \sin(\pi m x_1) \cos(\pi n x_2) \\
    -(m/\sqrt{m^2+n^2}) \cos(\pi m x_1) \sin(\pi n x_2) \\
    \sin(\pi m x_1) \sin(\pi n x_2)
  \end{pmatrix},
\end{equation*}  
with $a \in \R$ and $m, n \in \N$ being parameters (we shall choose $m=n=1$).
For any choice of the parameters, $\tilde{A}$ is divergence-free and a linear
Beltrami field, periodic is all variables; it is an eigenvalue of $\curl$
corresponding to the eigenvalue $\lambda_{m,n} = \pi(m^2+n^2)^{1/2}$.
We localize this field in the unit cube $\ol{\Omega} = [0,1]^3$ by means of the
cut-off function 
\begin{equation*}
  \eta(x) \coloneqq \chi(x_1) \chi(x_2) \chi(x_3), \qquad
  \chi(y) \coloneqq y^2 (1-y)^2, \quad y \in [0,1].
\end{equation*}
We have $\eta(x) = 0$ and $\nabla \eta(x) = 0$ for $x \in \partial \Omega$ since
both $\chi$ and its derivative $\chi'$ vanish for $y=0$ and $y=1$. We construct
the initial condition on the domain $\ol{\Omega} = [0,1]^3$ as
\begin{equation}
  \label{eq:Beltrami-ic}
  \begin{aligned}
    A_0 &\coloneqq \eta \tilde{A}, \\
    B_0 &\coloneqq \curl A_0
    = \nabla\eta \times \tilde{A} + \eta \curl \tilde{A},
  \end{aligned}
\end{equation}
and $A_0 \in H_0(\curl,\Omega)$, $B_0 \in H_0(\div,\Omega)$ with $\div B_0=0$.
As for magnetic helicity, 
\begin{equation*}
  H_m(B_0) = 2\fun{H}(B_0) = \int_\Omega A_0 \cdot B_0 dx =
  \lambda_{m,n} \|A_0\|^2_{L^2(\Omega)} > 0.
\end{equation*}
We can control the initial helicity by means of the parameters $a \in \R$ and
$m,n \in \N$. The magnetic field (\ref{eq:Beltrami-ic}) is represented in
Fig.~\ref{fig:B_sol} (top row), by means of a Poincar\'e plot using the plane
$x_2=1/2$ as a Poincar\'e section. From the plot (Fig.~\ref{fig:B_sol}, top-left
panel), one can identify a rather complex topology of the field lines, with, in
particular, four large islands of period two that are rendered in three
dimensions in Fig~\ref{fig:B_sol}, top-right panel, by tracing a few selected
field lines for each island.

\begin{figure}
  \centering
  \includegraphics[scale=0.195]{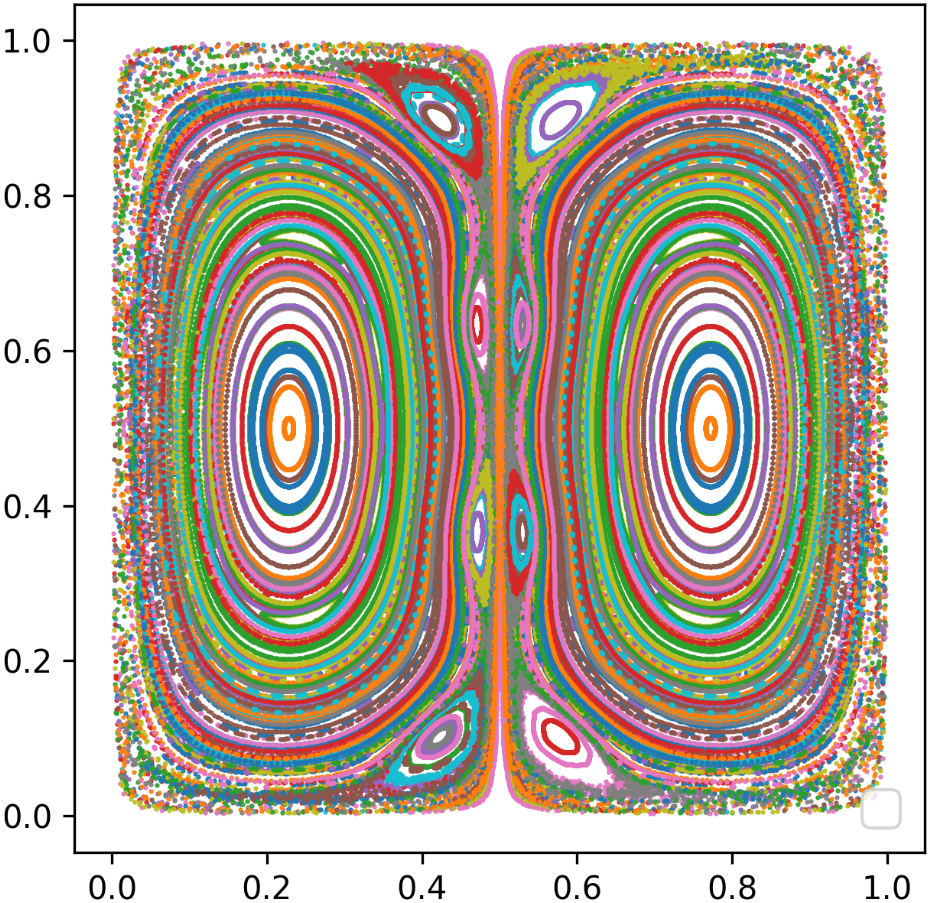}
  \includegraphics[scale=0.25]{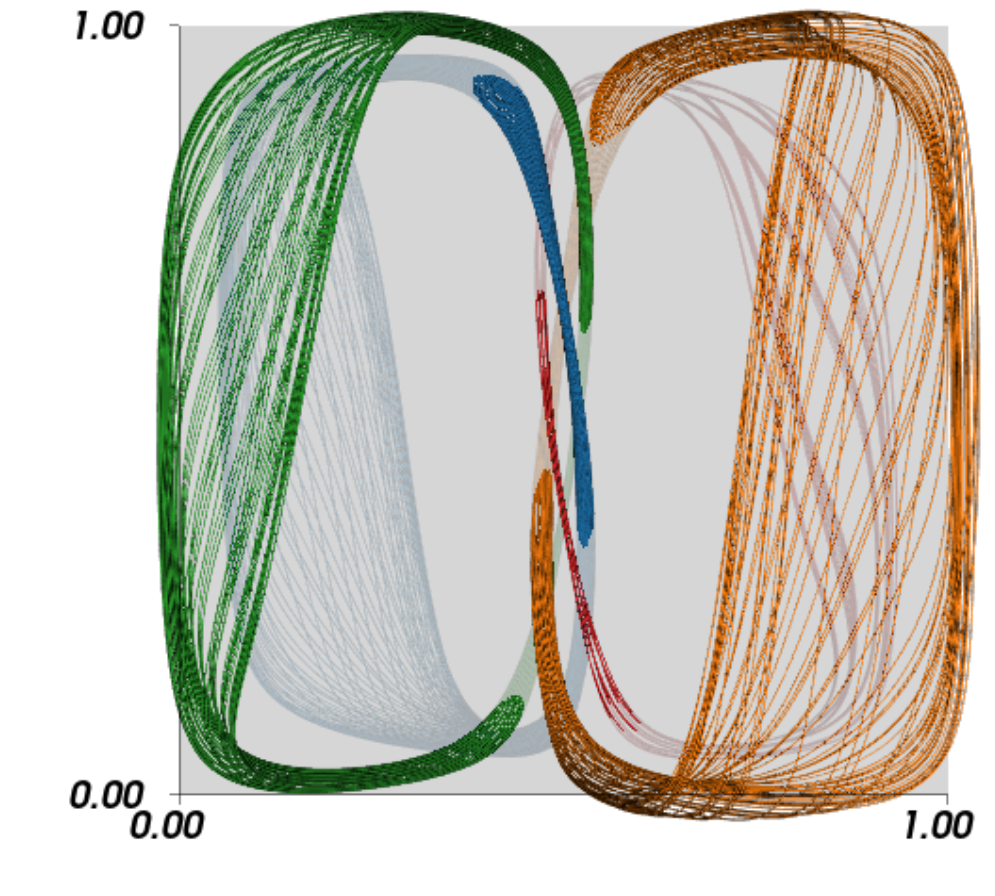}
  \includegraphics[scale=0.2]{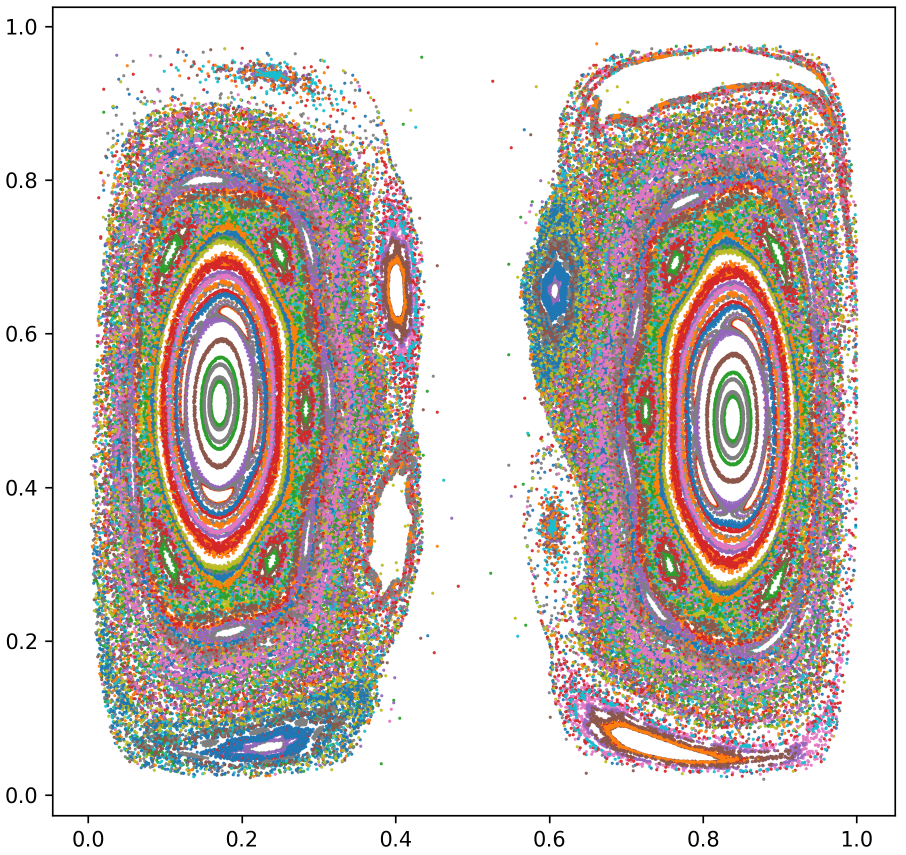}
  \includegraphics[scale=0.245]{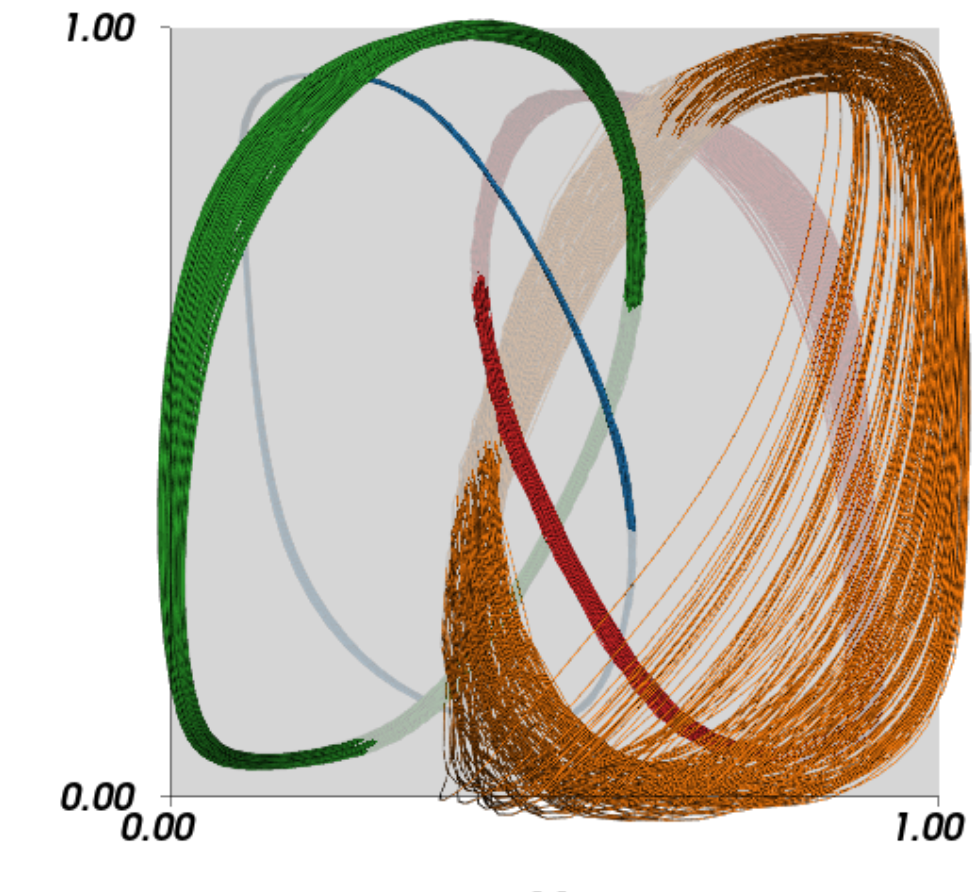}
  \caption{\label{fig:B_sol} Poincar\'e plot in the plane $x_1$-$x_2$ and
    selected field lines of the magnetic field $B$, for the initial condition
    (top row) and the final state (bottom row), after the relaxation process.
    The initial condition is given in Eq.~(\ref{eq:Beltrami-ic}) with $m=n=1$
    and $a=1$, while the evolution equation is the magneto-frictional
    method~(\ref{eq:CS-method}). The selected field lines correspond to the four
    large islands visible in the Poincar\'e plot around $x_1=1/2$. The
    Poincar\'e section is defined by $x_2=1/2$ and it is shown in light gray in
    the panels on the right-hand side. The initial points of the field lines are
    sampled differently for the initial and finals state of the field, hence
    they are not exactly the evolution of one another.} 
\end{figure}

\begin{figure}
  \centering
  \includegraphics[scale=0.21]{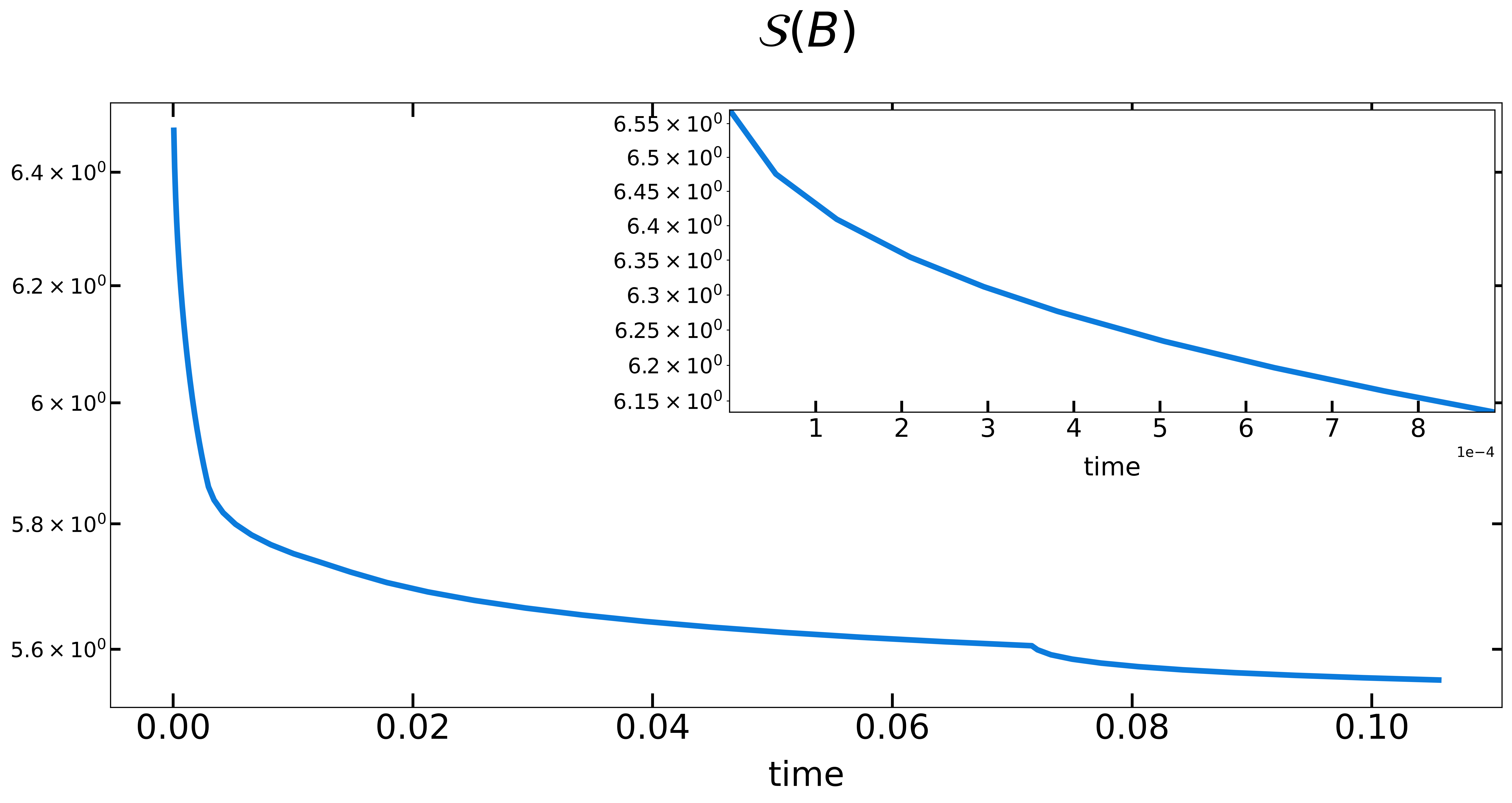}
  \includegraphics[scale=0.21]{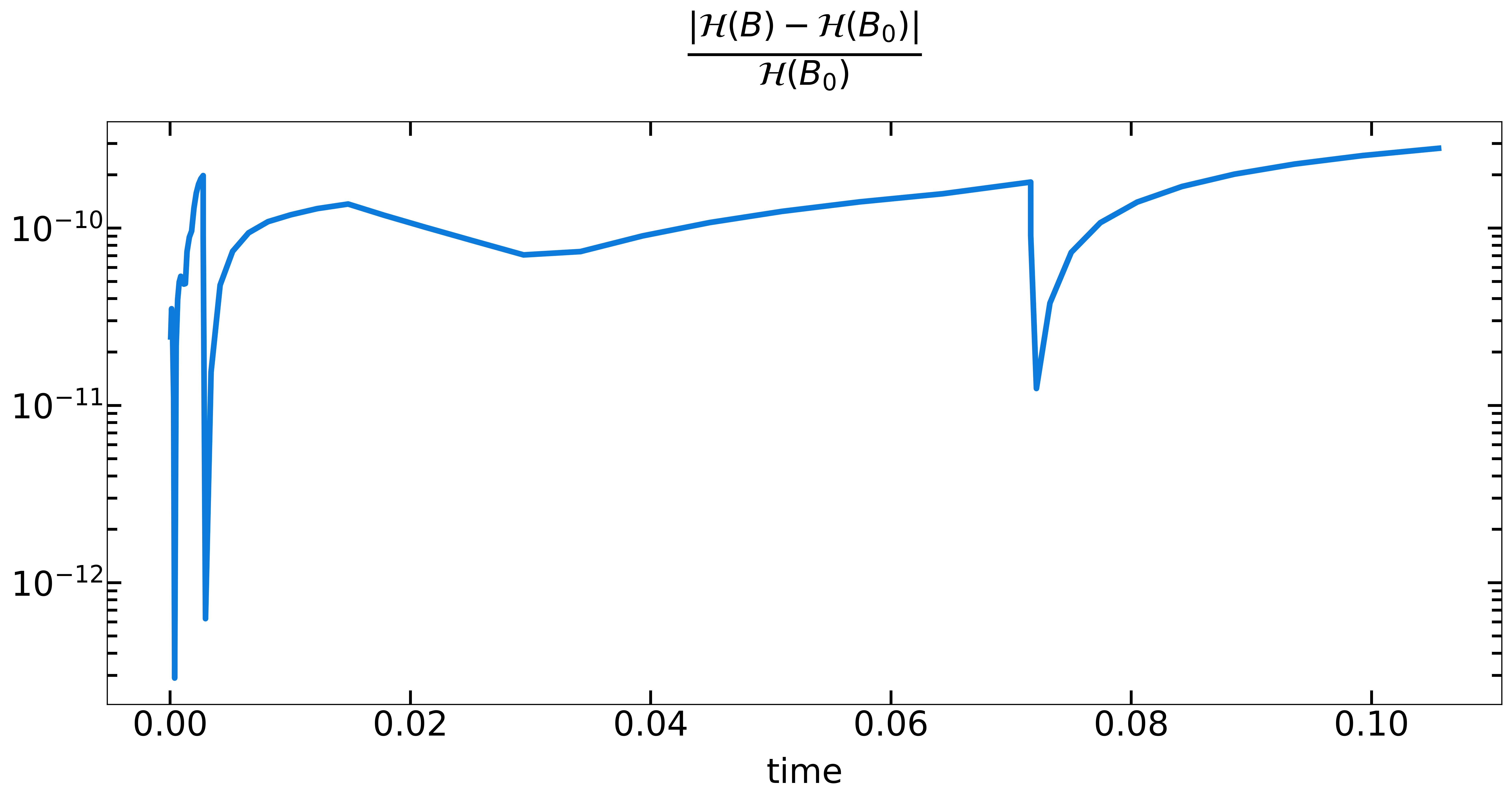}
  \includegraphics[scale=0.21]{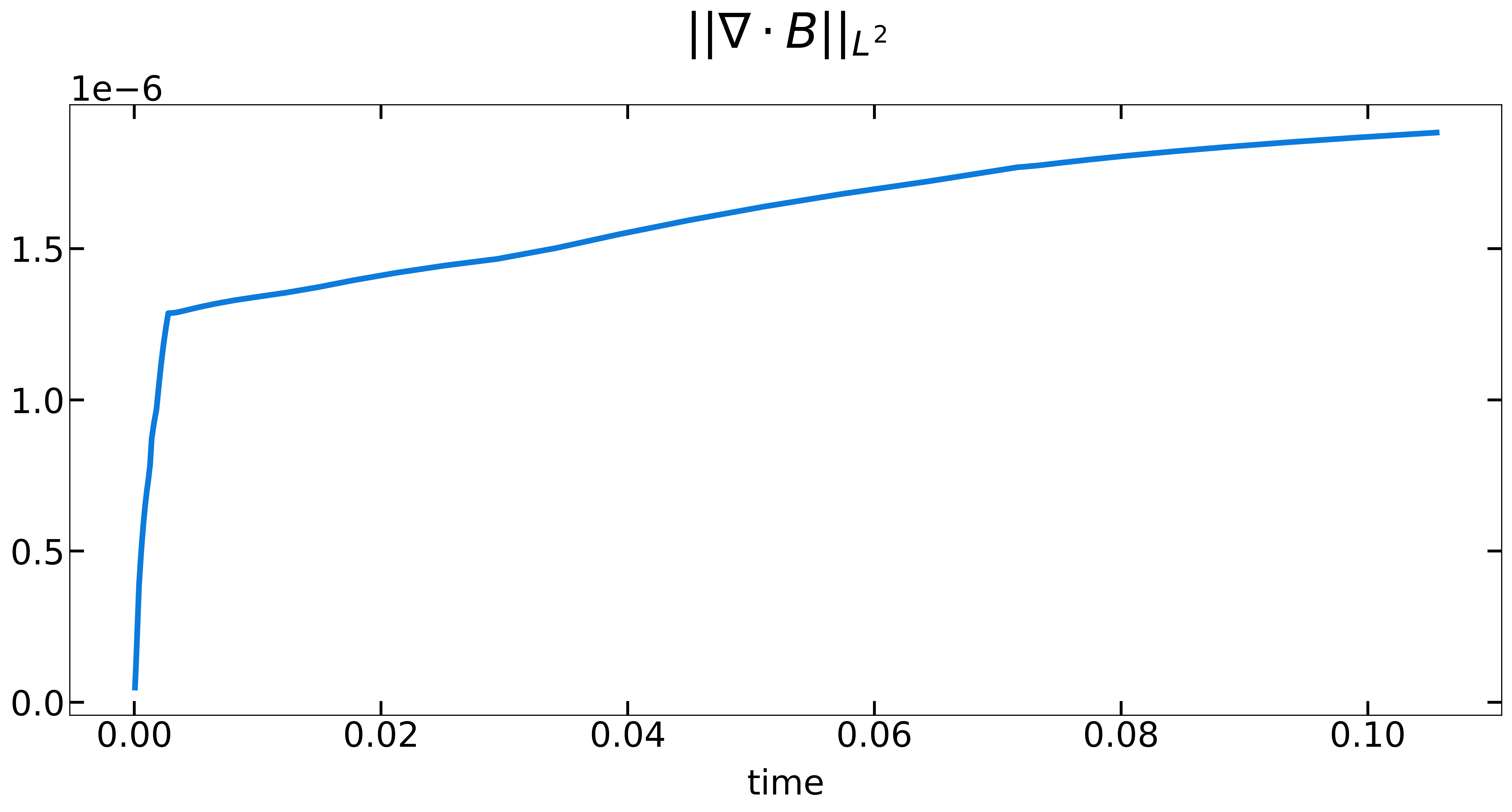}
  \includegraphics[scale=0.21]{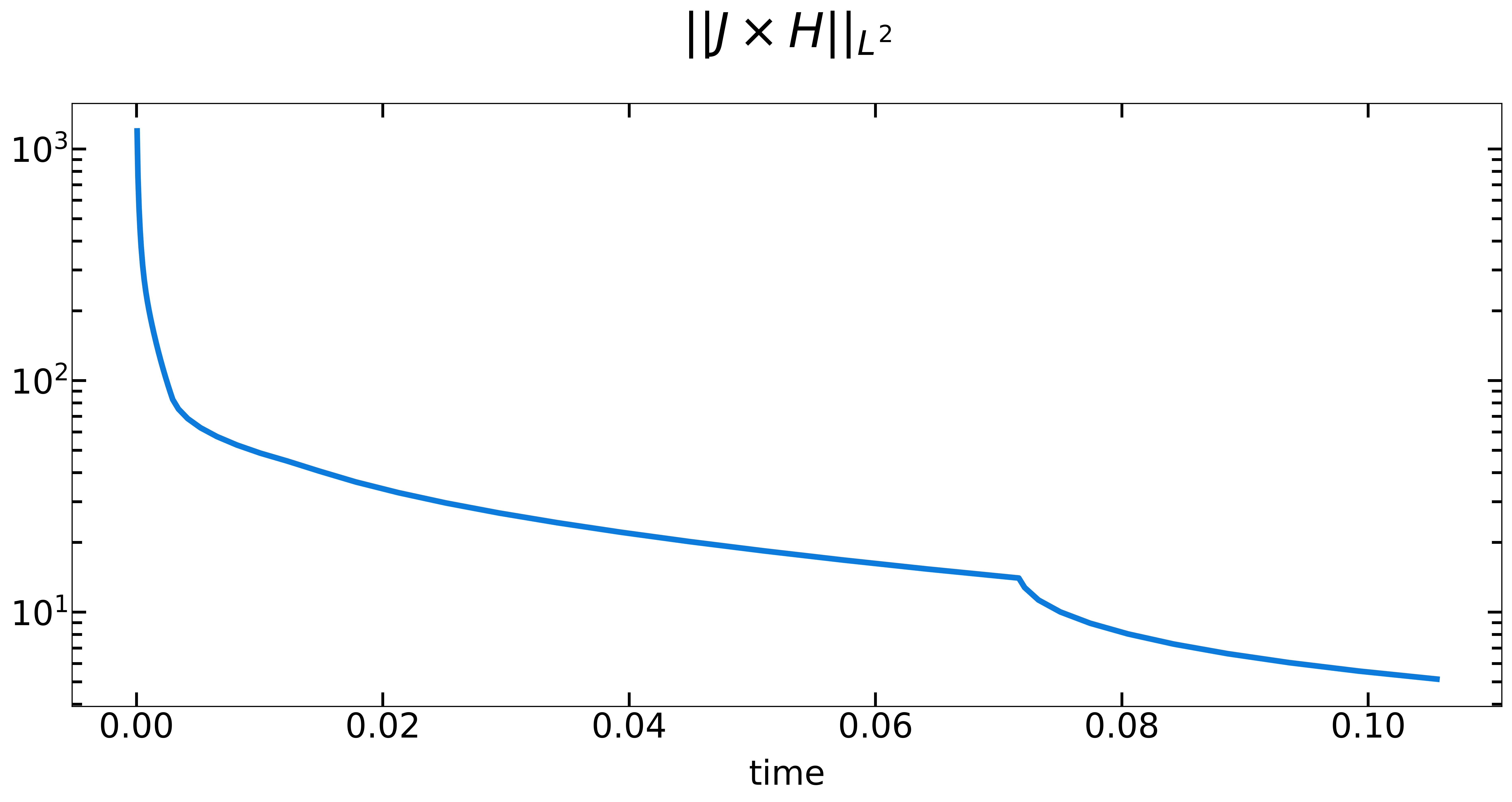}
  \caption{\label{fig:B_diag} From top to bottom, time evolution of the entropy
    (magnetic energy), the relative variation of the Hamiltonian (magnetic
    helicity), the $L^2$ norm of the divergence, and of the vector $j \times H$,
    where $j = \widetilde{\curl} B$ is the current density (computed with the
    weak $\curl$ operator), and $H$ is the $L^2$-orthogonal projection of $B$
    onto $H_0(\curl,\Omega)$. The equilibrium condition for the considered
    numerical scheme reads $j \times H = 0$. Corners and jumps in the time
    traces corresponds to restarts with larger time steps.}
\end{figure}

The time evolution of the initial condition~(\ref{eq:Beltrami-ic}) is obtained
numerically by means of the FEEC scheme, which has been implemented in FEniCS as
in the case of the tests reported in Section~\ref{sec:coll-like-metr}. Here, we
use a relatively course resolution: the domain $\ol{\Omega} = [0,1]^3$ has been
discretized by a uniform grid of $32^3$ nodes. The time step is adapted, but
limited to a maximum of $10^{-3}$. The obtained final state is represented in
Fig.~\ref{fig:B_sol} (bottom row), again by means of a Poincar\'e plot,
Fig.~\ref{fig:B_sol}, bottom-left panel. The four large islands appear to have
been preserved by the relaxation process and are rendered in
Fig.~\ref{fig:B_sol}, bottom-right panel. Other large islands appear to have
been preserved, but a closer analysis shows that the field line topology is not
exactly preserved \cite{Bressan2023}. Indeed the numerical scheme preserves
magnetic helicity only, and this alone does not completely guarantee the exact
preservation of the field line topology. 

When using Poincar\'e plots for the visualization of the field line topology,
one should address the effect of the error of the projection onto the
finite-element space used for the representation of the magnetic field. In this
case, $B$ is approximated in the space of linear Raviart-Thomas elements for
computations, and the discrete approximation is further projected onto the space
of linear Lagrange elements, for visualization purposes (precise definitions of
these finite-element spaces can be found,  e.g.,  in the FEniCS book 
\cite{Logg2012}). Lagrange elements are nodal so that the degrees of freedom
coincide with the value of the field at the grid nodes and can be directly
interpolated. We have qualitatively checked the effect of all those operations
on the results by comparing the Poincar\'e sections of the analytical field
(\ref{eq:Beltrami-ic}) with that of its projection onto the finite element space
on the considered grid.

Figure \ref{fig:B_diag} shows the evolution in time of the main quantities of
interest, namely, the entropy, the relative variation of the Hamiltonian with
respect to its initial value, and the $L^2$-norms of $\div B$ and of the vector
$j \times H$, where $j = \widetilde{\curl} B$ is (proportional to) the current
density (computed weakly), while $H(t)$ is the $L^2$-orthogonal
projection of $B(t) \in H_0(\div,\Omega)$ onto the space
$H_0(\curl,\Omega)$.
For sufficiently regular fields, we have $j \times H = (\curl B) \times B$, but
in general $B(t) \in H_0(\div,\Omega)$ and
$H(t) \in H_0(\curl,\Omega)$ are different and the numerical equilibrium
condition is $j \times H = 0$. We recall that in this application the entropy
and the Hamiltonian coincide with the magnetic energy and the magnetic helicity,
respectively. From Fig.~\ref{fig:B_diag}, we see that the qualitative properties
of the relaxation method are preserved: the entropy decreases monotonically, the
Hamiltonian is constant within a relative error of $10^{-10}$, and $\div B = 0$
within an absolute error of $2 \times 10^{-6}$ measured by the norm in
$L^2(\Omega)$. We also see that $j\times H$ decreases, which indicates that the
solution is approaching a configuration that satisfies the Beltrami condition
$j \times H = 0$.

\section{Summary and conclusions}
\label{sec:conclusions}

We have considered the question of whether, given a metriplectic system with
Hamiltonian function $\fun{H}$ and (dissipated) entropy $\fun{S}$, the orbit
with initial condition $u_0$, in the long-time limit, converges to a minimum of
$\fun{S}$ on the surface of constant Hamiltonian
$\{ u \colon \fun{H}(u) = \fun{H}(u_0)\}$.
This question is interesting in itself, since many physical systems are
metriplectic, but our work is mainly motivated by the idea of utilizing
artificial metriplectic systems as relaxation methods for the computation of
equilibria of fluids and plasmas. 

We have shown that, in general, the answer is negative. For finite-dimensional
metriplectic systems, we have given a sufficient condition,
Proposition~\ref{th:mod-Lyapunov}, under which an orbit relaxes to a constrained
entropy minimum. These results are proven by means of a natural extension of the
Lyapunov stability theorem for systems with constants of motion. One key
assumption in Proposition~\ref{th:mod-Lyapunov} is that the metric bracket 
should be specifically degenerate with respect to a given finite set of
constants of motion, or minimally degenerate if the Hamiltonian is the only
constant of motion. 
Recall, a metric bracket is specifically degenerate if its null space is spanned
by the gradients of a finite number of invariant functions $\fun{I}^\alpha$,
the constants of motion, and minimally degenerate if its null space is spanned
by the gradient of the Hamiltonian alone. The introduction of the concepts of
specifically and minimally degenerate brackets is justified by
Proposition~\ref{th:mod-Lyapunov} for finite dimensional systems, and
generalized without proofs to the case of infinite-dimensional systems in
section~\ref{sec:infinite-dim}.
In addition, we have generalized the Polyak--{\L}ojasiewicz condition for the
exponential convergence of gradient flows, to the case of metriplectic system.
The finite-dimensional results have been extended to the infinite-dimensional
case without proof, and supported by a number of examples.
In Section~\ref{sec:simple}, we have studied quite in detail a specific
equilibrium problem for the reduced Euler equation. We constructed 
two relaxation methods based upon two metriplectic systems: one is specifically
degenerate and the other one is not. The results of our numerical experiments
with these two relaxation mechanisms can be explained in terms of the
theoretical results. 

In the second part of the paper, we have proposed a class of metric
brackets that have been put forward as a generalization of the Landau
collision operator, which is included in the class as a special case,
Section~\ref{sec:coll-like-metr}. For this reason we propose the term
``collision-like'' brackets. Checking if collision-like brackets are
specifically degenerate reduces to checking two separate conditions, and this
is usually simpler.

We demonstrate the use of such brackets 
as the basis for relaxation methods for various equilibrium problems for both
the reduced Euler equations and axisymmetric MHD equilibria, the latter being
equivalent to solving the Grad-Shafranov equation. These are well-known
equilibrium problems, for which various methods of solution exist, and are used
here only as a proof of concept. From a purely computational point of view, the
direct solution of the Grad-Shafranov equation, in particular, is usually faster
than the relaxation method constructed here, but the latter provides a recipe
that can be adapted to more complicated equilibrium problems.  
Specifically, we have in mind equilibrium problems in kinetic theories,
such as the Vlasov-Maxwell system, drift- and gyro-kinetic equations. There is
also the possibility of generalizing the variational formulation for the
Grad-Shafranov equation to non-monotonic equilibrium profiles, which have not
been treated in the present work. 

At last, we have discussed a simplified class of brackets, in which the
nonlocal nature of collision-like brackets is removed,
Section~\ref{sec:diff-like-metr}. These brackets lead to
diffusion-like evolution equations and reduce to a number of known brackets in
special cases. Without the nonlocality, characteristic of collision-like
operators, these diffusion-like brackets are not minimally degenerate, but they
can still be used to construct relaxation methods for equilibrium problems with
stronger constraints. In the simplest case, we recover the known magnetic
relaxation method of Chodura and Schl\"uter, and we give a numerical example
based upon a structure-preserving numerical scheme obtained via finite-element
exterior calculus (FEEC).

The results presented in this paper are far from complete. The proof of
convergence for infinite-dimensional systems has not been addressed, and the
formal argument for the minimal degeneracy of the collision-like bracket in
example~\ref{ex:old-clb} is incomplete as it relies on a technical condition
being true. 
  
Our results however can have consequences in designing relaxation methods for
equilibrium problems.
Specifically, for equilibrium problems that can be characterized by a
variational principle of the form (\ref{eq:entropy-principle}), that is
finding a (local) minimum of entropy subject to the constraint of energy and 
possibly other quantities being constant, one should make sure that the
relaxation method is based on specifically degenerate brackets, with the
kernel generated by the gradients of exactly the same quantities defining the
constraints. If this is not the case, examples show that the relaxation method
may not find a solution of the considered problem. Equilibria of the reduced
Euler equations, Grad-Shafranov equilibria, and linear Beltrami fields are
examples of problems that belong to this category. Alternatively, one might
search for equilibria that can be characterized by entropy minima subject to
stronger constraints. This is the case, for instance, of nonlinear Beltrami
fields, which can be characterized as minima of magnetic energy over the set of
fields that are smooth deformations (push-forward) of a given initial
configuration, cf. Section~\ref{sec:beltrami-problem} and
Appendix~\ref{sec:VP}. For this type of problem the metric bracket cannot be
specifically degenerate, but must be designed to satisfy the needed constraints.
  
The actual implementation of the relaxation methods can be more subtle, due to
the fact that, depending on the initial condition, the corresponding evolution
equation might not admit a smooth solution, so that low-regularity solutions
need to be considered. As an example in Section~\ref{sec:app-Beltrami}, we
have discussed the case of a relaxation method for Beltrami fields, and the
need for a weak formulation of the evolution equation and the corresponding
equilibrium condition.

As for physical metriplectic systems, example~\ref{ex:Landau} in
Section~\ref{sec:c-general} shows how the techniques developed here could be  
used to study physically relevant metric brackets (in this example, Morrison's
bracket for the Landau collision operator). By checking if the bracket is
specifically degenerate, one can gain some information on the long-time limit
of the solution. We find that, in general, Morrison's bracket is not
specifically degenerate with respect to the three collision invariants, but
this is only due to the fact that the Landau collision operator acts pointwise
in space.

\appendix

\section{Bilinear forms and Leibniz identity}
\label{sec:Leibniz}

In this appendix we recall the basic definition of a Poisson structure and give
a formal derivation of~(\ref{eq:J-K-kernels}). The material is standard but not
always easy to find in textbooks. Let $V \subseteq L^2(\R^d,\mu;\R^N)$ be a
Banach space of squared-integrable functions over a domain
$\Omega \subseteq \R^d$ with values in $\R^N$. We define functional
derivatives of a function in $C^1(V)$ as an element of $L^2(\Omega,\mu;\R^N)$
according~(\ref{eq:functional-derivatives}) with the pairing given by the
$L^2$ scalar product,  
\begin{equation*}
  \langle w,v \rangle = \int_\Omega  w(x) \cdot  v(x) d\mu(x),
\end{equation*}
for all $v \in V$ and $w \in L^2(\Omega,\mu;\R^N)$, where $\mu$ is a
measure on $\Omega$, e.g., the Lebesgue measure.

We consider in particular the consequences of the Leibniz identity on the
structure of the bracket.
Let $\alpha: C^\infty(V) \times C^\infty(V) \to C^\infty(V)$
be a generic bi-linear map satisfying the symmetry condition 
\begin{equation}
  \label{eq:alpha-symmetry}
  \alpha(\fun{F},\fun{G}) = \sigma \alpha(\fun{G},\fun{F}), \quad \sigma^2 = 1,
\end{equation}
and the Leibniz identity,
\begin{equation}
  \label{eq:alpha-Leibniz}
  \alpha(\fun{F},\fun{G}\fun{H}) = \alpha(\fun{F},\fun{G}) \fun{H} +
  \fun{G} \alpha(\fun{F},\fun{H}).
\end{equation}
Hence, $\alpha$ is either symmetric or antisymmetric, depending on whether
$\sigma = 1$ or $\sigma = -1$, and a derivation in all arguments. Both Poisson
and metric brackets, cf. Section~\ref{sec:metriplectic}, are special cases of
such a bilinear form. 

Given the restriction functions $\fun{R}_{x,j}(u) = u_j(x)$, which are defined
for functions $u$ that are at least continuous, let 
\begin{equation*}
  \mathscr{A}_{ij}(u;x,x') \coloneqq
  \alpha(\fun{R}_{x,i}, \fun{R}_{x',j})(u),
\end{equation*}
and we have
\begin{equation*}
  \mathscr{A}_{ij}(u;x,x') = \sigma \mathscr{A}_{ji}(u;x',x).
\end{equation*}
We claim that a bilinear form $\alpha$ satisfying~(\ref{eq:alpha-symmetry})
and~(\ref{eq:alpha-Leibniz}) is such that 
\begin{itemize}
\item[1)] it vanishes on constants, that is,
  \begin{equation*}
    \fun{G}(v) = a \in \R \text{ for all $v$} \implies
    \alpha(\fun{F}, \fun{G}) = 0 \text{ for all $\fun{F}$;}
  \end{equation*}
\item[2)] it has the representation, for any $\fun{F},\fun{G}$,
  \begin{equation*}
    \alpha(\fun{F},\fun{G}) = \sum_{i,j} \int_\Omega \int_\Omega
    \frac{\delta \fun{F}(u)}{\delta u_i} (x)
    \mathscr{A}_{ij}(u;x,x') \frac{\delta \fun{F}(u)}{\delta u_j} (x') 
    d\mu(x') d\mu(x).
  \end{equation*}
\end{itemize}
In particular, 2) implies Eqs.~(\ref{eq:J-K-kernels}).

Claim 1) follows from the bilinearity of the form and the Leibniz property.
If $\fun{G}(v) = a$ is a constant function, for any pair of functions $\fun{F}$
and $\fun{H}$, one has
\begin{align*}
  a \alpha(\fun{F}, \fun{H}) &= \alpha(\fun{F}, a\fun{H}) =
  \alpha(\fun{F}, \fun{G} \fun{H}) \\
  &= \alpha(\fun{F}, \fun{G} ) \fun{H} + a \alpha(\fun{F}, \fun{H}),
\end{align*}
from which one deduces $\alpha(\fun{F}, \fun{G} ) = 0$ as claimed.

Claim 2) requires Taylor's formula: For any $u,u_0 \in V$
\begin{equation*}
  \fun{F}(u) = \fun{F}(u_0) + \theta_{\fun{F}}(u_0,u)(u-u_0),
\end{equation*}
where, with $v = u-u_0$,
\begin{align*}
  \theta_{\fun{F}}(u_0,u)v &= \int_0^1 D\fun{F}\big((1-t)u_0 + tu\big)v dt \\
  &= \int_0^1 \int_\Omega  \frac{\delta \fun{F}}{\delta u}
  \big((1-t)u_0 + tu\big) \cdot v d\mu(x) dt \\
  &= \sum_{i} \int_0^1 \int_\Omega  \frac{\delta \fun{F}}{\delta u_i}
  \big((1-t)u_0 + tu\big)  v_i d\mu(x) dt.
\end{align*}
For a fixed $u_0$, $\fun{F}(u_0)$ and $\fun{G}(u_0)$ are constants and using
claim 1) we have 
\begin{align*}
  \alpha(\fun{F},\fun{G})
  &= \alpha\big(\fun{F} - \fun{F}(u_0),\fun{G}-\fun{G}(u_0)\big) \\
  &=\alpha\big(\theta_{\fun{F}}(u_0,\cdot)(\cdot-u_0),
  \theta_{\fun{G}}(u_0,\cdot)(\cdot-u_0) \big) \\
  &= \sum_{i,j} \int_0^1 \int_0^1 \int_\Omega \int_\Omega
  A_{ij} dt ds d\mu(x) d\mu(x'),
\end{align*}
where we have \emph{formally} exchanged the integrals and the bi-linear form and
for brevity we have defined
\begin{multline*}
  A_{ij} = \alpha\Big(
  \frac{\delta \fun{F}((1-t)u_0 + tu)}{\delta u_i}(x) (u(x)-u_0(x))_i, \\
  \frac{\delta \fun{G}((1-s)u_0 + su)}{\delta u_j}(x') (u(x')-u_0(x'))_j\Big).
\end{multline*}
In the latter expression the first argument of $\alpha$ is the product of the
functions $u \mapsto \delta \fun{F}((1-t)u_0 + tu)/\delta u_i|_x$ and
$u \mapsto (u - u_0)_j|_x$; analogously for the second argument. We can use
Leibniz identity and evaluate at $u=u_0$ with the result that
\begin{equation*}
  A_{ij} = \frac{\delta \fun{F}(u_0)}{\delta u_i}(x)
  \alpha\big( \fun{R}_{x,i} , \fun{R}_{x',j} \big)(u_0)
  \frac{\delta \fun{G}(u_0)}{\delta u_j}(x').
\end{equation*}
Therefore,
\begin{equation*}
  \alpha(\fun{F},\fun{G})(u_0) =
  \sum_{i,j} \int_\Omega \int_\Omega
  \frac{\delta \fun{F}(u_0)}{\delta u_i}(x)
  \mathscr{A}_{ij}(u_0,x,x') \frac{\delta \fun{G}(u_0)}{\delta u_j}(x')
  d\mu(x') d\mu(x),
\end{equation*}
and since the point $u_0 \in V$ is arbitrary this is claim 2). This argument
however is purely formal: The restriction function $\mathcal{R}_{x,j}$ is
defined only for functions that can be evaluated at a point $x$, e.g. continuous
functions, thus excluding $L^p$ functions for any $p$. We have assumed that
functional derivative exists and in exchanging the integral with the form
$\alpha$ one needs some continuity in order to pass to the limit after
approximating the integrals by finite sums. We have also freely exchanged the
integration order.

\section{On continuous bilinear forms on Hilbert spaces}
\label{sec:bilinear-forms}

Let $H$ be a Hilbert space over the fields of real numbers and with scalar
product $(\cdot,\cdot)$ and induced norm $\|\cdot\|$.
If $a : H \times H \to \R$ is a continuous, positive
bilinear form, where continuity means that there exists $C>0$ such that
\begin{equation*}
  0 \leq a(u,v) \leq C \|u\| \|v\|,
\end{equation*}
for all $u,v \in H$, then one can find a bounded, symmetric, positive-definite
linear operator $A:H\to H$ such that
\begin{equation*}
  a(u,v) = (u, Av),
\end{equation*}
for all $u,v \in H$. In order to find $A$ one observes that, for any $u$ fixed,
$\ell(v) = a(u,v)$ is a bounded linear function from $H \to \R$. The Riesz
representation theorem \cite[Theorem~8.12]{Hunter2001} yields a unique element
$w \in H$ such that 
\begin{equation*}
  \ell(v) = (w, v),
\end{equation*}
for all $v \in H$ and $\|w\| = \sup\{\ell(v): v\in H,\; \|v\|=1\} \leq C \|u\|$.
Since $w$ is unique, one can set $w = Au$, and $A$ is a bounded linear operator
on $H$. Then
\begin{equation*}
  a(u,v) = \ell(v) = (Au, v).
\end{equation*}
Positivity and symmetry follow from the positivity and symmetry of $a$.

\section{A Lagrangian variational principle for Beltrami fields}
\label{sec:VP}

In this appendix, we give a self-contained overview of the variational
principle for Beltrami fields. This is the constant-pressure version of the
variational principle for full MHD equilibria obtained by Kendall
\cite{Kendall1960}, and formulated in a modern language.

While the special case of linear Beltrami fields obey Woltjer's principle of
least magnetic energy at constant magnetic helicity,
cf. Section~\ref{sec:beltrami-problem}, general Beltrami fields minimize energy
under a much stronger constraint. 

On a bounded (not necessarily simply connected) domain $\Omega \subset \R^3$, we
fix a reference magnetic field $B_0 \in V$, where $V$ is the space of vector
fields $B \in [L^2(\Omega)]^3$ satisfying the conditions
\begin{equation}
  \label{eq:B0-psi0-conditions}
  \begin{aligned}
    \div B &= 0,  &&\text{ in } \Omega, \\
    n \cdot B &= 0, &&\text{ on } \partial \Omega.
  \end{aligned}
\end{equation}
For any $\Phi : \Omega \to \Omega$ an element of the group
$\mathrm{Diff}(\Omega)$ of diffeomorphisms of the domain $\Omega$, we define 
\begin{equation}
  \label{eq:push-forwards}
  B = \Phi_* B_0 = \frac{D\Phi B_0}{\det D\Phi} \circ \Phi^{-1}, 
\end{equation}
where $D\Phi$ is the Jacobian matrix of $\Phi$ (defined by
$(D\Phi)_{ij}=\partial_{x_j}\Phi_i$) and $\det D\Phi \not=0$ is its
determinant. Then $B$ is the push-forward of the fields $B_0$ with the map
$\Phi$. A direct calculation show that  
\begin{equation*}
  (\det D\Phi) \div B = \div B_0,
\end{equation*}
hence $\div B_0 = 0$ imply $\div B = 0$. Analogously one can show that the
boundary condition $B_0 \cdot n = 0$ on $\partial \Omega$ is preserved by the
diffeomorphism since if $x = \Phi(x_0)$ and $x_0 \in \partial \Omega$,
then $x \in \partial\Omega$, 
\begin{align*}
  n(x) &=
  \frac{\transpose{D\Phi^{-1}(x)} n(x_0)}{|\transpose{D\Phi^{-1}(x)}n(x_0)|},\\
  n(x) \cdot B(x) &= \frac{B_0(x_0) \cdot n(x_0)}
  {\det D\Phi \cdot |\transpose{D\Phi^{-1}(x)}n(x_0)|}.
\end{align*}
(This can be proven by recalling that for a sufficiently regular domain, near a
point $x_0 \in \partial \Omega$ there is a function $f$ such that $f>0$ in
$\Omega$ and $f=0$ on $\partial \Omega$; then $n(x_0) \propto \nabla f(x_0)$ and
this transforms like a $1$-form under $\Phi$.) 
Therefore the push-forward formula~(\ref{eq:push-forwards}) maps $B_0 \in V$
into $B \in V$.

Given $B_0 \in V$, we define the entropy functional on the group
$\mathrm{Diff}(\Omega)$ as the magnetic energy stored in $B$, that is
\begin{equation}
  \label{eq:full-MHD-entropy}
  \fun{S}(\Phi) = \fun{S}(\Phi; B_0)
  = \int_\Omega \frac{|B|^2}{8\pi} dx,
\end{equation}
where $B = \Phi_*B_0$.
The entropy depends parametrically on the initial field $B_0$.

We can now state the variational principle for~(\ref{eq:nonlinear-Beltrami}).
For any $B_0 \in V$ fixed, if $\Phi$ is a critical point
of~(\ref{eq:full-MHD-entropy}), then 
$B = \Phi_* B_0 \in V$ satisfies~(\ref{eq:nonlinear-Beltrami}). More explicitly,
this means that if 
\begin{equation}
  \label{eq:full-MHD-principle}
  \frac{d}{d\varepsilon} \fun{S}(\Phi^\varepsilon)\Big|_{\varepsilon = 0} = 0,
\end{equation}
for any curve $\varepsilon \mapsto \Phi^\varepsilon \in \mathrm{Diff}(\Omega)$
such that $\Phi^\varepsilon|_{\varepsilon=0} = \Phi$, then $B = \Phi_* B_0$ is a
Beltrami field obtained by mapping the given field $B_0$ by the action of the
diffeomorphism $\Phi$. 

In order to prove the variational principle~(\ref{eq:full-MHD-principle}) let
$B^\varepsilon (x) = \Phi^\varepsilon_* B_0(x)$, and introduce the
displacement field 
\begin{equation}
  \label{eq:v-epsilon}
  \xi^\varepsilon =
  \partial_\varepsilon \Phi^\varepsilon \circ (\Phi^\varepsilon)^{-1}.
\end{equation}
The definition is equivalent to
\begin{equation*}
  \partial_\varepsilon \Phi^\varepsilon (x_0) = \xi^\varepsilon(x), \quad
  x = \Phi^\varepsilon (x_0).
\end{equation*}
Then, one obtains
\begin{equation}
  \label{eq:dragging-B}
  \partial_\varepsilon B^\varepsilon =
  \curl(\xi^\varepsilon \times B^\varepsilon),
  \quad B^\varepsilon|_{\varepsilon=0} = B, \\
\end{equation}
and we compute from~(\ref{eq:full-MHD-principle}),
\begin{equation*}
  \frac{d}{d\varepsilon} \fun{S}(\Phi^\varepsilon)\Big|_{\varepsilon = 0} =
  - \int_\Omega \Big[\frac{1}{4\pi} (\curl B)\times B \Big]
  \cdot \xi dx 
  + \int_\Omega n\cdot[(\xi\times B)\times B] d\sigma= 0,
\end{equation*}
where $\xi = \xi^\varepsilon|_{\varepsilon=0}$. The boundary term vanishes due
to the identity $(\xi\times B)\times B = (B\cdot \xi) B - B^2 \xi$ and the
boundary condition $B\cdot n = 0$, $\xi \cdot n=0$; the latter follows from the
fact that $\Phi$ preserves the boundary, hence $\xi|_{\partial\Omega}$ must be
tangent to $\partial\Omega$. Since the derivative of $\fun{S}(\Phi^\varepsilon)$
has to vanish for any curve $\Phi^\varepsilon\in \mathrm{Diff}(\Omega)$ and thus
for every $\xi$, we deduce that $B$ satisfies~(\ref{eq:nonlinear-Beltrami}).   

We remark that, since $B$ is the push-forward of a known field $B_0$, the
field-line topology of $B$ is the same as that of $B_0$. Magnetic helicity is
also preserved.

This is a variant of the variational principle~(\ref{eq:full-MHD-principle}) at
the basis of the relaxation method of Chodura and Schl\"uter
\cite{Chodura1981, Wiegelmann2012}, Moffatt \cite{Moffatt2021}, and of the
SIESTA code \cite{Hirshman2011}.

\bibliographystyle{elsarticle-num}
\bibliography{metriplectic_relaxation_bibliography}

@Article{Morrison1984,
  author  = {P. J. Morrison},
  title   = {Bracket formulation for irreversible classical fields},
  journal = {Physics Letters A},
  year    = {1984},
  volume  = {100},
  number  = {8},
  pages   = {423 - 427},
  issn    = {0375-9601},
  doi     = {10.1016/0375-9601(84)90635-2},
}

@techreport{pjm84b,
  author =     "Morrison, P. J.",
  title =       {Some Observations Regarding Brackets and Dissipation},
  journal =     "{C}enter for {P}ure and {A}pplied {M}athematics {R}eport",
institution= {{U}niversity of {C}alifornia at {B}erkeley},
  number =      {{PAM}--228},
  month={March},
  pages =       "{}",
  year =        1984,
  Note={Available at arXiv:2403.14698v1 [mathph] 15 Mar 2024}
}

@Article{Morrison1986,
  author  = {P. J. Morrison},
  title   = {A paradigm for joined {H}amiltonian and dissipative systems},
  journal = {Physica D: Nonlinear Phenomena},
  year    = {1986},
  volume  = {18},
  pages   = {410--419},
  issn    = {0167-2789},
  doi     = {10.1016/0167-2789(86)90209-5},
}

@Article{Morrison1998,
  author    = {Morrison, P. J.},
  title     = {Hamiltonian description of the ideal fluid},
  journal   = {Rev. Mod. Phys.},
  year      = {1998},
  volume    = {70},
  issue     = {2},
  month     = {Apr},
  pages     = {467--521},
  doi       = {10.1103/RevModPhys.70.467},
  publisher = {American Physical Society},
}

@Article{pjmU24,
  author  = {P. J. Morrison and M. H. Updike},
  title   = {Inclusive Curvature-Like Framework for Describing Dissipation: Metriplectic 4-Bracket Dynamics},
  journal = {Physical Review E},
  year    = {2024},
  volume  = {109},
  pages   = {045202},
}

@Book{Taylor3,
  author    = {Michael E. Taylor},
  title     = {Partial Differential Equations III: Nonlinear Equations},
  year      = {2011},
  volume    = {117},
  series    = {Applied Mathematical Sciences},
  publisher = {Springer New York},
  doi       = {10.1007/978-1-4419-7049-7},
}

@Book{Freidberg2014,
  author    = {Freidberg, Jeffrey P.},
  title     = {Ideal MHD},
  year      = {2014},
  publisher = {Cambridge University Press},
  doi       = {10.1017/CBO9780511795046},
  place     = {Cambridge},
}

@Article{Wiegelmann2012,
  author  = {Wiegelmann, Thomas and Sakurai, Takashi},
  title   = {Solar Force-free Magnetic Fields},
  journal = {Living Reviews in Solar Physics},
  year    = {2012},
  volume  = {9},
  number  = {1},
  month   = sep,
  pages   = {5},
  issn    = {1614-4961},
  doi     = {10.1007/s41116-020-00027-4},
}

@Article{Grad1958,
  author    = {Grad, Harold and Rubin, Hanan},
  title     = {Hydromagnetic equilibria and force-free fields},
  journal   = {Journal of Nuclear Energy (1954)},
  year      = {1958},
  volume    = {7},
  number    = {3-4},
  pages     = {284--285},
  doi       = {10.1016/0891-3919(58)90139-6},
  publisher = {Pergamon},
}

@Article{Grad1964,
  author  = {Grad, Harold},
  title   = {Some New Variational Properties of Hydromagnetic Equilibria},
  journal = {Physics of Fluids},
  year    = {1964},
  volume  = {7},
  number  = {8},
  pages   = {1283-1292},
  doi     = {10.1063/1.1711373},
}

@Article{Chodura1981,
  author  = {R. Chodura and A. Schlueter},
  title   = {A {3D} code for {MHD} equilibrium and stability},
  journal = {Journal of Computational Physics},
  year    = {1981},
  volume  = {41},
  number  = {1},
  pages   = {68 - 88},
  issn    = {0021-9991},
  doi     = {10.1016/0021-9991(81)90080-2},
}

@Article{Kruskal1958a,
  author  = {Kruskal, M. D. and Kulsrud, R. M.},
  title   = {Equilibrium of a Magnetically Confined Plasma in a Toroid},
  journal = {Physics of Fluids},
  year    = {1958},
  volume  = {1},
  number  = {4},
  pages   = {265-274},
  doi     = {10.1063/1.1705884},
}

@Article{Bruno1996,
  author  = {Bruno, Oscar P. and Laurence, Peter},
  title   = {Existence of three-dimensional toroidal {MHD} equilibria with nonconstant pressure},
  journal = {Communications on Pure and Applied Mathematics},
  year    = {1996},
  volume  = {49},
  number  = {7},
  pages   = {717-764},
  doi     = {10.1002/(SICI)1097-0312(199607)49:7<717::AID-CPA3>3.0.CO;2-C},
}

@Article{Yoshida2016,
  author  = {Yoshida, Z. and Morrison, P. J.},
  title   = {Hierarchical structure of noncanonical {H}amiltonian systems},
  journal = {Physica Scripta},
  year    = {2016},
  volume  = {91},
  number  = {2},
  pages   = {024001},
  doi     = {10.1088/0031-8949/91/2/024001},
}

@Article{Amari2009,
  author  = {Amari, Tahar and Boulbe, C\'{e}dric and Boulmezaoud, Tahar Zam\`{e}ne},
  title   = {Computing {B}eltrami Fields},
  journal = {SIAM Journal on Scientific Computing},
  year    = {2009},
  volume  = {31},
  number  = {5},
  pages   = {3217-3254},
  doi     = {10.1137/070700942},
}

@Article{Yoshida1990a,
  author  = {Yoshida, Zensho and Giga, Yoshikazu},
  title   = {Remarks on spectra of operator rot},
  journal = {Mathematische Zeitschrift},
  year    = {1990},
  volume  = {204},
  number  = {1},
  month   = dec,
  pages   = {235--245},
  issn    = {1432-1823},
  doi     = {10.1007/BF02570870},
}

@Book{Bauer1978,
  author    = {Bauer, Frances and Betancourt, Octavio and Garabedian, Paul},
  title     = {A Computational Method in Plasma Physics},
  year      = {1978},
  publisher = {Springer-Verlag},
  location  = {New York},
  doi       = {10.1007/978-3-642-85470-5},
}

@Article{Hirshman1983,
  author  = {Hirshman, S. P. and Whitson, J. C.},
  title   = {Steepest-descent moment method for three-dimensional magnetohydrodynamic equilibria},
  journal = {Physics of Fluids},
  year    = {1983},
  volume  = {26},
  number  = {12},
  pages   = {3553-3568},
  doi     = {10.1063/1.864116},
}

@Article{Reiman1986,
  author  = {A. Reiman and H. Greenside},
  title   = {Calculation of three-dimensional {MHD} equilibria with islands and stochastic regions},
  journal = {Computer Physics Communications},
  year    = {1986},
  volume  = {43},
  number  = {1},
  pages   = {157 - 167},
  issn    = {0010-4655},
  doi     = {10.1016/0010-4655(86)90059-7},
}

@Article{Suzuki2006,
  author  = {Yasuhiro Suzuki and Noriyoshi Nakajima and Kiyomasa Watanabe and others},
  title   = {Development and application of {HINT2} to helical system plasmas},
  journal = {Nuclear Fusion},
  year    = {2006},
  volume  = {46},
  number  = {11},
  month   = {sep},
  pages   = {L19},
  doi     = {10.1088/0029-5515/46/11/L01},
}

@Article{Harafuji1989,
  author   = {Harafuji, Kenji and Hayashi, Takaya and Sato, Tetsuya},
  title    = {Computational study of three-dimensional magnetohydrodynamic equilibria in toroidal helical systems},
  journal  = {Journal of Computational Physics},
  year     = {1989},
  volume   = {81},
  number   = {1},
  month    = mar,
  pages    = {169--192},
  issn     = {0021-9991},
  doi      = {10.1016/0021-9991(89)90069-7},
}

@Article{Moffatt2021,
  author    = {Moffatt, H.K.},
  title     = {Some topological aspects of fluid dynamics},
  journal   = {Journal of Fluid Mechanics},
  year      = {2021},
  volume    = {914},
  pages     = {P1},
  doi       = {10.1017/jfm.2020.230},
  publisher = {Cambridge University Press},
}

@Article{Moffatt1985,
  author    = {Moffatt, H. K.},
  title     = {Magnetostatic equilibria and analogous {E}uler flows of arbitrarily complex topology. Part 1. Fundamentals},
  journal   = {Journal of Fluid Mechanics},
  year      = {1985},
  volume    = {159},
  pages     = {359–378},
  doi       = {10.1017/S0022112085003251},
  publisher = {Cambridge University Press},
}

@Article{Alfven1942,
  author  = {Alfv{\'e}n, H.},
  title   = {Existence of electromagnetic-hydrodynamic waves},
  journal = {Nature},
  year    = {1942},
  volume  = {150},
  pages   = {405--406},
  doi     = {10.1038/150405d0},
}

@Article{Hirshman2011,
  author  = {Hirshman, S. P. and Sanchez, R. and Cook, C. R.},
  title   = {{SIESTA: A scalable iterative equilibrium solver for toroidal applications}},
  journal = {Physics of Plasmas},
  year    = {2011},
  volume  = {18},
  number  = {6},
  note    = {062504},
  issn    = {1070-664X},
  doi     = {10.1063/1.3597155},
}

@Article{Brenier2018,
  author  = {Brenier, Yann and Duan, Xianglong},
  title   = {An integrable example of gradient flow based on optimal transport of differential forms},
  journal = {Calculus of Variations and Partial Differential Equations},
  year    = {2018},
  volume  = {57},
  number  = {5},
  month   = {Aug},
  pages   = {125},
  issn    = {1432-0835},
  doi     = {10.1007/s00526-018-1370-6},
}

@Book{Hirsch2013,
  author    = {Hirsch, M. W. and Smale, S. and Devaney, R. L.},
  title     = {Differential Equations, Dynamical Systems, and an Introduction to Chaos},
  year      = {2013},
  edition   = {Third Edition},
  publisher = {Elsevier},
  doi       = {10.1016/C2009-0-61160-0},
}

@Article{Materassi2012,
  author  = {Materassi, Massimo and Tassi, Emanuele},
  title   = {Metriplectic framework for dissipative magneto-hydrodynamics},
  journal = {Physica D: Nonlinear Phenomena},
  year    = {2012},
  volume  = {241},
  number  = {6},
  month   = mar,
  pages   = {729--734},
  issn    = {0167-2789},
  doi     = {10.1016/j.physd.2011.12.013},
}

@Article{Coquinot2020,
  author    = {Coquinot, Baptiste and Morrison, Philip J.},
  title     = {A general metriplectic framework with application to dissipative extended magnetohydrodynamics},
  journal   = {Journal of Plasma Physics},
  year      = {2020},
  volume    = {86},
  number    = {3},
  pages     = {835860302},
  issn      = {0022-3778},
  doi       = {10.1017/s0022377820000392},
  database  = {Cambridge Core},
  edition   = {2020/05/19},
  publisher = {Cambridge University Press},
}

@Article{Materassi2015,
  author  = {Materassi, M.},
  title   = {Metriplectic Algebra for Dissipative Fluids in {L}agrangian Formulation},
  journal = {Entropy},
  year    = {2015},
  volume  = {17},
  number  = {3},
  month   = {Mar},
  pages   = {1329–1346},
  issn    = {1099-4300},
  doi     = {10.3390/e17031329},
}

@Article{Morrison1982,
  author    = {Morrison, Philip J.},
  title     = {Poisson brackets for fluids and plasmas},
  journal   = {Cont. Math.},
  year      = {1982},
  volume    = {88},
  month     = jul,
  pages     = {13--46},
  issn      = {0094-243X},
  doi       = {10.1063/1.33633},
  comment   = {doi: 10.1063/1.33633},
  publisher = {American Institute of Physics},
}

@Article{Kraus2017a,
  author  = {Michael Kraus and Eero Hirvijoki},
  title   = {Metriplectic integrators for the {L}andau collision operator},
  journal = {Physics of Plasmas},
  year    = {2017},
  volume  = {24},
  number  = {10},
  pages   = {102311},
  doi     = {10.1063/1.4998610},
}

@Article{Flierl2011,
  author  = {G.R. Flierl and P.J. Morrison},
  title   = {Hamiltonian {D}irac simulated annealing: {A}pplication to the calculation of vortex states},
  journal = {Physica D: Nonlinear Phenomena},
  year    = {2011},
  volume  = {240},
  number  = {2},
  pages   = {212 - 232},
  note    = {“Nonlinear Excursions” Symposium and Volume in Physica D to honor Louis N. Howard’s scientific career},
  issn    = {0167-2789},
  doi     = {10.1016/j.physd.2010.08.011},
}

@Article{Chikasue2015,
  author  = {Chikasue, Y. and Furukawa, M.},
  title   = {Simulated annealing applied to two-dimensional low-beta reduced magnetohydrodynamics},
  journal = {Physics of Plasmas (1994-present)},
  year    = {2015},
  volume  = {22},
  number  = {2},
  eid     = {022511},
  doi     = {10.1063/1.4913234},
}

@Article{Chikasue2015a,
  author  = {Chikasue,Y. and Furukawa,M.},
  title   = {Adjustment of vorticity fields with specified values of {C}asimir invariants as initial condition for simulated annealing of an incompressible, ideal neutral fluid and its MHD in two dimensions},
  journal = {Journal of Fluid Mechanics},
  year    = {2015},
  volume  = {774},
  month   = {7},
  pages   = {443--459},
  issn    = {1469-7645},
  doi     = {10.1017/jfm.2015.263},
}

@Article{Furukawa2017,
  author  = {M. Furukawa and P. J. Morrison},
  title   = {Simulated annealing for three-dimensional low-beta reduced {MHD} equilibria in cylindrical geometry},
  journal = {Plasma Physics and Controlled Fusion},
  year    = {2017},
  volume  = {59},
  number  = {5},
  pages   = {054001},
  doi     = {10.1088/1361-6587/aa5863},
}

@Article{Furukawa2018,
  author  = {Furukawa,M. and Watanabe,Takahiro and Morrison,P. J. and others},
  title   = {Calculation of large-aspect-ratio tokamak and toroidally-averaged stellarator equilibria of high-beta reduced magnetohydrodynamics via simulated annealing},
  journal = {Physics of Plasmas},
  year    = {2018},
  volume  = {25},
  number  = {8},
  pages   = {082506},
  doi     = {10.1063/1.5038043},
}

@Article{Gay-Balmaz2013,
  author  = {François Gay-Balmaz and Darryl D Holm},
  title   = {Selective decay by {C}asimir dissipation in inviscid fluids},
  journal = {Nonlinearity},
  year    = {2013},
  volume  = {26},
  number  = {2},
  pages   = {495},
  doi     = {10.1088/0951-7715/26/2/495},
}

@Article{Gay-Balmaz2014,
  author  = {François Gay-Balmaz and Darryl D Holm},
  title   = {A geometric theory of selective decay with applications in {MHD}},
  journal = {Nonlinearity},
  year    = {2014},
  volume  = {27},
  number  = {8},
  pages   = {1747},
  doi     = {10.1088/0951-7715/27/8/1747},
}

@Article{Adams2017,
  author  = {Adams, Mark F. and Hirvijoki, Eero and Knepley, Matthew G. and others},
  title   = {Landau Collision Integral Solver with Adaptive Mesh Refinement on Emerging Architectures},
  journal = {SIAM Journal on Scientific Computing},
  year    = {2017},
  volume  = {39},
  number  = {6},
  pages   = {C452-C465},
  doi     = {10.1137/17M1118828},
}

@Book{Arnold1998,
  author    = {V. I. Arnold and B. A. Khesin},
  title     = {Topological Methods in Hydrodynamics},
  year      = {1998},
  publisher = {Springer-Verlag},
}

@Book{Hunter2001,
  author    = {Hunter, J. K. and Nachtergaele, B.},
  title     = {Applied Analysis},
  year      = {2001},
  publisher = {World Scientific Publishing Company},
}

@Book{Marsden2001,
  author    = {Marsden, J. E. and Ratiu, T. and Abraham, R.},
  title     = {Manifolds, Tensor Analysis, and Applications},
  year      = {2001},
  edition   = {Third Edition},
  publisher = {Springer},
}

@Article{Yoshida2002,
  author  = {Yoshida, Z. and Mahajan, S. M.},
  title   = {Variational Principles and Self-Organization in Two-Fluid Plasmas},
  journal = {Phys. Rev. Lett.},
  year    = {2002},
  volume  = {88},
  issue   = {9},
  month   = {Feb},
  pages   = {095001},
  doi     = {10.1103/PhysRevLett.88.095001},
}

@Article{Hasegawa1985,
  author    = {Akira Hasegawa},
  title     = {Self-organization processes in continuous media},
  journal   = {Advances in Physics},
  year      = {1985},
  volume    = {34},
  number    = {1},
  pages     = {1-42},
  doi       = {10.1080/00018738500101721},
  publisher = {Taylor \& Francis},
}

@Article{Woltjer1958a,
  author  = {Woltjer, L.},
  title   = {On hydromagnetic equilibrium},
  journal = {Proceedings of the National Academy of Sciences},
  year    = {1958},
  volume  = {44},
  number  = {9},
  pages   = {833-841},
  doi     = {10.1073/pnas.44.9.833},
}

@Article{Taylor1974,
  author    = {Taylor, J. B.},
  title     = {Relaxation of Toroidal Plasma and Generation of Reverse Magnetic Fields},
  journal   = {Phys. Rev. Lett.},
  year      = {1974},
  volume    = {33},
  issue     = {19},
  month     = {Nov},
  pages     = {1139--1141},
  doi       = {10.1103/PhysRevLett.33.1139},
  publisher = {American Physical Society},
}

@Article{Taylor1986,
  author    = {Taylor, J. B.},
  title     = {Relaxation and magnetic reconnection in plasmas},
  journal   = {Rev. Mod. Phys.},
  year      = {1986},
  volume    = {58},
  issue     = {3},
  month     = {Jul},
  pages     = {741--763},
  doi       = {10.1103/revmodphys.58.741},
  numpages  = {0},
  publisher = {American Physical Society},
}

@Article{Qin2012,
  author  = {Qin, Hong and Liu, Wandong and Li, Hong and others},
  title   = {{W}oltjer-{T}aylor State without {T}aylor's Conjecture: Plasma Relaxation at all Wavelengths},
  journal = {Phys. Rev. Lett.},
  year    = {2012},
  volume  = {109},
  issue   = {23},
  month   = {Dec},
  pages   = {235001},
  doi     = {10.1103/PhysRevLett.109.235001},
}

@Article{Bloch2013,
  author  = {Anthony M. Bloch and Philip J. Morrison and Tudor S. Ratiu},
  title   = {Flows in the Normal and {K}aehler Metrics and Triple Bracket Generated Metriplectic Systems},
  journal = {in Recent Trends in Dynamical Systems, eds. A. Johann et al., Springer Proceedings in Mathematics \& Statistics},
  year    = {2013},
  volume  = {35},
  pages   = {371–415},
}

@Article{Polyak1963,
  author  = {B.T. Polyak},
  title   = {Gradient methods for the minimisation of functionals},
  journal = {USSR Computational Mathematics and Mathematical Physics},
  year    = {1963},
  volume  = {3},
  number  = {4},
  pages   = {864-878},
  issn    = {0041-5553},
  doi     = {10.1016/0041-5553(63)90382-3},
}

@InCollection{Lojasiewicz1984,
  author    = {{\L}ojasiewicz, S.},
  title     = {Sur les trajectoires du gradient d’une fonction analytique},
  booktitle = {Seminari di Geometria 1982–1983},
  year      = {1984},
  publisher = {Universita di Bologna, Istituto di Geometria, Dipartamento di Matematica},
  pages     = {115–117},
}

@Article{Takeda1991,
  author    = {Tatsuoki Takeda and Shinji Tokuda},
  title     = {Computation of {MHD} equilibrium of tokamak plasma},
  journal   = {Journal of Computational Physics},
  year      = {1991},
  volume    = {93},
  number    = {1},
  pages     = {1 - 107},
  issn      = {0021-9991},
  doi       = {10.1016/0021-9991(91)90074-u},
}

@Article{LoDestro1994,
  author   = {LoDestro, L. L. and Pearlstein, L. D.},
  title    = {{On the {G}rad-Shafranov equation as an eigenvalue problem, with implications for q solvers}},
  journal  = {Physics of Plasmas},
  year     = {1994},
  volume   = {1},
  number   = {1},
  month    = {01},
  pages    = {90-95},
  issn     = {1070-664X},
  doi      = {10.1063/1.870464},
}

@Article{Pataki2013,
  author   = {Pataki, Andras and Cerfon, Antoine J. and Freidberg, Jeffrey P. and others},
  title    = {A fast, high-order solver for the {Grad}-{Shafranov} equation},
  journal  = {Journal of Computational Physics},
  year     = {2013},
  volume   = {243},
  month    = jun,
  pages    = {28--45},
  issn     = {0021-9991},
  doi      = {10.1016/j.jcp.2013.02.045},
}

@Article{Brockett1991,
  author   = {R.W. Brockett},
  title    = {Dynamical systems that sort lists, diagonalize matrices, and solve linear programming problems},
  journal  = {Linear Algebra and its Applications},
  year     = {1991},
  volume   = {146},
  pages    = {79-91},
  issn     = {0024-3795},
  doi      = {10.1109/cdc.1988.194420},
}

@Article{Bloch1992,
  author   = {Bloch, Anthony M. and Brockett, Roger W. and Ratiu, Tudor S.},
  title    = {Completely integrable gradient flows},
  journal  = {Communications in Mathematical Physics},
  year     = {1992},
  volume   = {147},
  number   = {1},
  month    = jun,
  pages    = {57--74},
  issn     = {1432-0916},
  doi      = {10.1007/bf02099528},
  refid    = {Bloch1992},
}

@Article{Bressan2018,
  author    = {C. Bressan and M. Kraus and P.  J. Morrison and others},
  title     = {Relaxation to magnetohydrodynamics equilibria via collision brackets},
  journal   = {Journal of Physics: Conference Series},
  year      = {2018},
  volume    = {1125},
  month     = {nov},
  pages     = {012002},
  doi       = {10.1088/1742-6596/1125/1/012002},
  publisher = {{IOP} Publishing},
}

@Article{Boulmezaoud1999,
  author  = {Boulmezaoud, Tahar-Zamène and Maday, Yvon and Amari, Tahar},
  title   = {On the linear force-free fields in bounded and unbounded three-dimensional domains},
  journal = {ESAIM: M2AN},
  year    = {1999},
  volume  = {33},
  number  = {2},
  pages   = {359-393},
  doi     = {10.1051/m2an:1999121},
}

@Book{Taylor1,
  author    = {Michael E. Taylor},
  title     = {Partial Differential Equations I: Basic Theory},
  year      = {2011},
  volume    = {115},
  series    = {Applied Mathematical Sciences},
  publisher = {Springer New York},
  doi       = {10.1007/978-1-4419-7055-8},
}

@Article{Laurence1991,
  author   = {Laurence, Peter and Avellaneda, Marco},
  title    = {{On {W}oltjer's  variational principle for force-free fields}},
  journal  = {Journal of Mathematical Physics},
  year     = {1991},
  volume   = {32},
  number   = {5},
  month    = {05},
  pages    = {1240-1253},
  issn     = {0022-2488},
  doi      = {10.1063/1.529321},
}

@Article{Dewar2008,
  author    = {Dewar, Robert L. and Hole, Matthew J. and McGann, Mathew and others},
  title     = {Relaxed Plasma Equilibria and Entropy-Related Plasma Self-Organization Principles},
  journal   = {Entropy},
  year      = {2008},
  volume    = {10},
  number    = {4},
  month     = {Nov},
  pages     = {621–634},
  issn      = {1099-4300},
  doi       = {10.3390/e10040621},
  publisher = {MDPI AG},
}

@Article{Malhotra2019,
  author   = {Malhotra, Dhairya and Cerfon, Antoine and Imbert-Gérard, Lise-Marie and others},
  title    = {Taylor states in stellarators: A fast high-order boundary integral solver},
  journal  = {Journal of Computational Physics},
  year     = {2019},
  volume   = {397},
  month    = nov,
  pages    = {108791},
  issn     = {0021-9991},
  doi      = {10.1016/j.jcp.2019.06.067},
}

@Article{Hudson2012,
  author    = {S. R. Hudson and R. L. Dewar and G. Dennis and others},
  title     = {Computation of multi-region relaxed magnetohydrodynamic equilibria},
  journal   = {Physics of Plasmas},
  year      = {2012},
  volume    = {19},
  number    = {11},
  pages     = {112502},
  doi       = {10.1063/1.4765691},
}

@Article{Moffatt1969,
  author    = {Moffatt,H. K.},
  title     = {The degree of knottedness of tangled vortex lines},
  journal   = {Journal of Fluid Mechanics},
  year      = {1969},
  volume    = {null},
  issue     = {01},
  month     = {1},
  pages     = {117--129},
  issn      = {1469-7645},
  doi       = {10.1017/s0022112069000991},
  numpages  = {13},
}

@Article{Vogel2003,
  author  = {Vogel, Thomas},
  title   = {On the asymptotic linking number},
  journal = {Proc. of the American Mathematical Society},
  year    = {2003},
  volume  = {131},
  number  = {7},
  pages   = {2289--2297},
  doi     = {10.1090/s0002-9939-02-06792-8},
}

@InBook{Arnold2014,
  author     = {Arnold, Vladimir I.},
  title      = {The asymptotic {H}opf invariant and its applications},
  booktitle  = {Vladimir I. Arnold - Collected Works: Hydrodynamics, Bifurcation Theory, and Algebraic Geometry 1965-1972},
  year       = {2014},
  bookauthor = {Givental, Alexander B. and Khesin, Boris A. and Varchenko, Alexander N. and Vassiliev, Victor A. and Viro, Oleg Ya.},
  publisher  = {Springer Berlin Heidelberg},
  isbn       = {978-3-642-31031-7},
  pages      = {357--375},
  doi        = {10.1007/978-3-642-31031-7_32},
  address    = {Berlin, Heidelberg},
}

@Book{Palis1982,
  author    = {Jacob Palis and Welington Melo},
  title     = {Geometric Theory of Dynamical Systems},
  year      = {1982},
  publisher = {Springer New York},
  doi       = {10.1007/978-1-4612-5703-5},
}

@Book{Henry1981,
  author    = {Henry, Dan},
  title     = {{Geometric Theory of Semilinear Parabolic Equations}},
  year      = {1981},
  publisher = {Springer, Berlin},
}

@Book{Temam1998,
  author    = {Temam, R.},
  title     = {{Infinite-Dimensional Dynamical Systems in Mechanics and Physics}},
  year      = {1998},
  publisher = {Springer-Verlag New York},
}

@Book{Wiggins2003,
  author    = {Wiggins, S.},
  title     = {Introduction to Applied Nonlinear Dynamical Systems and Chaos},
  year      = {2003},
  publisher = {Springer New York, NY},
}

@Book{Moretti2023,
  author    = {Moretti, Valter},
  title     = {Analytical Mechanics},
  year      = {2023},
  publisher = {Springer Nature Switzerland},
  isbn      = {9783031276125},
  issn      = {2532-3318},
  journal   = {UNITEXT},
}

@Article{Beretta1986,
  author   = {Beretta, Gian Paolo},
  title    = {A theorem on {L}yapunov stability for dynamical systems and a conjecture on a property of entropy},
  journal  = {J. Math. Phys.},
  year     = {1986},
  volume   = {27},
  number   = {1},
  month    = jan,
  pages    = {305--308},
  issn     = {0022-2488},
  doi      = {10.1063/1.527390},
}

@InProceedings{Karimi2016,
  author    = {Karimi, Hamed and Nutini, Julie and Schmidt, Mark},
  title     = {Linear Convergence of Gradient and Proximal-Gradient Methods Under the {P}olyak-Łojasiewicz Condition},
  booktitle = {Machine Learning and Knowledge Discovery in Databases},
  year      = {2016},
  publisher = {Springer International Publishing},
  pages     = {795--811},
  address   = {Cham},
  issn      = {978-3-319-46128-1},
  refid     = {10.1007/978-3-319-46128-1_50},
}

@Book{Rauch2012,
  author    = {Rauch, J.},
  title     = {Hyperbolic Partial Differential Equations and Geometric Optics},
  year      = {2012},
  series    = {Graduate studies in mathematics},
  publisher = {American Mathematical Society},
  isbn      = {9780821872918},
  lccn      = {2011046666},
}

@PhdThesis{Bressan2023,
  author   = {Bressan, Camilla},
  title    = {Metriplectic relaxation for calculating equilibria: theory and structure-preserving discretization},
  year     = {2023},
  language = {en},
  pages    = {168},
  url      = {https://mediatum.ub.tum.de/1686142},
  school   = {Technische Universität München},
}

@Book{Evans1998,
  author    = {Evans, Lawrence C.},
  title     = {{Partial Differential Equations}},
  year      = {1998},
  series    = {Graduate studies in mathematics},
  publisher = {American Mathematical Society},
  isbn      = {0-8218-0772-2},
  address   = {Providence (R.I.)},
}

@Book{Monk2003,
  author    = {Monk, P.},
  title     = {{Finite Element Methods for Maxwell's equations}},
  year      = {2003},
  publisher = {Oxford Un. Press},
  location  = {Oxford},
}

@Article{Alnaes2015,
  author  = {Alnaes, Martin S. and Blechta, Jan and Hake, Johan and Johansson, August and Kehlet, Benjamin and Logg, Anders and Richardson, Chris N. and Ring, Johannes and Rognes, Marie E. and Wells, Garth N.},
  title   = {The {FEniCS} Project Version 1.5},
  journal = {Archive of Numerical Software},
  year    = {2015},
  volume  = {3},
  doi     = {10.11588/ans.2015.100.20553},
}

@Book{Logg2012,
  author    = {Logg, Anders and Mardal, Kent-Andre and Wells, Garth N. and others},
  title     = {Automated Solution of Differential Equations by the Finite Element Method},
  year      = {2012},
  publisher = {Springer},
  doi       = {10.1007/978-3-642-23099-8},
}

@Article{Alnaes2014,
  author  = {Alnaes, Martin S. and Logg, Anders and {\O}lgaard, Kristian B. and Rognes, Marie E. and Wells, Garth N.},
  title   = {Unified Form Language: A domain-specific language for weak formulations of partial differential equations},
  journal = {{ACM} Transactions on Mathematical Software},
  year    = {2014},
  volume  = {40},
  doi     = {10.1145/2566630},
}

@InCollection{Logg2012a,
  author    = {Logg, Anders and Wells, Garth N. and Hake, Johan},
  title     = {{DOLFIN:} a {C++/Python} Finite Element Library},
  booktitle = {Automated Solution of Differential Equations by the Finite Element Method},
  year      = {2012},
  editor    = {Logg, Anders and Mardal, Kent-Andre and Wells, Garth N.},
  volume    = {84},
  series    = {Lecture Notes in Computational Science and Engineering},
  publisher = {Springer},
  chapter   = {10},
}

@Article{Logg2010,
  author  = {Logg, Anders and Wells, Garth N.},
  title   = {{DOLFIN:} {A}utomated Finite Element Computing},
  journal = {{ACM} Transactions on Mathematical Software},
  year    = {2010},
  volume  = {37},
  doi     = {10.1145/1731022.1731030},
}

@Article{Carthy1999,
  author   = {Mc~Carthy, P. J.},
  title    = {Analytical solutions to the {G}rad-{S}hafranov equation for tokamak equilibrium with dissimilar source functions},
  journal  = {Phys. Plasmas},
  year     = {1999},
  volume   = {6},
  number   = {9},
  month    = sep,
  pages    = {3554--3560},
  issn     = {1070-664X},
  doi      = {10.1063/1.873630},
}

@Article{Maschke1973,
  author   = {Maschke, E. K.},
  title    = {Exact solutions of the {MHD} equilibrium equation for a toroidal plasma},
  journal  = {Plasma Physics},
  year     = {1973},
  volume   = {15},
  number   = {6},
  month    = jun,
  pages    = {535},
  issn     = {0032-1028},
  doi      = {10.1088/0032-1028/15/6/006},
}

@InProceedings{Herrnegger1972,
  author    = {Herrnegger, F.},
  booktitle = {Proceedings of 5th European Conference on Controlled Fusion and Plasma Physics, Grenoble},
  year      = {1972},
  volumes   = {I},
  publisher = {European Physical Society},
  pages     = {26},
}

@Article{Hu2021,
  author   = {Hu, Kaibo and Lee, Young-Ju and Xu, Jinchao},
  title    = {Helicity-conservative finite element discretization for incompressible {MHD} systems},
  journal  = {Journal of Computational Physics},
  year     = {2021},
  volume   = {436},
  month    = jul,
  pages    = {110284},
  issn     = {0021-9991},
  doi      = {10.1016/j.jcp.2021.110284},
}

@Article{Candelaresi2014,
  author   = {Candelaresi, S. and Pontin, D. and Hornig, G.},
  title    = {Mimetic Methods for {Lagrangian} Relaxation of Magnetic Fields},
  journal  = {SIAM Journal on Scientific Computing},
  year     = {2014},
  volume   = {36},
  number   = {6},
  pages    = {B952-B968},
  doi      = {10.1137/140967404},
}

@Article{Dixon1989,
  author  = {Dixon, A. M. and Berger, M. A. and Browning, P. K. and Priest, E. R.},
  title   = {A generalization of the {W}oltjer minimum-energy principle},
  journal = {Astron. Astrophys.},
  year    = {1989},
  volume  = {225},
  pages   = {156--166},
}

@Article{Kendall1960,
  author  = {Kendall, P. C.},
  title   = {The variational formulation of the magneto-hydrostatic equations},
  journal = {Astrophys. J.},
  year    = {1960},
  volume  = {131},
  pages   = {681},
}

@Article{Yang1986,
  author  = {Yang, W. H. and Sturrock, P. A. and Antiochos, S. K.},
  title   = {Force-free magnetic fields: The magneto-frictional method},
  journal = {Astrophys. J.},
  year    = {1986},
  volume  = {309},
  pages   = {383--391},
}

@Article{Padilla2022,
  author     = {Padilla, Marcel and Gross, Oliver and Kn\"{o}ppel, Felix and Chern, Albert and Pinkall, Ulrich and Schr\"{o}der, Peter},
  title      = {Filament based plasma},
  journal    = {ACM Trans. Graph.},
  year       = {2022},
  volume     = {41},
  number     = {4},
  month      = {jul},
  issn       = {0730-0301},
  doi        = {10.1145/3528223.3530102},
  address    = {New York, NY, USA},
  articleno  = {153},
  issue_date = {July 2022},
  numpages   = {14},
  publisher  = {Association for Computing Machinery},
}

@Article{Gross2023,
  author   = {Gross, Oliver and Pinkall, Ulrich and Schröder, Peter},
  title    = {Plasma knots},
  journal  = {Physics Letters A},
  year     = {2023},
  volume   = {480},
  month    = aug,
  pages    = {128986},
  issn     = {0375-9601},
}

@Misc{Klimchuk1992,
  author   = {Klimchuk, J. A. and Sturrock, P. A.},
  title    = {Three-dimensional force-free magnetic fields and flare energy buildup},
  year     = {1992},
  comment  = {Export Date: 12 August 2024; Cited By: 44},
  database = {Scopus},
  journal  = {Astrophysical Journal},
  number   = {1},
  pages    = {344--353},
  volume   = {385},
}

@Article{Guo2016,
  author    = {Guo, Y. and Xia, C. and Keppens, R. and Valori, G.},
  title     = {MAGNETO-FRICTIONAL MODELING OF CORONAL NONLINEAR FORCE-FREE FIELDS. {I. T}ESTING WITH ANALYTIC SOLUTIONS},
  journal   = {The Astrophysical Journal},
  year      = {2016},
  volume    = {828},
  number    = {2},
  month     = sep,
  pages     = {82},
  issn      = {0004-637X},
  doi       = {10.3847/0004-637X/828/2/82},
  publisher = {The American Astronomical Society},
}

@Article{Valori2007,
  author   = {Valori, G. and Kliem, B. and Fuhrmann, M.},
  title    = {Magnetofrictional Extrapolations of {L}ow and {L}ou’s Force-Free Equilibria},
  journal  = {Solar Physics},
  year     = {2007},
  volume   = {245},
  number   = {2},
  month    = oct,
  pages    = {263--285},
  issn     = {1573-093X},
  doi      = {10.1007/s11207-007-9046-y},
  refid    = {Valori2007},
}

@Article{Valori2010,
  author  = {Valori, G. and Kliem, B. and Török, T. and Titov, V. S.},
  title   = {Testing magnetofrictional extrapolation with the {Titov-Démoulin} model of solar active regions},
  journal = {A\&A},
  year    = {2010},
  volume  = {519},
  month   = sep,
  doi     = {10.1051/0004-6361/201014416},
  refid   = {10.105100046361201014416},
}

@Article{Yeates2022,
  author    = {Yeates, A. R.},
  title     = {On the limitations of magneto-frictional relaxation},
  journal   = {Geophysical \& Astrophysical Fluid Dynamics},
  year      = {2022},
  volume    = {116},
  number    = {4},
  month     = jul,
  pages     = {305--320},
  issn      = {0309-1929},
  doi       = {10.1080/03091929.2021.2021197},
  comment   = {doi: 10.1080/03091929.2021.2021197},
  publisher = {Taylor \& Francis},
}

@Article{Zoni2019,
  author   = {Zoni, Edoardo and Güçlü, Yaman},
  title    = {Solving hyperbolic-elliptic problems on singular mapped disk-like domains with the method of characteristics and spline finite elements},
  journal  = {Journal of Computational Physics},
  year     = {2019},
  volume   = {398},
  month    = dec,
  pages    = {108889},
  issn     = {0021-9991},
  doi      = {10.1016/j.jcp.2019.108889},
}

@Book{LaSalle1961,
  author    = {La Salle, J. and Lefschetz, S.},
  title     = {Stability by Liapunov's Direct Method with Applications},
  year      = {1961},
  publisher = {Academic Press},
  location  = {New York},
}

@article{pjmS24,
  title={A collision operator for describing dissipation in noncanonical phase space},
  author={Sato, N. and Morrison, P. J.},
  journal={Fund.  Plasma Phys.},
  volume={10},
  pages={100054},
  year={2024},
  publisher={Elsevier}
}

@Article{pjmZB24,
  author  = {A. Zaidni and P J. Morrison and S. Benjelloun},
  title   = {Thermodynamically consistent {C}ahn-{H}illiard-{N}avier-{S}tokes equations using the metriplectic dynamics formalism},
  journal = {Physica D},
  year    = {2024},
  volume  = {468},
  pages   = {134303},
  issn    = {0167-2789},
  doi     = {10.1016/j.physd.2024.134303},
}

@Article{Vallis-1989,
  author  = {G. K. Vallis and G. G. Carnevale and W. R. Young},
  title   = {Extremal energy properties and construction of stable solutions of the {E}uler equations},
  journal = {J. Fluid Mech.},
  year    = {1989},
  volume  = {207},
  pages   = {133--152},
  doi     = {10.1017/S0022112089002533},
}

@Article{Carnevale-1990,
  author  = {G. G. Carnevale and G. K. Vallis},
  title   = {Pseudo-advective relaxation to stable states of inviscid two-dimensional fluids},
  journal = {J. Fluid Mech.},
  year    = {1990},
  volume  = {213},
  pages   = {549--571},
  doi     = {10.1017/S0022112090002440},
}

@Article{pjmE86,
  author  = {P. J. Morrison and S. Eliezer},
  title   = {Spontaneous Symmetry Breaking and Neutral Stability in the Noncanonical {H}amiltonian Formalism},
  journal = {Phys. Rev. A},
  year    = {1986},
  volume  = {33},
  pages   = {4205--4214},
}

@article{pjmY20,
author = {Z. Yoshida and P. J. Morrison},
title = {Deformation of {L}ie-{P}oisson Algebras and Chirality},
year = {2020},
journal = {J. Math. Phys.},
volume = {61},
number = {},
pages = {082901}
	}

@article{pjmYT17,
	author = {Z. Yoshida and T. Tokieda and P. J. Morrison},
	title = {Rattleback: A model of how geometric singularity induces dynamic chirality},
	year = {2017},
	journal = {Phys. Lett.  A},
	volume = {381},
	number = {},
	pages = {2772--2777}
	}

@Article{He2025,
  author    = {He, Mingdong and Farrell, Patrick E. and Hu, Kaibo and Andrews, Boris D.},
  title     = {Topology-preserving discretization for the magneto-frictional equations arising in the {P}arker conjecture},
  year      = {2025},
  doi       = {10.48550/arXiv.2501.11654},
  copyright = {Creative Commons Attribution 4.0 International},
  publisher = {arXiv},
}

@Book{Imbert-Gerard2024,
  author    = {Imbert-Gérard, Lise-Marie and Paul, Elizabeth J. and Wright, Adelle M.},
  title     = {An Introduction to Stellarators: From Magnetic Fields to Symmetries and Optimization},
  year      = {2024},
  publisher = {Society for Industrial and Applied Mathematics},
  isbn      = {9781611978223},
  doi       = {10.1137/1.9781611978223},
  month     = jan,
}

@Article{pjmBZ24,
  author   = {Barham, William and Morrison, Philip J. and Zaidni, Azeddine},
  title    = {A thermodynamically consistent discretization of {1D} thermal-fluid models using their metriplectic 4-bracket structure},
  journal  = {Communications in Nonlinear Science and Numerical Simulation},
  year     = {2025},
  volume   = {145},
  month    = jun,
  pages    = {108683},
  issn     = {1007-5704},
  doi      = {10.2139/ssrn.5035052},
}

@Article{pjmF24,
  author   = {Furukawa, M. and Morrison, P. J.},
  title    = {Simulated annealing of reduced magnetohydrodynamic systems},
  journal  = {Reviews of Modern Plasma Physics},
  year     = {2025},
  volume   = {9},
  number   = {1},
  month    = mar,
  pages    = {15},
  issn     = {2367-3192},
  doi      = {10.1007/s41614-025-00185-8},
  refid    = {Furukawa2025},
}

@Book{Holm2009,
  author    = {Holm, D.D. and Schmah, T. and Stoica, C. and Ellis, D.C.P.},
  title     = {Geometric mechanics and symmetry: from finite to infinite dimensions},
  year      = {2009},
  series    = {Oxford texts in applied and engineering mathematics},
  publisher = {Oxford University Press},
  isbn      = {9780199212910},
  lccn      = {2009019331},
}

@Book{Marsden1999,
  author    = {Marsden, Jerrold E. and Ratiu, Tudor S.},
  title     = {Introduction to Mechanics and Symmetry: A Basic Exposition of Classical Mechanical Systems},
  year      = {1999},
  publisher = {Springer New York},
  isbn      = {9780387217925},
  doi       = {10.1007/978-0-387-21792-5},
  issn      = {0939-2475},
  journal   = {Texts in Applied Mathematics},
}

@Article{Denny2016,
  author    = {Denny, Diane},
  title     = {Existence of a solution to a semilinear elliptic equation},
  journal   = {AIMS Mathematics},
  year      = {2016},
  volume    = {1},
  number    = {3},
  pages     = {208--211},
  issn      = {2473-6988},
  doi       = {10.3934/Math.2016.3.208},
  publisher = {American Institute of Mathematical Sciences (AIMS)},
}

@Article{Villani1999,
  author    = {Villani, Cédric},
  title     = {Conservative forms of Boltzmann's collision operator : Landau revisited},
  journal   = {ESAIM: Modélisation mathématique et analyse numérique},
  year      = {1999},
  language  = {en},
  volume    = {33},
  number    = {1},
  pages     = {209--227},
  doi       = {10.1051/m2an:1999112},
  publisher = {EDP Sciences},
  refid     = {M2AN_1999__33_1_209_0},
}

@Article{Lenard1960,
  author  = {Lenard, A.},
  title   = {On Bogoliubov’s kinetic equation for a spatially homogeneous plasma},
  journal = {Annals of Physics},
  year    = {1960},
  volume  = {10},
  number  = {3},
  pages   = {390--400},
}

@Article{Villani2007,
  author    = {Villani, Cédric},
  title     = {Hypocoercive Diffusion Operators},
  journal   = {Bollettino dell'Unione Matematica Italiana},
  year      = {2007},
  language  = {eng},
  volume    = {10-B},
  number    = {2},
  month     = {6},
  pages     = {257-275},
  doi       = {10.4171/022-3/25},
  publisher = {Unione Matematica Italiana},
}

@Article{Enciso2016,
  author   = {Enciso, Alberto and Peralta-Salas, Daniel},
  title    = {Beltrami Fields with a Nonconstant Proportionality Factor are Rare},
  journal  = {Archive for Rational Mechanics and Analysis},
  year     = {2016},
  volume   = {220},
  number   = {1},
  month    = apr,
  pages    = {243--260},
  issn     = {1432-0673},
  doi      = {10.1007/s00205-015-0931-5},
  refid    = {Enciso2016},
}

@Article{Grad1967,
  author    = {Harold Grad},
  title     = {Toroidal Containment of a Plasma},
  journal   = {The Physics of Fluids},
  year      = {1967},
  volume    = {10},
  number    = {1},
  pages     = {137-154},
  doi       = {10.1063/1.1761965},
}

@Article{Dudt2020,
  author    = {Dudt, D. W. and Kolemen, E.},
  title     = {DESC: A stellarator equilibrium solver},
  journal   = {Physics of Plasmas},
  year      = {2020},
  volume    = {27},
  number    = {10},
  month     = oct,
  pages     = {102513},
  issn      = {1070-664X},
  doi       = {10.1063/5.0020743},
  comment   = {doi: 10.1063/5.0020743},
  publisher = {American Institute of Physics},
}

@Article{Hindenlang2025,
  author    = {Hindenlang, Florian and Plunk, Gabriel G. and Maj, Omar},
  title     = {Computing MHD equilibria of stellarators with a flexible coordinate frame},
  journal   = {Plasma Physics and Controlled Fusion},
  year      = {2025},
  volume    = {67},
  number    = {4},
  month     = mar,
  pages     = {045002},
  issn      = {0741-3335},
  doi       = {10.1088/1361-6587/adba11},
  publisher = {IOP Publishing},
}

@Article{Garabedian1998,
  author    = {Garabedian, P. R.},
  title     = {Magnetohydrodynamic stability of fusion plasmas},
  journal   = {Communications on Pure and Applied Mathematics},
  year      = {1998},
  volume    = {51},
  number    = {9-10},
  pages     = {1019--1033},
  issn      = {1097-0312},
  doi       = {10.1002/(SICI)1097-0312(199809/10)51:9/10<1019::AID-CPA4>3.0.CO;2-G},
  publisher = {Wiley Subscription Services, Inc., A Wiley Company},
}

@Article{Garabedian2008,
  author   = {Garabedian, Paul R.},
  title    = {Three-dimensional analysis of tokamaks and stellarators},
  journal  = {Proceedings of the National Academy of Sciences},
  year     = {2008},
  volume   = {105},
  number   = {37},
  pages    = {13716-13719},
  doi      = {10.1073/pnas.0806354105},
}

@Article{Boulmezaoud2000,
  author   = {Boulmezaoud, T. Z. and Amari, T.},
  title    = {On the existence of non-linear force-free fields in three-dimensional domains},
  journal  = {Zeitschrift für angewandte Mathematik und Physik ZAMP},
  year     = {2000},
  volume   = {51},
  number   = {6},
  month    = nov,
  pages    = {942--967},
  issn     = {1420-9039},
  doi      = {10.1007/pl00001531},
  refid    = {Boulmezaoud2000},
}

@Article{Enciso2025,
  author   = {Enciso, Alberto and Peralta-Salas, Daniel},
  title    = {Obstructions to Topological Relaxation for Generic Magnetic Fields},
  journal  = {Archive for Rational Mechanics and Analysis},
  year     = {2025},
  volume   = {249},
  number   = {1},
  month    = dec,
  pages    = {6},
  issn     = {1432-0673},
  doi      = {10.1007/s00205-024-02078-5},
  refid    = {Enciso2024},
}

@Article{Hu2017,
  author   = {Hu, Kaibo and Ma, Yicong and Xu, Jinchao},
  title    = {Stable finite element methods preserving $\nabla \cdot {B}=0$ exactly for MHD models},
  journal  = {Numerische Mathematik},
  year     = {2017},
  volume   = {135},
  number   = {2},
  month    = feb,
  pages    = {371--396},
  issn     = {0945-3245},
  doi      = {10.1007/s00211-016-0803-4},
  refid    = {Hu2017},
}

@Article{Parker1972,
  author    = {Parker, E. N.},
  title     = {Topological Dissipation and the Small-Scale Fields in Turbulent Gases},
  journal   = {Astrophysical Journal},
  year      = {1972},
  volume    = {174},
  month     = {jun},
  pages     = {499},
  doi       = {10.1086/151512},
  publisher = {American Astronomical Society},
}

@Article{pjmZ24,
  author        = {Zaidni, Azeddine and Morrison, Philip J.},
  title         = {Metriplectic four-bracket algorithm for constructing thermodynamically consistent dynamical systems},
  journal       = {Phys. Rev. E},
  year          = {2025},
  volume        = {112},
  number        = {2},
  month         = aug,
  pages         = {025101},
  doi           = {10.1103/r2lb-xkq6},
  __markedentry = {[omaj:]},
  publisher     = {American Physical Society},
  refid         = {10.1103/r2lb-xkq6},
}

@Book{Davidson2001,
  author    = {Davidson, P.A.},
  title     = {An Introduction to Magnetohydrodynamics},
  year      = {2001},
  series    = {Cambridge Texts in Applied Mathematics},
  publisher = {Cambridge University Press},
  isbn      = {9780521794879},
  lccn      = {00033733},
}

@Article{Enciso2025spe,
  author    = {Enciso, Alberto and Luque, Alejandro and Peralta-Salas, Daniel},
  title     = {MHD equilibria with nonconstant pressure in nondegenerate toroidal domains},
  journal   = {Journal of the European Mathematical Society},
  year      = {2025},
  volume    = {27},
  number    = {6},
  month     = dec,
  pages     = {2251--2291},
  issn      = {1435-9863},
  doi       = {10.4171/JEMS/1410},
  publisher = {European Mathematical Society - EMS - Publishing House GmbH},
}

@InCollection{Meyer1973,
  author        = {Meyer, K. R.},
  title         = {Symmetries and Integrals in Mechanics},
  booktitle     = {Dynamical Systems},
  year          = {1973},
  editor        = {Peixoto, M. M.},
  publisher     = {Academic Press},
  pages         = {259--272},
  doi           = {10.1016/B978-0-12-550350-1.50025-4},
  url           = {https://www.sciencedirect.com/science/article/pii/B9780125503501500254},
  __markedentry = {[omaj:]},
  abstract      = {Publisher Summary This chapter discusses the concept of symmetries and integrals in mechanics. The literature of Hamiltonian mechanics has several special theorems dealing with Hamiltonian systems that admit additional integrals in involution. Two examples are: (1) If a Hamiltonian system admits p additional independent integral in involution, then the algebraic multiplicity of + 1 as a characteristic multiplier of a periodic solution is greater than or equal to 2(p + 1); and (2) the integration of a Hamiltonian system of n degrees of freedom that admits p independent integrals in involution can be reduced to the integration of a Hamiltonian system of n − p degrees of freedom with p parameters and additional quadratures. After restating (2) in modern terminology, these theorems are generalized by dropping the assumption that the integrals are in involution.},
  issn          = {978-0-12-550350-1},
  month         = jan,
}

@Article{Marsden1974a,
  author   = {Marsden, Jerrold and Weinstein, Alan},
  title    = {Reduction of symplectic manifolds with symmetry},
  journal  = {Reports on Mathematical Physics},
  year     = {1974},
  volume   = {5},
  number   = {1},
  month    = feb,
  pages    = {121--130},
  issn     = {0034-4877},
  url      = {https://www.sciencedirect.com/science/article/pii/0034487774900214},
}

@Article{Berk1986,
  author        = {Berk, H. L. and Freidberg, J. P. and Llobet, X. and Morrison, P. J. and Tataronis, J. A.},
  title         = {Existence and calculation of sharp boundary magnetohydrodynamic equilibrium in three‐dimensional toroidal geometry},
  journal       = {Phys. Fluids},
  year          = {1986},
  volume        = {29},
  number        = {10},
  month         = oct,
  pages         = {3281--3290},
  issn          = {0031-9171},
  url           = {https://doi.org/10.1063/1.865845},
  __markedentry = {[omaj:]},
  abstract      = {The problem of sharp boundary, ideal magnetohydrodynamic equilibria in three‐dimensional toroidal geometry is addressed. The sharp boundary, which separates a uniform pressure, current‐free plasma from a vacuum, is determined by a magnetic surface of a given vacuum magnetic field. The pressure balance equation has the form of a Hamilton-Jacobi equation with a Hamiltonian that is quadratic in the momentum variables, which are the two covariant components of the magnetic field on the outer surface of the plasma. The condition of finding a unique solution on the outer surface is identical with finding phase‐space tori in nonlinear dynamics problems, and the Kolmogorov-Arnold-Moser (KAM) theorem guarantees that such solutions exist for a wide band of parameters. Perturbation theory is used to calculate the properties of the magnetic field just outside the plasma. Special perturbation theory is needed to treat resonances and it is explicitly shown that there are bands of pressure where there are no solutions.},
}

@InBook{Llave2001,
  author    = {de la Llave, Rafael},
  title     = {A tutorial on KAM theory},
  year      = {2001},
  publisher = {American Mathematical Society},
  isbn      = {9780821893746},
  pages     = {175--292},
  doi       = {10.1090/pspum/069/1858536},
  issn      = {0082-0717},
  journal   = {Smooth Ergodic Theory and Its Applications},
}

@Article{Sato2025,
  author    = {Sato, Naoki and Morrison, Philip J.},
  title     = {Scattering theory in noncanonical phase space: A Drift-Kinetic collision operator for weakly collisional plasmas},
  journal   = {Physics of Plasmas},
  year      = {2025},
  volume    = {32},
  number    = {10},
  month     = oct,
  issn      = {1089-7674},
  doi       = {https://doi.org/10.1063/5.0289410},
  publisher = {AIP Publishing},
}

\end{document}